\documentclass[11pt,reqno]{amsart}
\usepackage[utf8]{inputenc}
\usepackage[T1]{fontenc}
\usepackage{amsmath}
\usepackage{amsfonts}
\usepackage{amssymb}
\usepackage{graphicx}
\usepackage[export]{adjustbox}
\usepackage{mathrsfs}
\usepackage{xcolor}
\usepackage{slashed}
\usepackage{tikz}
\usetikzlibrary{shadings,intersections}
\usepackage{enumitem}
\usepackage{verbatim}
\usepackage{caption}
\usepackage{hyperref}
\usepackage{amsthm}
\usepackage[final,expansion=false]{microtype}
\usepackage[a4paper,top=2.5cm,bottom=2.5cm,left=3cm,right=3cm]{geometry}

\theoremstyle{plain}
\newtheorem{lemma}{Lemma}[section]
\newtheorem{definition}[lemma]{Definition}
\newtheorem{proposition}[lemma]{Proposition}
\newtheorem{theorem}[lemma]{Theorem}
\newtheorem{corollary}[lemma]{Corollary}
\newtheorem{remark}{Remark}
\numberwithin{equation}{section}

\newcommand{\Lie}{\mathcal{L}}
\newcommand{\dv}{{\rm div}}
\newcommand{\tr}{{\rm tr}}
\newcommand{\cl}{{\rm curl}}

\theoremstyle{plain}
\newtheorem{conj}{Conjecture}[section]

\def\Hb {\underline{H}}
\def\chib {\underline{\chi}}
\def\chih {\hat{\chi}}
\def\chibh {\hat{\underline{\chi}}}
\def\omegab {\underline{\omega}}
\def\etab {\underline{\eta}}
\def\betab {\underline{\beta}}
\def\alphab {\underline{\alpha}}

\def\i {\infty}
\def\l {\bigg(}
\def\r {\bigg)}
\def\S {S_{u,v}}
\def\SS {\mathbb{S}}

\def\LS{{\mathscr{L}}}
\def\LSS{\LS_{II}}
\def\LT{\LS_{III}}
\def\Rc{{\rm Ric}}

\def\lnm{\left\|}
\def\rnm{\right\|}
\def\z{\left(\frac{v}{(-u)^{1-\k}}\right)}
\def\R{\mathcal{R}}
\def\D{\mathcal{D}}

\def\V{{\rm Vol}_{\mathbb{S}}}
\def\k{\kappa}

\newcommand{\wt}{\widetilde}
\newcommand{\ol}{\overline}
\newcommand{\ot}{\overset{\triangle}}
\newcommand{\oc}{\overset{\circ}}
\newcommand{\df}[1]{\left[\overline{#1}\right]_0^v}
\newcommand{\ddf}[1]{\left\{ #1\right\}_0^v }
\newcommand{\ddfl}[1]{\ddf{\ol{#1}}}

\newcommand{\mub}{\underline{\mu}}

\newcommand{\Rb}{\underline{\mathcal{R}}}

\newcommand{\sg}{g\mkern-9mu /}
\newcommand{\seps}{\epsilon\mkern-8mu /}

\title[Naked Singularities beyond Spherical Symmetry]{Naked Singularities beyond Spherical Symmetry: Singular Inner Cauchy Horizons\\ for the Einstein-Scalar Field System}

\author{Xinliang An}
\address{Department of Mathematics, National University of Singapore, Singapore}
\email{matax@nus.edu.sg}

\author{Shengrong Wu}
\address{Department of Mathematics, National University of Singapore, Singapore}
\email{shengrong\_wu@u.nus.edu}

\date{}


\begin{document}

\maketitle

\begin{abstract}
In this work, we investigate the formation of naked singularities for the $3+1$-dimensional Einstein-scalar field system without symmetry assumptions. We generalize the spherically symmetric and self-similar naked-singularity solution constructed by Christodoulou in \cite{Christodoulou1994} by prescribing non-spherically symmetric initial data along both incoming and outgoing initial null hypersurfaces. We then establish global existence for the resulting solutions and analyze the singular structure of the inner Cauchy horizon. Our construction is based on employing a notion of four-type differences and designing a system of scale-invariant weighted norms to control the corresponding geometry. We show that the constructed spacetimes retain a global naked-singularity structure, characterized by an incomplete future null infinity and a singular inner Cauchy horizon. Moreover, we derive detailed asymptotics near the inner Cauchy horizon and prove
the desired $C^{1, \frac{\kappa}{1-\kappa}+}$inextendibility of these solutions, where $\kappa\in (0,1/3)$ is the self-similar parameter. This indicates a connection between weak and strong cosmic censorship: for the class of non-spherically symmetric solutions constructed here, the failure of weak cosmic censorship in its strict formulation is accompanied by a quantitative inextendibility mechanism at the inner Cauchy horizon.
\end{abstract}

\section{Introduction}

In this article, we study the Einstein-scalar field equations (ESE):
\begin{equation}\label{ESE}
    \left\{
    \begin{aligned}
        \text{Ric}_{\mu\nu}(g)-\frac{1}{2}{R}(g)g_{\mu\nu} &= D_\mu\phi D_\nu\phi-\frac{1}{2}D_\lambda\phi D^\lambda \phi g_{\mu\nu},\\
        \square_{g}\phi &= 0,
    \end{aligned}
    \right.
\end{equation}
where $(\mathcal{M},g)$ is a $3+1$-dimensional Lorentzian manifold with Ricci curvature $\text{Ric}_{\mu\nu}(g)$ and scalar curvature $R(g)$, and $\phi:\mathcal{M}\rightarrow\mathbb{R}$ is a real-valued scalar function. 

In \cite{Christodoulou1994} Christodoulou constructed the celebrated naked singularity solution for this system. This solution arises under the assumption of continuous self-similarity (CSS) and spherical symmetry. In renormalized double null coordinates, the background solution is solved under the following ansatz:
\begin{equation}\label{eq:symmetry_ansatz}
    \begin{split}
        \text{Spherical Symmetry:}&\quad 
    g = -2\Omega^2(u,v)(du\otimes dv+dv\otimes du) + r^2(u,v)d\sigma_{\mathbb{S}^2},\\
    \kappa\text{-Self-Similarity:}&\quad  S = u\partial_u+v\partial_v,\quad \mathcal{L}_S g=2g,\quad S(\phi)=\sqrt{2\kappa}.
    \end{split}
\end{equation}
Both constructing more examples of naked singularity solutions beyond these symmetries and proving the stability or instability of these naked singularities are challenging and are closely tied to the \textit{weak cosmic censorship conjecture} (WCCC), one of the central problems in mathematical general relativity.

\begin{conj}[Weak Cosmic Censorship, \cite{penrose1969gravitational,christodoulou1999global}] 
    The maximal future development of generic, regular, asymptotically flat initial data for the Einstein field equations possesses a complete future null infinity.
\end{conj}

For the system \eqref{ESE}, under spherically symmetry, in a series of foundational works \cite{christodoulou1991,christodoulou1993,christodoulou1999}  Christodoulou showed that all naked singularities are \textit{unstable}, i.e., \textit{each naked singularity initial datum is associated with two unstable directions and the generic perturbation leads to black hole formation}. Recently, the first author \cite{An2024} proved that for Christodoulou's naked-singularity solution, even a small anisotropic perturbation can produce an anisotropic spacelike apparent horizon and thereby censors the singularities behind the achronal dynamical horizon. In \cite{An2024} it also shows that, for any $k\in\mathbb{Z}^+$, Christodoulou's naked-singularity solution has at least codimensional $2k$ nonlinear instability under outgoing anisotropic characteristic perturbations. Here we also note the work \cite{LL2} by Li-Liu and their result shows that the spherically symmetric singularities considered in \cite{LL2} are unstable (as a result of trapped surface formation) subject to close-to-isotropic gravitational perturbations and the instability there is in the sense of first category.

\vspace{1mm}
These above instability results raise a complementary set of questions: \textit{Can naked singularities persist under a suitable class of non-spherically symmetric perturbations? If so, since weak cosmic censorship in its strict formulation, without the genericity requirement, appears to be violated, is there another censorship mechanism that takes effect?} In this article, we prove that \textit{naked singularities form from non-spherically symmetric initial data lying in a small neighborhood of Christodoulou's self-similar solution \cite{Christodoulou1994}}. We further show that \textit{the singular boundary (the inner Cauchy horizon) of the corresponding spacetime admits no H\"older extension}. For our initial data, we prescribe non-spherically symmetric $\kappa$-self-similar perturbations of \cite{Christodoulou1994} on the initial incoming null hypersurface and small perturbations on the initial outgoing null hypersurface. In the setting considered there,  we find that \textit{when weak cosmic censorship fails in its strict formulation, strong cosmic censorship comes into force.}

\vspace{1mm}
Our results here also connect another direction. In the vacuum setting, Rodnianski and Shlapentokh-Rothman \cite{Rodnianski2019NakedSF} proved the formation of naked singularities from their introduced $\kappa$-self similar initial data for the Einstein vacuum equations. Their result \cite{Rodnianski2019NakedSF} gives the first mathematical construction of naked-singularity solutions without symmetry assumption, and their results and methods inspired the present paper. We also mention that in the spherically symmetric Einstein-scalar field setting, Singh \cite{singh2024construction} proved the stability of Christodoulou's solution under (non-generic) small spherical perturbations. See also a further extension of \cite{singh2024construction} by Singh-Zheng \cite{SZ} and Zheng \cite{Zheng}.

\vspace{1mm}
The present paper can be viewed as a non-spherically symmetric stability result and the nonlinear singularity analysis of Christodoulou's naked-singularity geometry beyond spherical symmetry. In this paper, the scalar field enters as the leading terms in the analysis: it supplies the singular background, enters the Ricci tensor through $\text{Ric}_{\mu\nu}=D_\mu\phi D_\nu\phi$, modifies the Bianchi equations, creates a derivative-loss difficulty at the top order, and ultimately yields the quantitative obstruction to extending the solution across the inner Cauchy horizon. Thus, this paper both advances the study of naked singularities and establishes a precise nonlinear description of the singular inner Cauchy horizon beyond spherical symmetry.

Besides overcoming the new difficulties arising from the matter field, compared with \cite{Rodnianski2019NakedSF} we additionally  derive the precise asymptotic behavior of the solution as it approaches the singular inner Cauchy horizon and establish the $C^{1,\frac{\kappa}{1-\kappa}+}$-inextendibility of the spacetime. This hence links to
 
\begin{conj}[Strong Cosmic Censorship, \cite{christodoulou1999global}] 
    For generic initial data for suitable Einstein-matter system or for Einstein vacuum equations, the maximal Cauchy development arising from initial data is inextendible.
\end{conj}

\subsection{Definitions and Main Result}

To rigorously state our results, we first introduce the geometric framework of the double null foliation and the precise definition of the naked singularity structure.

\begin{definition}[Double Null Foliation]\label{def:double_null}
    We introduce a \textbf{double null foliation} on the $(3+1)$-dimensional spacetime $(\mathcal{M},g)$. Consider a spacelike $2$-sphere $S_{-1,0}$, and let $u$ and $v$ be future-directed optical functions (which satisfy the eikonal equation):
    \[
    g^{\mu\nu}\partial_\mu u\partial_\nu u = g^{\mu\nu}\partial_\mu v\partial_\nu v = 0,
    \]
    subject to the initial conditions $u|_{S_{-1,0}}=-1$ and $v|_{S_{-1,0}}=0$.
    
    We define the null geodesic vector fields $L' = -2\nabla u^\sharp$ and $\underline{L}' = -2\nabla v^\sharp$. The lapse function $\Omega$ is defined by the relation $2\Omega^{-2} = -g(L', \underline{L}')$. We define the normalized null frame as $e_3 = \Omega \underline{L}'$ and $e_4 = \Omega L'$.
    
    Let $H_u$ and $\underline{H}_v$ denote the level sets of $u$ and $v$, respectively, where $H_u$ corresponds to outgoing null cones and $\underline{H}_v$ to incoming null cones. The intersection $S_{u,v} := H_u \cap \underline{H}_v$ is a topological 2-sphere.
    
    We construct the coordinate system $(u, v, \theta^A)$ as follows: let $(\theta^A)$ denote arbitrary coordinates on $S_{-1,0}$, and propagate them throughout the spacetime by imposing the gauge condition $[e_4, \partial_{\theta^A}] = 0$. Let $b^A$ be the shift vector defined by $e_3 = \Omega^{-1}(\partial_u + b^A \partial_A)$. In these coordinates, the metric takes the form:
    \begin{equation}
        g = -2\Omega^2(u,v,\theta)(du\otimes dv + dv\otimes du) + \slashed{g}_{AB}(u,v,\theta)(d\theta^A - b^A du)\otimes(d\theta^B - b^B du).
    \end{equation}
    The associated null frame is given by $e_3 = \Omega^{-1}(\partial_u + b^A\partial_A)$, $e_4 = \Omega^{-1}\partial_v$, and $e_A = \partial_{\theta^A}$ for $A=1,2$.
\end{definition}

\begin{definition}[Incomplete Future Null Infinity and Naked Singularity]\label{def:naked_singularity}
    Suppose that, for a fixed $u_0$, the outgoing null hypersurface $H_{u_0}$ extends to an asymptotically flat region. The spacetime $(\mathcal{M},g)$ is said to possess an \textbf{incomplete future null infinity} if there exists a sequence of points $\{p_i\} \subset H_{u_0}$ with $v(p_i) \to \infty$ such that the future-directed incoming null geodesics emanating from $p_i$ and generated by $e_3$ have uniformly bounded affine length. 
    
    \noindent A spacetime is said to contain a \textbf{naked singularity} if it arises as the maximal globally hyperbolic development of regular initial data and possesses an incomplete future null infinity $\mathcal{I}^+$.
    
\end{definition}

We denote Christodoulou's spherically symmetric, $\kappa$-self-similar background solution by $(\mathcal{M}^c, g^c, \phi^c)$, which satisfies the symmetry conditions in \eqref{eq:symmetry_ansatz}. We are now ready to state our main theorem.

\begin{theorem}\label{Main_Theorem}
    Let $0 < \kappa < 1/3$ be fixed and $\delta > 0$ be a small parameter. For any $\delta>0$, there exists a sufficiently small $\epsilon = \epsilon(\delta) > 0$ such that the following holds:
    
    There exists a spacetime $(\mathcal{M}, g, \phi)$ solving the Einstein-scalar field equations \eqref{ESE} with the following properties:
    
    \begin{enumerate}
        \item \textbf{Global Structure:} The manifold $\mathcal{M}$ is diffeomorphic to $[-1,0) \times [0,\infty) \times \mathbb{S}^2$ with coordinates $(u,v,\theta)$. The metric $g$ takes the double null form described in Definition \ref{def:double_null}. In these coordinates, the metric of Christodoulou's solution \cite{Christodoulou1994} is written by $g^c = -2(\Omega^c)^2 (du\otimes dv + dv\otimes du) + \slashed{g}^c_{AB} d\theta^A \otimes d\theta^B$.
        
        \item \textbf{Initial Data:} Along the initial incoming hypersurface $v=0$ we prescribe the non-spherically symmetric $\kappa$-self similar initial data as described in Section \ref{Construction of Self-Similar Initial Data on Incoming Cone} and along the initial outgoing null hypersurface $u=-1$ we impose small perturbations as specified in Section \ref{Initial Values along the Initial Outgoing Null Cone}. In particular, along $v=0$ the initial data satisfy the Lie-scaling condition $u\mathcal{L}_{\partial_u} g = 2g$, the quantity $u\partial_u \phi + \sqrt{2\kappa}$ remains constant up to an error of order $O(\epsilon)$, and the deviation of the metric from the background obeys the bound
        \[
        \|\slashed{g} - \slashed{g}^c\|_{L^2(\slashed{g}^c)} \leq \epsilon.
        \]
        
        \item \textbf{Asymptotic Behavior at Infinity:} The metric $g$ and the scalar field $\phi$ admit a continuous extension to the region boundary $\{u=0, v>0\}$. In the regime where the ratio $v/(-u)$ is sufficiently large, the continuous extensions satisfy the following decay estimates:
        \begin{equation*}
            |\partial_v(\phi - \phi^c)| \leq C\epsilon^{1-\delta} v^{-1}, \quad |\slashed{g} - \slashed{g}^c|_{\slashed{g}^c} \leq C\epsilon^{1-\delta}.
        \end{equation*}
        
        \item \textbf{Inextendibility at the Inner Cauchy Horizon:} As $u \to 0$, that is, as the solution approaches the inner Cauchy horizon, the solution exhibits singular behaviors. Specifically,
        \[
        \lim_{u \to 0} (-u)^{1-2\kappa-\delta} |(\Omega^{-1}e_3)^2 \phi| = \infty.
        \]
        
        \noindent This blow-up behavior indicates that the scalar field fails to be $C^{1, \frac{\kappa+\delta}{1-\kappa}}$-continuous, thereby rendering the the corresponding spacetime inextendibility across the inner Cauchy horizon.
        
        \item \textbf{Causal Structure:} The spacetime admits the Penrose diagram depicted in Figure \ref{Fig:Penrose Diagram}. The boundary $\mathcal{I}^+$ represents an incomplete future null infinity, confirming that the singularity $\mathcal{O}$ is naked.
        
    \end{enumerate}
\end{theorem}

From the perspective of cosmic censorship, the above theorem has two complementary features. First, future null infinity of the solution is incomplete in the sense of Definition \ref{def:naked_singularity}; hence the singularity is visible from null infinity within the constructed class of data. Second, the inner Cauchy horizon is itself singular. Thus, the same solutions exhibit both a \textit{violation of weak cosmic censorship in its strict formulation} and a \textit{manifestation of strong cosmic censorship} at the inner Cauchy horizon.
\vspace{2mm}

The regularity threshold is dictated by the self-similar background. The precise H\"older regularity $C^{1,\frac{\kappa}{1-\kappa}}$ arises from Christodoulou's spherically symmetric solution, while the detailed analysis in this paper yields a quantitative obstruction at the level $C^{1,\frac{\kappa+\delta}{1-\kappa}}$, for every sufficiently small $\delta>0$, under non-spherically symmetric perturbations prescribed on both the initial incoming and outgoing null hypersurfaces. This regularity threshold is hence desired in the present non-spherically symmetric setting and is summarized as a $C^{1,\frac{\kappa}{1-\kappa}+}$ inextendibility at the inner Cauchy horizon for the coupled Einstein-scalar field system.

\subsection{New Ingredients and Key Strategies}

The proof relies on a global bootstrap argument coupled with a domain decomposition strategy. Based on the ratio of the null coordinates $v/|u|$—which reflects the self-similar nature of the background flow—we partition the spacetime into three characteristic regions.

\begin{figure}[h]
    \begin{center}
    \begin{tikzpicture}[scale=1.5]
        \fill[gray!10] (0,0) -- (2,2) -- (4,0) -- (2,-2) -- cycle;
        
        \draw[dashed, thick] (0,0) -- (2,2); 
        \draw[dashed] (2,2) -- (4,0);         
        \draw[thick] (4,0) -- (2,-2);        
        \draw[thick] (2,-2) -- (0,0);        
        \draw[dashed] (0,0) -- (2.6,-1.4);
        \draw[dashed] (0,0) -- (2.6,1.4);
        \node[rotate=45, above] at (1,1) {\small Cauchy Horizon ($u = 0$)};
        
        \node[rotate=-45, above] at (3,1) {\small $\mathcal{I}^+$};
        
        \node[rotate=-45, below] at (1,-1) {\small Initial Incoming ($v = 0$)};
        
        \node[rotate=45, below] at (3,-1) {\small Initial Outgoing ($u = -1$)};
        \node at (2,1.4) {$\mathcal{R}_{III}$};
        \node at (2,-1.4) {$\mathcal{R}_{I}$};
        \node at (3,0) {$\mathcal{R}_{II}$};

        \filldraw[white] (0,0) circle (0.10);
        \draw[thick] (0,0) circle (0.10);
        \node[left] at (-0.1,0) {\small $\mathcal{O}$};
    \end{tikzpicture}
    \end{center}
    \caption{\small{Penrose diagram and domain decomposition for the proof, showing the naked singularity $\mathcal{O}$, the incomplete future null infinity $\mathcal{I}^+$, and the singular inner Cauchy horizon. The spacetime is foliated by the ratio $v/|u|$. \textbf{Region I} is the initial layer where the solution is close to initial data. \textbf{Region II} is the wave zone dominated by self-similar dynamics. \textbf{Region III} is the asymptotic region where the singularity and the inner Cauchy horizon locate.}}
    \label{Fig:Penrose Diagram}
\end{figure}
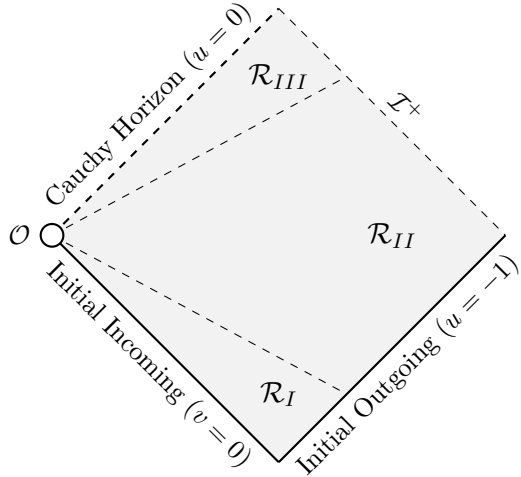
Precisely, we let $z=\frac{v}{-u}$. With $0<\epsilon_1\ll 1$ the regions are defined via:
\[\R_I:=\{0\leq z\leq \epsilon_1\},\,\R_{II}:=\{\epsilon_1\leq z\leq \epsilon_1^{-1}\},\,\R_{III}:=\{\epsilon_1^{-1}\leq z\}.\]
The overarching logic of the proof involves constructing an \textit{approximating solution} via Lie transport along the homothetic vector field $S$. We then estimate various differences, including the difference between the actual solution and the approximations, using a system of \textit{scale-invariant weighted norms}. This allows us to propagate regularity from the initial data surfaces (Region I) through the wave zone (Region II) and finally establish the singular asymptotics at the horizon (Region III). 

\subsubsection{Lie Propagation Equation and the Approximating Solution}

A primary challenge in Region I lies in constructing an approximating solution that is sufficiently close to the exact solution, thereby ensuring that the error terms remain controllable with respect to both the perturbation size $\epsilon$ and the distance parameter $z$. The underlying self-similar symmetry plays an indispensable role in this construction. 

\vspace{1mm}
Consider a tensor field $\psi(u,v,\theta)$ satisfying the scaling law $\psi \sim (-u)^a \check{\psi}(v/(-u), \theta)$. The geometric transport of such a tensor along the incoming null direction $e_3$ is intimately related to its variation along the similarity flow. Specifically, the characteristic relation $S = u\partial_u + v\partial_v$ yields the following identity:
\[
\Omega \Lie_{e_3}\psi = -\frac{a}{-u}\psi + \Lie_b\psi + \frac{v}{-u}\Lie_{\partial_v}\psi.
\]
Motivated by this, we define the \textit{Lie propagation equation} for a general tensor $f$ as:
\begin{equation}\label{Lie_propagation_eq}
    \frac{v}{-u}\Lie_{\partial_v} f + \Lie_b f - \frac{\lambda}{-u} f = F,
\end{equation}
where $\lambda$ is a weight constant related to the scaling dimension.

\vspace{1mm}
For geometric quantities governed by a transport equation along $e_3$—specifically the renormalized shear $\Omega^{-1}\chih$ and the transversal scalar derivative $\Omega^{-1}e_4\phi$—we can derive their corresponding Lie propagation equations. A crucial simplification arises on the initial incoming cone $S_{u,0}$ (where $v=0$). Provided the regularity condition $\lim_{v\to 0} v\Lie_{\partial_v}\psi = 0$ holds, the equation reduces to a purely spatial system on the 2-sphere, which we refer to as the \textit{degenerate Lie propagation equation}.

For instance, by substituting the Ricci evolution equation for $\nabla_3 \chih$ into the self-similarity relation, we derive the following defining equation for the approximating shear restricted to the initial cone:
\begin{equation}\label{Intro_Lie_prop_eq_chih}
    \begin{aligned}
        -\Lie_b(\Omega^{-1}\chih_{AB}) &+ \left(\frac{1}{2}\dv b - 2\Lie_b\log\Omega + \frac{\k}{-u}\right)\Omega^{-1}\chih_{AB} \\
        &= -(\nabla\hat{\otimes}\eta)_{AB} - (\eta\hat{\otimes}\eta)_{AB} + \frac{1}{2}(\nabla\hat{\otimes} b)_{AB}(\Omega^{-1}\tr\chi) - \frac{1}{2}(\nabla\phi\hat{\otimes}\nabla\phi)_{AB} \\
        &\quad - \left[ \frac{1}{2}(\nabla\hat{\otimes} b)^C_{\ A}(\Omega^{-1}\chih)_{BC} + \frac{1}{2}(\nabla\hat{\otimes} b)^C_{\ B}(\Omega^{-1}\chih)_{AC} \right].
    \end{aligned}
\end{equation}
Note that the $v\partial_v$ terms here vanish at $v=0$. This system admits a unique solution provided the divergence of the shift, $\dv b$, is sufficiently small relative to the spectral parameter $\lambda$. The existence theory and regularity estimates for such equations were established in Section 4 of \cite{Rodnianski2019NakedSF}, allowing us to prescribe the initial values along $v=0$.

\vspace{1mm}

Once the boundary values on $v=0$ are suitably determined, we extend the approximating solution into the interior domain ($v>0$). For any fixed $V>0$ and $\lambda > 0$, the full Lie propagation equation \eqref{Lie_propagation_eq} admits a unique solution on the interval $v \in (0, V)$ corresponding to the prescribed boundary data at $v=0$.

\vspace{1mm}
A high degree of accuracy is required for this approximation because the curvature tensor $\alpha$ depends on $\nabla_4 \chih$. Following the strategy in \cite{Rodnianski2019NakedSF}, we define the approximating shear $\ot\chih$ by solving \eqref{Intro_Lie_prop_eq_chih} where all coefficients (like $b, \eta, \Omega$) are evaluated at their limiting values on $v=0$. The approximating tensor is then extended via:
\begin{equation*}
    \ot\chih_{AB}(u,v) = \frac{1}{2}\left( \ot\chih_{AC}(u,0)g^{CD}(u,0)g_{DB}(u,v) + \ot\chih_{BC}(u,0)g^{CD}(u,0)g_{DA}(u,v) \right).
\end{equation*}
We apply an analogous construction procedure to define $\ot{e_4\phi}$. However, $e_4\phi$ does not vanish as $\chih$ in the spherical symmetric case. To capture the increasing rate, we use the following equation to define it:
\begin{equation}
    \begin{aligned}
        &\left((\Omega\nabla_3)(0)+\frac{1}{2}(\Omega\tr\chib)(0)
+\frac{1}{2}\left((\Omega\tr\chib)^c(v)-(\Omega\tr\chib)^c(0)\right)
\right)\ot{\Omega^{-1}e_4\phi}\\
&\qquad-\left(4(\Omega\omegab)(0)+4\left((\Omega\omegab)^c(v)-(\Omega\omegab)^c(0)\right)\right)\ot{\Omega^{-1}e_4\phi}\\
=&
\left(\dv\nabla\phi-\frac{1}{2}(\Omega^{-1}\tr\chi)(\Omega e_3\phi)
+2\eta\cdot\nabla\phi\right)(0)\\
&\qquad-\frac{1}{2}\left((\Omega^{-1}\tr\chi\,\Omega e_3\phi)^c(v)-(\Omega^{-1}\tr\chi\,\Omega e_3\phi)^c(0)\right).
    \end{aligned}
\end{equation}
Here $(\cdot)^c$ stands for corresponding quantity in Christodoulou's solution, and $\psi(v)$ means $\psi$ takes value at $v$.
The success of this approximation scheme fundamentally hinges on the smallness of the shift vector, namely $|b| \ll \kappa$. In the context of the scalar field system, this condition ensures that the non-spherical perturbations (driven by the shift $b$) do not overwhelm the damping effect provided by the self-similarity parameter $\kappa$. Besides, we define the approximating curvature tensor $\ot{\Omega^{-2}\alpha}$ via
\[\ot{\Omega^{-2}\alpha}_{AB}=-(\Omega_0)^{-2}\Lie_{\partial_v}\ot{\Omega^{-1}\chih}_{AB},\]
which captures the singular part of $\nabla_4\chih$.

\subsubsection{Gauge Choice and the Scale-Invariant Weighted Norm System}

In the standard double null gauge $(u,v)$, the background solution exhibits singular scaling behavior near the origin. Specifically, expressed in terms of the self-similar variable $z = v/(-u)$, the background conformal factor behaves as $(\Omega^c)^2 \sim z^\kappa$ as $z \to 0$, while the connection coefficient scales according to $(\Omega^{-1}\omega)^c \sim v^{-(1-\kappa)}(-u)^{-\kappa}$.

\vspace{1mm}
To impose a consistent functional framework where the Ricci coefficients depend primarily on the time coordinate, we introduce a \textit{renormalized coordinate system}. Near the initial cone ($z \approx 0$, Region $\mathcal{R}_I$), we employ the transformation $(\hat{u}, \hat{v}) = (u, v^{1-\kappa})$. Symmetrically, away from $z=0$ (Regions $\mathcal{R}_{II}$ and $\mathcal{R}_{III}$), we utilize $(U, V) = (-(-u)^{1-\kappa}, v)$. For notational simplicity, within the analysis of Region $\mathcal{R}_I$, we will continue to use $(u, v)$ to refer to the renormalized gauge.

\vspace{1mm}
To close the energy estimates and cancel borderline terms arising from the background geometry, we introduce a system of \textit{weighted norms}. Since the background quantities scale as powers of $-u$ or $v$ along characteristics, we assign a specific weight---termed the \textit{signature}---to each geometric quantity. The signature $s(\psi)$ represents the \textit{expected} decay rate of the $L^2$ norm of a tensor $\psi$.

\vspace{1mm}
In Region $\mathcal{R}_I = \{0 < z < \epsilon_1, -1 < u < 0\}$, we define a modified $L^2$ norm over the 2-spheres $S_{u,v}$ by approximating the geometric volume form $d\text{Vol}_g$ with $(-u)^2 d\mathcal{V}_{\mathbb{S}^2}$. We incorporate an additional weight function $w(u,v)$ to adequately handle the fine structure of the error terms. For example, since the background Gaussian curvature evaluates to $\|K^c\|_{L^2(S_{u,v})} \sim (-u)^{-1}$, we assign the geometric signature $s(K) = -1$. The corresponding spatial norms take the form:
\begin{equation*}
    \lnm K\rnm_{\LS_w^2(S_{u,v})} := \left(\int_{\mathbb{S}^2} (-u)^{2-2s(K)} |K|^2 w(u,v) d\mathcal{V} \right)^{1/2},
\end{equation*}
and the pointwise norm is given by $\lnm \cdot\rnm_{\LS_w^\infty(S_{u,v})} := (-u)^{1-s(\cdot)}w^{1/2}\lnm \cdot\rnm_{L^\infty(\mathbb{S}^2)}$.

We control the evolution of the system using the following integrated energy norms along null hypersurfaces:
\[
\lnm\cdot \rnm_{\LS_w^2(H_{{u}}^{{v}})} := \left(\int_{0}^{{v}}\frac{1}{(-{u})^{1-\kappa}}\lnm\Omega\ \cdot \rnm^2_{\LS_w^2(S_{{u},v^\prime})}dv^\prime\right)^{1/2},
\]
\[
\lnm\cdot \rnm_{\LS_w^2(\Hb_{{v}}^{{u}})} := \left(\int_{-1}^{{u}}\frac{1}{-{u^\prime}}\lnm \cdot \rnm^2_{\LS_w^2(S_{{u^\prime},{v}})}du^\prime\right)^{1/2}.
\]
In Section \ref{Norms and Prework}, we will adapt the standard methodologies of the characteristic initial value problem---including Sobolev inequalities, transport estimates, and elliptic estimates for Hodge systems---to this weighted framework.

Finally, in Regions $\mathcal{R}_{II}$ and $\mathcal{R}_{III}$, we transition to the coordinates $(U,V)$. The weight factors $(-u)^{2-2s}$ are replaced by $V^{2-2S}$, where $S(\cdot)$ denotes the $V$-signature. This reflects the change in the dominant scaling variable from $-u$ (near $v=0$) to $v$ (near $u=0$).

\subsubsection{Hierarchy of Differences and Stability Estimates in Regions I and II}

To obtain optimal control over the perturbations, we employ a hierarchy of \textit{four distinct difference quantities}. The choice of difference variable depends on the geometric interaction and the region of spacetime.

\begin{enumerate}
    \item \textbf{Approximation Difference} ($\wt{\psi}$): 
    Let $\ot\psi$ be the approximating solution constructed via the Lie propagation equations. We define $\wt{\psi} = \psi - \ot{\psi}$. This quantity is the primary variable for the shear $\chih$, the scalar derivative $e_4\phi$, and the curvature components $\alpha, \beta$, capturing the deviation from the imagined self-similar evolution.
    
    \item \textbf{Background Difference} ($\ol{\psi}$): 
    We define $\ol{\psi} = \psi - \psi^c$, measuring the deviation from the actual Christodoulou background solution. This is effective in Regions $\R_{II}$ and $\R_{III}$. However, it is insufficient in Region $\R_I$ because it fails to capture the precise $z$-dependence with $z = v/(-u)$ near the initial cone $z=0$.
    
    \item \textbf{Initial Value Difference} ($\df{\psi}$): 
    To resolve the $z$-dependence issue in $\R_I$, we define $\df{\psi} := \ol{\psi} - \ol{\psi}|_{v=0}$. This subtracts the boundary value, simplifying the transport equations. We use this for quantities decoupled from the shear $\chih$.
    
    \item \textbf{Second-Order Difference} ($\ddfl{\psi}$): 
    For the expansion scalar $\tr\chib$, the estimates require a higher degree of cancellation near $z=0$ which $\df{\psi}$ cannot provide. We therefore define the second-order subtraction:
    \[
    \ddfl{\psi} := \ol{\psi} - \ol{\psi}|_{v=0} - v \cdot (\Lie_v \ol{\psi})|_{v=0}.
    \]
    This removes both the value and the linear growth rate at the boundary, allowing us to close the estimates with higher powers of $z$.
\end{enumerate}
\begin{figure}[h]
    \centering
    \includegraphics[width=0.8\textwidth]{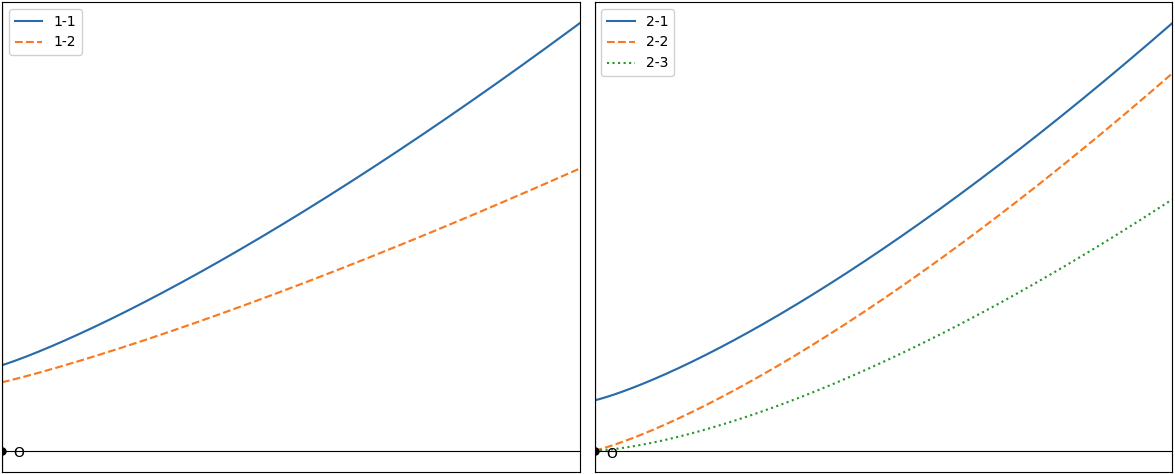}
    \caption{This is a visualized comparison of $\ol{\psi},\df{\psi},\ddfl{\psi}$. The curve 1-1 stands for $\psi$, and 1-2 for $\psi^c$. Curve 2-1 is defined by the difference $\ol{\psi}=\psi-\psi^c$, and 2-2, 2-3 represent the residues of zeroth and first order Taylor expansion of $\ol{\psi}$ respectively.}
    \label{fig:my_image}
\end{figure}

\paragraph{\textbf{Comparison with the Vacuum Case.}}
A critical distinction between this work and the vacuum result in \cite{Rodnianski2019NakedSF} lies in the regularity requirements. In the vacuum case, one can rely on the simpler difference $\psi - \psi|_{v=0}$. However, for the Einstein-scalar field system \eqref{ESE}, the curvature estimates at order $n$ depend on the $n+1$-th derivatives of the scalar field (via terms like $D_\mu \phi D_\nu \phi$). Consequently, we must control the $n+1$-th order energy estimates for the Ricci coefficients. 
More specifically, we use $\psi$ to represent the first order derivatives related to $g,\Omega,\phi$ and $\Psi$ the curvature components. The equations of Einstein vacuum systems are simply written as
$$\nabla_3\Psi_1+\mathcal{D}\Psi_2\sim\psi\Psi,\ \nabla_4\Psi_2-{}^{*}\mathcal{D}\Psi_1\sim\psi\Psi,$$
where $\mathcal{D}$ is a differencial operator with Hodge dual ${}^*\mathcal{D}$. However, within the Einstein scalar-field system, the right hand sides of these equations will be:
$$\nabla_3\Psi_1+\mathcal{D}\Psi_2\sim\psi\Psi+\psi\nabla\psi,\ \nabla_4\Psi_2-{}^{*}\mathcal{D}\Psi_1\sim\psi\Psi+\psi\nabla\psi.$$
This \textit{loss-of-derivative structure} necessitates the finer control provided by our hierarchy of differences $\df{\psi}$ and $\ddfl{\psi}$.

\vspace{2mm}
\paragraph{\textbf{Bootstrap Strategy in Region $\R_I$.}}
We improve the bootstrap assumptions by carefully designing the weights to exploit the decay of the background. In $\R_I$, we employ the weight $w(u,v) = ((-u)^{1-\k}/v)^{1-2\delta}$. Our goal is to prove estimates of the form:
\begin{equation*}
    \lnm \varphi_1\rnm^2_{\LS_w^2(H)}+\lnm \varphi_2\rnm^2_{\LS_w^2(\Hb)}\leq \epsilon^2\z^{2\tau},
\end{equation*}
where $\tau > 0$ determines the decay rate near the cone.

Consider a simplified model equation: $\nabla_{\partial_u}\varphi + \lambda\Omega\tr\chib\varphi = F$, where we assume $F$ has good estimate. In the background, $\Omega\tr\chib \approx -\frac{2}{-u}$. The weighted energy identity along the incoming null direction takes the form:
\begin{equation*}
    \lnm \varphi\rnm^2_{\LS_w^2(S_{u,v})} = \lnm \varphi\rnm^2_{\LS_w^2(S_{-1,v})} + \int_{-1}^u \mathcal{C}(u',v) \lnm \varphi\rnm^2_{\LS_w^2(S_{u',v})} du'+\mathcal{E}(F),
\end{equation*}
where the coefficient function satisfies $(-u)\mathcal{C} \approx -2 + 2s(\varphi) + (-u)\partial_u\log w + 4\lambda$, and $\mathcal{E}(F)$ denotes the terms arising from $F$.

If $\mathcal{C}$ is negative, the norm is bounded by the initial data. However, for bad-behavior quantities like $\Omega\tr\chib$ itself, the corresponding coupling constant $\lambda$ is large, making $\mathcal{C}$ positive. 
Provided the forcing term $F$ increases sufficiently slow, i.e., $\mathcal{E}(F) \lesssim \epsilon^2 z^{2\tau}$ and $2\tau(1-\k) > (-u)\mathcal{C}$, the integral term can be absorbed, yielding the desired bound $$\lnm \varphi\rnm^2_{\LS_w^2(S_{u,v})} \leq C \epsilon^2 \z^{2\tau}.$$ This leads to an extra requirement that 
$\lnm \varphi\rnm^2_{\LS_w^2(S_{-1,v})}\lesssim \epsilon^2 v^{2\tau}$,
so instead of using standard quantities, such as $\chi,\chib,$ we use differences defined above, like $\df{\chi},\ddfl{\chib}$, to provide the $\epsilon,v$ smallness. 

With $\LS_I$ standing for the norm in $\R_I$, where the weight function is $w_I=\z^{2\delta-1}$, we briefly write down the choice of $\tau$ for each quantity. There are three types of energy estimates in this region. Firstly, we have for the most singular one $\wt{\Omega^{-2}\alpha}=\Omega^{-2}\alpha-\ot{\Omega^{-2}\alpha}$ that 
\[\sum_{i\leq 5}\lnm \nabla^i\wt{\Omega^{-2}\alpha}\rnm^2_{\LS_I^2(H_u^v)}+\lnm \nabla^{i+1}\wt{\Omega^{-1}\chih}\rnm^2_{\LS_I^2(\Hb_v^u)}\lesssim\epsilon^2\z^{2\delta}.\]
For the Bianchi pair without $\alpha$, we have the second type
\[\sum_{i\leq 5}\lnm \nabla^i\df{\Psi_1}\rnm^2_{\LS_I^2(H_u^v)}+\lnm \nabla^{i+1}\df{\Psi_2}\rnm^2_{\LS_I^2(\Hb_v^u)}\lesssim\epsilon^2\z^{2\k/(1-\k)}.\]
For the scalar field and $\tr\chib$, we need more accurate estimate:
\begin{equation*}
    \begin{aligned}
        &\sum_{i\leq 6}\left\|\nabla^{i} \wt{\Omega^{-1}D_4\phi},\nabla^i\df{\nabla\phi},\nabla^i\ddfl{\Omega\tr\chib}\right\|_{\LS_I^{2}(H_u^v)}^{2}\\
&\quad\quad\quad+\left\|\nabla^{i} \ddfl{\Omega D_3\phi},\nabla^i\ddfl{\nabla\phi} \right\|_{\LS_I^{2}(\Hb_v^u)}^{2} \lesssim\epsilon^2\left(\frac{v}{(-u)^{1-\k}}\right)^{1+2\k/(1-\k)-2\delta}.
    \end{aligned}
\end{equation*}

\vspace{2mm}
\paragraph{\textbf{Bootstrap Strategy in Region $\R_{II}$.}}
In Region $\R_{II}$, where the self-similarity variable $z$ is no longer small, we switch to an exponential weight inspired by \cite{Rodnianski2019NakedSF}:
\[
W(U,V) = V^a \exp\left(-\frac{V^{1-\k}}{-U}D\right).
\]
The large parameter $D=D(\epsilon_1)$ provides a strong damping mechanism. The bootstrap estimates allow us to show that the error terms are bounded by $o(1) \epsilon^{2-\delta} W(U,V)$ as $D \to \infty$. This large parameter effectively suppresses the non-linear interactions, allowing us to extend the existence result towards the singularity. More precisely, we take 
$\partial_U \psi=\frac{1}{-U}\psi$
as a toy model. 
Supposing 
\[\lnm \psi\rnm^2_{\LS_W^2(S_{U,V})}\leq C_{boot} V^a\exp\left(\frac{V^{1-\k}}{-U}D\right),\]
then one can actually prove that 
\[\lnm \psi\rnm^2_{\LS_W^2(\Hb_V^U)}\leq o_D(1)C_{boot} V^a\exp\left(\frac{V^{1-\k}}{-U}D\right)\]
with $o_D(1)$ tends to $0$ as $D$ becoming large. The transport estimate further yields that 
\[\lnm \psi\rnm^2_{\LS_W^2(S_{U,V})}\lesssim \lnm \psi\rnm^2_{\LS_W^2(S_{-1,V})}+\lnm \psi\rnm^2_{\LS_W^2(\Hb_V^U)}\leq o_D(1)C_{boot} V^a\exp\left(\frac{V^{1-\k}}{-U}D\right).\]
For demonstration, here we assume the initial value of $\psi$ satisfies good estimate. By letting $D$ be sufficiently large, we can then improve the bootstrap assumption.

\subsubsection{Asymptotic Analysis in Region III: Renormalization and Bulk Integration}

Region $\R_{III}$ presents the most significant analytic challenge. As the solution approaches the inner Cauchy horizon ($z \to \infty$), the curvature component $\alpha$ is expected to blow up. Unlike in Region I, we cannot subtract a smooth approximating solution $\ot{\alpha}$ because the background itself becomes singular. Consequently, we cannot estimate the absolute difference directly.

To resolve this, we employ a strategy of renormalization. We assign specific weights $\Omega^{t(\psi)}$ to each Ricci coefficient $\psi$ and curvature component $\Psi$ (excluding the transversal components $\omegab$ and $\alphab$) to neutralize their expected blow-up rates. In Section \ref{Estimates in Region III}, we prove that the renormalized difference decays as:
\begin{equation*}
\lnm\Omega^{t(\psi)}\ol{\psi},\; \Omega^{t(D\phi)}\ol{D\phi},\; \Omega^{t(\Psi)}\ol{\Psi}\rnm_{\LS^\infty(S_{U,V})}\lesssim \epsilon^{1-\delta}.
\end{equation*}
For the geometric quantities $\eta, b,$ and $\Omega^2$, we achieve even stronger asymptotic controls. We demonstrate that for any $\lambda > 0$, provided $\epsilon$ is sufficiently small, we have:
\[
\lnm \eta,b,\ol{\Omega^2}\rnm_{\LS^\infty(S_{U,V})}\lesssim \epsilon^{1-\delta}\left(\frac{V}{(-U)^{1/(1-\k)}}\right)^\lambda.
\]
Also note that the weight function in this region is set to be
$$W_{III}(U,V) = V^{a-\tau(1-\k)}(-U)^\tau,$$ 
where $0 < a \ll \tau \ll 1$.

\vspace{2mm}
\paragraph{\textbf{From $\R_{II}$ to $\R_{III}$.}} The transition from $\R_{II}$ to $\R_{III}$ requires careful handling of the weight parameter $a$. We prove the estimates in $\R_{II}$ for an arbitrary positive weight power. This flexibility allows us to verify the initial conditions for the bulk norms $\LS^2(\R)$ and $\LS^2(\Rb)$ along the future boundary of $\R_{II}$ by utilizing the $\R_{II}$ estimates with parameter $a/2$.

\vspace{2mm}
\paragraph{\textbf{The Integrated Bulk Norm Method.}}
A \textit{new} feature of our analysis in $\R_{III}$ is introducing the \textit{spacetime integrated norms} to close the energy estimates. Standard energy estimates along hypersurfaces $H$ and $\Hb$ are insufficient in $\R_{III}$ because the damping terms in the transport equations are too weak to control the error growth directly.

To overcome this difficulty, we introduce norms integrated over the bulk regions $\R_{U,V}$ (outgoing bulk) and $\Rb_{U,V}$ (incoming bulk):
\[
\lnm \varphi\rnm_{\LS^2(\R_{U,V})} \sim \left(\int_{V} (V^\prime)^{-1} \lnm \varphi\rnm^2_{\LS^2(\Hb_{V^\prime}^U)} dV^\prime \right)^{\frac{1}{2}},\] 
\[
\lnm \varphi\rnm_{\LS^2(\Rb_{U,V})} \sim \left(\int_{U} V^{-1+\k} \lnm \varphi\rnm^2_{\LS^2(H_{U^\prime}^V)} dU^\prime \right)^{\frac{1}{2}}.
\]
By integrating over the bulk, we exploit the structure of the $\nabla_3$ equations to gain an effective decay factor of $(-U)$, which allows us to break the circularity in the estimates. We apply this method to the Bianchi pairs $(\Psi_1, \Psi_2)$. Instead of estimating them individually on the cones, we control the coupled quantity:
\[
\lnm \Psi_1 \cdot z^\lambda\rnm_{\LS^2(\R)} + \lnm \Psi_2 \cdot z^\lambda\rnm_{\LS^2(\Rb)},
\]
where $z = V/(-U)^{1/(1-\k)}$ and $\lambda \in [0, (1-\k)\tau/2)$. This novel \textit{bulk-first} strategy is crucial for us to establish the global existence up to the horizon.

\vspace{2mm}
\paragraph{\textbf{Bootstrap Strategy in $\R_{III}$.}} 
We first illustrate our target estimates, which can be roughly divided into four types in this part. 
Firstly, for those connection components except $\Omega^{-1}\chibh$, and derivatives of scalar field with $\nabla_3$ equation, for example $\Omega^{-1}\tr\chib,\etab,\Omega e_4\phi$, we have estimates up to $4$th order derivative
\[\sum_{i\leq 4}\lnm \nabla^i\ol{\psi}\rnm^2_{\LS_1^2(S_{U,V})}\lesssim \epsilon^{2-2\delta}.\]
Moreover, for $\Omega^{-1}e_3\phi,\Omega^{-1}\chibh$, the estimates can only be proved up to $3$rd order:
\[\sum_{i\leq 3}\lnm \nabla^i\ol{\psi}\rnm^2_{\LS_1^2(S_{U,V})}\lesssim \epsilon^{2-2\delta}.\]
Secondly, the standard spherical norms estimates can be derived as
\[\sum_{i\leq 5}\lnm \nabla^i\ol{\psi}\rnm^2_{\LT^2(S_{U,V})}\lesssim \epsilon^{2-2\delta}V^a\]
for connection components and scalar-field derivatives.
Here $\LT$ corresponds to the weight function $W_{III}=V^a\left(\frac{-U}{V^{1-\k}}\right)^\tau$. If we rewrite the second type in $\LS_1$ norm, the right hand side should be $\epsilon^{2-2\delta}\left(\frac{V^{1-\k}}{-U}\right)^\tau$, which is worse than the first type.
The third type is about $b,\Omega,\eta$. We will prove that for any fixed $\lambda>0$, it holds that 
\[\sum_{i\leq 3}\lnm \nabla^i\left(b,\eta,\ol{\Omega}\right)\rnm^2_{\LS_1^2(S_{U,V})}\lesssim \epsilon^{2-2\delta} \left(\frac{V^{1-\k}}{-U}\right)^\lambda,\]
which is weaker than the first type but stronger than the second type. Lastly we will prove the energy estimates for Bianchi pairs, such as $(\nabla(\Omega e_4\phi),\nabla^2\phi)$, $(\Omega^2\alpha,\Omega\beta)$,
\[\sum_{i\leq 5}\lnm \nabla^i\ol{\Psi_1}\rnm^2_{\LT^2(H_U^V)}+\lnm \nabla^i\ol{\Psi_2}\rnm^2_{\LT^2(\Hb_V^U)}\lesssim\epsilon^{2-2\delta}V^a.\]
The key strategy for the energy estimate is to do another bulk energy estimate with the $\LS^2(\R_{U,V})$ and $\LS^2(\Rb_{U,V})$ norms defined above. Specifically, we will show that 
\[\sum_{i\leq 5}\lnm \nabla^i\ol{\Psi_1}\cdot\left(\frac{V^{1-\k}}{-U}\right)^\lambda\rnm^2_{\LT^2(H_U^V)}+\lnm \nabla^i\ol{\Psi_2}\cdot\left(\frac{V^{1-\k}}{-U}\right)^\lambda\rnm^2_{\LT^2(\Hb_V^U)}\ll \epsilon^{2-2\delta}V^a,\]
which provides the smallness to close the energy bootstrap estimates.

\subsection{The Naked Singularity and Singular Inner Cauchy Horizon}
With the global existence and a priori estimates, we can control the difference between our solution and Chrostodoulou's background solution. However, the difference $\Omega^2-(\Omega^2)^c$ does not guarantee the uniform boundness for the affine length of the related null geodesics. We will further estimate $\Omega^2(U,V)\Omega^{-2}(-1,V)-\left(\Omega^2(U,V)\Omega^{-2}(-1,V)\right)^c$. From the estimates of $\ol{\Omega\omega}$, the lapse $\Omega$ satisfies
\[\left|\ol{\log\left(\Omega^2(U,V)\Omega^{-2}(-1,V)\right)}\right|\leq\log\left((-U)^{-\epsilon^{1/4}/(1-\k)}\right)+\epsilon^{1/4}.\]
This indicates that
\[\frac{1}{2}(-U)^{-\frac{\epsilon^{1/4}}{1-\k}}\leq \Omega^{2}(U,V)\Omega^{-2}(-1,V) \leq 2(-U)^{-\frac{\epsilon^{1/4}}{1-\k}},\]
and thus the affine length is uniformly bounded, which means the singularity is naked. 

\vspace{1mm}
Moreover, with the control of difference $(\Omega^{-1}e_3)^2\phi-\left((\Omega^{-1}e_3)^2\phi\right)^c$, we further prove that 
\begin{equation}
    (-U)^{\frac{2-4\k+2\delta}{1-\k}}\left|\partial_U^2\phi\right|^2(U,V)\lesssim_{\epsilon, V} 1,\ \ \lim_{-U\rightarrow 0}(-U)^{\frac{2-4\k-2\delta}{1-\k}}\left|\partial_U^2\phi\right|^2(U,V) =\infty,
\end{equation}
which yields the $C^{1,\frac{\k+\delta}{1-\k}}$-inextendibility across the inner Cauchy Horizon. We remark that the above blow-up estimates rely crucially on the growth of $\partial_U^2\phi^c$ and on the precise estimates we obtain around the singular background solution. In the Einstein vacuum setting, the available information on the singular background solution is more limited, and establishing the corresponding discontinuity and inextendibility statements across the inner Cauchy horizon remains a delicate question.

\subsection{Outline of the Paper}
The rest of the paper is organized as follows.

Section \ref{Preliminaries} fixes the double null formalism for the Einstein-scalar field equations. We record the null structure equations, the renormalized Bianchi equations, the relevant properties of Christodoulou's self-similar background in the double null gauge, and the characteristic local existence theorem used to start the construction.

Section \ref{Initial Values and Approximating Solution} constructs the characteristic initial data. We solve the degenerate Lie propagation equations on the initial incoming cone, construct the approximating shear and scalar-field derivative along the outgoing cone, and prescribe data that remain close to the Christodoulou background while allowing non-spherical perturbations.

Section \ref{Norms and Prework} develops the weighted functional framework. We introduce the regional gauges, the scale-invariant signatures, the weighted sphere, null, and bulk norms, and the Sobolev, transport, and elliptic estimates used throughout the proof.

Section \ref{Estimates in Region I} carries out the Region I analysis near the initial cone. This section proves the weighted estimates for the four-type difference hierarchy, including the second-order subtraction for $\tr\chi$, and culminates in the interface estimates on the fixed self-similar boundary.

Section \ref{Estimates in Region II} treats Region II. Using the Region I interface data, we prove high-order stability estimates around the Christodoulou background with the $V$-signature norm system and the arbitrary factor $V^a$. The section also contains the top-order scalar-field estimates, the auxiliary treatment of $\Omega\omega$ and $\Omega^{-1}\omega$, and the renormalized mass-aspect estimates needed to recover the Ricci coefficients.

Section \ref{Estimates in Region III} proves the Region III estimates near the inner Cauchy horizon. The analysis uses weighted background-renormalized differences and spacetime-integrated bulk norms to close the Bianchi and scalar-field estimates in the singular asymptotic region.

Section \ref{Further Result near Cauchy Inner Horizon} identifies the limiting causal and regularity structure. We prove the affine-length estimate giving the incompleteness of future null infinity and derive the blow-up of the incoming transverse second derivative of the scalar field. This yields the singular inner Cauchy horizon and the $C^{1,\frac{\k+\delta}{1-\k}}$ inextendibility statement.

\subsection{Acknowledgments}
XA is supported by MOE Tier 1 grant A-8002933-00-00 and MOE Tier 2 grant A-8000977-00-00. SW is supported by the NUS President Graduate Fellowship.

\section{Preliminaries}\label{Preliminaries}

In this section, we review the geometric setup of the double null foliation alongside the relevant field equations. Additionally, we summarize the fundamental properties of the background solution and formulate the characteristic initial value problem.
\subsection{Equations in the Double Null Foliation}

We work within the double null framework defined in Definition \ref{def:double_null}. Let $D$ denote the Levi-Civita connection on $\mathcal{M}$ and $\nabla$ the induced connection on the surfaces $S_{u,v}$. The projections of the null covariant derivatives $D_{3}$ and $D_{4}$ onto $S_{u,v}$ are denoted by $\nabla_3$ and $\nabla_4$. Let $\slashed{\epsilon}$ be the volume form associated with the induced metric $\slashed{g}$.

We define the Ricci coefficients on $S_{u,v}$ for indices $A,B\in\{1,2\}$ as follows:
\[
\chi_{A B}=g\left(D_A e_4, e_B\right), \quad \underline{\chi}_{A B}=g\left(D_A e_3, e_B\right), \quad 
\eta_A=-\frac{1}{2} g\left(D_3 e_A, e_4\right), \quad \underline{\eta}_A=-\frac{1}{2} g\left(D_4 e_A, e_3\right),
\]
\[
{\omega}=-\frac{1}{4} g\left(D_4 e_3, e_4\right), \quad \underline{\omega}=-\frac{1}{4} g\left(D_3 e_4, e_3\right),\quad 
\zeta_A=\frac{1}{2}g\left(D_A e_4,e_3\right).
\]
Direct computation yields the derivatives of the conformal factor:
\[
\nabla_4\log\Omega^2=-4\omega,\, \nabla_3\log\Omega^2=-4\omegab,\, \nabla\log\Omega^2=\eta+\etab.
\]

We decompose the Weyl curvature $W$ into its null components:
\[
\alpha_{AB}=W_{4A4B},\, \alphab_{AB}=W_{3A3B},\, \beta_A=\frac{1}{2}W_{A434},\, \betab_A=\frac{1}{2}W_{A334},\, \rho=\frac{1}{4}W_{3434},\, \sigma=\frac{1}{4}{}^*W_{3434}.
\]
Here the Weyl tensor is defined as
\begin{equation}
    W_{i k l m}:=R_{i k l m}+\frac{1}{2}\left(R_{i m} g_{k l}-R_{i l} g_{k m}+R_{k l} g_{i m}-R_{k m} g_{i l}\right)+\frac{1}{6} R\left(g_{i l} g_{k m}-g_{i m} g_{k l}\right),
\end{equation}
and ${}^*W_{ijkl}=\frac{1}{2}\slashed{\epsilon}_{ij}^{\ \ mn}W_{mnkl}$ denotes the Hodge dual. 

\vspace{2mm}
To avoid the second derivatives of the scalar field in the evolution equations, we introduce the \textit{renormalized curvature components}:
\begin{equation}
        {\beta^r}_A = \beta_A-\frac{1}{2} \nabla_4 \phi \nabla_A \phi, \,\, {\underline{\beta}^r}_A=\underline{\beta}_A+\frac{1}{2} \nabla_3 \phi \nabla_A \phi,\,\,
        \sigma^r = \sigma+\frac{1}{2} \hat{\chi} \wedge \underline{\hat{\chi}}.
\end{equation}
Note that for the Einstein-scalar field system, the Ricci tensor satisfies $\text{Ric}_{\mu\nu} = D_\mu\phi D_\nu\phi$.

\paragraph{\textbf{Geometric Operators and Null Structure.}}
For 1-forms $\phi^{(1)}, \phi^{(2)}$ and symmetric 2-tensors, we denote the symmetric traceless product $\hat{\otimes}$ and the wedge product $\wedge$ as:
\begin{equation*}
    \left(\phi^{(1)} \widehat{\otimes} \phi^{(2)}\right)_{A B}:=\phi_A^{(1)} \phi_B^{(2)}+\phi_B^{(1)} \phi_A^{(2)}-\sg_{A B}\left(\phi^{(1)} \cdot \phi^{(2)}\right),
\end{equation*}
\begin{equation*}
    \phi^{(1)} \wedge \phi^{(2)}:=\slashed{\epsilon}^{A B}\sg^{C D} \phi_{A C}^{(1)} \phi_{B D}^{(2)}.
\end{equation*}
We also define the standard divergence, curl, and trace operators on $S_{u,v}$:
\begin{equation*} 
    \begin{aligned}
    (\operatorname{div} \phi)_{A_1 \ldots A_r}:=\nabla^B \phi_{B A_1 \ldots A_r}, \quad
    (\operatorname{curl} \phi)_{A_1 \ldots A_r}:=\slashed{\epsilon}^{B C} \nabla_B \phi_{C A_1 \ldots A_r} ,
    \end{aligned}
\end{equation*}
\begin{equation*}
    (\operatorname{tr} \phi)_{A_1 \ldots A_{r-1}}:=g^{B C} \phi_{B C A_1 \ldots A_{r-1}} .
\end{equation*}

Let $\chih$ and $\chibh$ denote the traceless parts of $\chi$ and $\chib$ respectively. The Ricci coefficients satisfy the following null structure equations:
\begin{equation}
    \begin{aligned}
    \nabla_4 \operatorname{tr} \chi+\frac{1}{2}(\operatorname{tr} \chi)^2 & =-|\hat{\chi}|^2-2 \omega \operatorname{tr} \chi-D_4 \phi D_4 \phi, \\
    \nabla_4 \hat{\chi}_{A B}+\operatorname{tr} \chi \hat{\chi}_{A B} & =-2 {\omega} \hat{\chi}_{A B}-\alpha_{A B}, \\
    \nabla_3 \operatorname{tr} \underline{\chi}+\frac{1}{2}(\operatorname{tr} \underline{\chi})^2 & =-2 \underline{\omega} \operatorname{tr} \underline{\chi}-|\underline{\hat{\chi}}|^2-D_3 \phi D_3 \phi, \\
    \nabla_3 {\underline{\hat\chi}}_{A B}+\operatorname{tr} \underline{\chi} \underline{\hat{\chi}}_{A B} & =-2 {\omegab} \chibh_{A B}-\underline{\alpha}_{A B}, \\
    \nabla_4 \operatorname{tr} \underline{\chi}+ \operatorname{tr} \chi \operatorname{tr} \underline{\chi} & =2 {\omega} \operatorname{tr} \underline{\chi}-2 K+2 \operatorname{div} \underline{\eta}+2|\underline{\eta}|^2+ 2D_A \phi D^A \phi, \\
    \nabla_4 \underline{\hat{\chi}}_{A B}+\frac{1}{2} \operatorname{tr} \chi \hat{{\chib}}_{A B} & =(\nabla \hat{\otimes} \underline{\eta})_{A B}+2 {\omega} \underline{\hat{\chi}}_{A B}-\frac{1}{2} \operatorname{tr} \underline{\chi} {\hat{\chi}}_{A B}+\left(\underline{\eta} \widehat{\otimes}\etab\right)_{A B}+\frac{1}{2}(\nabla\phi\hat\otimes\nabla\phi)_{AB}, \\
    \nabla_3 \operatorname{tr} \chi+ \operatorname{tr} \underline{\chi} \operatorname{tr} \chi & =2 \underline{\omega} \operatorname{tr} \chi-2K+2 \operatorname{div} \eta+2|\eta|^2+ 2D_A \phi D^A \phi, \\
    \nabla_3 \hat{\chi}_{A B}+\frac{1}{2} \operatorname{tr} \underline{\chi} \hat{\chi}_{A B} & =(\nabla \hat{\otimes} \eta)_{A B}+2 \underline{\omega} \hat{\chi}_{A B}-\frac{1}{2} \operatorname{tr} \chi \underline{\hat{\chi}}_{A B}+(\eta \widehat{\otimes} \eta)_{A B}+\frac{1}{2}(\nabla\phi\hat\otimes\nabla\phi)_{AB}.
    \end{aligned}
\end{equation}
The torsion coefficients and conformal weights satisfy:
\begin{equation}
    \begin{aligned}
    \nabla_4 \eta_A & =-\chi_{A B} \cdot(\eta-\underline{\eta})_B-\beta_A-\frac{1}{2} D_A \phi D_4 \phi, \\
    \nabla_3 \underline{\eta}_A & =-\underline{\chi}_{A B} \cdot(\underline{\eta}-\eta)_B+\underline{\beta}_A-\frac{1}{2} D_A \phi D_3 \phi, \\
    \nabla_4 \underline{\omega} & =2 {\omega} \underline{\omega}-\eta_A \cdot \underline{\eta}^A+\frac{1}{2}|\eta|^2-\frac{1}{2}K-\frac{1}{8}\tr\chi\tr\chib+\frac{1}{4}\chih\cdot\chibh+\frac{1}{4}D_3\phi D_4\phi+\frac{1}{4}D_A\phi D^A\phi , \\
    \nabla_3 {\omega} & =2 {\omega} \underline{\omega}-\eta_A \cdot \underline{\eta}^A+\frac{1}{2}|\underline{\eta}|^2-\frac{1}{2}K-\frac{1}{8}\tr\chi\tr\chib+\frac{1}{4}\chih\cdot\chibh+\frac{1}{4}D_3\phi D_4\phi+\frac{1}{4}D_A\phi D^A\phi.
    \end{aligned}
\end{equation}
The Gauss-Codazzi equations relate the intrinsic geometry to the null coefficients:
\begin{equation}\label{Gauss_Codazzi}
    \begin{aligned}
    &\operatorname{div} \hat{\chi}_A =\frac{1}{2} \nabla_A \operatorname{tr} \chi-\zeta^B\cdot\left(\hat{\chi}_{B A}-\frac{1}{2} \operatorname{tr} \chi g_{B A}\right)-\beta_A^r, \\
    &\operatorname{div} \underline{\hat{\chi}}_A =\frac{1}{2} \nabla_A \operatorname{tr} \underline{\chi}+\zeta^B \cdot\left(\underline{\hat{\chi}}_{B A}-\frac{1}{2} \operatorname{tr} \underline{\chi} g_{B A}\right)+\underline{\beta}_A^r, \\
    &\operatorname{curl} \eta  =-\operatorname{curl} \underline{\eta}=\sigma^r, \\
    &K =-\rho-\frac{1}{4} \operatorname{tr} \chi \operatorname{tr} \underline{\chi}+\frac{1}{2} \hat{\chi} \cdot \underline{\hat{\chi}}+\frac{1}{6} D_3 \phi D_4 \phi+\frac{1}{3} D_A \phi D^A \phi .
    \end{aligned}
\end{equation}

\paragraph{\textbf{Wave Equation.}}
The wave equation $\square_g \phi=0$ decomposes as:
\begin{equation}\label{intrinsic_wave_equation}
    \begin{aligned}
    & e_3\left(D_4 \phi\right)+\frac{1}{2}  \operatorname{tr} \underline{\chi}  D_4 \phi= \Delta_g \phi+2\omegab D_4\phi -\frac{1}{2}  \operatorname{tr} \chi  e_3 \phi+2 \eta^A e_A \phi, \\
    & e_4\left(D_3 \phi\right)+\frac{1}{2}  \operatorname{tr} \chi  D_3 \phi= \Delta_g \phi+2\omega D_3\phi-\frac{1}{2}  \operatorname{tr} \underline{\chi}  D_4 \phi+2  \underline{\eta}^A e_A \phi.
    \end{aligned}
\end{equation}

\paragraph{\textbf{Renormalized Bianchi Equations.}}
Finally, we present the evolution equations for the renormalized curvature components. For the Gaussian curvature $K$ and renormalized curl $\sigma^r$, there are transport equations:
\begin{equation}\label{Equations_K}
    \begin{aligned}
        \Omega\nabla_3 K+\Omega\tr\chib K &=\operatorname{div} \left(\Omega\underline{\beta}^r\right)-\frac{1}{2}\operatorname{div}\left(\eta\Omega\tr\chib\right)+\operatorname{div}\left(\eta\cdot\Omega\chibh\right), \\
        \Omega\nabla_4 K+\Omega\tr\chi K &=-\operatorname{div} \left(\Omega\beta^r\right)+\frac{1}{2}\operatorname{div}\left(\etab\Omega\tr\chi\right)-\operatorname{div}\left(\etab\cdot\Omega\chih\right),
    \end{aligned}
\end{equation}
and
\begin{equation}\label{Equations_reduced_sigma}
  \begin{aligned}
      \Omega\nabla_3\sigma^r+\Omega\tr\chib\sigma^r &=\operatorname{curl}(\Omega\underline{\beta}^r)+\operatorname{curl}\left(\Omega\chibh\cdot\eta\right), \\
      \Omega\nabla_4\sigma^r+\Omega\tr\chi\sigma^r &=-\operatorname{curl}(\Omega\beta^r)+\operatorname{curl}\left(\Omega\chih\cdot\etab\right).
  \end{aligned}
\end{equation}
For reduced $\beta^r$ and $\betab^r$, we have equations: 
\begin{equation}\label{Equations_nab3beta_nab4betab}
    \begin{aligned}
        \nabla_3 \beta^r_A &+\left(\tr\chib-2\omegab\right) \beta^r_A = -\nabla_A K+{}^*\nabla_A\sigma^r+2(\chih\cdot\betab)_A \\
        &+\frac{1}{2}\left(\nabla_A(\chih\cdot\chibh)-{}^*\nabla_A(\chih\wedge\chibh)\right)-\frac{1}{4}\nabla_A\left(\tr\chi\tr\chib\right)\\
        &+3(\eta\rho+{}^*\eta\sigma)_A + \nabla_A\left(g^{CD}\text{Ric}_{CD}\right)-\nabla^B \text{Ric}_{BA}-\frac{1}{2}\tr\chib \text{Ric}_{4A},\\
        \nabla_4 \underline{\beta}^r_A &+\left(\tr\chi -2\omega\right)\underline{\beta}^r_A =\nabla_A K+{}^*\nabla_A\sigma^r+2(\chibh\cdot\beta)_A \\
        &-\frac{1}{2}\left(\nabla_A(\chih\cdot\chibh)+{}^*\nabla_A(\chih\wedge\chibh)\right)+\frac{1}{4}\nabla_A\left(\tr\chi\tr\chib\right)\\
        &-3\left(\etab\rho-{}^*\etab\sigma\right)_A -\nabla_A(g^{CD}\text{Ric}_{CD})+\nabla^B \text{Ric}_{BA}+\frac{1}{2}\tr\chi \text{Ric}_{3A},
    \end{aligned}
\end{equation}
and 
\begin{equation}\label{Equations_nab3betab_nab4beta}
    \begin{aligned}
\nabla_3 \underline{\beta}^r_A &+\left(2 \operatorname{tr} \chib+2\omegab\right)\underline{\beta}^r_A =-\operatorname{div} \underline{\alpha}_A+\tr\chib\text{Ric}_{3A}+(\underline{\eta} \cdot \underline{\alpha})_A+\frac{1}{2}D_A \text{Ric}_{33},\\
\nabla_4 \beta^r_A &+\left(2 \operatorname{tr} \chi+2\omega\right)\beta^r_A =\operatorname{div} {\alpha}_A+\tr\chi\text{Ric}_{4A}+({\eta} \cdot {\alpha})_A-\frac{1}{2}D_A \text{Ric}_{44}.
    \end{aligned}
\end{equation}
The equations for $\alpha$ and $\alphab$ involve the mixed $\beta$ terms:
\begin{equation}\label{Equations_nab3alpha_nab4alphab}
    \begin{split}
        &\nabla_3\alpha +\left(\frac{1}{2}\tr\chib-4\omegab\right)\alpha -\nabla\hat\otimes\beta=-3(\chih\rho+{}^*\chih\sigma)+(\zeta+4\eta)\hat\otimes\beta,\\
        &\nabla_4\alphab+\left(\frac{1}{2}\tr\chi-4\omega\right)\alphab +\nabla\hat\otimes\betab=-3(\chibh\rho-{}^*\chibh\sigma)+(\zeta-4\etab)\hat\otimes\betab.
    \end{split}
\end{equation}

\begin{remark}
    The original equations for $\beta,\betab,\sigma,\rho$ contain many terms related to $D_3 D_3\phi$ and $D_4 D_4\phi$, which are expected to blow up near the singularity. By using the renormalized quantities $\beta^r$ and $\sigma^r$, we effectively remove these second-derivative terms from the evolution system.
\end{remark}

\begin{remark}
    For the mass aspect functions $\operatorname{div}\etab-K$ and $\operatorname{div}\eta-K$, we have the following equations:
    \begin{equation}
        \begin{aligned}
            &\Omega\nabla_3\left(\operatorname{div} \etab-K\right)+\Omega\tr\chib\left(\operatorname{div} \etab-K\right)=-2\operatorname{div}(\Omega\chibh\cdot\etab)-\operatorname{div}(\Omega\Rc_{3\cdot})+2\operatorname{div}(\eta\cdot\Omega\chibh),\\
            &\Omega\nabla_4\left(\operatorname{div} \eta-K\right)+\Omega\tr\chi\left(\operatorname{div} \eta-K\right)=-2\operatorname{div}(\Omega\chi\cdot\eta)-\operatorname{div}(\Omega\Rc_{4\cdot})+2\operatorname{div}(\etab\cdot\Omega\chih).
        \end{aligned}
    \end{equation}
\end{remark}

\subsection{Christodoulou's Naked Singularity}\label{subsec:Christodoulou_Background}

The background spacetime for our perturbation analysis is the spherically symmetric, continuously self-similar naked singularity solution constructed by Christodoulou in his seminal work \cite{Christodoulou1994}. This construction effectively reduces the Einstein-scalar field equations to a system of ordinary differential equations with respect to the self-similar variable, which are subsequently analyzed employing standard dynamical systems techniques.

A spherically symmetric configuration is described by the triple $(r, m, \phi)$, where $r$ is the area radius, $m$ is the Hawking mass, and $\phi$ is the scalar field. The system admits a scaling symmetry under the group $\mathbb{R}_+ \times \mathbb{R}$, where an element $(a, b)$ maps a solution $(r, m, \phi)$ to $(ar, am, \phi+b)$. A solution is said to be **$k$-continuously self-similar ($k$-CSS)** if there exists a homothetic Killing vector field $S$ such that the spacetime is invariant under the one-parameter diffeomorphism group $\Phi_a$ generated by $S$, up to a shift in the scalar field:
\[
\Phi_a^*(g) = a^2 g, \quad \Phi_a^*(\phi) = \phi + \sqrt{2\kappa} \log a.
\]
In terms of the Lie derivative, this condition is expressed as:
\begin{equation}\label{eq:self_similar_Lie}
    \mathcal{L}_S g = 2g, \quad S(\phi) = \sqrt{2\kappa}.
\end{equation}
Although the original construction in \cite{Christodoulou1994} employed self-similar Bondi coordinates, we adapt this exact solution to the double null gauge to facilitate our stability analysis. 
We summarize the key analytic properties of this background solution relevant to our proof, and detailed derivations concerning the background geometry in double null coordinates are provided in Appendix \ref{Appendix_Chr94} (see also Appendix A of \cite{singh2024construction}).

\begin{theorem}[Properties of the Background Solution]\label{thm:background_properties}
    For any fixed $\kappa \in (0, 1/3)$, there exists an asymptotically flat, spherically symmetric solution $(\mathcal{M}^c, g^c, \phi^c)$ to the Einstein-scalar field equations in the causal domain $\mathcal{Q} = \{u < 0, v > u\}$, satisfying the following properties:
    
    \begin{enumerate}
        \item \textbf{Global Structure:} The spacetime admits a global double null coordinate system satisfying the symmetry condition \eqref{eq:self_similar_Lie}.
        
        \item \textbf{Regularity:} The scalar function $\phi$ and the conformal factor $\Omega^2$ are smooth in the interior $\mathcal{Q} \setminus \{v=0\}$. Near the boundaries $\{v=0\}$ and $\{u=0\}$, the solution belongs to the H\"older class $C^{1, \frac{\kappa}{1-\kappa}}$.
        
        \item \textbf{Asymptotics at the Past Null Cone ($v=0$):} 
        The geometric quantities exhibit specific power-law behaviors. Along the initial incoming cone $v=0$:
        \begin{equation*}
            \begin{split}
                 &\lim_{v\rightarrow 0} v^\kappa \Omega^2(u,v) = (-u)^\kappa, \quad \Omega\operatorname{tr}\underline{\chi}(u,0) = \frac{2}{u}, \quad \Omega^{-1}\operatorname{tr}\chi(u,0) = -\frac{2}{(1+\kappa)u}, \\
                 &K(u,0) = (-u)^{-2}, \quad \Omega D_3\phi(u,0) = \frac{\sqrt{2\kappa}}{u}, \\
                 &\Omega^{-1}D_4\phi (u,0) = \frac{1}{(1+\kappa)u}\sqrt{\frac{2}{\kappa}}, \quad \Omega^{-1}D_4\log\left(v^\kappa\Omega^2\right)(u,0) = \frac{2}{(1+\kappa)u}.
            \end{split}
        \end{equation*}
        
        \item \textbf{Asymptotics at the Cauchy Horizon ($u=0$):} 
        Along the singular boundary $u=0$, the solution satisfies:
        \begin{equation*}
            \begin{split}
                &\lim_{u\rightarrow 0} (-u)^\kappa \Omega^2(u,v) = v^\kappa, \quad \Omega\operatorname{tr}\chi(0,v) = \frac{2}{v}, \quad \Omega^{-1}\operatorname{tr}\underline{\chi}(0,v) = -\frac{2}{(1+\kappa)v}, \\
                &K(0,v) = v^{-2}, \quad \Omega D_4\phi(0,v) = \frac{\sqrt{2\kappa}}{v}, \\
                &\Omega^{-1}D_3\phi(0,v) = \frac{1}{(1+\kappa)v}\sqrt{\frac{2}{\kappa}}, \quad \Omega^{-1}D_3\log ((-u)^\kappa \Omega^2)(0,v) = \frac{2}{(1+\kappa)v}.
            \end{split}
        \end{equation*}
        
        \item \textbf{Singular Behavior:} 
        The second derivatives of the scalar field blow up near the boundaries. Specifically,
        \[
        |(\Omega^{-1}e_4)^2\phi| \sim \left(\frac{-u}{v}\right)^{1-2\kappa}(-u)^{-2} \quad \text{as } \frac{v}{-u} \to 0,
        \]
        \[
        |(\Omega^{-1}e_3)^2\phi| \sim \left(\frac{v}{-u}\right)^{1-2\kappa}v^{-2} \quad \text{as } \frac{-u}{v} \to 0.
        \]
    \end{enumerate}
\end{theorem}

\subsection{The Characteristic Initial Value Problem}

We now review the local well-posedness of the characteristic initial value problem (CIVP) for the Einstein-scalar field system. For the Einstein vacuum equations, interested readers are referred to the work \cite{Luk12} by Luk, which is an extension of Rendall's result \cite{Rendall90} with harmonic gauge. Here, we formulate the analogous result for the coupled Einstein-scalar field system.

\vspace{2mm}
The underlying strategy relies on reducing the full system to a set of quasilinear wave equations formulated in harmonic coordinates. We cite the relevant local existence theorem due to Rendall:

\begin{theorem}[Rendall \cite{Rendall90}]\label{Rendall_thm}
    Consider a quasilinear wave equation of the form
    \begin{equation}\label{qswveq}
        g^{\mu\nu}(\phi)\partial^2_{\mu\nu}\phi = F(\phi, \partial\phi),
    \end{equation}
    where $g$ and $F$ are smooth functions of their arguments. Let data be prescribed on two intersecting null hypersurfaces $H_{-1}$ and $\underline{H}_0$. If the initial data are smooth and satisfy the requisite compatibility conditions at the intersection sphere $S_{-1,0}$, then there exists a unique smooth solution in a neighborhood to the future of $S_{-1,0}$.
\end{theorem}

Building upon this theorem, we formulate the corresponding local existence proposition for the geometric CIVP associated with the Einstein-scalar field system.

\begin{proposition}\label{Local existence of CIVP of ESE} 
    Let $H_{-1}^{(0,I_1)}$ and $\underline{H}_0^{(-1,I_2)}$ be outgoing and incoming null hypersurfaces emanating from a 2-sphere $S_{-1,0}$. 
    
    Consider a regular initial data set $(g_{AB}, \Omega, \zeta_A, \chi_{AB}, \underline{\chi}_{AB}, \phi)$ defined as follows:
    \begin{itemize}
        \item The degenerate metric $g_{AB}$ and the scalar field $\phi$ are defined on $H_{-1} \cup \underline{H}_0$.
        \item The null second fundamental forms $\chi_{AB}$ and $\underline{\chi}_{AB}$ are defined on $H_{-1}$ and $\underline{H}_0$, respectively.
        \item The torsion $\zeta_A$ and the conformal factor $\Omega$ are defined on the intersection sphere $S_{-1,0}$.
    \end{itemize}
    
    Suppose the data satisfy the constraint equations along the initial hypersurfaces:
    \begin{align}
        \text{Along } H_{-1}:& \quad \mathcal{L}_{{L}'} g_{AB} = 2{\chi}_{AB}, \quad \mathcal{L}_{L'} \operatorname{tr}\chi = -|\chi|^2_{g} - (L'(\phi))^2, \label{CIVP_constraint_outgoing} \\
        \text{Along } \underline{H}_{0}:& \quad \mathcal{L}_{\underline{L}'} g_{AB} = 2\underline{\chi}_{AB}, \quad \mathcal{L}_{\underline{L}'} \operatorname{tr}\underline{\chi} = -|\underline{\chi}|^2_{g} - (\underline{L}'(\phi))^2, \label{CIVP_constraint_incoming}
    \end{align}
    and the normalization condition $g(L', \underline{L}')|_{S_{-1,0}} = -2\Omega^{-2}$.

    Then the Einstein-scalar field equations admit a unique smooth solution in a sufficiently small spacetime neighborhood to the future of $S_{-1,0}$.
\end{proposition}

\begin{proof}
    The proof closely follows the harmonic coordinate construction employed in Section 3 of \cite{Luk12} for the vacuum case. In this framework, the Einstein-scalar field system satisfies the reduced Ricci condition $\widetilde{R}_{\mu\nu} = D_\mu\phi D_\nu\phi$, where $\widetilde{R}_{\mu\nu}$ denotes the Ricci tensor expressed in harmonic coordinates.
    
    The coupled system for the variables $\Phi = (g_{\mu\nu}, \phi)$ assumes the form of a quasilinear wave system:
    \begin{equation}
        g^{\mu\nu}\partial_{\mu\nu}^2\Phi = F(g, \partial g, \phi, \partial\phi).
    \end{equation}
    Local existence then follows directly from Theorem \ref{Rendall_thm}. To ensure that the obtained solution to the reduced system additionally satisfies the full Einstein equations, one must verify the propagation of the gauge constraints (i.e., the vanishing of the harmonic coordinate quantities $\Gamma^\mu = 0$). Invoking the Bianchi identities alongside energy conservation guarantees that the vector field $\Gamma^\mu$ satisfies a homogeneous linear wave equation:
    \begin{equation}
        \nabla_{\alpha}(G^{\alpha\beta} - T^{\alpha\beta}) = -\frac{1}{2}g^{\alpha\lambda}\partial^2_{\alpha\lambda}\Gamma^\beta + A_{\alpha}^{\beta\lambda}(g, \partial g, \Gamma)\partial_{\lambda}\Gamma^\alpha = 0.
    \end{equation}
    Given that these constraints are rigorously satisfied on the initial null boundary surfaces, the uniqueness of solutions implies $\Gamma^\mu \equiv 0$ throughout the domain of dependence, thereby concluding the proof.
\end{proof}

We can further refine this result by explicitly specifying the propagation of the lapse function, which simplifies the practical application of the existence theorem within our double null gauge framework.

\begin{corollary}\label{Local existence of CIVP of ESE cor} 
    Consider the setup of Proposition \ref{Local existence of CIVP of ESE}, but with the initial data set augmented to include the lapse $\Omega$ and connection coefficients $\omega, \underline{\omega}$ along the null hypersurfaces:
    \begin{itemize}
        \item $\Omega, \chi_{AB}, \omega$ defined on $H_{-1}$.
        \item $\Omega, \underline{\chi}_{AB}, \underline{\omega}$ defined on $\underline{H}_0$.
    \end{itemize}
    Suppose the data satisfy the propagation equations:
    \begin{align*}
        \text{Along } H_{-1}:& \, \mathcal{L}_{e_4} g_{AB} = 2\chi_{AB}, \, \mathcal{L}_{e_4}\log\Omega = -2\omega, \, \mathcal{L}_{e_4} \operatorname{tr}\chi = -|\chi|^2_{g} - 2\omega\operatorname{tr}\chi - (e_4(\phi))^2. \\
        \text{Along } \underline{H}_{0}:& \, \mathcal{L}_{e_3} g_{AB} = 2\underline{\chi}_{AB}, \, \mathcal{L}_{e_3}\log\Omega = -2\underline{\omega}, \, \mathcal{L}_{e_3} \operatorname{tr}\underline{\chi} = -|\underline{\chi}|^2_{g} - 2\underline{\omega}\operatorname{tr}\underline{\chi} - (e_3(\phi))^2.
    \end{align*}

    Then the Einstein-scalar field equations admit a unique smooth solution in a sufficiently small spacetime neighborhood to the future of $S_{-1,0}$.
\end{corollary}

\begin{proof}
    We introduce the respective harmonic coordinates $x^4$ and $x^3$ by integrating the lapse function along the null generators:
    \[
    x^4(u,v,\theta) = \int_0^v \Omega^{2}(-1,v',\theta) dv', \quad x^3(u,v,\theta) = \int_{-1}^u \Omega^{2}(u',0,\theta) du'.
    \]
    In these coordinates, the propagation equations governing the metric components and the lapse function reduce to standard transport equations. Consequently, the conditions stipulated in Proposition \ref{Local existence of CIVP of ESE} are thereby satisfied, smoothly guaranteeing local existence.
\end{proof}

\section{Initial Values and Approximating Solution}\label{Initial Values and Approximating Solution}

In this section, we construct exactly $\kappa$-self-similar initial data along the incoming null hypersurface $\underline{H}_0$, and approximately self-similar initial data along the outgoing null hypersurface $H_{-1}$.

\subsection{Construction of Self-Similar Initial Data on $\underline{H}_0$}\label{Construction of Self-Similar Initial Data on Incoming Cone}

We begin by constructing the data along $v=0$ to strictly satisfy the $\kappa$-self-similarity condition. Suppose a spacetime $(\mathcal{M},g)$ admits a homothetic Killing field $S=u\partial_u+v\partial_v$ satisfying the $\kappa$-self-similarity condition. Expressed in the double null gauge, the metric
\begin{equation*}
    g = -2\Omega^2(du\otimes dv+dv\otimes du) + g_{AB}(d\theta^A-b^Adu)\otimes(d\theta^B-b^Bdu),
\end{equation*}
admits a scaling reduction in terms of the self-similar variable $z = v/(-u)$. Specifically, the metric components exhibit the following scaling behavior:
\begin{equation*}
    (v^\kappa\Omega^2, b^A, g_{AB}) = ((-u)^\kappa\check{\Omega}^2(z,\theta), (-u)^{-1}\check{b}^A(z,\theta), (-u)^2\check{g}_{AB}(z,\theta)),
\end{equation*}
where the checked quantities $\check{\cdot}$ depend only on $z$ and $\theta$.

\vspace{2mm}
For a general tensor field taking the form $F_{A_1\cdots A_r}(u,v,\theta) = (-u)^{-a}\check{F}_{A_1\cdots A_r}(z,\theta)$, its geometric derivative along the incoming null direction $e_3$ is related to the tangential derivatives on the cone $\underline{H}_v$ via the identity:
\begin{equation*}
    \Omega \mathcal{L}_{e_3}F_{A_1\cdots A_r} = -\frac{a}{-u}F_{A_1\cdots A_r} + \mathcal{L}_b F_{A_1\cdots A_r} + \frac{v}{-u} \mathcal{L}_{\partial_v}F_{A_1\cdots A_r}.
\end{equation*}
Applying this fundamental relation to the geometric quantities yields the following algebraic constraints among the Ricci coefficients.

\begin{lemma}[Algebraic Relations in Self-Similar Spacetimes]
    In a self-similar spacetime, the following relations hold:
    \begin{equation}\label{relation_chib_and_chi_omegab_and_omega}
        \underline{\chi}_{AB} + \frac{v}{u}\chi_{AB} = \Omega^{-1}u^{-1}g_{AB} + \frac{1}{2}\Omega^{-1}(\mathcal{L}_b g)_{AB}, \quad 
        \underline{\omega} = -\frac{v}{u}\omega + \frac{1}{2}b^A\partial_A(\Omega^{-1}).
    \end{equation}
    Suppose $\check{\Omega}^2$ remains bounded near $z=0$. Evaluating the limit towards the initial cone gives $\lim_{v\rightarrow 0}\frac{v}{u}\Omega\omega = \frac{\kappa}{4u}$. In particular, along the hypersurface $v=0$ we have:
    \begin{equation*}
        \Omega\operatorname{tr}\underline{\chi} = \frac{2}{u} + \operatorname{div} b, \quad \Omega\underline{\hat{\chi}}_{AB} = \frac{1}{2}(\nabla\hat{\otimes}b)_{AB}, \quad \Omega\underline{\omega} = -\frac{\kappa}{4u} - \frac{1}{2}\mathcal{L}_b\log \Omega.
    \end{equation*}
\end{lemma}

If a given Einstein-scalar field spacetime $(\mathcal{M},g,\phi)$ obeys the $\kappa$-self-similarity condition, restricting these geometric relations to the hypersurface $v=0$ yields a coupled system of Lie derivative equations.

\begin{lemma}[Constraint Equations on $S_{u,0}$]
    Assuming the regularity condition \\
    $\lim_{v\rightarrow 0}v\mathcal{L}_{\partial_v}\psi=0$ holds for $\psi \in \{\Omega\underline{\chi}, \zeta, \Omega^{-1}\chi, \Omega^{-1}D_4\phi\}$, the following evolution equations are satisfied on the initial spheres $S_{u,0}$:
    \begin{equation}\label{eq Lie b div b}
        \begin{aligned}
            \frac{2\kappa}{(-u)^2}& + \frac{1}{-u}\operatorname{div} b - \mathcal{L}_b\operatorname{div} b - \frac{1}{2}(\operatorname{div} b)^2 - \frac{1}{4}|\nabla\hat{\otimes} b|^2 \\
            &- \frac{4}{-u}\mathcal{L}_b\log\Omega + 2\operatorname{div} b \cdot \mathcal{L}_b\log\Omega - (\Omega D_3\phi)^2 = 0,
        \end{aligned}
    \end{equation}
    \begin{equation}\label{eq_etab(-1,0)}
        \begin{split}
             \mathcal{L}_b\zeta_A + &\left(-\frac{2}{-u} + \operatorname{div} b\right)\zeta_A\\
            =& 2(\nabla\hat{\otimes} b)_A^C \nabla_C\log\Omega - \frac{\kappa}{-2u} - \mathcal{L}_b\log\Omega - \frac{1}{2}\operatorname{div}(\nabla\hat{\otimes} b)_A + \frac{1}{2}\nabla_A\operatorname{div} b,
        \end{split}
    \end{equation}
    \begin{equation}\label{eq_trchi(-1,0)}
        \begin{aligned}
            \mathcal{L}_b(\Omega^{-1}\operatorname{tr}\chi) + \Omega^{-1}\operatorname{tr}\chi\left(-\frac{1+\kappa}{-u} + \operatorname{div} b + 2\mathcal{L}_b\log\Omega\right) = -2K + 2\operatorname{div}\eta + 2|\eta|^2 + |\nabla\phi|^2,
        \end{aligned}
    \end{equation}
    \begin{equation}\label{eq_chih(-1,0)}
        \begin{split}
            \mathcal{L}_b(\Omega^{-1}\hat{\chi})_{AB} &- \left(\frac{1}{2}\operatorname{div} b + \frac{\kappa}{-u} - 2\mathcal{L}_b\log\Omega\right)\Omega^{-1}\hat{\chi}_{AB} \\
            &= (\nabla\hat{\otimes} b)^C_{\ A}(\Omega^{-1}\hat{\chi})_{BC} + (\nabla\hat{\otimes} b)^C_{\ B}(\Omega^{-1}\hat{\chi})_{AC} + (\nabla\hat{\otimes}\eta)_{AB} + (\eta\hat{\otimes}\eta)_{AB} \\
            &\quad - \frac{1}{2}(\nabla\hat{\otimes} b)_{AB}(\Omega^{-1}\operatorname{tr}\chi) + \nabla_A\phi \nabla_B\phi - \frac{1}{2}g_{AB}|\nabla\phi|^2,
        \end{split}
    \end{equation}
    \begin{equation}\label{eq_D4phi(-1,0)}
        \mathcal{L}_b(\Omega^{-1} D_4\phi) - \left(\frac{1}{2}\operatorname{div} b + \frac{\kappa}{-u}\right)\Omega^{-1} D_4\phi = \Delta_g \phi - \frac{1}{2}\Omega^{-1}\operatorname{tr}\chi\Omega D_3\phi - \mathcal{L}_b\log\Omega \cdot \Omega^{-1} D_4\phi.
    \end{equation}
    Detailed derivations are provided in Appendix \ref{appen_derive_eq}.
\end{lemma}

The construction of the initial data critically hinges on solving this coupled system. Specifically, the master equation \eqref{eq Lie b div b} determines the shift vector $b$ alongside the spatial metric structure. Predicated on the analytic framework developed in Appendix A of \cite{Rodnianski2019NakedSF}, we establish the following solvability result:

\begin{lemma}[Solution to the Master Equation]\label{Solve Lie b div b eq}
    Let $N \in \mathbb{N}$ and $0 < \epsilon \ll 1$. There exists a solution set $(g_{AB}, b^A, \phi, v^\kappa\Omega^2)$ on $S_{-1,0}$ such that:
    \begin{enumerate}
        \item The master equation \eqref{eq Lie b div b} is satisfied.
        \item The metric admits the conformal decomposition $g_{AB} = e^{2\varphi} \check{g}_{AB}$, where $\check{g}$ is the standard metric on $\mathbb{S}^2$.
        \item The solution is a small perturbation of the background in the following sense:
        \begin{equation}
            \|\varphi, \log(v^\kappa\Omega^2), b, \Omega D_3\phi + \sqrt{2\kappa}, \nabla\phi\|_{H^{N+2}(\mathbb{S}^2, g)} \leq \epsilon.
        \end{equation}
        \item The scalar field derivative $\partial_u\phi$ is constant on $S_{-1,0}$.
    \end{enumerate}
\end{lemma}

\begin{proof}
    The proof mirrors the corresponding vacuum case. Suppose that the tuple\\
    $(\kappa_1, b, g, \log\Omega)$ solves the vacuum equation (i.e., Eq. \eqref{eq Lie b div b} evaluated with $\phi=0$). Then, the extended tuple $(\kappa, b, g, \log\Omega, \phi)$ subject to the relation $\phi = \sqrt{2(\kappa - \kappa_1)}\log(-u)$ provides a valid solution to the coupled system \eqref{eq Lie b div b}.
\end{proof}

Once the fundamental variables $(g, b, \phi, v^\kappa\Omega^2)$ are fully determined, the remaining geometric quantities are successively obtained by solving their respective linear transport equations. To this end, we invoke the following solvability lemma governing Lie transport equations on the sphere.

\begin{lemma}[Proposition 4.5 in \cite{Rodnianski2019NakedSF}]\label{Lie_Propagation_Sol}
    Consider the linear equation on the sphere $\mathbb{S}^2$:
    \begin{equation}\label{Lie_propagate_eq}
        u + \mathcal{L}_X u + h \cdot u = F,
    \end{equation}
    where $X$ is a vector field in $H^{M+1}(\mathbb{S}^2)$, and $h, F$ are tensor fields in $H^M(\mathbb{S}^2)$. Define the control parameter $A_M = \|X\|_{H^{M+1}} + \|h\|_{H^M}$. If $A_2 \ll 1$, then there exists a unique solution $u \in H^M(\mathbb{S}^2)$ satisfying the estimate:
    \begin{equation}\label{Lie_propagate_est}
        \|u\|_{H^M} + \|\mathcal{L}_X u\|_{H^M} \leq C(A_M)\|F\|_{H^M}.
    \end{equation}
\end{lemma}

With this lemma, we can construct the initial data along $\Hb_0$.
\begin{corollary}[Construction of Initial Data]
    Let $\oc{\psi}$ denote the solution constructed via Lemma \ref{Solve Lie b div b eq} and Lemma \ref{Lie_Propagation_Sol}. These quantities define a self-similar initial data set along $\underline{H}_0$, obeying the scaling laws:
    \begin{gather*}
        \oc{b^A}(u) = (-u)^{-1}\oc{b^A}(-1), \, \oc{g}_{AB}(u) = (-u)^2\oc{g}_{AB}(-1), \\
        \oc{v^\kappa\Omega^2}(u) = (-u)^\kappa\oc{v^\kappa\Omega^2}(-1),  \, \oc{\zeta}_A(u) = \oc{\zeta}_A(-1),\\
        \oc{\Omega^{-1}\operatorname{tr}\chi}(u) = (-u)^{-1}\oc{\Omega^{-1}\operatorname{tr}\chi}(-1), \, \oc{\Omega^{-1}\hat{\chi}}_{AB}(u) = (-u)\oc{\Omega^{-1}\hat{\chi}}_{AB}(-1), \\
        \oc{\Omega D_3\phi}(u) = (-u)^{-1}\oc{\Omega D_3\phi}(-1), \, \oc{\Omega^{-1}D_4\phi}(u) = (-u)^{-1}\oc{\Omega^{-1}D_4\phi}(-1), \, \oc{\nabla_A\phi}(u) = \oc{\nabla_A\phi}(-1).
    \end{gather*}
    Furthermore, the constructed data remain close to the Christodoulou background solution. For $\psi \in \{\Omega^{-1}\chi, \zeta\}$ and $D\phi \in \{\Omega D_3\phi, \nabla\phi, \Omega^{-1}D_4\phi\}$, we have the estimates:
    \begin{equation}
        \|\nabla^i(\oc{\psi}(u) - \psi^c(u,0), \oc{D\phi}(u) - D\phi^c(u,0))\|_{L^2(\oc{g}(u))} \lesssim \epsilon (-u)^{-i}, \quad i=0, \dots, N.
    \end{equation}
\end{corollary}

\begin{proof}
    These precise estimates readily follow from applying Lemma \ref{Lie_Propagation_Sol} to the appropriate difference equations below. 
    \begin{equation*}
        \begin{split}
            &\Lie_b\left(\Omega^{-1}\tr\chi-\Omega^{-1}\tr\chi^c\right)+\left(\Omega^{-1}\tr\chi-\Omega^{-1}\tr\chi^c\right)\left(-\frac{1+\k}{-u}+\dv b+2\Lie_b\log\Omega\right)\\
            =& -2\left(e^{-2\varphi}-1\right)K^c +2e^{-2\varphi}\Delta_g\varphi+2\dv\eta+2|\eta|^2+|\nabla\phi|^2-\Omega^{-1}\tr\chi^c\left(\dv b+2\Lie_b\log\Omega\right),\\
            &\Lie_b(\Omega^{-1} D_4\phi-(\Omega^{-1} D_4\phi)^c)-\l\frac{1}{2}\dv b+\frac{\k}{-u}-\Lie_b\log\Omega\r\left(\Omega^{-1} D_4\phi-(\Omega^{-1} D_4\phi)^c\right)\\
            =&\left(\frac{1}{2}\dv b-\Lie_b\log\Omega\right)(\Omega^{-1} D_4\phi)^c+ \Delta \phi
    -\frac{1}{2}\left(\Omega^{-1}\tr\chi\Omega D_3\phi-(\Omega^{-1}\tr\chi\Omega D_3\phi)^c\right).
        \end{split}
    \end{equation*}
    The existence and regularity bounds are then a direct consequence of the linear theory.
\end{proof}

\subsection{Initial Values along the Initial Outgoing Null Cone}\label{Initial Values along the Initial Outgoing Null Cone}

Having explicitly constructed the requisite data on the incoming cone $\underline{H}_0$, we now proceed to formulate the initial data along the outgoing cone $H_{-1}$. The construction strategy involves three distinct steps:
\begin{enumerate}
    \item Establish a general solvability theory for parabolic-type Lie propagation equations.
    \item Construct the approximating shear $\ot{\hat{\chi}}$ and scalar derivative $\ot{e_4\phi}$ by solving these equations.
    \item Integrate the transport equations along $H_{-1}$ to define the full metric $g_{AB}$, lapse $\Omega$, and scalar field $\phi$.
\end{enumerate}

\subsubsection{General Theory for Lie Propagation Equations}

We first establish the corresponding well-posedness and regularity estimates for a broader class of transport equations that naturally arise in the geometric construction.

\begin{lemma}[Parabolic Lie Propagation]\label{Lie_propagation_eq_II}
    Let $\epsilon\ll 1$, and $0<\lambda\sim O(1)$.
    For tensor fields $f_{A_1\cdots A_r}(v,\theta)$ and $F_{A_1\cdots A_r}(v,\theta,f)$,
    consider the equation 
    \begin{equation}\label{Lie_prop_parabolic_eq}
        v\partial_v f(v,\theta)+\Lie_{b(v,\theta)} f(v,\theta)-\lambda f(v,\theta)+h(v,\theta)\cdot f(v,\theta)=F(v,\theta,f),
    \end{equation}
    where $b$ is a vector field and $h=\sum_{i\leq r}{h^{(i)}}_{A_1\cdots A_i}^{\quad\quad \ B_1\cdots B_i}$.
    Suppose $b,h,F$ are continuous in $v$ and bounded, and for any $v$ it holds that  $\|b(v)\|_{{H}^{N+1}(\SS)}+\|\nabla h\|_{{H}^{N}(\SS)}+\|(\nabla,\partial_f) F\|_{{H}^{N}(\SS\times\mathbb{R}^r)}\lesssim\epsilon$, $\|h(v)\|_{L^\infty(\SS)}+\|(\nabla,\partial_f) F\|_{L^\infty(\SS\times \mathbb{R}^r)}\leq \delta\ll 1$, then for any $V>0$, and $f_{V}(\theta)\in\mathring{H}^{N}$, there exists a solution $f$ defined on $(0,V)\times \SS$ solving \eqref{Lie_prop_parabolic_eq} with $f(V)=f_V$. The following estimate holds 
    \begin{equation}
        \|\nabla f\|_{L^\infty_v{H}^{N}(\SS)}\lesssim \epsilon+\|\nabla f_V\|_{{H}^{N}(\SS)},\quad \| f\|_{L^\infty_v{H}^{N}(\SS)}\lesssim \lnm F(f(0))\rnm_{L^\infty_{v}L^2(\SS)}+\| f_V\|_{{H}^{N}(\SS)}.
    \end{equation}
    Moreover, if we assume that all coefficients are continuous at $v=0$ and $\partial_v(b,\nabla b,F,\partial_f F,h)$ is of order $o_{\mathring{H}^4(\SS)}(\frac{1}{v})$, then $\lim_{v\rightarrow 0}f(v)$ exists and coincides with the solution to 
    \begin{equation}\label{Lie_prop_eq_s=0}
        \Lie_{b(0)} f-\lambda f+h(0)f=F(0,f(0)).
    \end{equation}
    Specifically, suppose that $\|\partial_vh,\partial_v\nabla b,\partial_vF\|_{\mathring{H}^4(\SS)}=o(v^{a-1})$. Then for any positive constant $a<\lambda$ and $\delta\ll \lambda-a$, we have $(v^{1-a}\partial_v f)$ is continuous at $0$, with limit satisfying 
    \begin{equation}
        \begin{split}
            {b^A(0)}\nabla_A(v^{1-a}\partial_vf) (0)-&\left(\lambda-a-h(0)+\partial_f F(0,f(0))\right)(v^{1-a}\partial_vf) (0)\\
            =&(v^{1-a}\partial_v)F(0,f(0))-(v^{1-a}\partial_v)b^A(0)\partial_A f(0)-(v^{1-a}\partial_v)h(0)f(0).
        \end{split}
    \end{equation}
\end{lemma}
\begin{proof}
    Consider reparametrization $s=-\log(\frac{v}{V})$, then the equation turns to 
    \begin{equation}
    \left\{
    \begin{aligned}
        &\partial_s f+(\lambda-h)f=\Lie_b f-F(f),\\
        &f|_{s=0}=f_V.
    \end{aligned}
    \right.
    \end{equation}
    For any $n\in\mathbb{N}$, the equation for $\nabla^n f$ reads 
    \begin{equation*}
        \begin{split}
            \partial_s\nabla^n f+&(\lambda-h)\nabla^n f-\Lie_b \nabla^n f-\partial_f F\nabla^{n}f\\
            =&\sum_{\substack{i_1+i_2+i_3=n\\ i_3\leq n-2}}\nabla^{i_1}F\partial_f^{i_2}F(f)\nabla^{i_3}(\nabla f)^{i_2}+\sum{\substack{i_1+i_2=n\\ i_2\leq n-1}}\nabla^{i_1}h\cdot \nabla^{i_2}f+[\nabla^n,\Lie_b]f.
        \end{split}
    \end{equation*}
    If we multiply the equation by $\nabla^n f$ and integrate on $\SS$, the left hand side is  
    \begin{equation}
        \begin{split}
            &\frac{1}{2}\partial_s\lnm \nabla^n f\rnm^2_{L^2(\SS)}+\int_{\SS} (\lambda-h+\frac{1}{2}\dv b -\partial_fF)\left|\nabla^n f\right|^2\\
            \geq & \lnm \nabla^n f\rnm_{L^2(\SS)}\left(\partial_s\lnm \nabla^n f\rnm_{L^2(\SS)}+(\lambda-C_0\delta)\lnm \nabla^n f\rnm_{L^2(\SS)}\right),
        \end{split}
    \end{equation}
    and the right hand side is bounded by:
    \begin{equation}
        \begin{split}
            &C\left(\lnm f\rnm_{H^{n-1}(\SS)}\right)\cdot\left(1+\lnm (\nabla,\partial_s)F\rnm_{H^{n-1}(\SS\times \mathbb{R}^r)}+\lnm \nabla h\rnm_{H^{n-1}(\SS)}+\lnm b\rnm_{H^n(\SS)}\right)\cdot \lnm \nabla^n f\rnm_{L^2(\SS)}\\
            \leq & C\left(\lnm f\rnm_{H^{n-1}(\SS)}\right)\cdot \epsilon \lnm \nabla^n f\rnm_{L^2(\SS)}.
        \end{split}
    \end{equation}
    Especially, for $n=0$ we have $\left|F(v,\theta,f)-F(v,\theta,0)\right|\leq \lnm \partial_f F\rnm_{L^\infty(\SS\times \mathbb{R}^r)}\cdot\left|f\right|$, thus
    \begin{equation}
        \begin{split}
            &\left(\partial_s\lnm f\rnm_{L^2(\SS)}+(\lambda-C_0\delta)\lnm f\rnm_{L^2(\SS)}\right)\lnm f\rnm_{L^2(\SS)}\\
            \leq & C_0\lnm F\rnm_{L^\infty(\SS)}\lnm f\rnm_{L^2(\SS)}+\lnm \partial_f F\rnm_{L^\infty(\SS)}\lnm f\rnm_{L^2(\SS)}^2,
        \end{split}
    \end{equation}
    which means 
    \begin{equation}\label{Lie_prop_eq_lemma nameless eq 2}
        \lnm f(s)\rnm_{L^2(\SS)}\leq \frac{\lnm f_V\rnm_{L^2(\SS)}}{e^{(\lambda-C\delta)s}}+C_0\lnm F\rnm_{L^\infty(\SS)}\leq C_0.
    \end{equation}
    Inductively, we obtain that for $1\leq n\leq N+1$, 
    \begin{equation}
        \lnm \nabla^nf\rnm_{L^2(\SS)}\leq C_0\epsilon.
    \end{equation}
    Since $\|\partial_v(b,\nabla b,F,\partial_f F,h)\|_{\mathring{H}^4(\SS)}=o(\frac{1}{v})$, we have $\|\partial_s(b,\nabla b,F,\partial_f F,h)\|_{W^{2,\infty}(\SS)}=o(1)$.
    Differentiating the equation in $s$ and writing $u=\partial_s f$, we get 
    \begin{equation}
        \partial_su+(\lambda-h+\partial_f F)u-\Lie_b u=f\partial_s h+\Lie_{\partial_A} f\partial_s b^A-\partial_s F=o(1).
    \end{equation}
    For any $\mu>0$, there exists $s_0>0$ such that for any $s>s_0$, $$\|f\partial_s h+\Lie_A f\partial_s b^A-\partial_s F\|_{W^{2,\infty}(\SS)}<\mu.$$
    We immediately obtain the bound for $u$:
    \begin{equation}\label{Nameless equation sec 3.2 eq 1}
        \|u\|_{H^2(\SS)}\leq \frac{\|(\Lie_b f-F(s,f)-(\lambda-h)f)|_{s=s_0}\|_{H^2(\SS)}}{e^{(\lambda-C\delta)(s-s_0)}}+C_0\mu.
    \end{equation}
    Taking the limit as $s\to\infty$ yields $\limsup_{s\to\infty}\|\partial_s f\|_{H^2(\SS)}\le \mu$ for any arbitrarily chosen $\mu>0$, which in turn implies that $\partial_s f$ converges to $0$ uniformly. Given the uniform standard bound for $\|f\|_{H^N(\SS)}$, it follows that the sequence $f(s)$ possesses a subsequential limit satisfying the limit equation \eqref{Lie_prop_eq_s=0}.
    If there are two convergent subsequences $f(s^{(1)_n})$ and $f(s^{(2)_n})$, then 
    \begin{align*}
        &\partial_sf(s^{(1)}_n)-\partial_sf(s^{(2)}_n)-(h(s_n^{(1)})-h(s_n^{(2)}))f(s_n^{(1)})-\nabla_A f(s^{(2)}_n)(b^A(s_n^{(1)})-b^A(s_n^{(2)}))\\
        &+F(s_n^{(1)},f(s_n^{(2)}))-F(s_n^{(2)},f(s_n^{(2)}))\\
        =& (h(s_n^{(2)})-\lambda+b^A(s_n^{(1)})\nabla_A)(f(s_n^{(1)})-f(s_n^{(2)}))-\left(F(s_n^{(1)},f(s_n^{(1)}))-F(s_n^{(1)},f(s_n^{(2)}))\right).
    \end{align*}
    Let $w_n=f(s_n^{(1)})-f(s_n^{(2)})$. Since the expression on the left-hand side converges sequentially to $0$ as $n\rightarrow\infty$ for any $\mu>0$, we can choose $n$ sufficiently large to ensure that 
    \[\lnm \left(h(s_n^{(2)})-\lambda-\lnm \partial_s F\rnm_{L^\infty} +b^A(s_n^{(1)})\nabla_A\right)w_n\rnm_{H^2(\SS)}\leq \mu.\]
    By Lemma \ref{Lie_Propagation_Sol}, we have $\|w_n\|_{H^2}\leq C_0\mu$. Therefore the two subsequential limits coincide.

\end{proof}

Building upon this lemma, we derive two crucial corollaries concerning the asymptotic behavior in the limit $v\to 0$.

\begin{corollary}[Derivative Decay]\label{Lie_prop_eq_II_Cor_II}
    Suppose that
    $\partial_vh,\partial_v\nabla b,\partial_vF$ is of order $o_{\mathring{H}^{N+1}(\SS)}(\frac{1}{v})$, then order $o_{\mathring{H}^{N+1}(\SS)}(\frac{1}{v})$ also holds for $f$,
    \begin{equation}\label{Lie_prop_eq_II_Cor_II eq1}
        \lim_{v\rightarrow 0} \lnm v\partial_v f(v,\theta)\rnm_{H^{N+1}(\SS)}=0.
    \end{equation}
    Moreover, if we assume
    \begin{equation}\label{Lie_prop_eq_II_Cor_II eq2}
        \|v\partial_v b(v)\|_{L^\infty_v{H}^{N+1}(\SS)}+\|v\partial_v\nabla h\|_{L^\infty_v{H}^{N}(\SS)}+\|v\partial_v (\nabla,\partial_f) F\|_{L^\infty_v{H}^{N}(\SS\times\mathbb{R}^r)}\lesssim\epsilon,
    \end{equation}
    then the following estimate holds :
    \begin{equation}
        \lnm v\partial_v f(v,\theta)\rnm_{H^{N+1}(\SS)}+\lnm \Lie_b f(v,\theta)\rnm_{H^{N+1}(\SS)}\lesssim \epsilon.
    \end{equation}
\end{corollary}
\begin{proof}
    Repeating the argument for \eqref{Nameless equation sec 3.2 eq 1} gives \eqref{Lie_prop_eq_II_Cor_II eq1}. If we have \eqref{Lie_prop_eq_II_Cor_II eq2}, then applying Lemma \ref{Lie_propagation_eq_II} to $v\partial_v f$ gives 
    \begin{equation}
        \lnm v\partial_v f(v,\theta)\rnm_{H^{N+1}(\SS)}\lesssim\epsilon.
    \end{equation}
    Since $\nabla^n\Lie_b f=\nabla^n\left(-v\partial_v f+F(f)+(\lambda-h)f\right)$, we obtain that for $1\leq n\leq N+1$, 
    \begin{equation}
        \lnm \nabla^n\Lie_b f(v,\theta)\rnm_{L^2(\SS)}\lesssim\epsilon.
    \end{equation}
    Together with $\lnm \Lie_b f\rnm_{L^2(\SS)}\leq \lnm b\rnm_{L^\infty(\SS)}\lnm \nabla f\rnm_{L^2(\SS)}$, we finish the proof.
\end{proof}

The next corollary reveals that solution's increasing rate is determined by coefficient $\lambda$.
\begin{corollary}[H\"older Continuity]\label{Lie_prop_eq_II_Cor_I}
    Suppose that $\| F(v,f(0))-F(0,f(0))\|_{H^N(\SS\times\mathbb{R}^r)}+\| b(v)-b(0),h(v)-h(0)\|_{H^N(\SS)}\leq C_1 v^{a}$, where $a<\lambda$ is a fixed number, then for $$\lnm h,F(f(0))\rnm_{L^\infty(S)}\ll \lambda-a,$$ it holds that 
    \begin{equation}
        \|f(v)-f(0)\|_{H^N}\leq \frac{C_0}{\lambda-a}\cdot\left(C_1+\frac{1}{V^{a}}\|f_V-f(0)\|_{H^N}\right) v^{a}.
    \end{equation}
\end{corollary}
\begin{proof}
    Let $w(v)=f(v)-f(0)$. After making difference of equations \eqref{Lie_prop_parabolic_eq} and \eqref{Lie_prop_eq_s=0}, we obtain that 
    \begin{equation}
        \begin{aligned}
            &v\partial_v w(v)+b^A(v)\nabla_A w(v)-\lambda w(v)+h(v)w(v)-\left(F(v,f(0)+w)-F(v,f(0))\right)\\
            =&-\left(b^A(v)-b^A(0)\right)\nabla_A f(0)-\left(h(v)-h(0)\right)f(0)+F(v,f(0))-F(0,f(0)).
        \end{aligned}
    \end{equation}
    Since the right hand side is bounded by $C_1 v^a O_{H^N}(1)$, we can obtain that 
    \begin{equation}
        \partial_s\lnm w(s)\rnm_{H^N(\SS)}+(\lambda-C\delta)\lnm w(s)\rnm_{H^N(\SS)}\leq C_0\cdot C_1 V^a e^{-as},
    \end{equation}
    where $s=-\log\left(\frac{v}{V}\right)$. 
    Thus after integration, we obtain 
    \begin{equation}
        \begin{split}
            \lnm w(s)\rnm_{H^N(\SS)}\leq & \frac{\lnm f_V-f(0)\rnm_{H^N(\SS)}}{e^{(\lambda-C\delta)s}}+C_0\cdot C_1 V^a\int_0^s e^{-as^\prime-(\lambda-C\delta)(s-s^\prime)}ds^\prime\\
            \leq & \frac{\lnm f_V-f(0)\rnm_{H^N(\SS)}}{e^{(\lambda-a-C\delta)s}}\left(\frac{v(s)}{V}\right)^a+C_0\cdot C_1 v^a\int_{0}^s e^{-(\lambda-a-C\delta)(s-s^\prime)}ds^\prime.
        \end{split}
    \end{equation}

\end{proof}

\subsubsection{Construction of Approximating Quantities}

We define the approximating shear $\ot{\Omega^{-1}\hat{\chi}}$ combined with the approximating scalar derivative $\ot{\Omega^{-1}e_4\phi}$ evaluated along $H_{-1}$ by solving the corresponding Lie propagation equations uniquely derived from the self-similar background ansatz. Specifically, we consider the coupled system:
\begin{equation}
    \begin{aligned}
        &\left(\oc{\Omega\nabla_3}+\frac{1}{2}\oc{\Omega\tr\chib}-4\oc{\Omega\omegab}\right)\ot{\Omega^{-1}\chih}= \left(\nabla\hat\otimes\eta+\eta\hat\otimes\eta+\frac{1}{2}\nabla\phi\hat\otimes\nabla\phi-\frac{1}{2}\Omega^{-1}\tr\chi \Omega\chibh\right)^{\circ},\\
        &\left(\oc{\Omega\nabla_3}+\frac{1}{2}\oc{\Omega\tr\chib}+\frac{1}{2}\left[(\Omega\tr\chi)^c\right]_0^v-4\oc{\Omega\omegab}-4\left[(\Omega\omegab)^c\right]_0^v\right)\ot{\Omega^{-1}e_4\phi}\\
        &\qquad\qquad = \left(\dv\nabla\phi-\frac{1}{2}\Omega^{-1}\tr\chi\Omega e_3\phi+2\eta\cdot\nabla\phi\right)^\circ -\left[\left(\frac{1}{2}\Omega^{-1}\tr\chi\Omega e_3\phi\right)^c\right]_0^v,
    \end{aligned}
\end{equation}
which reads in the self-similar context along $H_{-1}$ that 
\begin{equation}
    \begin{aligned}
        \left(v\Lie_v+\Lie_{\oc{b}}\right)\ot{\Omega^{-1}\chih}_{AB}-&\left(\k-2\Lie_{\oc{b}}\log\oc{\Omega}+\frac{1}{2}\oc{\dv b}\right)\ot{\Omega^{-1}\chih}_{AB}-\frac{1}{2}\left(\oc{\nabla\hat\otimes b}\right)^C_{\ \left(A\right.}\ot{\Omega^{-1}\chih}_{\left. B\right)C}\\
        =& \left(\nabla\hat\otimes\eta+\eta\hat\otimes\eta+\frac{1}{2}\nabla\phi\hat\otimes\nabla\phi-\frac{1}{2}\Omega^{-1}\tr\chi \Omega\chibh\right)^{\circ},\\
        \left(v\Lie_v+\Lie_{\oc{b}}\right)\ot{\Omega^{-1}e_4\phi}-&\left(\k+\left[(\Omega\omega)^c\right]_0^v-2\Lie_{\oc{b}}\log\oc{\Omega}-\frac{1}{2}\oc{\dv b}-\frac{1}{2}(v\Omega\tr\chi)^c\right)\ot{\Omega^{-1}e_4\phi}
        \\
        =& \left(\dv\nabla\phi-\frac{1}{2}\Omega^{-1}\tr\chi\Omega e_3\phi+2\eta\cdot\nabla\phi\right)^\circ-\frac{1}{2}\left[\left(\Omega^{-1}\tr\chi\Omega e_3\phi\right)^c\right]_0^v.
    \end{aligned}
\end{equation}
Here we use notation $\left[\psi\right]_0^v(u,v)=\psi(u,v)-\psi(u,0)$.
We impose the boundary conditions at $v=\epsilon_1$:
\[
\ot{\Omega^{-1}\hat{\chi}}(\epsilon_1) = 0, \quad \ot{\Omega^{-1}e_4\phi}(\epsilon_1) = (\Omega^{-1}e_4\phi)^c(\epsilon_1).
\]

\begin{proposition}[Properties of Approximating Solutions]
    For the above approximating quantities, the following properties hold:
    \begin{enumerate}
        \item They are continuous at $v=0$, and $$\ot{\Omega^{-1}\chih}(-1,0)=\oc{\Omega^{-1}\chih}(-1,0),\ot{\Omega^{-1}e_4\phi}(-1,0)=\oc{\Omega^{-1}e_4\phi}(-1,0).$$
        \item For any $\delta>0$, we can choose $\epsilon,\epsilon_1$ sufficiently small so that \begin{equation}\label{ID_triangle_chih_e4phi}
            \lnm \left[\ot{\Omega^{-1}\chih}\right]_0^v(-1,v),\left[\ot{\Omega^{-1}e_4\phi}\right]_0^v(-1,v)-\left[\left(\Omega^{-1}e_4\phi\right)^c\right](-1,v)\rnm_{H^{N-1}(\SS)}\lesssim\epsilon v^{\k-\delta}.
        \end{equation}
        \item $\ot{\Omega^{-1}\chih}_{AB}$ is $\oc{g}$-trace-free.
        \item $\ot{\Omega^{-1}\chih}_{AB}(u,v)$ and $\ot{\Omega^{-1}e_4\phi}(u,v)$ satisfy
        \[\ot{\Omega^{-1}\chih}_{AB}(u,v)=(-u)\ot{\Omega^{-1}\chih}_{AB}(-1,v/(-u)),\ \ot{\Omega^{-1}e_4\phi}(u,v)=\frac{1}{-u}\ot{\Omega^{-1}e_4\phi}(-1,v/(-u)).\]
    \end{enumerate}
\end{proposition}
\begin{proof}
    
    To compute the convergent speed at $v=0$, we first notice that 
    \begin{equation}
        \begin{aligned}
            &\left(v\Lie_v-4(v\Omega\omega)^c+\frac{1}{2}(v\Omega\tr\chi)^c\right)\left({\Omega^{-1}e_4\phi}\right)^c=-\frac{1}{2}\left(\Omega^{-1}\tr\chi\Omega e_3\phi\right)^c,
        \end{aligned}
    \end{equation}
    and we have equation for $\ot{\Omega^{-1}e_4\phi}-(\Omega^{-1}e_4\phi)^c$
    \begin{equation}
        \begin{aligned}
            &\left(v\Lie_v+\Lie_{\oc{b}}-\left(4(v\Omega\omega)^c-2\Lie_{\oc{b}}\log\oc{\Omega}-\frac{1}{2}\oc{\dv b}-\frac{1}{2}(v\Omega\tr\chi)^c\right)\right)\left(\ot{\Omega^{-1}e_4\phi}-(\Omega^{-1}e_4\phi)^c\right)\\
            =& -\left(2\Lie_{\oc{b}}\log\oc\Omega+\frac{1}{2}\oc\dv b\right)+\left(\dv\nabla\phi+2\eta\nabla\phi\right)^\circ\\ & -\frac{1}{2}\left(\oc{\Omega^{-1}\tr\chi}\oc{\Omega e_3\phi}(-1,0)-{\left(\Omega^{-1}\tr\chi \Omega e_3\phi\right)^c}(-1,0)\right).
        \end{aligned}
    \end{equation}
    Because $H^{N-1}(\SS)$ norm of the right hand side in above equation is of size $O(\epsilon)$,
    then the estimate follows by Corollary \ref{Lie_prop_eq_II_Cor_II}, and $\Omega^{-1}\chih$ is similar. 
    Since $\oc{\Omega\nabla_3}$ is commutable with $\oc{g}$, we obtain 
    \begin{equation}
        \left(\oc{\Omega\nabla_3}+\frac{1}{2}\oc{\Omega\tr\chib}-4\oc{\Omega\omegab}\right)\left(\oc{g^{AB}}\ot{\Omega^{-1}\chih}_{AB}\right)=0.
    \end{equation}
    Hence $\ot{\Omega^{-1}\chih}_{AB}$ is $\oc{g}$-trace-free
\end{proof}

For $v\in (0,\epsilon_1)$, we define ${v^\k\Omega^2}(v)$ by solving 
\begin{equation}
    v\partial_v\log\left({v^\k\Omega^2}\right)=\k-4\left(v\Omega\omega\right)^c(v),\ {v^\k\Omega^2}(0)=\oc{v^\k\Omega^2}(-1,0).
\end{equation}
Define $\phi(-v)$ by $$\phi(v)=\oc{\phi}(-1,0)+\int_0^v \frac{1}{(v^\prime)^\k}\left(v^\k\Omega^2\right)(v^\prime)\ot{\Omega^{-1}e_4\phi}(-1,v^\prime)dv^\prime,$$
and $g_{AB}(v),(\Omega^{-1}\tr\chi)(v)$ by 
\begin{equation}
    \left\{\begin{aligned}
        &v^\k\Lie_v g_{AB}=(v^\k\Omega^2)(\Omega^{-1}\tr\chi)+2(v^\k\Omega^2)(\Omega^{-1}\chih)_{AB},\\
        &v^\k\partial_v\left(\Omega^{-1}\tr\chi\right)=-\left(v^\k\Omega^2\right)\left(\frac{1}{2}\left(\Omega^{-1}\tr\chi\right)^2+\Omega^{-1}\chih_{AB}\Omega^{-1}\chih_{CD}g^{AC}g^{BD}+\left({\Omega^{-1}e_4\phi}\right)^2\right) ,\\
        &\Omega^{-1}\chih_{AB}(v)=\frac{1}{2}\left(g_{AC}\ot{\Omega^{-1}\chih}_{DB}+g_{BC}\ot{\Omega^{-1}\chih}_{DA}\right)\oc{\left(g^{CD}\right)}(-1,v),\\ 
        &\Omega^{-1}e_4\phi(v)=\ot{\Omega^{-1}e_4\phi}(-1,v)\\
        &g_{AB}(0)=\oc{g}_{AB}(-1,0),\ \Omega^{-1}\tr\chi(0)=\oc{\Omega^{-1}\tr\chi}(-1,0).
    \end{aligned}\right.
\end{equation}

The following lemma systematically compares $g,\Omega$ with their respective backgrounds $g^c,\Omega^c$ restricted on $H_{-1}\cap\{0<v<\epsilon_1\}$.
\begin{lemma}[Comparison with Background]
    For $u=-1,$ $v\in(0,\epsilon_1)$, we have
    \begin{enumerate}
        \item $(\Omega/\Omega^c)(v)=(\Omega/\Omega^c)(0)$, $v\Omega\omega=(v\Omega\omega)^c$.
        \item For $\delta$ satisfying $ \epsilon\ll\delta\ll \k$, the following inequalities hold 
        \begin{equation}\label{333444000}
            \lnm \left[g\right]_0^v-\left[g^c\right]_0^v,\left[\Omega^{-1}\tr\chi\right]_0^v-\left[(\Omega^{-1}\tr\chi)^c\right]_0^v\rnm_{H^{N}(\SS)}\lesssim \epsilon v^{1-\k},
        \end{equation}
        \begin{equation}
            \lnm \left[\Omega^{-1}\chih\right]_0^v,\left[\Omega^{-1}e_4\phi\right]_0^v-\left[(\Omega^{-1}e_4\phi)^c\right]_0^v\rnm_{H^{N}(\SS)}\lesssim \epsilon v^{\k-\delta}.
        \end{equation}
    \end{enumerate}
\end{lemma}
\begin{proof}
    We write the equations for differences. 
    \begin{equation}
        \left\{\begin{aligned}
            v^\k\Lie_v\left(g-g^c\right) =&\left(v^\k\Omega^2\right)\left(\Omega^{-1}\tr\chi-(\Omega^{-1}\tr\chi)^c+2\Omega^{-1}\chih-2\left(\Omega^{-1}\chih\right)^c\right) \\
            &+\left((\Omega/\Omega^c)^2-1\right)(v^\k\Omega\tr\chi)^c,\\
            v^\k\partial_v\left(\Omega^{-1}\tr\chi-(\Omega^{-1}\tr\chi)^c\right)=&-\left(v^\k\Omega^2\right)\left(\frac{1}{2}\left(\Omega^{-1}\tr\chi\right)^2+\left|\Omega^{-1}\chih\right|^2+(\Omega^{-1}e_4\phi)^2\right)\\
            &+\left(v^\k\Omega^2\right)\left(\frac{1}{2}\left(\Omega^{-1}\tr\chi\right)^2+\left|\Omega^{-1}\chih\right|^2+(\Omega^{-1}e_4\phi)^2\right)^c\\
            &-\left((\Omega/\Omega^c)^2-1\right)v^\k\left(\frac{1}{2}((\tr\chi)^c)^2+((e_4\phi)^c)^2\right).
        \end{aligned}\right.
    \end{equation}
    The right hand side are all of size $\epsilon$, hence \eqref{333444000} follows by direct integration. 
    Since 
    \begin{equation}
        \left[\Omega^{-1}\chih\right]_0^v\sim g\oc{g^{-1}}\left[\ot{\Omega^{-1}\chih}\right]_0^v,
    \end{equation}
    the estimates of $\left[\Omega^{-1}\chih,\Omega^{-1}e_4\phi\right]_0^v$ comes from the corresponding $\ot{\cdot}$ estimates \eqref{ID_triangle_chih_e4phi}.
\end{proof}
In the exterior region $v\geq \epsilon_1$, we can prescribe the data components $g_{AB},\phi,\Omega^2$ freely, provided that the following geometric estimates reliably hold:
\begin{equation}
    \begin{aligned}
        &\sum_{j=0}^1\sum_{i =0}^7\lnm v^{i+j-1}\partial_v^j\partial_\theta^i\left(g-g^c,\Omega/\Omega^c-1\right)\rnm^2_{L^2(\SS,{g^c(-1,v)})}\\
        &\qquad +\sum_{j=0}^1\sum_{i=1-j}^7\lnm v^{i+j-1}\partial_v^j\partial_\theta^i\left(\phi-\phi^c\right)\rnm^2_{L^2(\SS,{g^c(-1,v)})}\leq \epsilon^2.
    \end{aligned}
\end{equation}
With Corollary \ref{Local existence of CIVP of ESE cor}, we study the ESE spacetime $(\mathcal{M},g,\phi)$ near $\Hb_0\cup H_{-1}$ arising from initial data 
\begin{equation}\label{334466}
    \begin{aligned}
        &g_{AB}(u,0)=\oc{g}_{AB}(u,0),\ g_{AB}(-1,v)=g_{AB}(v),\ \lim_{v\rightarrow 0}\left(v^\k\Omega^2\right)(u,v)=\oc{v^\k\Omega^2}(u,0),\\
        &(v^\k\Omega^2)(-1,v)=(v^\k\Omega^2)(v),\ \phi(u,0)=\oc{\phi}(u,0),\ \phi(-1,v)=\phi(v).
    \end{aligned}
\end{equation}
In this article, we rigorously prove the global existence result formulated as follows.
\begin{theorem}[Global Existence]\label{Main_Existence_Theorem_whole_paper}
    The spacetime $(\mathcal{M},g,\phi)$ associated with the characteristic initial data \eqref{334466} admits global a priori energy estimates. Consequently, the spacetime can be globally developed over the entire exterior characteristic initial region $-1<u<0$ and $0<v<\infty$, revealing a globally naked singularity located at $u=0,v=0$.
\end{theorem}

\section{Weighted Scale Invariant Norms and Useful Lemmas}\label{Norms and Prework}
\subsection{Weighted norm system} We first define \textit{signatures} and a \textit{scale-invariant weighted norm system} to simplify the calculation in later sections.
\begin{definition}
    We define the \textit{signature} for the geometric quantities of interest:
    \[s(\Omega^2)=\k+1,\quad s(b)=s(g)=s(\phi)=1,\]
    \[s(\Omega^{-1}\tr\chi)=s(\Omega\tr\chib)=s(\Omega^{-1}\chih)=s(\Omega^{-1}\omega)=s(\Omega\omegab)=s(\eta)=s(\etab)=0,\]
    \[s(\Omega^2\alphab)=s(\Omega^{-2}\alpha)=s(\Omega^{-1}\beta)=s(\Omega\betab)=s(\rho)=s(\sigma)=-1.\]
    Furthermore, we establish the following conventions for derivatives and contractions:
    \[s(\nabla)=s(\Omega D_3)=s(\Omega^{-1}D_4)=0,\]
    \[s(\Phi_1\Phi_2)=s(\Phi_1)+s(\Phi_2)-1,\]
    where $\Phi_1$ can be $\nabla,\Omega D_3,\Omega^{-1} D_4$.
    For difference quantities, the signatures of the differences, such as $s({\psi}-\psi^c)$ and $s(\left[\psi\right]_0^v)$, are all defined to be equal to $s(\psi)$. Additionally, we define 
    \[s(v^a(-u)^b\psi)=s(\psi)+a(1-\k)+b.\]
\end{definition}
\begin{remark}
    By direct calculation, one can readily verify that within every evolution equation, each additive term possesses the same signature. For instance, in the equation 
    $\Omega^{-1}\nabla_4\psi=\psi^2+\Psi$, we have $s(\Omega^{-1}\nabla_4\psi)=-1=s(\psi^2)=s(\Psi)$.
\end{remark}

\begin{definition}
    We now define the scale-invariant norms incorporated with the weight function $w(u,v)$. 
\[\lnm\psi\rnm_{\LS_w^2(S_{u,v})}^2:=\int_{\SS} (-u)^{2-2s(\psi)}|\psi|_g^2\cdot w(u,v)d\V,\]
\[\|\psi\|_{\LS_w^2(H_u^v)}^2:=\int_{v_0(u)}^v\int_{\SS} (-u)^{1-2s(\psi)}\Omega^{2}w(u,v^\prime)\cdot |{\psi}|_g^2d\V dv^\prime,\]
and 
\[\|\psi\|_{\LS_w^2(\Hb_v^u)}^2:=\int_{u_0(v)}^u\int_{\SS} (-u^\prime)^{1-2s(\psi)}w(u^\prime,v)\cdot |{\psi}|_g^2d\V du^\prime.\]
We furthermore define $\LS_w^\infty(S_{u,v})$ norm by 
\[\lnm \psi\rnm_{\LS_w^\infty(S_{u,v})}=(-u)^{1-s(\psi)}w^{1/2}\sup_{S_{u,v}}|\psi|_{g}.\]
$\LS_w^2$ norm on the $\k$ similar hypersurfaces is defined by 
\[\|\Phi\|^2_{\LS_w^2(\Sigma^z_{u,v})}:=\int_{-v^{1/(1-\k)}/z}^u\|\Phi\|^2_{\LS_w^2\left(S_{s,z(-s)^{1-\k}}\right)}\frac{1}{(-s)}ds.\]
Especially, we write
$\|\Phi\|^2_{\LS_w^2(\Sigma^z)}:=\|\Phi\|^2_{\LS_w^2\left(\Sigma^z_{0,z^{1-\k}}\right)}.$
Next, we introduce the bulk integrated norms tailored for the energy estimates:
\[\|\Phi\|^2_{\LS_w^2({\mathcal{R}}_{u,v})}:=\int_{u_0(v)}^u\int_{v_0(u^\prime)}^v\int_{\SS} (-u^\prime)^{-2s(\Phi)}|{\Phi}|^2\Omega^{2} w\  d\V dv^\prime du^\prime .\]
Whenever there is no risk of ambiguity, we shall write $\LS$ in place of $\LS_w$.

\end{definition}
\begin{remark}
    For $X_{u,v}\in\{S_{u,v},H_u^v,\Hb_v^u,\R_{u,v}\}$, we deduce that 
    \begin{equation}
        \lnm v^a(-u)^{b}\psi\rnm_{\LS_w^2(X_{u,v})}\leq \left(\frac{v}{(-u)^{1-\k}}\right)^a\lnm \psi\rnm_{\LS_w^2(X_{u,v})}.
    \end{equation}
    Especially, the equality holds for $X=S$,
    \begin{equation}
        \lnm v^a(-u)^{b}\psi\rnm_{\LS_w^2(S_{u,v})}= \left(\frac{v}{(-u)^{1-\k}}\right)^a\lnm \psi\rnm_{\LS_w^2(S_{u,v})}.
    \end{equation}
\end{remark}

\begin{remark}
    We will demonstrate later that within Region I, the metric satisfies $g(u,v)=(-u)^2g^{\SS}+(-u)^2o_{\epsilon,\epsilon_1}(1)$. Consequently, the ${\LS^2(S_{u,v})}$-norm is formally equivalent to the standard weighted norm $L^2(S_{u,v},w\cdot d{\rm Vol}_g)$:
    \begin{equation}
        (1-o_{\epsilon,\epsilon_1}(1))\lnm \psi\rnm^2_{\LS^2(S_{u,v})}\leq w(u,v)\cdot \int_{S_{u,v}}|\psi|^2_{g}d{\rm Vol}_{g}\leq (1+o_{\epsilon,\epsilon_1}(1))\lnm \psi\rnm^2_{\LS^2(S_{u,v})}.
    \end{equation}
\end{remark}

\subsection{Inequalities of the weighted norm system}
In this part, we establish some useful inequalities, which appear frequently in the standard double null analysis, for our weighted norm system. 
\begin{lemma}{\rm (Sobolev Embedding)}
    Assume that $C_1^{-1}g^{\SS}_{AB}\leq (-u)^{-2}g_{AB}\leq C_1g^{\SS}_{AB}$. For tensors $\psi=\psi_{i_1\cdots i_n}$ on $S_{u,v}$, it follows that 
    \begin{equation}\label{Sobolev_W22}
            \|{\psi}\|_{\LS_w^\infty(S_{u,v})} 
            \leq C(C_1)\sum_{i=0}^2\|\nabla^i{\psi}\|_{\LS_w^2(S_{u,v})}^2,
    \end{equation}

\end{lemma}
\begin{proof}
    We begin by recalling the standard Sobolev inequality on the round sphere $\SS_r$ of radius $r$. For $r=1$, we have the familiar bound:
$$
\begin{aligned}
\|\psi\|_{L^{\infty}\left(\SS_{1}\right)} & \leq C_{0}\left(\|\psi\|_{L^{2}\left(\SS_{1}\right)}+\| \nabla \psi\|_{L^{2}\left(\SS_{1}\right)}+\| \nabla^{2} \psi\| _{L^{2} \left(\SS_{1}\right)}\right) 
\end{aligned}.
$$
For a scalar function $\psi$ defined on $\SS_r$, we can naturally interpret it as a corresponding function on $\SS_1$. Since the metric on $\SS_r$ is obtained by a constant conformal rescaling of the metric on $\SS_1$, we have $g^{\SS_r}_{ij}=r^2g^{\SS_1}_{ij}$, which implies
$$
\begin{aligned}
\| \nabla^{k} \psi\|_{L^{2}\left(\SS_{r}\right)}^{2} & =\int_{\SS_{1}} \nabla_{i_{1}} \cdots \nabla_{i_{k}} \psi \nabla_{j_{1}} \cdot \nabla_{j_{k}} \psi g^{i_1j_1} . . g^{i_kj_k} r^{2} d V =r^{-2 k+2}\left\|\nabla^{k} \psi\right\|_{L^{2}\left(\SS_{1}\right)}^{2}
\end{aligned}
$$
Consequently, we arrive at the following rescaled Sobolev embedding on $\SS_r$:
$$
\|\psi\|_{L^{\infty}\left(\SS_{r}\right)} \leq C_{0}\left(r^{-1}\|\psi\|_{L^{2}\left(\SS_{r}\right)}+\|\nabla \psi\|_{L^{2}\left(\SS_{r}\right)}+r \| \nabla^{2}\psi\|_{L^{2}\left(\SS_{r}\right)}\right).
$$
Next, we carefully derive the corresponding Sobolev embedding expressed in the $\LS^2(\S)$ norm. 
For a generic tensor $\psi=\psi_{i_{1} \ldots i_{n}}$, we note the pointwise equivalence $|\psi|_g\sim|\psi|_{\SS} \cdot (-u)^{-n}$. 
$$
\begin{aligned}
 &\|\psi\|_{L^{\infty}(\S)}\sim(-u)^{s(\psi)-1-n}\left\|u^{1-s(\psi)} \psi\right\|_{L^{\infty}(\SS)}  \\
 \lesssim& (-u)^{s(\psi)-1-n}\left(\int_{\SS} (-u)^{2-2 s(\psi)}|\psi|_{\SS}^{2}+\int_{\SS} (-u)^{2-2 s(\psi)} \mid \nabla\psi|_{\SS} ^{2}+\int_{\SS} (-u)^{2-2 s(\psi)}\left|\nabla^2 \psi\right|_{\SS}^{2}\right)^{1 / 2}  \\
\lesssim& (-u)^{s(\psi)-1}\left(\int_{\SS} (-u)^{2-2 s(\psi)}|\psi|_{g}^{2}  +\int_{\mathbb{S}} (-u)^{2-2 s(\nabla \psi)}|\nabla \psi|_{g}^{2}  +\int_{\mathbb{S}} (-u)^{2-2 s\left(\nabla^{2} \psi\right)}\left|\nabla^{2} \psi\right|_{g}^{2} \right)^{1 / 2}. \end{aligned}
$$

\end{proof}

We furthermore reformulate the standard product Sobolev inequality 
\[\lnm \psi_1\psi_2\rnm_{H^n(\SS)}\lesssim \lnm \psi_1\rnm_{H^{n-1}(\SS)}\lnm \psi_2\rnm_{H^n(\SS)}+\lnm \psi_1\rnm_{H^{n}(\SS)}\lnm \psi_2\rnm_{H^{n-1}(\SS)},\, {\rm for}\, n\geq 3,\]
in terms of $\LS$-norms.
\begin{lemma}\label{Lemma_Sobolev_Ineq_product}
    With the scale invariant weighted norm, for $n\geq 3$ the following holds,
    \begin{equation}\label{Sobolev_Product_form}
        \begin{aligned}
            \sum_{j\leq n}&\lnm \nabla^j(\psi_1\psi_2)\rnm_{\LS_w^2(S_{u,v})}\\
            &\lesssim \sum_{\substack{i\leq n-1,j\leq n}}\left(\lnm \nabla^i\psi_1\rnm_{\LS_1^2(S_{u,v})}\lnm \nabla^j\psi_2\rnm_{\LS_w^2(S_{u,v})}+\lnm \nabla^i\psi_2\rnm_{\LS_1^2(S_{u,v})}\lnm \nabla^j\psi_1\rnm_{\LS_w^2(S_{u,v})}\right),
        \end{aligned}
    \end{equation}
    where $\LS_1$, $\LS_w$ stand for the norm systems with respect to weight $1$ and $w$. Moreover, for $n\geq 3$, $X=H_{u}^v,\Hb_v^u,\R_{u,v}$, one obtains
    \begin{equation}\label{Calcu_weighted_norm_Sobolev}
        \begin{aligned}
            \sum_{j\leq n}\lnm \nabla^j(\psi_1\psi_2)\rnm_{\LS_w^2(X)}\lesssim &\sup_{S_{u,v}\subset X}\left({\sum_{i\leq n-1}\lnm \nabla^i\psi_1\rnm_{\LS_1^2(S_{u,v})}}\right)\sum_{j\leq n}\lnm \nabla^j\psi_2\rnm_{\LS_w^2(X)}\\
            &+\sup_{S_{u,v}\subset X}\left({\sum_{i\leq n-1}\lnm \nabla^i\psi_2\rnm_{\LS_1^2(S_{u,v})}}\right)\sum_{j\leq n}\lnm \nabla^j\psi_1\rnm_{\LS_w^2(X)}.
        \end{aligned}
    \end{equation}
\end{lemma}
\begin{proof}
    With weight $w$, we have inequality:
    \begin{equation}
        \begin{aligned}
            \sum_{j\leq n}\lnm \nabla^j(\psi_1\psi_2)\rnm_{\LS_w^2(S_{u,v})}\lesssim & \sum_{i\leq n-1}\lnm \nabla^i\psi_1\rnm_{\LS_1^2(S_{u,v})}\sum_{j\leq n}\lnm \nabla^j\psi_1\rnm_{\LS_w^2(S_{u,v})}\\
            &+\sum_{i\leq n-1}\lnm \nabla^i\psi_2\rnm_{\LS_1^2(S_{u,v})}\sum_{j\leq n}\lnm \nabla^j\psi_2\rnm_{\LS_w^2(S_{u,v})}\\
            \lesssim & \sup_{S_{u,v}\subset X}\left(\sum_{i\leq n-1}\lnm \nabla^i\psi_1\rnm_{\LS_1^2(S_{u,v})}\right)\sum_{j\leq n}\lnm \nabla^j\psi_1\rnm_{\LS_w^2(S_{u,v})}\\
            &+\sup_{S_{u,v}\subset X}\left(\sum_{i\leq n-1}\lnm \nabla^i\psi_2\rnm_{\LS_1^2(S_{u,v})}\right)\sum_{j\leq n}\lnm \nabla^j\psi_2\rnm_{\LS_w^2(S_{u,v})}.
        \end{aligned}
    \end{equation}
    After integration, desired estimates for $H_u^v,\Hb_v^u,\R_{u,v}$ are derived.
\end{proof}

The transport estimates can be written as below.
\begin{lemma}{\rm (Transport Estimates)}
    Assume that $s(\psi)\leq 0$, $C_1^{-1}\leq\frac{\Omega^2}{(-u)^\k}\leq C_1$, $\partial_u\log w\leq 0$, and $\partial_v\log w\leq 0$. The following estimates hold.
    \begin{equation}\label{transport_3}
            \|\psi\|_{\LS_w^2(S_{u,v})}^2\leq \|\psi\|_{\LS_w^2(S_{u_0(v),v})}^2+\|\Omega\nabla_3{\psi},{\psi}|\nabla b|^{1/2}\|_{\LS_w^2(\Hb_v^u)}^2,
    \end{equation}
    \begin{equation}\label{transport_4}
        \|\psi-\psi|_{v_0(u)}\|_{\LS_w^2(S_{u,v})}^2\lesssim \frac{v}{(-u)^{1-\k}}\|\Omega^{-1}\nabla_4{\psi}, \Omega^{-1}\chi\cdot\psi\|_{\LS_w^2(H_u^v)}^2.
    \end{equation}
\end{lemma}
\begin{proof}
    
By integration, it follows that 
\begin{align*}
    \|\psi\|_{\LS_w^2(S_{u,v})}^2 = & \|\psi\|_{\LS_w^2(S_{u_0,v})}^2+\int_{u_0}^u\int_{\SS} \partial_u |\psi|_g^2d\V du^\prime\\
    =& \|\psi\|_{\LS_w^2(S_{u_0,v})}^2+\int_{u_0}^u\int_{\SS} (\Omega e_3-b^Ae_A)\left( (-u)^{2-2s(\psi)}|\psi|_g^2\cdot w(u,v)\right)d\V du^\prime\\
    =& \|\psi\|_{\LS_w^2(S_{u_0,v})}^2+\int_{u_0}^u\int_{\SS} 2\langle\Omega\nabla_3{\psi}, \psi\rangle (-u)^{2-2s(\psi)} w(u,v) d\V du^\prime\\
    &+\int_{u_0}^u\int_{\SS} (-u)^{1-2s(\psi)}|{\psi}|_g^2 w(-(1-2s(\psi))+(-u)\partial_u \log w) d\V du^\prime\\
    &+\int_{u_0}^u\int_{\SS} (-u)^{2-2s(\psi)}|\psi|_g^2\cdot w(u,v)\partial_A b^A d\V du^\prime\\
    \leq & \|\psi\|_{\LS_w^2(S_{u_0,v})}^2+\int_{u_0}^u\int_{\SS} |\Omega\nabla_3{\psi}| (-u)^{1-2s(\Omega\nabla_2\psi)} w(u,v) d\V du^\prime\\
    &+\int_{u_0}^u\int_{\SS} (-u)^{1-2s(\psi)}|{\psi}|_g^2 w\cdot (-(1-2s(\psi))+1+(-u)\partial_u \log w) d\V du^\prime\\
    &+\int_{u_0}^u\int_{\SS} (-u)^{2-2s(\psi)}|\psi|_g^2\cdot w(u,v)|\nabla b| d\V du^\prime,
\end{align*}
and 
\begin{align*}
    \|\psi-\psi|_{v_0(u)}\|_{\LS_w^2(S_{u,v})}^2 \lesssim & \int_{\SS} (-u)^{2-2s(\psi)}w(u,v)\left|\int_{v_0(u)}^v\Lie_v\psi(u,v^\prime)\right|^2\\
    \lesssim & \int_{\SS}(-u)^{2-2s(\psi)}\int_0^v \frac{1}{(-u)^{1-\k}}\int_{v_0(u)}^v(-u)^{1+\k}w(u,v^\prime)\left|\Omega^{-1}\Lie_4{\psi}\right|^2\\
    \lesssim & \frac{v}{(-u)^{1-\k}}\lnm \Omega^{-1}\nabla_4{\psi},\Omega^{-1}\chi\cdot\psi\rnm^2_{\LS_w^2(H_u^v)}.
\end{align*}

\end{proof}
\begin{remark}
    If $\lnm \Omega^{-1}\chi\rnm_{\LS_1^\infty(S_{u,v})}\lesssim 1$, which will be direct from bootstrap assumption in later sections, and $\psi|_{v_0(u)}=0$, \eqref{transport_4} can be written as
    \begin{equation}
        \begin{aligned}
            \|\psi\|_{\LS_w^2(S_{u,v})}^2\lesssim & \frac{v}{(-u)^{1-\k}}\|\Omega^{-1}\nabla_4{\psi}|_{v_0(u)}\|_{\LS_w^2(H_u^v)}^2+\frac{v}{(-u)^{1-\k}}\int_{v_0(u)}^v\frac{1}{(-u)^{1-\k}}\|\psi\|_{\LS_w^2(S_{u,v^\prime})}^2.
        \end{aligned}
    \end{equation}
    The last term is ignorable with the smallness of $\frac{v}{(-u)^{1-\k}}$ in region I.

\end{remark}

For Bianchi pair, we next establish the energy estimates.
\begin{lemma}{\rm (Energy Estimates)}
    Consider equations 
    \begin{equation}
        \left\{
            \begin{aligned}
                &\Omega\nabla_3{\Phi_1}+(\lambda\Omega\tr\chib-\Lambda\Omega\omegab){\Phi_1}-\mathcal{D}{\Phi_2}=F,\\
                &\Omega^{-1}\nabla_4{\Phi_2}+^*\mathcal{D}{\Phi_1}=G,
            \end{aligned}
        \right.
    \end{equation}
    where $\mathcal{D}$ is a first order differential operator and ${}^*\mathcal{D}$ is its Hodge dual. 
    Suppose that $$\lnm\frac{(-u)^2}{\sqrt{\det g(u,v)}}-1\rnm_{L^2(\SS)}\ll 1,\ \lnm\nabla\left(\frac{(-u)^2}{\sqrt{\det g(u,v)}}\right)\rnm_{L^\infty(\SS)}\lesssim \epsilon$$ and 
    \begin{equation}
        \lnm (-u)\dv b, (-u)\Omega\tr\chib+2,(-u)4\Omega\omegab-\k\rnm_{L^\infty(\SS)}\sim o_{\epsilon,\epsilon_1}(1), 
    \end{equation}
    where $f\sim o_{\epsilon,\epsilon_1}(1)$ means $\lim_{\epsilon,\epsilon_1\rightarrow 0}f=0$.
    Then it follows that 
    \begin{equation}\label{Energy Estimates}
        \begin{split}
            &\lnm {\Phi_1}\rnm_{\LS_w^2(H_{u}^v)}^2+\lnm {\Phi_2}\rnm_{\LS_w^2(\Hb_{v}^u)}^2\\
            \lesssim & \left(o_{\epsilon,\epsilon_1}(1)-1+2s(\Phi_1)-\k+\partial_u\log w+4\lambda+\Lambda\k/2\right)\lnm {\Phi_1}\rnm^2_{\LS_w^2(\mathcal{R}_{u,v})}\\
            &+\lnm {\Phi_1}\rnm_{\LS_w^2(\mathcal{R}_{u,v})}\cdot \lnm F\rnm_{\LS_w^2(\mathcal{R}_{u,v})}+\lnm G \rnm_{\LS_w^2(\mathcal{R}_{u,v})}^2+\lnm {\Phi_1}\rnm_{\LS_w^2(H_{-1}^v)}^2+\lnm {\Phi_2}\rnm_{\LS_w^2(\Hb_{0}^u)}^2.
        \end{split}
    \end{equation}
    If $o_{\epsilon,\epsilon_1}(1)-1+2s(\Phi_1)-\k+\partial_u\log w+4\lambda+\Lambda\k/2\leq -\delta<0$, the following inequality holds 
    \begin{equation}
        \begin{split}
            &\lnm {\Phi_1}\rnm_{\LS_w^2(H_{u}^v)}^2+\lnm {\Phi_2}\rnm_{\LS_w^2(\Hb_{v}^u)}^2
            \lesssim \lnm F, G \rnm_{\LS_w^2(\mathcal{R}_{u,v})}^2+\lnm {\Phi_1}\rnm_{\LS_w^2(H_{-1}^v)}^2+\lnm {\Phi_2}\rnm_{\LS_w^2(\Hb_{0}^u)}^2.
            \end{split}
    \end{equation}
\end{lemma}
\begin{proof}
    First notice that 
    \begin{equation}
        \begin{aligned}
            &\int_{\SS}(-u)^2 \langle {\Phi_1},\D {\Phi_2}\rangle_gd\V\\
            =&\int_{S_{u,v}} \langle {\Phi_1},\D {\Phi_2}\rangle_g\frac{(-u)^2}{\sqrt{\det g(u,v)}}d{\rm Vol}_g\\
            =&\int_{\SS} (-u)^2\langle {}^*\D{\Phi_1}, {\Phi_2}\rangle_gd\V+\int_{S_{u,v}} \langle{}^*\D\left(\frac{(-u)^2}{\sqrt{\det g(u,v)}}\right)\cdot{\Phi_1},{\Phi_2} \rangle d{\rm Vol}_g,
        \end{aligned}
    \end{equation}
    so we obtain 
    \begin{equation*}
        \left|\int_{\SS}(-u)^2 \langle {\Phi_1},\D {\Phi_2}\rangle_gd\V-\int_{\SS} (-u)^2\langle {}^*\D{\Phi_1}, {\Phi_2}\rangle_gd\V\right|\lesssim \epsilon \int_{\SS}(-u)^2 \left|{\Phi_1}\right|\cdot\left|{\Phi_2}\right|d\V
    \end{equation*}
    It therefore follows that
    \begin{align*}
        &\lnm {\Phi_1}\rnm_{\LS^2(H_{u}^v)}^2-\lnm {\Phi_1}\rnm_{\LS^2(H_{-1}^v)}^2+\lnm {\Phi_2}\rnm_{\LS^2(\Hb_{v}^u)}^2-\lnm {\Phi_2}\rnm_{\LS^2(\Hb_{0}^u)}^2\\
        \leq&\int_{-1}^u\int_0^v \int_{\SS}\nabla_A b^A (-u^\prime)^{1-2s(\Phi_1)}w\Omega^2\left|{\Phi_1}\right|^2\\
        &+\int_{-1}^u\int_0^v\int_{\SS} (-u^\prime)^{1-2s(\Phi_1)}w\Omega^2\left|{\Phi_1}\right|^2\left(-\frac{1-2s(\Phi_1)}{-u}+\partial_u\log w-2\lambda \Omega\tr\chib+(2\Lambda-4)\Omega\omegab\right)\\
        &+\int_{-1}^u\int_{-1}^v\int_{\SS} (-u^\prime)^{1-2s(\Phi_2)}w\Omega^2\left|{\Phi_2}\right|^2\left(-\frac{2\k-3\delta}{1-\k}\frac{1}{v\Omega^2}-2\mu\Omega^{-1}\tr\chi\right)  \\
        &+2\int_{-1}^u\int_0^v\int_{\SS} (-u^\prime)^{1-2s(\Phi_1)}w\Omega^2\langle{\Phi_1},F \rangle \\
        &+2\int_0^v\int_{-1}^u\int_{\SS} (-u^\prime)^{1-2s(\Phi_2)}w\Omega^2\langle {\Phi_2}, G \rangle +C_0\epsilon \lnm {\Phi_1}\rnm_{\LS^2(\R_{u,v})}\lnm {\Phi_2}\rnm_{\LS^2(\R_{u,v})} \\
        \leq & \left(C_0\epsilon-1+2s(\Phi_1)-\k+\partial_u\log w+4\lambda+\Lambda\k/2\right)\lnm {\Phi_1}\rnm^2_{\LS^2(\mathcal{R}_{u,v})}\\
        &+\lnm {\Phi_1}\rnm_{\LS^2(\mathcal{R}_{u,v})}\cdot\lnm F \rnm_{\LS^2(\mathcal{R}_{u,v})}+\lnm {\Phi_2}\rnm_{\LS^2(\mathcal{R}_{u,v})}\cdot\lnm G \rnm_{\LS^2(\mathcal{R}_{u,v})}.
    \end{align*}
\end{proof} 

We next state the elliptic estimates for difference.
\begin{lemma}{\rm (Elliptic Estimates for Hodge System)}
    Let $\phi$ be a totally symmetric $r+1$ covariant tensor on $S_{u,v}$ and correspondingly $\phi^{(c)}$ on $S^{(c)}_{u,v}$. We denote $\ol{\phi}=\phi-\phi^{(c)}$. Suppose that 
    \[\dv\phi=f,\quad \cl\phi=l,\quad \tr\phi=h.\]
    Assume that $C_1^{-1}\leq (-u)^2 K_{\S}\leq C_1$, $C_1^{-1}g^{\SS}_{AB}\leq (-u)^{-2}g_{AB}\leq C_1 g^{\SS}_{AB}$.
    Moreover, assume that $\phi^{(c)}$ is spherically symmetric.
    Then it holds that 
\begin{equation}\label{Ellip Est}
    \begin{split}
        \|\nabla^i&\ol\phi\|_{\LS_w^2(\S)}
        \leq   C\left(C_1,\|K\|_{\LS_1^2(\S)},\|\nabla K\|_{\LS_1^4(\S)}\right)\cdot \sum_{j=0}^{i-1} \|\nabla^j\left(\ol{f},\ol{l},\ol{h},\ol\phi,\ol{g}\right)\|_{\LS_w^2(\S)} .
    \end{split}
\end{equation}
\end{lemma}

\begin{proof}
    Applying difference gives that 
    \[\dv \ol{\phi} = \ol{f},\ \cl\ol\phi = \ol{l},\ \tr\ol\phi = \ol{h}-\ol{g}\cdot\phi^c.\]
    Standard Hodge system estimates give the inequality.
\end{proof}

\begin{remark}
    If we denote $\lnm \psi\rnm^4_{\LS^4_w(S_{u,v})}=\int_{\SS}(-u)^{4-4s(\psi)}\left|\psi\right|^4 w^2d\V$, then one can derive Sobolev inequality
    \begin{equation}
        \lnm \psi\rnm_{\LS^\infty_w(S_{u,v})}\lesssim \lnm \psi\rnm_{\LS^2_w(S_{u,v})}+\lnm \nabla\psi\rnm_{\LS^4_w(S_{u,v})},
    \end{equation}
    and trace formulae
        \begin{equation}
        \begin{aligned}
            &\|{\psi}\|_{\LS_w^4(S_{u,v})}^4-\|{\psi}\|_{\LS_w^4(S_{u,v_0(u)})}^4\\
        &\qquad\lesssim \|\Omega\nabla_4{\psi}\|_{\LS_w^2(H_u^v)}\left(\|{\psi}\|_{\LS_w^2(H_u^v)}+\|\nabla{\psi}\|_{\LS_w^2(H_u^v)}\right)\cdot\sup_{v_0(u)\leq v^\prime\leq v}\|{\psi}\|^2_{\LS_w^4(S_{u,v^\prime})},\\
            &\|{\psi}\|_{\LS_w^4(S_{u,v})}^4-\|{\psi}\|_{\LS_w^4(S_{u_0(v),v})}^4-\int_{u_0(v)}^u\int_{\SS}(-u)^{4-4s(\psi)}|{\psi}|^4w^2 \partial_Ab^A \,d\V \, du^\prime\\
        &\qquad\lesssim \|\Omega\nabla_3{\psi}\|_{\LS_w^2(\Hb_v^u)}\left(\|{\psi}\|_{\LS_w^2(\Hb_v^u)}+\|\nabla{\psi}\|_{\LS_w^2(\Hb_v^u)}\right)\cdot\sup_{u_0(v)\leq u^\prime\leq u}\|{\psi}\|^2_{\LS_w^4(S_{u^\prime,v})}.
        \end{aligned}
    \end{equation}

    These inequalities are highly useful when addressing the CIVP under low-regularity conditions. Since we will not actively employ the $\LS_w^4(S)$-norm in the subsequent analysis, we conservatively omit the detailed proof here.
\end{remark} 

\subsection{Commutation Formulae}
Lastly, we rigorously establish the fundamental commutation formulae below.

\begin{lemma}\label{LEMMA_COMMUTATION_FORMULAE_NABLA}{\rm (Commutation Formulae)}
    The following relation holds 
    \begin{equation}\label{Exact_Comm_Formula}
        \begin{aligned}
        \left[\Omega \nabla_4,  \nabla_A\right] \phi_{B_1 \cdots B_k}
        =&\sum_{i=1}^k\left(-\nabla_{B_i}(\Omega\chi)_A^C+\nabla^C(\Omega\chi)_{AB_i}\right) \phi_{B_1 \cdots \hat{B}_i C \cdots B_k}-\Omega \chi_A{ }^C  \nabla_C \phi_{B_1 \cdots B_k}, \\
        \left[\Omega \nabla_3, \nabla_A\right] \phi_{B_1 \cdots B_k}
        = &\sum_{i=1}^k\left(-\nabla_{B_i}(\Omega\chib)_A^C+\nabla^C(\Omega\chib)_{AB_i}\right) \phi_{B_1 \cdots \hat{B}_i C \cdots B_k}-\Omega \underline{\chi}_A{ }^C \nabla_C\phi_{B_1 \cdots B_k} .
        \end{aligned}
        \end{equation}
        Let $\nabla_4\phi=G_0$. Write $G_i=\nabla_4\nabla^i\phi$, then the following relation holds  
        \begin{equation}\label{Comm_Formula_D4,nab}
            \Omega G_i=\nabla^i(\Omega G_0)-\frac{i}{2}\Omega\tr\chi\nabla^i\phi-i\Omega\chih\cdot \nabla^i\phi+\sum_{\substack{i_1+i_2=i\\ i_2\leq i-1}}\nabla^{i_1}\Omega\chi\nabla^{i_2}\phi.
        \end{equation}
        Similarly, suppose $\nabla_3\phi=F_0$ and $\nabla_3\nabla^i\phi=F_i$, then we find that 
        \begin{equation}\label{Comm_Formula_D3,nab}
            \Omega F_i=\nabla^i(\Omega F_0)-\frac{i}{2}\Omega\tr\chib\nabla^i\phi-i\Omega\chibh\cdot \nabla^i\phi+\sum_{\substack{i_1+i_2=i\\ i_2\leq i-1}}\nabla^{i_1}\Omega\chib\nabla^{i_2}\phi.
        \end{equation}
    We also have the commutation formulae to operators $\Omega^{-1}\nabla_3$ and $\Omega^{-1}\nabla_4$.
    \begin{equation}
        \begin{aligned}
            \Omega^{-1}F_i=&\sum_{i_1+i_2+i_3=i}\nabla^{i_1}(\eta+\etab)^{i_2}\nabla^{i_3}(\Omega^{-1}F_1)-\frac{i}{2}\Omega^{-1}\tr\chib\nabla^i\phi -i\Omega^{-1}\chibh\nabla^i\phi\\
            &+\sum_{\substack{i_1+i_2+i_3+i_4=i\\i_4\leq i-1}}\nabla^{i_1}(\eta+\etab)^{i_2}\nabla^{{i_3}}(\Omega^{-1}\chib)\nabla^{i_4}\phi,
        \end{aligned}
    \end{equation}
    \begin{equation}
        \begin{aligned}
            \Omega^{-1}G_i=&\sum_{i_1+i_2+i_3=i}\nabla^{i_1}(\eta+\etab)^{i_2}\nabla^{i_3}(\Omega^{-1}G_1)-\frac{i}{2}\Omega^{-1}\tr\chi\nabla^i\phi -i\Omega^{-1}\chih\nabla^i\phi\\
            &+\sum_{\substack{i_1+i_2+i_3+i_4=i\\i_4\leq i-1}}\nabla^{i_1}(\eta+\etab)^{i_2}\nabla^{{i_3}}(\Omega^{-1}\chi)\nabla^{i_4}\phi.
        \end{aligned}
    \end{equation}
\end{lemma}
\begin{proof}
    As an illustrative example, we prove the formula for $\left[\Omega\nabla_4,\nabla_A\right]\phi_B$. The derivations for general $(0,k)$-tensors and for the conjugate commutator $\left[\Omega\nabla_3,\nabla_A\right]$ proceed in a strictly analogous manner. We compute directly:
    \begin{equation}
    \begin{aligned}
        &\left[\Omega \nabla_4,\nabla_A\right]\phi_B=\Omega\nabla_4\nabla_A\phi_B-\nabla_A(\Omega \nabla_4\phi)_B \\
        =&\Omega e_4\left(\nabla_A\phi_B\right)-\nabla^C\phi_B \langle D_{\Omega e_4} e_A,e_C\rangle-\nabla_A\phi^C\langle e_C,D_{\Omega e_4}e_B\rangle\\
        &-e_A\left(\Omega \nabla_4\phi_B\right)+\langle \Omega\nabla_4\phi,\nabla_A e_B\rangle\\
        =& \Omega e_4(e_A(\phi_B))-\Omega e_4\left(\phi_C\Gamma_{AB}^C\right)-e_A\left(\Omega e_4(\phi_B)\right)+e_A\langle\phi, D_{\Omega e_4}e_B\rangle\\
        &-\nabla^C\phi_B\left(\Omega\chi_{AC}+\langle\left[\Omega e_4,e_A\right],e_C\rangle\right)-\nabla_A\phi^C\langle e_C, D_{\Omega e_4}e_B\rangle\\
        &+\langle \Omega D_4\phi,\nabla_A e_B\rangle.
    \end{aligned}
\end{equation}
For the first line, we obtain
\begin{equation}
    \begin{aligned}
        &\Omega e_4(e_A(\phi_B))-\Omega e_4\left(\phi_C\Gamma_{AB}^C\right)-e_A\left(\Omega e_4(\phi_B)\right)+e_A\langle\phi, D_{\Omega e_4}e_B\rangle\\
        =&\underline{\left[\Omega e_4,e_A\right]\phi_B}_{(1)} -\underline{\Omega e_4\left(\phi_C\Gamma^C_{AB}\right)}_{(2)}+\underline{\langle D_A\phi,D_{\Omega e_4}e_B\rangle}_{(3)} +\underline{\langle \phi,D_A D_{\Omega e_4}e_B\rangle}_{(4)},
    \end{aligned}
\end{equation}
and the second line can be written as
\begin{equation}
    \begin{aligned}
        &-\nabla^C\phi_B\Omega\chi_{AC}-\langle\left[\Omega e_4,e_A\right],\nabla\phi_B\rangle-\nabla_A\phi^C\langle e_C, D_{\Omega e_4}e_B\rangle\\
        =&-\underline{\nabla^C\phi_B\Omega\chi_{AC}}_{(5)}-\underline{\nabla_{\left(\left[\Omega e_4,e_A\right]\right)^{T}}\phi_B}_{(6)} -\underline{\langle\nabla_A\phi,D_{\Omega e_4}e_B\rangle}_{(7)}.
    \end{aligned}
\end{equation}
Here, for an arbitrary vector field $X$, $X^T$ denotes its $S_{u,v}$-tangent projection and $X^V$ its vertical component.
Since the Gauss formula yields $D_A e_B-\nabla_A e_B=\frac{1}{2}\chi_{AB}e_3+\frac{1}{2}\chib_{AB}e_4$,
the third line evaluates to
\begin{equation}
    \begin{aligned}
        &\langle \Omega D_4\phi,\nabla_A e_B\rangle\\
        =&\Omega e_4\langle \phi,D_A e_B\rangle-\langle\phi,D_{\Omega e_4}D_A e_B\rangle-\frac{1}{2}\Omega\chi_{AB}\langle D_4\phi,e_3\rangle-\frac{1}{2}\Omega\chib_{AB}\langle D_4\phi,e_4\rangle\\
        =&\underline{\Omega e_4\left(\phi_C\Gamma^C_{AB}\right)}_{(8)}-\underline{\langle \phi, D_{\Omega e_4}D_{A}e_B\rangle}_{(9)}+\underline{\Omega\chi_{AB}\phi^C\etab_C}_{(10)}.
    \end{aligned}
\end{equation}
For those terms that do not contribute to the explicit derivatives of $\phi$, we successfully recover the necessary curvature components via the standard Ricci identity:
\begin{equation}
    \begin{aligned}
        &(1)+(4)+(6)+(9)\\
        =&\left[\Omega e_4,e_A\right]\phi_B-\nabla_{\left(\left[\Omega e_4,e_A\right]\right)^T}\phi_B+\langle\phi,D_{e_A}D_{\Omega e_4} e_B\rangle-\langle\phi,D_{\Omega e_4}D_A e_B\rangle\\
        =& \langle \phi,D_{\left(\left[\Omega e_4,e_A\right]\right)^T}e_B\rangle+\langle\phi,D_{e_A}D_{\Omega e_4} e_B\rangle-\langle\phi,D_{\Omega e_4}D_A e_B\rangle\\
        =&\phi^C \Omega R_{4ABC},
    \end{aligned}
\end{equation}
where we have utilized the fact that $$\left(\left[\Omega e_4,e_A\right]\right)^V=\left(\Omega D_4 e_A\right)^V-\left(D_A(\log\Omega) \Omega e_4+\Omega D_A e_4\right)=0,$$ which confirms that the commutator $\left[\Omega e_4,e_A\right]$ is purely tangent to $S_{u,v}$.
We can match the remaining geometric terms and obtain 
\begin{equation}
    \begin{aligned}
        (2)+(8)+(3)+(7)=&\langle D_A\phi,\left(D_{\Omega e_4} e_B\right)^V\rangle=\langle D_A\phi, \etab_B\Omega e_4\rangle\\
        =&\etab_B \phi^C\langle D_A e_C,\Omega e_4\rangle=-\Omega\chi_{AC}\etab_B \phi^C.
    \end{aligned}
\end{equation}
Gathering these results, we summarize:
\begin{equation}
    \left[\Omega \nabla_4,\nabla_A\right]\phi_B+\Omega\chi_{AC}\nabla^C\phi_B=\phi^C \left(\Omega R_{4ABC}+\Omega\chi_{AB}\etab_C-\Omega\chi_{AC}\etab_B\right).
\end{equation}
For the full second fundamental form $h_{AB}=\frac{1}{2}\Omega\chi_{AB}\Omega^{-1}e_3+\frac{1}{2}\Omega\chib_{AB}\Omega^{-1}e_4$, the associated Codazzi equation reads
\begin{equation}
    \langle\nabla_C h_{AB}-\nabla_B h_{AC},\Omega e_4 \rangle=\langle R(e_C,e_B)e_A,\Omega e_4\rangle,
\end{equation}
The left-hand side reduces to $-\nabla_C (\Omega\chi)_{AB}+\Omega\chi_{AB}\etab_C+\nabla_B(\Omega\chi)_{AC}-\Omega\chi_{AC}\etab_B$ while the corresponding right-hand side is simply $-\Omega R_{4ABC}$, completing the proof.

\end{proof}

We also derive the commutators with Lie-derivatives.
\begin{corollary}\label{CORROLLARY_COMMUTATION_FORMULA_LIE}
    The commutators $\left[\Lie_{\Omega e_3},\nabla_A\right]$ and $\left[\Lie_{\Omega e_4},\nabla_A\right]$ have the expressions below:
    \begin{equation}\label{Commutation_Formulae_Lie4}
        \left[\Lie_{\Omega e_4},\nabla_A\right]\phi_{B_1\cdots B_k}=\sum_{i=1}^k\left(-\nabla_{B_i}(\Omega\chi)_A^C+\nabla^C(\Omega\chi)_{AB_i}-\nabla_A(\Omega\chi)_{B_i}^C\right) \phi_{B_1 \cdots \hat{B}_i C \cdots B_k},
    \end{equation}
    \begin{equation}
        \left[\Lie_{\Omega e_3},\nabla_A\right]\phi_{B_1\cdots B_k}=\sum_{i=1}^k\left(-\nabla_{B_i}(\Omega\chib)_A^C+\nabla^C(\Omega\chib)_{AB_i}-\nabla_A(\Omega\chib)_{B_i}^C\right) \phi_{B_1 \cdots \hat{B}_i C \cdots B_k}.
    \end{equation}
    More schematically, we can express the higher-order commutator as $$\left[\Lie_{\Omega e_4},\nabla^j\right]\phi\sim\sum_{i_1+i_2=j}\nabla^{i_1}(\Omega\chi)\nabla^{i_2}\phi.$$
\end{corollary}
\begin{proof}
    Applying $\Omega\nabla_4\phi_{B_1\cdots B_k}=\Lie_{\Omega e_4}\phi_{B_1\cdots B_k}-\sum_{i=1}^k \Omega\chi_{B_i}^C\phi_{\cdots C\cdots}$ to $\left[\Omega \nabla_4,\nabla_A\right]$, we find that 
\begin{equation}
    \begin{aligned}
        &\left[\Omega \nabla_4,\nabla_A\right]\phi_{B_1\cdots B_k}\\
        =& \Lie_{\Omega e_4}\nabla_A\phi_{B_1\cdots B_k}-\Omega\chi_{A}^C\nabla_C\phi_{B_1\cdots B_k} -\sum_{i=1}^k\Omega\chi_{B_i}^C\nabla_A\phi_{\cdots B_{i-1}CB_{i+1}\cdots}\\
        &-\nabla_A\left(\Lie_{\Omega e_4}\phi\right)_{B_1\cdots B_k}+\sum_{i=1}^k\nabla_A\left(\Omega\chi_{B_i}^C\phi_{\cdots B_{i-1}CB_{i+1}\cdots}\right)\\
        =&\left[\Lie_{\Omega e_4},\nabla_A\right]\phi_{B_1\cdots B_k}-\Omega\chi_A^C\nabla_C\phi_{B_1\cdots B_k}+\sum_{i=1}^k\nabla_A(\Omega\chi)_{B_i}^C\cdot \phi_{\cdots B_{i-1}CB_{i+1}\cdots}.
    \end{aligned}
\end{equation}
\end{proof}
Finally, we explicitly compute the associated commutator for the angular derivatives operating on $S_{u,v}$.
\begin{lemma}
    Let $\phi$ be a tensor on $S_{u,v}$, then for $\mathcal{D}\in\{\dv,\cl,{}^*\nabla,\nabla,\nabla\hat\otimes\}$, this yields 
    \begin{equation}
        \left[\mathcal{D},\nabla^i\right]\phi=\sum_{i_1+i_2+i_3 = i-1}\nabla^{i_1}K^{i_2+1}\nabla^{i_3}\phi.
    \end{equation}
\end{lemma}

\section{Analysis in Region I}\label{Estimates in Region I}
In Region I, we conduct the coordinate change $(\hat{u},\hat{v},\theta)=(u,\frac{1}{1-\k}v^{1-\k},\theta)$. By direct computation, we obtain 
\begin{lemma}
    Suppose that $\hat{\psi}$ and $\psi$ denote geometric quantities in the new and old coordinate systems, respectively. Then, for any spacetime, we obtain
    \[\hat\Omega D_{\hat{e}_3}=\Omega D_{e_3},\quad \hat\Omega^{-1} D_{\hat{e}_4}=\Omega^{-1}D_{e_4},\]
    \[\hat{\Omega}^2=v^\k\Omega^2,\quad \hat\Omega\hat\omegab=\Omega\omegab,\quad \hat{\Omega}^{-1}\hat\omega=\frac{v\Omega\omega-\k/4}{v\Omega^2},\]
    \[\hat{g}=g,\quad \hat\Omega\hat\chib=\Omega\chib,\quad \hat{\Omega}^{-1}\hat\chi=\Omega^{-1}\chi,\quad \hat\eta=\eta,\quad \hat\etab=\etab.\]
\end{lemma}
For convenience, we omit $\hat{\cdot}$ for the new gauge and use $\psi^{old}$ to refer to the old gauge when it appears. 
This gauge change absorbs the leading singular dependence of the outgoing parameter into the lapse. In the estimates below this is why the same geometric quantity can be measured in scale-invariant norms without repeatedly translating between the old $v$-power and the new $v$-coordinate.

\subsection{Approximations and formulae of difference}
We adopt the notation $$\phi_0(u,v):=\lim_{v^\prime\rightarrow 0}\phi(u,v^\prime),$$
and $\ol{\psi}=\psi-\psi^c$, where $(\cdot)^c$ represents the corresponding quantity in Christodoulou's naked singularity solution.
Recall the definition of $\ot{\Omega^{-1}\chih}$ and $\ot{\Omega^{-1}e_4\phi}$, 
\begin{equation}
    \begin{aligned}
        &\left({(\Omega\nabla_3)_0}+\frac{1}{2}({\Omega\tr\chib})_0-4({\Omega\omegab})_0\right)\ot{\Omega^{-1}\chih}= \left(\nabla\hat\otimes\eta+\eta\hat\otimes\eta+\frac{1}{2}\nabla\phi\hat\otimes\nabla\phi-\frac{1}{2}\Omega^{-1}\tr\chi \Omega\chibh\right)_0,\\
        &\left((\Omega\nabla_3)_0+\frac{1}{2}({\Omega\tr\chib})_0+\frac{1}{2}\left[({\Omega\tr\chib})^c\right]_0^v-4({\Omega\omegab})_0-4\left[(\Omega\omegab)^c\right]_0^v\right)\ot{\Omega^{-1}e_4\phi}\\
        &\qquad\qquad\qquad\qquad= \left(\dv\nabla\phi-\frac{1}{2}\Omega^{-1}\tr\chi\Omega e_3\phi+2\eta\cdot\nabla\phi\right)_0 -\left[\left(\frac{1}{2}\Omega^{-1}\tr\chi\Omega e_3\phi\right)^c\right]_0^v,
    \end{aligned}
\end{equation}
where $\left[\psi\right]_0^v=\psi-\psi_0$. 
We define the approximation for the curvature component $\alpha$ by
\begin{equation}
    \ot{\Omega^{-2}\alpha}_{AB}=-(\Omega_0)^{-2}\Lie_{\partial_v}\left(\ot{\Omega^{-1}\chih}\right)_{AB},
\end{equation}
and $\ot{\Omega^{-1}\beta^r}_A$ by 
\begin{equation}\label{R_I_Def_ot_beta}
    \ot{\Omega^{-1}\beta^r}_A=-{\dv_0}\ot{(\Omega^{-1}\chih)}_A+\frac{1}{2}\nabla_A{(\Omega^{-1}\tr\chi)_0}-\left({\eta^B}\right)_0\ot{\Omega^{-1}\chih}_{AB}+\frac{1}{2}\left({\Omega^{-1}\tr\chi}{\eta_A}\right)_0.
\end{equation}
The lemma below reveals the key properties of $\ot{\Omega^{-2}\alpha}$ and $\ot{\Omega^{-1}\beta^r}$, which will be used to derive the equations for the differences.
\begin{lemma}
    The following equations hold for $\ot{\Omega^{-2}\alpha}$ and $\ot{\Omega^{-1}\beta^r}$:
    \begin{equation}\label{Eq_for_ot_alpha}
        \left(\Omega\nabla_3+\frac{1}{2}\Omega\tr\chib-8\Omega\omegab\right)_0\ot{\Omega^{-2}\alpha}_{AB}=0,
    \end{equation}
    and 
    \begin{equation}\label{Eq_for_ot_beta}
        (\Omega_0)^{-2}\Lie_v \ot{\Omega^{-1}\beta^r}_A=\dv_0\ot{\Omega^{-2}\alpha}_A+(2\eta^B+\etab^B)_0\ot{\Omega^{-2}\alpha}_{AB}.
    \end{equation}
\end{lemma}
\begin{proof}
    Because $[(\Omega \nabla_3)_0,\Lie_v]=0,[\psi_0,\Lie_v]=0$, we obtain
    \begin{equation}
        \begin{aligned}
            &\left(\Omega\nabla_3+\frac{1}{2}\Omega\tr\chib-4\Omega\omegab\right)_0\ot{\Omega^{-2}\alpha}\\
=&(\Omega\nabla_3)_0\log(\Omega_0)^{-2}\ot{\Omega^{-2}\alpha}-(\Omega_0)^{-2}\Lie_v\left(\left(\Omega\nabla_3+\frac{1}{2}\Omega\tr\chib-4\Omega\omegab\right)_0\ot{\Omega^{-1}\chih}\right)\\
=&4(\Omega\omegab)_0\ot{\Omega^{-2}\alpha}.
        \end{aligned}
    \end{equation}
    For $\ot{\Omega^{-1}\beta^r}$, we obtain 
    \begin{equation}
        \begin{aligned}
            &(\Omega_0)^{-2}\Lie_v \ot{\Omega^{-1}\beta^r}_A\\
            =&\left(\eta_0+\dv_0\right)\left(-(\Omega_0)^{-2}\Lie_v \ot{\Omega^{-1}\chih}\right)-\left(\nabla\log\Omega^{-2}\right)_0(\Omega_0)^{-2}\Lie_v\left(-\ot{\Omega^{-1}\chih}\right)\\
            =& \dv_0\ot{\Omega^{-2}\alpha}_A+(2\eta^B+\etab^B)_0\ot{\Omega^{-2}\alpha}_{AB}.
        \end{aligned}
    \end{equation}
\end{proof}
For $\ot{\Omega^{-1}\chih},\ot{\Omega^{-1}e_4\phi},\ot{\Omega^{-2}\alpha},\ot{\Omega^{-1}\beta^r}$, we define difference 
\begin{equation}
    \wt{\psi}=\psi-\ot\psi.
\end{equation}
We precisely know the initial values of $\Omega^{-1}\beta^r,\Omega\betab^r,\sigma^r$ along $\Hb_0$ via the intrinsic equations \eqref{Gauss_Codazzi} and $K_0=K(g_0)$; hence, we define the difference 
\begin{equation}
    \df{\psi}=\left[\psi\right]_0^v-\left[\psi^c\right]_0^v=\ol\psi-(\ol{\psi})_0=\psi-\psi^c-\psi_0+{\psi^c}_0.
\end{equation}
It is worth noting that the specific difference $\df{\psi}$ for $\Omega^{-1}\beta^r,\Omega\betab^r,\sigma^r$ is exactly $\psi-\psi_0$, because these quantities identically vanish in spherically symmetric settings.
For those geometric quantities admitting explicit $\nabla_4$ evolution equations that do not contain $\alpha$, we define the corresponding differences:
\begin{equation}
    \ddf{\psi}=\psi-\psi_0-v\cdot\left(\Lie_v\psi\right)_0,\ \ddfl{\psi}=\ol{\psi}-(\ol\psi)_0-v\cdot\left(\Lie_v\ol{\psi}\right)_0.
\end{equation}
\begin{remark}
    Heuristically, we expect the metric to remain close to the spherically symmetric and exactly self-similar background metric found in \cite{Christodoulou1994}. Therefore, $\ol{\cdot}$ serves as the appropriate basic difference, which will be primarily utilized in Regions II and III. In Region I, however, we require more accurate difference quantities, which are expected to satisfy bounds of the form $\epsilon \left(\frac{v}{(-u)^{1-\k}}\right)^C$ for suitable constants $C$. Consequently, we introduce $\df\cdot$ and $\ddfl{\cdot}$ to carefully isolate the remaining terms corresponding to the zeroth and first-order Taylor expansions. Nevertheless, the leading terms $\Omega^{-1}\chih$, $\Omega^{-1}e_4\phi$, and $\Omega^{-2}\alpha$ cannot be effectively estimated using these basic differences alone; we thus approximate them explicitly by solving their evolution equations. 
The three difference notations should be read as a hierarchy. The symbol $\wt{\cdot}$ is used after an explicit model term has been subtracted, $\df{\cdot}$ records a zeroth-order increment away from $v=0$, and $\ddfl{\cdot}$ records the remainder after the first $v$-Taylor term has also been removed. This hierarchy explains why later product estimates always separate explicit background increments from the genuinely small difference terms. 
\end{remark}

For these differences, we state the formulae to compute difference of product. They can be proved by direct calculation.
\begin{lemma}\label{Product_rule_for_difference}
    The following estimates hold for the relevant quantities:
    \begin{equation}
        \wt{\psi_1\psi_2}=\psi_1\wt{\psi_2}+\wt{\psi_2}\ot{\psi_1},\quad \ol{\psi_1\psi_2}=\ol{\psi_1}\psi_2+{\psi_1}^c\ol{\psi_2},
    \end{equation}
    \begin{equation}
    \begin{aligned}
            \left[\ol{\psi_1\psi_2}\right]_0^v=&\psi_1\left[\ol{\psi_2}\right]_0^v+\left[\ol{\psi_1}\right]_0^v\psi_2+\ol{\psi_1}\left[{\psi_2}^c\right]_0^v+\left[{\psi_1}^c\right]_0^v\ol{\psi_2}(0)\\
            =&\psi_1\left[\ol{\psi_2}\right]_0^v+\left[\ol{\psi_1}\right]_0^v\psi_2+\df{\psi_1}\left[{\psi_2}^c\right]_0^v+\ol{\psi_1}(0)\left[{\psi_2}^c\right]_0^v+\left[{\psi_1}^c\right]_0^v\ol{\psi_2}(0).
    \end{aligned}
    \end{equation}
    \begin{equation}
\begin{aligned}
            \ddf{\ol{\psi_1\psi_2}}=&\psi_1(0)\ddf{\ol{\psi_2}}+\ddf{\psi_1}\ol{\psi_2}(v)+v\Lie_v\psi_1(0)\df{\psi_2}+{\psi_2}^c(0)\ddf{\ol{\psi_1}}\\
            &+\ddf{{\psi_2}^c}\ol{\psi_1}(v)+v\Lie_v{\psi_2}^c(0)\df{\psi_1}\\
            =&\psi_1(0)\ddf{\ol{\psi_2}}+\ddfl{\psi_1}\ol{\psi_2}(v)+\ddf{\psi_1^c}\ol{\psi_2}(v)+v\Lie_v\psi_1(0)\ddfl{\psi_2}+{\psi_2}^c(0)\ddf{\ol{\psi_1}}\\
            &+\ddf{{\psi_2}^c}\ol{\psi_1}(v)+v\Lie_v{\psi_2}^c(0)\ddfl{\psi_1}+v^2\Lie_v\psi_2^c(0)\Lie_v\ol{\psi_1}(0)+v^2\Lie_v\psi_1(0)\Lie_v\ol{\psi_2}(0).
\end{aligned}
    \end{equation}
    Here we use $\psi(v)$ to represent $\psi(u,v)$.
\end{lemma}
\begin{remark}
    Although the expression for $\ddfl{\psi_1\psi_2}$ appears algebraically complicated, we can exploit the underlying self-similarity to simplify it significantly. Assuming $\psi_i(0)$, $\Lie_v\psi_i(0)$, and $\psi_i^c$ satisfy suitable scale-invariant bounds, and that $\left|\ol{\psi_i}\right|\ll 1$, we then obtain 
    \begin{equation}
        \begin{aligned}
            &\ddfl{\psi_1\psi_2}\\
            \sim & (-u)^{s(\psi_1)-1}\ddfl{\psi_2}+o(1)(-u)^{s(\psi_2)-1}\ddfl{\psi_1}+\frac{v}{(-u)^{1-\k}}(-u)^{s(\psi_1)-1}\ol{\psi_2}\\
            &+\frac{v}{(-u)^{1-\k}}(-u)^{s(\psi_1)-1}\ddfl{\psi_2}+(-u)^{s(\psi_2)-1}\ddfl{\psi_1}+\frac{v}{(-u)^{1-\k}}(-u)^{s(\psi_2)-1}\ol{\psi_1}\\
            &+\frac{v}{(-u)^{1-\k}}(-u)^{s(\psi_2)-1}\ddfl{\psi_1}+\epsilon\frac{v^2}{(-u)^{2-2\k}}(-u)^{s(\psi_1)+s(\psi_2)-2}\\
            \sim & (-u)^{s(\psi_1)-1}\ddfl{\psi_2}+(-u)^{s(\psi_2)-1}\ddfl{\psi_1}\\
            &+\frac{v}{(-u)^{1-\k}}(-u)^{s(\psi_2)-1}\ol{\psi_1}+\frac{v}{(-u)^{1-\k}}(-u)^{s(\psi_1)-1}\ol{\psi_2}+\epsilon\frac{v^2}{(-u)^{2-2\k}}(-u)^{s(\psi_1\psi_2)-1}.
        \end{aligned}
    \end{equation}
    The detailed statement will be presented and proved in Lemma \ref{RI_Lemma_Calculation_of_ddfl}.
\end{remark}
\begin{remark}
    The equations for differences $\df{\cdot}$ and $\ddfl{\cdot}$ can be written simply. If\\ 
    $\Omega^{-1}\nabla_4\psi=F$, the equations for the difference take the form
    \begin{equation}
        \Omega^{-1}\nabla_4\df{\psi} = \ol{F}-\ol{\Omega^{-2}}\Lie_v\psi^c+\ol{\Omega^{-1}\chi}\cdot \psi^c,
    \end{equation}
    \begin{equation}
        \Omega^{-1}\nabla_4\ddfl{\psi} = \Omega^{-1}\chi\ddfl{\psi}+\Omega^{-2}\left[\Lie_v\ol{\psi}\right]_0^v=\Omega^{-1}\chi\ddfl{\psi}+\Omega^{-2}\df{\Omega^2 \left(F+\Omega^{-1}\chi\psi\right)}.
    \end{equation}
    As for $\Omega\nabla_3\psi=F$, the equations for difference can be written as:
    \begin{equation}
        \begin{aligned}
            \Omega\nabla_3\df{\psi}=&\df{F}-\ol{\Omega\nabla_3}\left[\psi^c\right]_0^v-\left[\Omega\nabla_3\right]_0^v\ol{\psi}(0)-\df{\Omega\nabla_3}\psi^c(0)\\
            =&\df{F}+\ol{\Omega\chib}\left[\psi^c\right]_0^v+\left(\left[\Omega\chib\right]_0^v+\Lie_{\left[b\right]_0^v}\right)\ol{\psi}(0)+\df{\Omega\chib}\psi^c(0).
        \end{aligned}
    \end{equation}
     \begin{equation}
        \Omega\nabla_3\ddfl{\psi}=-\Lie_{\ddfl{b}}\left(\ol{\psi}(0)+v\Lie_v\ol{\psi}(0)\right)-\Lie_{v\Lie_v b(0)}v\Lie_v\ol{\psi}(0)+\ddfl{F}.
     \end{equation}
\end{remark}

To improve the readability of the subsequent computations, we first list several basic lemmas concerning the four types of differences, systematically based on Lemma \ref{Lemma_Sobolev_Ineq_product}. These lemmas allow one to largely bypass the complicated algebraic forms of the expanded equations and directly focus on anticipating the behavior of the leading terms.
\begin{lemma}\label{RI_Lemma_Calculation_of_ol}
    We consider difference pairs $(\psi_i,\psi_i^c,\ol{\psi}_i)$, $i\in\mathbb{N}$. Let $n\geq 3$ be an integer. Assume that 
    \begin{equation}
        \sum_{j\leq n-1}\sup_{u,v}\lnm\nabla^j \ol{\psi}_i\rnm_{\LS_1^2(S_{u,v})}+\sum_{j\leq n} \sup_{u,v}\lnm\nabla^j {\psi_i^c}\rnm_{\LS_1^2(S_{u,v})}\lesssim 1,
    \end{equation}
    then it holds that
    \begin{equation}\label{Calu_Diff_ol_1}
        \sum_{j\leq n}\lnm \nabla^j\left({\psi}_1\ol{\psi}_2\right)\rnm_{\LS^2_w(X)}\lesssim \sum_{j\leq n}\lnm \nabla^j\ol{\psi}_2\rnm_{\LS^2_w(X)}+\lnm\nabla^{n}\ol{\psi}_1 \rnm_{\LS_w^2(X)},
    \end{equation}
    where $X\in\{S_{u,v},H_{u}^v,\Hb_v^u,\R_{u,v}\}$.
    Also, since $\ol{\psi_1\psi_2}=\ol{\psi_1}\psi_2^c+\psi_1^c\ol{\psi_2}+\ol{\psi_1}\cdot\ol{\psi_2}$, it follows that 
    \begin{equation}\label{Calcu_Diff_product_ol}
        \sum_{j\leq n}\lnm \nabla^j\left(\ol{{\psi}_1{\psi}_2}\right)\rnm_{\LS^2_w(X)}\lesssim \sum_{j\leq n}\lnm \nabla^j\ol{\psi}_2\rnm_{\LS^2_w(X)}+\sum_{j\leq n}\lnm\nabla^j\ol{\psi}_1 \rnm_{\LS_w^2(X)}.
    \end{equation}
    All the arguments hold for difference pairs $(\psi_i,\ot\psi_i,\wt\psi_i)$.
\end{lemma}
\begin{proof}
    Note that the only ambiguous part of \eqref{Calu_Diff_ol_1} in $\nabla^j\psi_1$ is $\nabla^n\ol{\psi}_1$, and other terms about $\psi_1$ can be directly relaxed to $O(1)$ by the assumption. Using Lemma \ref{Lemma_Sobolev_Ineq_product}, it follows that 
    \begin{equation}
        \begin{aligned}
            &\sum_{j\leq n}\lnm \nabla^j\left({\psi}_1\ol{\psi}_2\right)\rnm_{\LS^2_w(X)}\lesssim  \sup_{S}\left(\sum_{i\leq n-1}\lnm \nabla^j{\psi}_1\rnm_{\LS^2_1(S)}\right)\sum_{j\leq n}\lnm \nabla^j\ol{\psi}_2\rnm_{\LS^2_w(X)}\\
            &\qquad +\sum_{k\leq n-1,j\leq n}\lnm\nabla^{j}\ol{\psi}_1\rnm_{\LS_w^2(X)}\lnm\nabla^{k}\ol{\psi}_2\rnm_{\LS_1^2(X)}+\lnm \nabla^j\psi_1^c\rnm_{\LS_1^2(X)}\lnm\nabla^{k}\ol{\psi}_2\rnm_{\LS_w^2(X)}\\
            \lesssim &\sum_{j\leq n}\lnm \nabla^j\ol{\psi}_2\rnm_{\LS^2_w(X)}+\lnm\nabla^{n}\ol{\psi}_1 \rnm_{\LS_w^2(X)}.
        \end{aligned}
    \end{equation}
    Here we note that all $\psi^c$-like quantities have uniform bound in $\LS_1^2$ norm.
\end{proof}

We proceed to prove the help lemma for $\df{\cdot}$.
\begin{lemma}\label{RI_Lemma_Calculation_of_df}
    We consider difference $\df{\psi_i}$. Let $n\geq 3$ be an integer. Suppose that 
    \begin{equation}
        \sum_{j\leq n-1}\sup_{u,v}\lnm \nabla^j{\psi_i}\rnm_{\LS_1^2(S_{u,v})}+\sum_{j\leq n}\sup_{u,v}\lnm \nabla^j{\psi_i^c},\nabla^j\left(\psi_i\right)_0\rnm_{\LS_1^2(S_{u,v})}\lesssim 1.
    \end{equation}
    Moreover, we assume 
    \begin{equation}
        \sum_{j\leq n}\sup_{u,v}\lnm \nabla^j\left(\ol{\psi_i}\right)_0\rnm_{\LS_1^2(S_{u,v})}\lesssim\epsilon,\ \sum_{j\leq n}\lnm \nabla^j\left[\psi_i^c\right]_0^v\rnm_{\LS_w^2(S_{u,v})}\lesssim \left(\frac{v}{(-u)^{1-\k}}\right)^a,
    \end{equation}
    where $a$ is some real number and we recall $(\ol{\psi_i})_0(u,v)=\ol{\psi_i}(u,0)$. Then we find that 
    \begin{equation}\label{Calcu_Diff_product_df}
        \sum_{j\leq n}\lnm \nabla^j\df{\psi_1\psi_2}\rnm_{\LS^2_w(X)}\lesssim\sum_{j\leq n}\lnm \nabla^j\df{\psi_1 }\rnm_{\LS^2_w(X)}+\sum_{j\leq n}\lnm \nabla^j\df{\psi_2}\rnm_{\LS^2_w(X)}+\epsilon\left(\frac{v}{(-u)^{1-\k}}\right)^a.
    \end{equation}
\end{lemma}
\begin{proof}
    We first recall that 
    \begin{equation}
        \left[\ol{\psi_1\psi_2}\right]_0^v= \psi_1\left[\ol{\psi_2}\right]_0^v+\left[\ol{\psi_1}\right]_0^v\psi_2+\df{\psi_1}\left[{\psi_2}^c\right]_0^v+\left(\ol{\psi_1}\right)_0\left[{\psi_2}^c\right]_0^v+\left[{\psi_1}^c\right]_0^v\left(\ol{\psi_2}\right)_0.
    \end{equation}
    For $\psi_1\df{\psi_2}$, we can decompose it as ${(\psi_1)_0}\df{\psi_2}+\left[{\psi_1}^c\right]_0^v\df{\psi_2}+\df{\psi_1}\df{\psi_2}$. With the assumptions, we find that
    \begin{equation}
        \begin{aligned}
            \sum_{j\leq n}\lnm \nabla^j\left(\psi_1\df{\psi_2}\right)\rnm_{\LS_w^2(X)}\lesssim & \sum_{j\leq n}\lnm \nabla^j\df{\psi_1}\rnm_{\LS_w^2(X)}+\lnm \nabla^j\df{\psi_2}\rnm_{\LS_w^2(X)}.
        \end{aligned}
    \end{equation}
    The second term can be estimated in the same way. The third term are bounded by 
    \[\sum_{j,k\leq n}\lnm \nabla^j\df{\psi_1}\rnm_{\LS_w^2(X)}\cdot\sup_{u,v}\lnm \nabla^k\left[{\psi_2}\right]_0^v\rnm_{\LS_1^2(S_{u,v})}\lesssim \sum_{j,k\leq n}\lnm \nabla^j\df{\psi_1}\rnm_{\LS_w^2(X)}.\]
    The last two terms are controlled by 
    \[\sum_{i_1,i_2=1,2}\sum_{j,k\leq n}\lnm \nabla^j\left(\ol{\psi_{i_1}}\right)_0\rnm_{\LS_1^2(X)}\cdot\sup_{u,v}\lnm \nabla^k\left[{\psi_{i_2}}\right]_0^v\rnm_{\LS_w^2(S_{u,v})}\lesssim\epsilon\left(\frac{v}{(-u)^{1-\k}}\right)^a.\]
\end{proof}
\begin{remark}
        The exponent $a$ is allowed to be negative. The proof uses only the stated weighted bound for the Christodoulou increment, not any monotonicity in the self-similar parameter.
\end{remark}
We finally state the product estimate for the first-order Taylor remainder. Compared with the preceding lemma, one must also track the explicit first Taylor coefficient at $v=0$.

\begin{lemma}\label{RI_Lemma_Calculation_of_ddfl}
    We consider difference $\ddfl{\psi_i}$. Let $n\geq 3$ be an integer. We still assume that 
    \begin{equation}\label{Calcu_Diff_Condition_1}
        \sum_{j\leq n-1}\sup_{u,v}\lnm \nabla^j{\psi_i}\rnm_{\LS_1^2(S_{u,v})}+\sum_{j\leq n}\sup_{u,v}\lnm \nabla^j{\psi_i^c},\nabla^j\left(\psi_i\right)_0\rnm_{\LS_1^2(S_{u,v})}\lesssim 1,
    \end{equation}
    and
    \begin{equation}\label{Calcu_Diff_Condition_2}
        \sum_{j\leq n}\sup_{u,v}\lnm \nabla^j\left(\ol{\psi_i}\right)_0\rnm_{\LS_1^2(S_{u,v})}\lesssim\epsilon,\ \sum_{j\leq n}\lnm \nabla^j\ddf{\psi_i^c}\rnm_{\LS_1^2(S_{u,v})}\lesssim \left(\frac{v}{(-u)^{1-\k}}\right)^a,
    \end{equation}
    where $a$ is some real number. 
    Moreover, we require 
    \begin{equation}\label{Calcu_Diff_Condition_3}
        \sum_{j\leq n}\sup_{u,v}\lnm \nabla^j\left(\Lie_v\psi_i\right)_0,\nabla^j\left(\Lie_v{\psi_i}^c\right)_0\rnm_{\LS_1^2(S_{u,v})}\lesssim 1,\ \sum_{j\leq n}\sup_{u,v}\lnm \nabla^j\left(\Lie_v\ol{\psi_i}\right)_0\rnm_{\LS_1^2(S_{u,v})}\lesssim \epsilon.
    \end{equation}
    Here the signature of $\Lie_v\psi_i$ is same as $s(\Omega^2\Omega^{-1}\nabla_4\psi_i)=s(\psi_i)+\k-1$ and $s(v\Lie_v\psi_i)=s((-u)^{1-\k}\Lie_v\psi_i)=s(\psi_i)$. 
    Then the following estimates hold:
    \begin{equation}\label{Calcu_Diff_product_ddfl}
        \begin{aligned}
            &\sum_{j\leq n}\lnm \nabla^j\ddfl{\psi_1\psi_2}\rnm_{\LS^2_w(X)}\\
            \lesssim & \sum_{j\leq n}\lnm \nabla^j\ddfl{\psi_1},\nabla^j\ddfl{\psi_2}\rnm_{\LS^2_w(X)}+\sum_{j\leq n}\left(\frac{v}{(-u)^{1-\k}}\right)^a\lnm \nabla^j\ol{\psi_1},\nabla^j\ol{\psi_2}\rnm_{\LS^2_w(X)}\\
            &+\left(\frac{v}{(-u)^{1-\k}}\right)^{2}\sum_{j\leq n}\lnm \nabla^j\left(\Lie_v\ol{\psi_1}\right)_0,\nabla^j\left(\Lie_v\ol{\psi_2}\right)_0\rnm_{\LS^2_w(X)}.
        \end{aligned}
    \end{equation}
\end{lemma}
\begin{proof}
    From the expression of $\ddfl{\psi_1\psi_2}$, we divide the terms into several types.
    \begin{enumerate}
        \item $(\psi_1)_0\ddfl{\psi_2},\ddfl{\psi_1}\ol{\psi_2},v\left(\Lie_v\psi_1\right)_0\ddfl{\psi_2},{\left({\psi_2}^c\right)_0}\ddfl{\psi_1},v\left(\Lie_v{\psi_2}^c\right)_0\ddfl{\psi_1}$.
        \item $\ddf{{\psi_1}^c}\ol{\psi_2},\ddf{{\psi_2}^c}\ol{\psi_1}$.
        \item $v^2\left(\Lie_v{\psi_2}^c\right)_0\left(\Lie_v\ol{\psi_1}\right)_0,v^2\left(\Lie_v{\psi_1}^c\right)_0\left(\Lie_v\ol{\psi_2}\right)_0$.
    \end{enumerate}
    For the first type, since we have the $\nabla^j$-estimate for $j\leq n$, for $(\psi_1)_0$, $v(\Lie_v\psi_1)_0$, $\left({\psi_2^c}\right)_0$, $v\left(\Lie_v{\psi_2}^c\right)_0$ we obtain that
    \begin{equation}
        \begin{aligned}
            &\sum_{j\leq n}\lnm \nabla^j\left((\psi_1)_0\ddfl{\psi_2},v\left(\Lie_v\psi_1\right)_0\ddfl{\psi_2},{\left({\psi_2}^c\right)_0}\ddfl{\psi_1},v\left(\Lie_v{\psi_2}^c\right)_0\ddfl{\psi_1}\right)\rnm_{\LS^2_w(X)}\\
            &\qquad\lesssim \sum_{k\leq n}\sup_{u,v}\lnm \nabla^k\left( (\psi_1)_0,v(\Lie_v\psi_1)_0,\left({\psi_2^c}\right)_0,v\left(\Lie_v{\psi_2}^c\right)_0\right)\rnm_{\LS^2_1(S_{u,v})}\cdot\\
            &\qquad\qquad\qquad\sum_{j\leq n}\lnm \nabla^j\left( \ddfl{\psi_2}, \ddfl{\psi_1} \right)\rnm_{\LS^2_w(X)}.
        \end{aligned}
    \end{equation}
    As for $\ddfl{\psi_1}\ol{\psi_2}$, note that $\ol{\psi_2}=\left(\ol{\psi_2}\right)_0+v\left(\Lie_v\ol{\psi_2}\right)_0+\ddfl{\psi_2}$. We have similar estimate 
    \begin{equation}
        \sum_{j\leq n}\lnm \nabla^j\left(\ddfl{\psi_1}\ol{\psi_2}\right)\rnm_{\LS^2_w(X)}
            \lesssim  \sum_{j\leq n}\lnm  \nabla^j\ddfl{\psi_1} \rnm_{\LS^2_w(X)}+\lnm \nabla^j \ddfl{\psi_2}  \rnm_{\LS^2_w(X)}.
    \end{equation}
    For the second type, we have by assumptions
    \begin{equation}
        \begin{aligned}
            &\sum_{j\leq n}\lnm \nabla^j\left(\ddf{{\psi_1}^c}\ol{\psi_2},\ddf{{\psi_2}^c}\ol{\psi_1}\right)\rnm_{\LS^2_w(X)}\\
            &\qquad\lesssim  \left(\frac{v}{(-u)^{1-\k}}\right)^a\sum_{j\leq n}\left(\lnm  \nabla^j\ol{\psi_1} \rnm_{\LS^2_w(X)}+\lnm \nabla^j \ol{\psi_2}  \rnm_{\LS^2_w(X)}\right).
        \end{aligned}
    \end{equation}
    Lastly, we estimate the third type. 
    \begin{equation}
        \begin{aligned}
            &\sum_{j\leq n}\lnm \nabla^j\left(v^2\left(\Lie_v{\psi_2}^c\right)_0\left(\Lie_v\ol{\psi_1}\right)_0\right)\rnm_{\LS^2_w(X)}\\
            =& \left(\frac{v}{(-u)^{1-\k}}\right)^2\sum_{j\leq n}\lnm \nabla^j\left(\left((-u)^{1-\k}\Lie_v{\psi_2}^c\right)_0\left((-u)^{1-\k}\Lie_v\ol{\psi_1}\right)_0\right)\rnm_{\LS^2_w(X)}\\
            \lesssim &\left(\frac{v}{(-u)^{1-\k}}\right)^2\sum_{j\leq n}\lnm \left(\Lie_v\ol{\psi_1}\right)_0\rnm_{\LS^2_w(X)},
        \end{aligned}
    \end{equation}
    and the remaining terms of the third type are treated in the same way.
\end{proof}
\begin{remark}
    Here we highlight the reasonability of conditions \eqref{Calcu_Diff_Condition_1}, \eqref{Calcu_Diff_Condition_2}, and \eqref{Calcu_Diff_Condition_3}. Because we know all necessary information of $(\cdot)^c$ quantities in \cite{christodoulou1993}, and we will show the differences between our solution and Christodoulou's solution are small, condition \eqref{Calcu_Diff_Condition_1} is natural. The initial values, $(\ol{\cdot})_0$, are prescribed already, which are bounded by $O(\epsilon)$ by previous arguments. The increasing rates of $\left[\psi^c\right]_0^v$ are exact, with minimal $a$ equal to $\frac{2\k}{1-\k}$, see Corollary \ref{increasing_rate_D4phi_c} and further arguments on the spherical symmetric self-similar naked singularity in Appendix \ref{Appendix_Chr94}. Similarly, condition \eqref{Calcu_Diff_Condition_3} only requires the initial values along $\Hb_0$, which is obviously true.
\end{remark}

\subsection{Bootstrap assumptions in Region I}
In region I, $\mathcal{R}_{I}=\{0\leq \frac{v^{1/(1-\k)}}{-u}\leq \epsilon_1\}$, we choose weight function $$w_I(u,v)=\left(\frac{(-u)^{1-\k}}{v}\right)^{1-2\delta}$$ and denote the correpsonding norms by $\LS_I$. 
We make bootstrap assumptions as following:
\begin{equation}\label{RI_Bootstrap_Assumptions}
    \left\{
    \begin{aligned}
        & \sum_{i=0}^{5}\left\|\nabla^{i} \widetilde{\Omega^{-2}\alpha}\right\|_{\LS_I^{2}(H_u^v)}^{2}+\left\|{\nabla}^i \wt{\Omega^{-1}\beta^r},\nabla^{i+1}\wt{\Omega^{-1}\chih}\right\|_{\LS_I^{2}(\Hb_v^u)}^{2}\leq B\epsilon^2 \left(\frac{v}{(-u)^{1-\k}}\right)^{2\delta}, \\
& \sum_{i=0}^{5}\left\|\nabla^{i} \df{\Psi_{1}}\right\|_{\LS_I^{2}(H_u^v)}^{2}+\left\|{\nabla}^i \df{\Psi_2}\right\|_{\LS_I^{2}(\Hb_v^u)}^{2}\leq B\epsilon^2 \left(\frac{v}{(-u)^{1-\k}}\right)^{2\k/(1-\k)}, \\
& \Psi_{1}= \Omega^{-1} \beta^r, \Omega \betab^r, K, \sigma^r, \quad \Psi_{2}=\Omega^{2} \alphab,  \Omega \betab^r, K, \sigma^r, \\
& \sum_{i=0}^{6}\lnm  \nabla^i\df{\Omega\chib},\nabla^i\df{\etab},\nabla^i\df{\Omega^{-1}\chi},\nabla^i\df{\eta},\nabla^i\df{\Omega^{-1}\omega}\rnm^2_{\LS_I^2(H_u^v)}\\
&\quad\quad + \lnm \nabla^i\df{\Omega\omegab} ,\nabla^i\df{\Omega^{-1}\tr\chi},\nabla^i\df{\Omega\chib},\nabla^i\df{\eta},\nabla^i\df{\etab}\rnm^2_{\LS_I^2(\Hb_v^u)} \\
&\quad\quad\quad\leq B \epsilon^2
\left(\frac{v}{(-u)^{1-\k}}\right)^{2\k/(1-\k)},\\ 
& \sum_{i=0}^{6}\left\|\nabla^{i} \wt{\Omega^{-1}D_4\phi},\nabla^i\df{\nabla\phi},\nabla^i\ddfl{\Omega\tr\chib}\right\|_{\LS_I^{2}(H_u^v)}^{2}\\
&\quad\quad+\left\|\nabla^{i} \ddfl{\Omega D_3\phi},\nabla^i\ddfl{\nabla\phi},\nabla^i \ddfl{\Omega\omegab}\right\|_{\LS_I^{2}(\Hb_v^u)}^{2} \leq B\epsilon^2\left(\frac{v}{(-u)^{1-\k}}\right)^{1+2\k/(1-\k)-2\delta}.
    \end{aligned}
    \right.
\end{equation}
We furthermore assume that 
\begin{equation}\label{RI_Bootstrap_Assumptions_2}
    \sum_{i=0}^5\lnm \nabla^i\ol{\psi}\rnm^2_{\LS_1^2(S_{u,v})}\lesssim B\epsilon^2\ll 1,\ \psi\in\{\Omega^{-1}\chi,\Omega\chib,\eta,\etab,\Omega\omegab,\Omega^{-1}\omega\}.
\end{equation}
With \eqref{RI_Bootstrap_Assumptions_2}, it follows immediately that all Ricci coefficients are bounded, which means 
\[\sum_{i=0}^5\lnm \nabla^i{\psi}\rnm^2_{\LS_1^2(S_{u,v})}\lesssim  1,\ \psi\in\{\Omega^{-1}\chi,\Omega\chib,\eta,\etab,\Omega\omegab,\Omega^{-1}\omega\}.\]
Similarly, \eqref{RI_Bootstrap_Assumptions} implies that for $\psi_1=\Omega\chib,\etab,\eta,\Omega^{-1}\chi,\Omega^{-1}\omega,\Omega^{-1}e_4\phi,\nabla\phi$, $\psi_2=\Omega\chib$, $\etab$, $\eta$, $\Omega^{-1}\tr\chi$, $\Omega\omegab$, $\Omega e_3\phi$, $\nabla\phi$, $\Psi_1=\Omega^{-1}\beta^r,\sigma^r,K,\Omega\betab^r$, $\Psi_2=\Omega^{-1}\beta^r,\sigma^r,K,\Omega\betab^r,\Omega^2\alphab$, we have the bound for their scale invariant norms with trivial weight $w=1$,
\[\sum_{i\leq 6,j\leq 5}\lnm \nabla^i\psi_1,\nabla^j\Psi_1\rnm^2_{\LS_1^2(H_u^v)}+\lnm \nabla^i\psi_2,\nabla^j\Psi_2\rnm^2_{\LS_1^2(H_v^u)}\lesssim 1.\] 
These two immediate consequences of the bootstrap assumptions will be used silently in the estimates below: they provide uniform control for all non-difference factors, while the weighted bootstrap bounds provide the small factor attached to whichever factor is a difference. Whenever a product is estimated, the relevant product lemma identifies which factor carries the difference and the uniform bound above controls the remaining factors.

In the rest of this section, we will prove that:
\begin{proposition}[Main Estimates in Region I]\label{Main Proposition R_I}
    Consider the Einstein scalar-field system $(\mathcal{M},g,\phi)$ in Theorem \ref{Main_Existence_Theorem_whole_paper}. If the bootstrap assumption \eqref{RI_Bootstrap_Assumptions} holds in region $\{-1\leq u<0,0\leq v^{\frac{1}{1-\k}}<\epsilon_1(-u)\}$, then we can improve the bootstrap assumptions so that $B\lesssim 1$ by choosing $\epsilon$, $\epsilon_1$ sufficiently small, which will not depend on each other.
\end{proposition}

\subsection{Estimates for metric components}
With defined differences, we first write the equations:
\begin{equation*}
    \begin{aligned}
        \Lie_v\ddfl{b^A}=2\df{\Omega^2(\etab^A-\eta^A)},\ \Lie_v\ddfl{g_{AB}}=2\df{\Omega^2(\Omega^{-1}\chi)_{AB}},
    \end{aligned}
\end{equation*}
\begin{equation*}
    \Lie_v\ddfl{\log\Omega^2}=-4\df{\Omega^2(\Omega^{-1}\omega)}.
\end{equation*}
From transport estimate \eqref{transport_4}, the following schematic equation holds for $\gamma=b^A,g_{AB}$, and $\log\Omega^2$:
\begin{equation}\label{R_I_EST_Eq_Lie_v_nabla_gamma}
    \Lie_v\nabla^j\ddfl{\gamma}=\sum_{i_1+i_2+i_3=j}\nabla^{i_1}(\Omega^2)\nabla^{i_2}(\Omega^{-1}\chi)\nabla^{i_3}\ddfl{\gamma}+\nabla^j\df{\Omega^2\psi},
\end{equation}
with $\psi$ satisfies bootstrap assumption
\begin{equation}
    \sum_{j\leq 6}\lnm \nabla^j\df\psi\rnm^2_{\LS^2_I(H_{u}^v)}\lesssim B\epsilon^2\left(\frac{v}{(-u)^{1-\k}}\right)^{2\k/(1-\k)}.
\end{equation}
We can now estimate the metric components.
\begin{proposition}\label{R_I_Estimate_metric}
    For $\gamma=b,g,\log\Omega^2$, we find that 
    \begin{equation}
         \sum_{i\leq 6}\lnm \nabla^i\ddfl{\gamma}\rnm^2_{\LS_I^2(S_{u,v})}\lesssim B\epsilon^2\left(\frac{v}{(-u)^{1-\k}}\right)^{1+2\k/(1-\k)}\ll  \epsilon^2\z^{1+2\delta}.
    \end{equation}
\end{proposition}
\begin{proof}
    We first make bootstrap assumption for this proposition: 
    \begin{equation}
        \sum_{i\leq 6}\lnm \nabla^i\ddfl{\gamma}\rnm^2_{\LS_I^2(S_{u,v})}\leq B_1\epsilon^2\z^{1+2\delta}\ll \epsilon^2.
    \end{equation}
    The metric $g$ is close to $g^c$, and especially we have
    $$\lnm \nabla\left({\sqrt{\det g(u,v)\cdot\det (g^c)^{-1}(u,v)}}\right)\rnm_{L^\infty(\SS)}\leq \epsilon.$$

    Because
$\log\left(\Omega/\Omega^c\right)=\log\left(\Omega/\Omega^c\right)_0-2v\left(\ol{\Omega\omega}\right)_0+\ddfl{\log\Omega},$ which gives in particular 
    $$\sum_{i\leq 6}\lnm \nabla^i\log(\Omega/\Omega^c)\rnm^2_{\LS_I^2(S_{u,v})}\ll \epsilon^2.$$
    Therefore the lapse $\Omega$ is close to $\Omega^c$ with estimate $|\Omega/\Omega^c-1|\sim |\log(\Omega/\Omega^c)|\ll \epsilon$. Moreover, the relation below holds for $j\in\mathbb{N}$,
    \begin{equation}
        \nabla^{j+1}\left(\Omega/\Omega^c\right)\sim\sum_{i=0}^j\left(\Omega/\Omega^c\right)\nabla^{i} \left(\nabla\ol{\log\Omega}\right)^{j-i+1}.
    \end{equation}
We hence obtain $\sum_{j\leq 6}\lnm \nabla^j(\Omega/\Omega^c-(\Omega/\Omega^c)_0)\rnm_{\LS_1^2(S_{u,v})}\sim \sum_{j\leq 6}\lnm \nabla^j\ol{\log\Omega}\rnm_{\LS_1^2(S_{u,v})}\ll \epsilon$, and 
\begin{equation}
    \sum_{j\leq 6}\lnm \nabla^j(\Omega^2)\rnm_{\LS_1^2(S_{u,v})}\lesssim \sum_{j\leq 6}\lnm \nabla^j\ol{\log\Omega}\rnm_{\LS_1^2(S_{u,v})}+\lnm (\Omega^c)^2\rnm_{\LS_1^2(S_{u,v})}\sim 1.
\end{equation} 
The estimate just before this point is used here to treat every factor $\Omega^2$ in \eqref{R_I_EST_Eq_Lie_v_nabla_gamma} as uniformly controlled. Thus the transport estimate \eqref{transport_4} is applied only to the unknown metric remainder, while the remaining source term is reduced to the difference of a product. 
    From equation \eqref{R_I_EST_Eq_Lie_v_nabla_gamma}, it follows that 
    \begin{equation}
        \begin{aligned}
            \sum_{j\leq 6}\lnm \nabla^j\ddfl{\gamma}\rnm^2_{\LS_I^2(S_{u,v})}
            \lesssim & \frac{v}{(-u)^{1-\k}}\sum_{j\leq 6}\lnm \Lie_v\nabla^j\ddfl{\gamma}\rnm^2_{\LS_I^2(H_u^v)}\\
            \lesssim & \frac{v}{(-u)^{1-\k}}\sum_{j\leq 6}\lnm \nabla^j\ddfl{\gamma},\nabla^j\df{\Omega^2(\Omega^{-1}\omega,\Omega^{-1}\chi,\eta,\etab)}\rnm^2_{\LS_I^2(H_u^v)}\\
            &+\frac{v}{(-u)^{1-\k}}\sum_{i\leq 2}\sup_{v^\prime\leq v}\lnm \nabla^i\ddfl{\gamma}\rnm^2_{\LS_1^2(S_{u,v^\prime})}\lnm \nabla^6(\Omega^2,\Omega^{-1}\chi)\rnm^2_{\LS_I^2(H_u^v)}.
        \end{aligned}
    \end{equation}
    The last term comes from the top order terms of $\nabla^6\Omega$ and $\nabla^6(\Omega^{-1}\chi)$ from the commutators. Using Lemma \ref{RI_Lemma_Calculation_of_df} with $a=\frac{\k}{1-\k}$, we then obtain  
    \begin{equation}
        \begin{aligned}
            &\sum_{j\leq 6}\lnm \nabla^j\ddfl{\gamma}\rnm^2_{\LS_I^2(S_{u,v})}\\
            \lesssim & \frac{v}{(-u)^{1-\k}}\sum_{i\leq 6}\lnm \nabla^6\left(\df{\Omega^2},\df{\Omega^{-1}\chi},\df{\Omega^{-1}\omega},\df{\eta},\df\etab\right)\rnm^2_{\LS_I^2(H_u^v)}\\
            &+\frac{v}{(-u)^{1-\k}}\sum_{i\leq 6}\sup_{v^\prime\leq v}\lnm \nabla^i\ddfl{\gamma}\rnm^2_{\LS_I^2(S_{u,v^\prime})}+\epsilon^2\z^{1+2\k/(1-\k)}\\
            \lesssim & (B+1)\epsilon^2\z^{1+2\k/(1-\k)}+B_1\epsilon^2\z^{2+2\delta} \ll \epsilon^2\z^{1+2\delta}.
        \end{aligned}
    \end{equation} 
The application of Lemma \ref{RI_Lemma_Calculation_of_df} converts each term of the form $\df{\Omega^2\psi}$ into a sum involving $\df{\Omega^2}$, the difference of $\psi$, and the known Christodoulou increment. The estimates used at this point are the bootstrap bounds \eqref{RI_Bootstrap_Assumptions}, the uniform norm consequence above, and the small metric bootstrap assumption made at the start of the proof.
\end{proof}

Following arguments strictly analogous to those for ${\Omega}$ in the preceding proof, we can efficiently estimate $\df{\Omega}$.
\begin{corollary}
    For $n\in \mathbb{Z}$, the following estimate holds for the difference $\df{\Omega^{2n}}$:
    \begin{equation}\label{Estimate_of_df_Omega}
        \sum_{j\leq 6}\lnm \nabla^j\left(\df{\Omega^{2n}}\right)\rnm^2_{\LS^2_1(S_{u,v})}\sim (-u)^{-2n\k}\lnm \df{\Omega^{2n}}\rnm^2_{H^6(\SS)}\lesssim \epsilon^2\z^{2}.
    \end{equation}
\end{corollary}
\begin{proof}
    Since $\left|\df{\log\Omega}\right|=\left|\log\left(\Omega(\Omega^c)^{-1}(\Omega^c)_0(\Omega_0)^{-1}\right)\right|\lesssim \epsilon\z$ and $\Omega_0/(\Omega^c)_0\sim 1$, the following relation holds 
    \[|(\Omega/\Omega^c)^{2n} - (\Omega_0/(\Omega^c)_0)^{2n}|\lesssim \epsilon\z.\]
    From identity 
    \begin{equation}
        \begin{aligned}
            &\nabla^{j+1}\left((\Omega/\Omega^c)^{2n}-(\Omega_0/(\Omega^c)_0)^{2n}\right)\\
            = &\sum_{i=0}^j C_i \left((\Omega/\Omega^c)^{2n}\nabla^{i}\left(\nabla\log(\Omega/\Omega^c)\right)^{j-i+1}-(\Omega_0/(\Omega^c)_0)^{2n}\nabla^{i}\left(\nabla\log(\Omega_0/(\Omega^c)_0)\right)^{j-i+1}\right)\\
            = & \sum_{i=0}^jC_i\left(\left((\Omega/\Omega^c)^{2n}-(\Omega_0/(\Omega^c)_0)^{2n}\right)\nabla^{i}\left(\nabla{\log(\Omega/\Omega^c)}\right)^{j-i+1}\right.\\
            &\qquad\qquad\qquad\left.+(\Omega_0/(\Omega^c)_0)^{2n}\nabla^{i}\left[\left(\nabla\log\left(\Omega/\Omega^c\right)\right)^{j-i+1}\right]_0^v\right),
        \end{aligned}
    \end{equation}
    where $C_i$'s are some constantg coefficients,
    we can obtain that $$\lnm (\Omega/\Omega^c)^{2n}-(\Omega_0/(\Omega^c)_0)^{2n}\rnm_{H^6(\SS)}\lesssim\epsilon\z.$$ Because of
    \begin{equation}
        \df{\Omega^{2n}}=\frac{\Omega^{2n}}{(\Omega^c)^{2n}}-\frac{(\Omega_0)^{2n}}{(\Omega^c)_0^{2n}}-(\Omega^c)^{2n}\left((\Omega_0)^{2n}-(\Omega^c)_0^{2n}\right)\left(\frac{1}{(\Omega^c)^{2n}}-\frac{1}{(\Omega^c)_0^{2n}}\right),
    \end{equation}
    and $\lnm \left((\Omega_0)^{2n}-(\Omega^c)_0^{2n}\right)\left(\frac{1}{(\Omega^c)^{2n}}-\frac{1}{(\Omega^c)_0^{2n}}\right)\rnm_{H^6(\SS)}\lesssim\epsilon\z$, we can further deduce that 
    $$\lnm \df{\Omega^{2n}}\rnm_{H^6(\SS)}\lesssim (\Omega^c)^{2n}\cdot\epsilon \frac{v}{(-u)^{1-\k}}.$$
    This gives the desired bound.
\end{proof}
\begin{remark}
    Because we have proved $\Omega/\Omega^c\sim 1$ and $\nabla^j\Omega\sim o_{L^2}(1)$ for $j\leq 6$, which means $\Omega^2(u,v)\sim (-u)^\k$, we can then ignore $\Omega$ in the scale-invariant norms, $$\lnm \Omega^s\psi\rnm^2_{\LS_w^2(S_{u,v})}=(1+o_{\epsilon,\epsilon_1}(1))\lnm \psi\rnm^2_{\LS_w^2(S_{u,v})},\ \sum_{j\leq 6}\lnm \nabla^j(\Omega^s\psi)\rnm^2_{\LS_w^2(S_{u,v})}\sim \sum_{j\leq 6}\lnm \nabla^j\psi\rnm^2_{\LS_w^2(S_{u,v})}. $$
\end{remark}

From the estimates above, the following estimates holds 
\begin{equation}
    \sum_{i\leq 6}\lnm \nabla^i\ol{g}\rnm^2_{\LS_1^2(S_{u,v})}\lesssim \epsilon^2,\ \sum_{i\leq 6}\lnm \nabla^i\left[g\right]_0^v\rnm^2_{\LS_1^2(S_{u,v})}\lesssim \z^2,
\end{equation}
which indicate the estimates for $\ol{\nabla}$ and $\left[\nabla\right]_0^v$.
\begin{corollary}\label{Estimate_of_tilde_nabla}
    For any tensor $\psi\in H^6(\SS)$, we find that 
    \begin{equation}
        \sum_{i\leq 6}\lnm \ol{{\nabla}^i}\psi\rnm_{L^2(\SS)}^2\lesssim \epsilon^2\lnm \psi\rnm^2_{H^5(\SS)},\ \sum_{i\leq 6}\lnm \left[{{\nabla}^i}\right]_0^v\psi\rnm_{L^2(\SS)}^2\lesssim  \z^2\lnm \psi\rnm^2_{H^5(\SS)},
    \end{equation}
    where $\ol{\nabla^i}_{A_1\cdots A_i}={(\nabla_g)}^i_{A_1\cdots A_i}-(\nabla_{{g^c}})^i_{A_1\cdots A_i}$ and $\left[{{\nabla}^i}\right]_0^v={(\nabla_g)}^i_{A_1\cdots A_i}-(\nabla_{{g_0}})^i_{A_1\cdots A_i}$. 
\end{corollary}
\begin{proof}
    Note that 
    \begin{equation}
        \begin{aligned}
            \ol{\nabla^i}\psi=&{\nabla_g}^{i-1}\ol{\nabla}\psi+ \ol{\nabla^{i-1}}{\nabla_{\ot{g}}}\psi            ={\nabla_g}^{i-1}\ \left(\ol{g^{-1}\partial_{\SS}g}\psi\right) + \ol{\nabla^{i-1}}\nabla_{\ot{g}}\psi.
        \end{aligned}
    \end{equation}
    We have for $i\leq 6$ 
    \begin{equation}
        \begin{aligned}
            \lnm \ol{\nabla^i}\psi\rnm^2_{L^2(\SS)}\lesssim &\sum_{i_1+i_2\leq i-1}\lnm \nabla_g^{i_1}\left(\ol{g^{-1}\partial_{\SS}g}\right)\nabla_g^{i_2}\psi\rnm^2_{L^2(\SS)}+\epsilon^2 \lnm \nabla_{\ot{g}}\psi\rnm^2_{H^{i-2}(\SS)}\\
            \lesssim & \sum_{i_1,i_2\leq 5}\lnm \nabla_g^{i_1}\psi\rnm^2_{L^2(\SS)}\lnm \nabla_g^{i_2}\left(\ol{g^{-1}\partial_{\SS}g}\right)\rnm^2_{L^2(\SS)}+\epsilon^2 \lnm \psi\rnm^2_{H^{i-1}(\SS)}\\
            \lesssim & \epsilon^2 \lnm \psi\rnm^2_{H^{5}(\SS)}.
        \end{aligned}
    \end{equation}
    Here we used that $\lnm g^{-1}\ol{g}\rnm_{H^6(\SS)}\lesssim\epsilon$. Similarly, the estimate for $\left[\nabla\right]_0^v$ follows by the same argument. 
This corollary is the mechanism used whenever an estimate contains $\left[\dv\right]_0^v$, $\left[\nabla\right]_0^v$, or a commutator produced by changing the connection. It should be read together with Proposition \ref{R_I_Estimate_metric} and will be invoked repeatedly in top-order arguments without restating the metric calculation. 
\end{proof}
We can get higher order regularity estimates if $\psi$ is a function.
\begin{corollary}
    For any function $f\in H^7(\SS)$, this yields 
    \begin{equation}
        \lnm \ol{{\nabla}^i} f\rnm_{H^7(\SS)}^2\lesssim \epsilon^2\lnm f\rnm^2_{H^6(\SS)},\ \lnm \left[\nabla^i\right]f\rnm_{H^7(\SS)}^2\lesssim \z^2\lnm f\rnm^2_{H^6(\SS)}.
    \end{equation}
\end{corollary}
\begin{proof}
    Note that $({\nabla_{g^\prime}})_A f=e_A(f)$ for any metric $g^\prime$. The result follows by the previous corollary.
\end{proof}

\subsection{Estimates of non top order terms.}
With the metric and operator comparisons in place, we next estimate all quantities below the top energy level. These bounds provide the lower-order input for the scalar-field and curvature energy estimates that follow.
\begin{proposition}\label{L2S_est_for_non_top_order_terms}
    With assumption \eqref{RI_Bootstrap_Assumptions}, the following estimates hold:
    \begin{equation}
        \begin{aligned}
            &\sum_{i\leq 5}\lnm \nabla^i\left(\ddfl{\Omega^{-1}\tr\chi},\ddfl\eta,\ddfl\etab,\ddfl{\Omega\chib},\ddfl{\Omega D_3\phi},\ddfl{\nabla\phi},\ddfl{\Omega\omegab}\right)\rnm^2_{\LS_I^2(S_{u,v})} \\
            &\qquad\qquad\qquad\qquad\qquad\qquad\qquad\qquad\qquad\lesssim B\epsilon^2\left(\frac{v}{(-u)^{1-\k}}\right)^{1+\frac{2\k}{1-\k}},
        \end{aligned}
    \end{equation}
    \begin{equation}
        \sum_{ j\leq 5}\lnm\nabla^{j}\wt{\Omega^{-1}\chih},\nabla^j\wt{\Omega^{-1}D_4\phi},\nabla^i\wt{\Omega^{-1}\beta^r} \rnm^2_{\LS_I^2(S_{u,v})}\lesssim \epsilon^2\left(\frac{v}{(-u)^{1-\k}}\right)^{1+2\delta},
    \end{equation}
    \begin{equation}
        \sum_{i\leq 4 }\lnm \nabla^i\left(\df{\sigma^r}, \df{K}, \df{\Omega\betab^r}, \df{\Omega^2\alphab}\right) \rnm^2_{\LS_I^2(S_{u,v})}\lesssim \epsilon^2\left(\frac{v}{(-u)^{1-\k}}\right)^{1+2\delta}.
    \end{equation}
\end{proposition}
\begin{proof}
    We use $\varphi$ below to represent Ricci coefficients except $\Omega^{-1}\omega$ and first order derivatives of $\phi$. The proof uses four inputs throughout: the metric bounds from Proposition \ref{R_I_Estimate_metric}, the lapse difference estimate \eqref{Estimate_of_df_Omega}, the product-difference estimate Lemma \ref{RI_Lemma_Calculation_of_df}, and the bootstrap assumptions \eqref{RI_Bootstrap_Assumptions}. The rates of the explicit Christodoulou increments are supplied by the construction of the approximate quantities and the background estimates quoted earlier. 
    For $\varphi_2=\Omega^{-1}\tr\chi$, $\Omega\chib,\eta$, $\Omega\omegab$, $\Omega D_3\phi,\nabla\phi$, the schematic equation reads
    $\Omega^{-1}\Lie_4{\varphi_2}\sim \Omega^2(\Omega^{-1}\beta^r+K+\nabla\eta+\varphi^2)$, and after taking difference and using commutation formula it turns to
    \begin{equation}
        \begin{aligned}
            \Lie_v \nabla^j\ddfl{\varphi_2}=&\sum_{i_1+i_2+i_3=j}\nabla^{i_1}(\Omega^2)\nabla^{i_2}(\Omega^{-1}\chi )\nabla^i_3\ddfl{\varphi_2}\\
            & +\nabla^j\df{\Omega^2\left(\Omega^{-1}\beta^r,K^r,\nabla\eta,\nabla^2\phi,\varphi^2\right)}.
        \end{aligned}
    \end{equation}
By transport estimate \eqref{transport_4}, and Lemma \ref{RI_Lemma_Calculation_of_df}, it follows that 
\begin{equation}
        \begin{aligned}
            &\sum_{i\leq 5}\lnm \nabla^i \ddfl{\varphi_2}\rnm^2_{\LS_I^2(S_{u,v})}\\
            \lesssim & \z \sum_{i\leq 5}\lnm \nabla^i \left(\df{\Omega^2},\df{\Omega^{-1}\beta^r},\df{K},\df{\nabla\eta},\df{\nabla^2\phi},\df{\Omega^{-1}\chi},\df{\varphi}\right)\rnm^2_{\LS_I^2(H_u^v)}\\
            &+\epsilon^2\z^{1+\frac{2\k}{1-\k}}.
        \end{aligned}
    \end{equation}
Here \eqref{transport_4} supplies the integration factor in the outgoing $v$-direction, while Lemma \ref{RI_Lemma_Calculation_of_df} is responsible for expanding the difference of the nonlinear source. The only term with a background increment that has to be watched is the one containing $\Omega^{-1}D_4\phi$, whose growth rate is recalled in the next sentence.
We also note that the increasing rate of $\left[\varphi^c\right]_0^v$ is controlled by $\z^{\frac{\k}{1-\k}}$, when $\varphi=\Omega^{-1}D_4\phi$. 
Hence the bootstrap assumption \eqref{RI_Bootstrap_Assumptions} implies the estimates for $\varphi_2$. 
For $\Omega^{-1}\chih$, the equation arising from the construction reads
\begin{equation}
    \Lie_v \wt{\Omega^{-1}\chi}=-\left[\Omega^2\right]_0^v \ot{\Omega^2\alpha}-\Omega^2\wt{\Omega^{-2}\alpha}-2\Omega^2\Omega^{-1}\chih\Omega^{-1}\chih.
\end{equation}
and we can derive the estimate,
    \begin{equation}
        \begin{aligned}
            \sum_{i\leq 5}\lnm \nabla^i \wt{\Omega^{-1}\chih}\rnm^2_{\LS_I^2(S_{u,v})}\lesssim \frac{v}{(-u)^{1-\k}}\sum_{i\leq 5}\lnm \nabla^i \left(\left[\Omega^2\right]_0^v\ot{\Omega^{-2}\alpha},\wt{\Omega^{-2}\alpha},\Omega^{-1}\chih\right)\rnm^2_{\LS_I^2(H_u^v)}.
        \end{aligned}
    \end{equation} 
The reason for isolating this product is that the approximating curvature may be singular near $v=0$, while the lapse increment supplies an additional $v$-gain. The estimate \eqref{Estimate_Omega/Omega_0} below is therefore paired with the explicit size of $\ot{\Omega^{-2}\alpha}$ coming from the construction of $\ot{\Omega^{-1}\chih}$. 
From $\sum_{i\leq 5}\lnm \nabla^i(\Omega^{-1}\chi)\rnm^2_{\LS_1^2(H_u^v)}\lesssim 1$, it follows that $\z\sum_{i\leq 5}\lnm \nabla^i(\Omega^{-1}\chi)\rnm^2_{\LS_I^2(H_u^v)}\lesssim \z^{1+2\delta}.$ $\wt{\Omega^{-2}\alpha}$ is directly estimated by the bootstrap assumption. For $\left[\Omega^2\right]_0^v\ot{\Omega^{-2}\alpha}$, we recall that 
\begin{equation}
    \log\left(\Omega/\Omega_0\right)=\log\left(\Omega^c/(\Omega^c)_0\right)+v\left(\Lie_v\ol{\log\Omega}\right)_0+\ddfl{\log\Omega},
\end{equation}
hence the previous estimate on $\ddfl{\log\Omega}$ implies 
\begin{equation}\label{Estimate_Omega/Omega_0}
    \lnm \Omega^2/(\Omega_0)^2-1\rnm^2_{H^6(\SS)}\lesssim \lnm \log(\Omega^c/(\Omega^c)_0)\rnm^2_{H^6(\SS)}+\epsilon^2\z^2\lesssim \epsilon^2\z^2.
\end{equation}
Because $\left|\nabla^j\ot{\Omega^{-2}\alpha}\right|_{g}\sim \epsilon\z^{\k/(1-\k)-1}(-u)^{s(\Omega^{-2}\alpha)-j-1}$, the following estimate holds:
\begin{equation}
    \begin{aligned}
        &\sum_{j\leq 6}\int_{0}^v\int_{\SS}(-u)^{1+\k-2s(\nabla^j\alpha)}\left(\frac{(-u)^{1-\k}}{v^\prime}\right)^{1-2\delta}\left|\nabla^j\left(\left[\Omega^{2}\right]_0^v\ot{\Omega^{-2}\alpha}\right)\right|^2\\
        \lesssim &\sum_{j\leq 6}\int_{0}^v\int_{\SS}(-u)^{1+\k-2s(\Omega^2)}\epsilon^2\left(\frac{v^\prime}{(-u)^{1-\k}}\right)^{-1+2\delta+2\k/(1-\k)-2}\left|\nabla^j\left(\Omega^2/(\Omega_0)^2-1\right)\right|^2(\Omega_0)^4\\
        \lesssim &\int_0^v \epsilon^2\frac{1}{(-u)^{1-\k}}\left(\frac{v^\prime}{(-u)^{1-\k}}\right)^{-1+2\delta+2\k/(1-\k)}\lesssim \epsilon^2 \z^{2\delta+2\k/(1-\k)}.
    \end{aligned}
\end{equation}
Combining the estimates above gives 
\[\sum_{i\leq 5}\lnm \nabla^i \wt{\Omega^{-1}\chih}\rnm^2_{\LS_I^2(S_{u,v})}\lesssim \epsilon^2\z^{1+2\delta}.\]
For $\ddfl{\etab},$ if suffices to apply the fifth order estimates of $\ddfl\eta$ and sixth order of $\ddfl{\log\Omega}$ to $\ddfl{\etab_A}=-\ddfl{\eta_A}+2e_A\ddfl{\log\Omega}$.  
For $\wt{\Omega^{-1}e_4\phi}$, this yields 
\begin{equation}
    \begin{aligned}
        &\Omega\nabla_3\wt{\Omega^{-1}e_4\phi}+\left(\frac{1}{2}\Omega\tr\chib-4\Omega\omegab\right)\wt{\Omega^{-1}e_4\phi}\\
        =& -\df{\left(\frac{1}{2}\Omega\tr\chib-4\Omega\omegab\right)}\ot{\Omega^{-1}e_4\phi}-\Lie_{\df{b}}\ot{\Omega^{-1}e_4\phi} +\df{\dv\nabla\phi-\frac{1}{2}\tr\chi e_3\phi-2\eta\nabla\phi}.
    \end{aligned}
\end{equation} 
The following transport step is the first place where the sign of the coefficient is essential. All source terms in the displayed equation are already controlled either by the preceding estimates for $\varphi_2$, by the metric estimate, or by the bootstrap assumptions; the remaining issue is to ensure that the homogeneous coefficient gives a favorable contribution. 
Since coefficient of the borderline term in transport estimate of $\Omega^{-1}e_4\phi$ is 
$$-2+2s(\Omega^{-1}D_4\phi)+(-u)\partial_u\log I-(-u)\Omega\tr\chib+8(-u)\Omega\omegab=-(1-\k)(1-2\delta)+2\k+o_{\epsilon,\epsilon_1}(1),$$ this term will be cancelled if we require $0<\delta\ll 1-3\k$ and $\epsilon,\epsilon_1$ sufficiently small. Hence the estimate follows: 
\begin{equation}
        \begin{aligned}
            &\sum_{i\leq 5}\lnm \nabla^i \wt{\Omega^{-1}D_4\phi}\rnm^2_{\LS_I^2(S_{u,v})}\\
            \lesssim & \sum_{i\leq 5}\lnm  \nabla^i \left(\df{b},\df{\Omega^{-1}\tr\chi},\df{\Omega\chib},\df{\eta},\df{\nabla^2\phi} ,\df{\Omega e_3\phi}\right)\rnm^2_{\LS_I^2(\Hb_v^u)}+\epsilon^2v^{1+2\delta}\\
            &-\delta\sum_{i\leq 5}\lnm \nabla^i\wt{\Omega^{-1}D_4\phi}\rnm^2_{\LS^2(\Hb_v^u)}\\
            \lesssim & \epsilon^2\z^{1+2\delta}.
        \end{aligned}
    \end{equation} 
Since $\ot{\Omega^{-1}\beta^r}$ is defined by $\ot{\Omega^{-1}\chih}$, the difference satisfies 
\begin{equation}
    \wt{\Omega^{-1}\beta^r}=-\dv\wt{\Omega^{-1}\chih}-\left[\dv\right]_0^v\ot{\Omega^{-1}\chih}+\frac{1}{2}e_A\left[\Omega^{-1}\tr\chi\right]_0^v-\eta\wt{\Omega^{-1}\chih}-\left[\eta\right]_0^v\ot{\Omega\chih}+\left[\Omega^{-1}\tr\chi\eta\right]_0^v.
\end{equation}
The estimate for $\wt{\Omega^{-1}\beta^r}$ follows by previous estimates for $\Omega^{-1}\chih$ and $\varphi_2$. 
More explicitly, the divergence term uses the just-proved estimate for $\wt{\Omega^{-1}\chih}$, the connection-difference term uses Corollary \ref{Estimate_of_tilde_nabla}, and all remaining $\df{\varphi_2}$ contributions are covered by the first part of this proposition. 
The final low-order curvature bounds are obtained from the constraint equations rather than from a separate energy argument. The estimates used here are the non-top-order Ricci bounds already proved in this proposition, the Codazzi and Gauss identities \eqref{Gauss_Codazzi}, and the bootstrap bounds for the curvature components appearing on the right-hand side. 
From $\etab^c\equiv 0$, it follows that $\df{\sigma^r}=-\left[\cl\right]_0^v\etab_0-\cl\df\etab$. The estimate of $\sigma^r$ can be obtained immediately. Similarly, $\Omega\betab^r$ can be estimated with $\Omega\chib,\eta$ and the Codazzi equation \eqref{Gauss_Codazzi}. From $\Omega^{-1}\nabla_4(\Omega^2\alphab)$ and $\Omega\nabla_4 \ol{K}+\frac{1}{2}\Omega\tr\chi \ol{K}$ equations, the low order terms on right hand side are all of size $O(\epsilon)$ in $\LS^2_1(S)$ norm, and the higher order terms, $\Omega\betab^r,\ol{K}$. are bounded by the bootstrap assumptions. Then we have estimate 
\begin{equation}
    \begin{aligned}
        &\sum_{i\leq 4}\lnm\nabla^i\left(\df{\Omega^{2}\alphab},\df{K}\right)\rnm^2_{\LS_I^2(S_{u,v})}\\
        \lesssim & \frac{v}{(-u)^{1-\k}} \sum_{i\leq 5}\lnm \nabla^{i}\left(\ol{\Omega^{-1}\chi},\etab,\eta,\ol{\Omega\chib},\ol{\Omega e_3\phi},\nabla\phi\right),\nabla^{i}\left(\ol{\Omega\betab^r},\ol{K}\right) \rnm^2_{\LS_I^2(H_u^v)}\\
        \lesssim &  \epsilon^2\left(\frac{v}{(-u)^{1-\k}}\right)^{1+2\delta}.
    \end{aligned}
\end{equation}
Here we only list the leading terms of the estimates and omit the expressions of the schematical equations. Readers can refer to Proposition \ref{R_I_Estimate_beta_sigma_K_betab_alphab} to check the equations.

\end{proof}

We can further estimate $\Omega^{-1}\omega$.
\begin{corollary}
    We can estimate $\Omega^{-1}\omega$, 
    \begin{equation}
        \sum_{i\leq 5}\lnm \nabla^i\df{\Omega^{-1}\omega}\rnm^2_{\LS^2_1(S_{u,v})}\lesssim \epsilon^2\z^{2\k/(1-\k)},
    \end{equation}
    and hence 
    \begin{equation}
        \sum_{i\leq 5}\lnm \nabla^i\df{\Omega^{-1}\omega}\rnm^2_{\LS^2_I(H_u^v)}\lesssim \epsilon^2\z^{2\k/(1-\k)+2\delta}.
    \end{equation}
    Note that we use weight $w=1$ when estimating the norms on sphere.
\end{corollary}
\begin{proof}
    
For $\df{\Omega^{-1}\omega}$, the schematic equation reads  
\begin{equation}
    \begin{aligned}
        \Omega\nabla_3\df{\Omega^{-1}\omega}\sim & \df{K+ e_3\phi e_4\phi+\nabla\phi\nabla\phi+\chi\chib+\eta\etab+\etab^2}\\
        &+\ol{\Omega\chib}\left[(\Omega^{-1}\omega)^c\right]_0^v+\df{\Omega\chib}{(\Omega^{-1}\omega)^c}_0+\left(\left[{\Omega\chib}\right]_0^v+\Lie_{\left[b\right]_0^v}\right)\ol{\Omega^{-1}\omega}_0.
    \end{aligned}
\end{equation}
The product expansion in this corollary uses the same mechanism as the preceding proposition. The potentially large background increments are only $(\Omega^{-1}e_4\phi)^c$ and $(\Omega^{-1}\omega)^c$; their rates are precisely $\z^{\k/(1-\k)}$. We can apply Lemma \ref{RI_Lemma_Calculation_of_df} to obtain that 
\begin{equation}
    \begin{aligned}
        & \sum_{i\leq 5}\lnm \nabla^i\df{\Omega^{-1}\omega}\rnm^2_{\LS^2_1(S_{u,v})}\\
        \lesssim&\sum_{i\leq 5}\lnm \nabla^i\left(\df{K},\df{\Omega e_3\phi},\df{\Omega e_4\phi},\df{\nabla\phi},\df{\Omega\chib},\df{\Omega^{-1}\chi},\df{\eta},\df{\etab}\right)\rnm^2_{\LS^2_1(\Hb_v^u)}\\
        &+\epsilon^2\z^{2\k/(1-\k)}
        \lesssim  \epsilon^2\z^{2\k/(1-\k)}.
    \end{aligned}
\end{equation}
\end{proof}

\begin{remark}
    It is not easy to refine the approximation for $\Omega^{-1}\omega$. If we mimic the approximation of $\Omega^{-1}\chih$, $\Omega^{-1}e_4\phi$ and consider equation
\begin{equation}
    \begin{aligned}
            &\left({(\Omega\nabla_3)_0}-4({\Omega\omegab})_0-4\left[(\Omega\omegab)^c\right]_0^v\right)\ot{\Omega^{-1}\omega}=\left(\frac{1}{4}|\nabla\phi|^2+|\etab|^2-\eta\etab\right)_0+\frac{1}{4}(\Omega\chibh)_0\cdot\ot{\Omega^{-1}\chih}\\
         & +\frac{1}{4}\left((\Omega e_3\phi)_0+\left[(\Omega e_3\phi)^c\right]_0^v\right)\ot{\Omega^{-1}e_4\phi}+\left(-\frac{1}{2}K+\frac{1}{8}\tr\chi\tr\chib\right)_0+\left[\left(-\frac{1}{2}K+\frac{1}{8}\tr\chi\tr\chib\right)^c\right]_0^v,
    \end{aligned}
\end{equation}
which can be written in the context of self-similar along $H_{-1}$,
\begin{equation}
    \left(v\Lie_v+\Lie_{b_0}\right)\ot{\Omega^{-1}\omega}+(1-\k-4v(\Omega\omega)^c)\ot{\Omega^{-1}\omega}=R.H.S.,
\end{equation}
the constant on the left hand side is positive. Therefore the equation cannot be solved backward in $v$. We can avoid dealing with $\Omega^{-1}\omega$ by choosing appropriate $\Omega^s$ in front of tensors.
However, we can define better approximation of $\Omega$ with another approach. Define 
$\ot{\Omega\omegab}$ by solving $\partial_v\ot{\Omega\omegab}$ equation with $\Omega^{-1}\chih,\Omega^{-1}e_4\phi$ on the right hand side substituted by $\ot{\Omega^{-1}\chih},\ot{\Omega^{-1}e_4\phi}$ and others by $(\psi)_0+\left[\psi^c\right]_0^v$. Solve $\ot{\Omega}$ by $\left(\partial_u+b_0+v(\Lie_v b)_0\right)\log\ot{\Omega}=-2\ot{\Omega\omegab}$ and define $\ot{\Omega^{-1}\omega}=-\frac{1}{2}\ot{\Omega^{-2}}\Lie_v\log\ot{\Omega}$. In such way one can show that along $H_{-1}$, $\wt{\Omega^{-1}\omega}$ is of size $O(\epsilon v)$, which is the same as $\wt{\Omega^{-1}\chih}$ and $\wt{\Omega^{-1}e_4\phi}$. 
\end{remark} 
 
\subsection{Top order energy estimates for scalar field}
We now turn from lower-order transport estimates to the top angular derivatives. The scalar-field estimates are coupled: the outgoing derivative of the scalar field is controlled together with the angular derivatives propagated along the incoming direction.
\begin{proposition}
    We establish the following estimates for the relevant quantities:
    \begin{equation}
        \lnm \nabla^6\wt{\Omega^{-1}D_4\phi}\rnm^2_{\LS^2(H_{u}^v)}+\lnm \nabla^6\ddfl{\nabla\phi}\rnm^2_{\LS^2(\Hb_{v}^u)}\leq \epsilon^2\left(\frac{v}{(-u)^{1-\k}}\right)^{1+2\k/(1-\k)-2\delta}.
    \end{equation}
\end{proposition}
\begin{proof}
    If we consider $\Omega\nabla_3\nabla^6\wt{\Omega^{-1}e_4\phi}$ equation, there will be $\nabla^6\left[\dv\right]_0^v(\nabla\phi)_0$ on the right hand side, which will contribute to $\left[g^{-1}\partial^7g\right]_0^v$ and cannot be estimated. Instead, we consider $\Omega \nabla_3$ equation for
    $\nabla_A \left(\ot{\Omega^{-1}e_4\phi}\right)=e_A\left(\ot{\Omega^{-1}e_4\phi}\right).$ For $e_A =\partial_{\theta^A}$, we have $\left[\Lie_{\Omega e_3},\Lie_{e_A}\right]=-\Lie_{\partial_A b^C\partial_C}$. This choice moves the dangerous top derivative of the divergence operator into terms involving the metric difference and the shift $b$. Those terms are exactly the ones controlled by Proposition \ref{R_I_Estimate_metric} and Corollary \ref{Estimate_of_tilde_nabla}.
    \begin{equation}
        \begin{aligned}
            &\Omega\nabla_3\nabla_A \left(\ot{\Omega^{-1}e_4\phi}\right)=\Lie_{\partial_u+b} (\nabla_A)_0\left(\ot{\Omega^{-1}e_4\phi}\right)-\Omega\chib_A^B\nabla_B\left(\ot{\Omega^{-1}e_4\phi}\right)\\
            =& \nabla_A \left(\Omega e_3\left(\ot{\Omega^{-1}e_4\phi}\right)\right)-\Lie_A b^C\nabla_C\ot{\Omega^{-1}e_4\phi} -\Omega\chib_A^B\nabla_B\left(\ot{\Omega^{-1}e_4\phi}\right)\\
            =&\left(-\frac{1}{2}\left(\Omega\tr\chib-\df{\Omega\tr\chib}\right)+4\left(\Omega\omegab-\df{\Omega\omegab}\right)\right)\nabla_A\left(\ot{\Omega^{-1}e_4\phi}\right)\\
            &+\nabla_A\left(-\frac{1}{2}\left(\Omega\tr\chib-\df{\Omega\tr\chib}\right)+4\left(\Omega\omegab-\df{\Omega\omegab}\right)\right)\left(\ot{\Omega^{-1}e_4\phi}\right)\\
            &+(\nabla_A\dv\nabla\phi)_0+2\left(\nabla_A\left(\eta\cdot\nabla\phi\right)\right)_0-\frac{1}{2}\nabla_A\left(\Omega^{-1}\tr\chi\Omega e_3\phi-\df{\Omega^{-1}\tr\chi\Omega e_3\phi}\right)\\
            &-\Lie_A b^C\nabla_C\ot{\Omega^{-1}e_4\phi} -\Omega\chib_A^B\nabla_B\left(\ot{\Omega^{-1}e_4\phi}\right).
        \end{aligned}
    \end{equation}
    Using $\left(\nabla_A\dv\nabla\phi\right)_0=\left(\dv\nabla_A\nabla\phi-K\nabla_A\phi\right)_0$ difference $e_A\left(\wt{\Omega^{-1}e_4\phi}\right)$ obeys
    \begin{equation}\label{eq:G_1_nabla_phi}
        \begin{aligned}
            F_1(e_4\phi):=&\Omega\nabla_3\nabla_A \left(\wt{\Omega^{-1}e_4\phi}\right)+\left(\Omega\tr\chib-4\Omega\omegab\right)\nabla_A\wt{\Omega^{-1}e_4\phi} -\dv \df{\nabla_A\nabla\phi}\\
            =& \left(-\frac{1}{2}\df{\Omega\tr\chib}+4\df{\Omega\omegab}\right)\nabla_A\ot{\Omega^{-1}e_4\phi}+\nabla_A\left(-\frac{1}{2}\df{\Omega\tr\chib}+4\df{\Omega\omegab}\right)\ot{\Omega^{-1}e_4\phi}\\
            &+\nabla_A \left(-\frac{1}{2}{\Omega\tr\chib}+4{\Omega\omegab}\right)\wt{\Omega^{-1}e_4\phi}-\df{K\nabla_A\phi}+\left[\dv\right]_0^v(\nabla_A\nabla\phi)_0\\
            &+2\df{\nabla_A(\eta\cdot\nabla\phi)}-\frac{1}{2}\nabla_A\df{\Omega^{-1}\tr\chi\Omega e_3\phi}-\left(\Omega\chibh_A^C+\Lie_A b^C\right) \nabla_C\wt{\Omega^{-1}e_4\phi}.
        \end{aligned}
    \end{equation}
    For $\nabla^2\phi$, it holds that $\Omega\nabla_4\ddfl{\nabla_B\phi}= e_B\df{(\Omega e_4\phi)}-\Omega\chi\cdot \ddfl{\nabla\phi}.$ We then compute: 
    \begin{equation}
        \begin{aligned}
            G_1(\nabla\phi):=&\Omega^{-1}\nabla_4\nabla_A\ddfl{\nabla_B\phi}-\nabla_B\nabla_A\left(\df{\Omega^{-1} e_4\phi}\right)\\
            \sim& \Omega^{-2}\nabla(\Omega\chi)\ddfl{\nabla\phi}+\Omega^{-1}\chi\nabla\ddfl{\nabla\phi}\\
            &+\left(\Omega^{-2}\nabla_A\nabla_B\df{\Omega e_4\phi}-\nabla_A\nabla_B\df{\Omega^{-1} e_4\phi}\right).
        \end{aligned}
    \end{equation}
    Due to spherical symmetry of $(\cdot)^c$, we have  $$\nabla^2\df{\Omega e_4\phi}=\nabla^2\left(\Omega^2\left[{\Omega^{-1}e_4\phi}\right]_0^v+\left[\Omega^2\right]_0^v(\Omega^{-1}e_4\phi)_0\right),$$ and the last term of \eqref{eq:G_1_nabla_phi} can be written as 
    \begin{equation}
        \begin{aligned}
            &\Omega^{-2}\nabla_A\nabla_B\df{\Omega e_4\phi}-\nabla_A\nabla_B\df{\Omega^{-1} e_4\phi}\\
            \sim&\sum_{i+j=2}\Omega^{-2}\nabla^i\left(\Omega^2-(\Omega_0)^2\right)\nabla^j(\Omega^{-1}e_4\phi)_0+\sum_{i+j=1}\nabla^i(\eta+\etab)\nabla^j\left[\Omega^{-1}e_4\phi\right]_0^v.
        \end{aligned}
    \end{equation}
    Using commutation formulae, we find that 
    \begin{equation}
        \begin{aligned}
            F_6(e_4\phi):=&\Omega\nabla_3\nabla^6\wt{\Omega^{-1}e_4\phi}+\left(\frac{7}{2}\Omega\tr\chib-4\Omega\omegab\right)\nabla^6\wt{\Omega^{-1}e_4\phi}-\dv\nabla^5\df{\nabla^2\phi}\\
            \sim& \left(\Omega\chibh+\nabla b\right)\nabla^6\wt{\Omega^{-1}e_4\phi}+\sum_{i_1+i_2=4}\nabla^{i_1+1}(\Omega\chib,\Omega\omegab)\nabla^{i_2+1}\wt{\Omega^{-1}e_4\phi}+\nabla^5F_1(e_4\phi) \\
            &+\sum_{i_1+i_2+i_3= 4}\nabla^{i_1}K^{i_2+1}\nabla^{i_3}\df{\nabla^2\phi},\\
            G_6(\nabla\phi):=&\Omega^{-1}\nabla_4 \nabla^6\ddfl{\nabla\phi}-\nabla\left(\nabla^6\df{\Omega^{-1}e_4\phi}\right)\\
            \sim & \nabla^5 G_1(\nabla\phi)+\sum_{i_1+i_2+i_3= 4}\nabla^{i_1}K^{i_2+1}\nabla^{i_3+1}\df{\Omega^{-1}e_4\phi}\\
            &+\sum_{i_1+i_2+i_3+i_4=5}\nabla^{i_1}(\eta+\etab)^{i_2}\nabla^{i_3}(\Omega^{-1}\chi)\nabla^{i_4}\ddfl{\nabla\phi}.
        \end{aligned}
    \end{equation} 
At this stage the pair $(\wt{\Omega^{-1}D_4\phi},\ddfl{\nabla\phi})$ has been arranged so that the top-order derivatives appear in a genuine energy pair. The terms denoted by $F_6$ and $G_6$ will be estimated after the energy inequality using the non-top-order bounds of Proposition \ref{L2S_est_for_non_top_order_terms} and the metric bounds already established. 
Using energy estimate \eqref{Energy Estimates}, we deduce that 
    \begin{equation}\label{energy est nabla 3 D4phi nabla 3 nablaphi}
        \begin{aligned}
            &\lnm \nabla^6\wt{\Omega^{-1}D_4\phi}\rnm_{\LS_I^2(H_u^v)}^2+\lnm \nabla^5\ddfl{\nabla^2\phi}\rnm_{\LS_I^2(\Hb_v^u)}^2\\
            \leq & \left(o_{\epsilon,\epsilon_1}(1)-1-\k-(1-\k)(1-2\delta)+2+2\k \right)\int_{-1}^u\frac{1}{-u^\prime}\lnm\nabla^6\wt{\Omega^{-1}D_4\phi}\rnm_{\LS_I^2(H_{u^\prime}^v)}^2\\
            &+\lnm \nabla^6\wt{\Omega^{-1}D_4\phi}\rnm_{\LS_I^2(\mathcal{R}_{u,v})}\cdot \lnm F_6(e_4\phi)\rnm_{\LS_I^2(\mathcal{R}_{u,v})}\\
            &+\lnm \nabla^5\ddfl{\nabla^2\phi}\rnm_{\LS_I^2(\mathcal{R}_{u,v})}\cdot \lnm G_6(\nabla\phi) \rnm_{\LS_I^2(\mathcal{R}_{u,v})}\\
            &+\lnm\nabla^7\left(\wt{\Omega^{-1}D_4\phi}-\df{\Omega^{-1}D_4\phi}\right)\rnm_{\LS_I^2(\R_{u,v})}\lnm \nabla^5\ddfl{\nabla^2\phi}\rnm_{\LS_I^2(\R_{u,v})}\\
            \leq &\left(2\delta+2\k-2\delta\k+o_{\epsilon,\epsilon_1}(1) \right)\int_{-1}^u\frac{1}{-u^\prime}\lnm\nabla^6\wt{\Omega^{-1}D_4\phi}\rnm_{\LS_I^2(H_{u^\prime}^v)}^2\\
            &+\lnm \nabla^6\wt{\Omega^{-1}D_4\phi}\rnm_{\LS_I^2(\mathcal{R}_{u,v})}\cdot \lnm F_6(e_4\phi)\rnm_{\LS_I^2(\mathcal{R}_{u,v})}\\
            &+\left(\frac{v}{(-u)^{1-\k}}\right)^{1/2}\lnm \nabla^5\ddfl{\nabla^2\phi}\rnm_{\LS_I^2(\Hb_v^u)}\cdot \lnm G_6(\nabla\phi) \rnm_{\LS_I^2(\mathcal{R}_{u,v})}\\
            &+\lnm\nabla^6\wt{\Omega^{-1}D_4\phi}\rnm_{\LS_I^2(\R_{u,v})}\z^{1+\delta}\epsilon+\lnm \nabla^5\ddfl{\nabla^2\phi}\rnm_{\LS_I^2(\R_{u,v})}\z^{1+\delta}\epsilon.
        \end{aligned}
    \end{equation}
We next prove that 
    \begin{equation}\label{Main Est of nab 3 D4phi}
        \lnm F_6(D_4\phi)\rnm^2_{\LS_I^2(\mathcal{R}_{u,v})}\lesssim \epsilon \lnm \nabla^6\wt{\Omega^{-1}D_4\phi}\rnm_{\LS_I^2(\mathcal{R}_{u,v})}^2+C(B)\frac{v}{(-u)^{1-\k}}\cdot \epsilon^2 \left(\frac{v}{(-u)^{1-\k}}\right)^{2\k/(1-\k)},
    \end{equation}
    and 
    \begin{equation}\label{Main Est of nab 3 nablaphi}
        \lnm G_6(\nabla\phi)\rnm^2_{\LS_I^2(\mathcal{R}_{u,v})}\lesssim B\epsilon^2 \left(\frac{v}{(-u)^{1-\k}}\right)^{2\k/(1-\k)}.
    \end{equation}
Since $\lnm \nabla^6\wt{\Omega^{-1}D_4\phi}\rnm_{\LS_I^2(H_{-1}^v)}^2\equiv 0$, soloving the ODE model $(-u)\partial_uF\leq \lambda F+o(1)(-u)^{-\lambda-\delta}$ gives the estimate
    \begin{equation}
        \lnm \nabla^6\wt{\Omega^{-1}D_4\phi}\rnm_{\LS_I^2(\mathcal{R}_{u,v})}^2\leq C_0\epsilon^2\left(\frac{v }{(-u)^{1-\k}}\right)^{1+2\k/(1-\k)-2\delta},
    \end{equation}
    and thus we can show that 
    \begin{equation}
        \lnm \nabla^6\wt{\Omega^{-1}D_4\phi}\rnm_{\LS_I^2(H_u^v)}^2+\lnm \nabla^5\ddfl{\nabla^2\phi}\rnm_{\LS_I^2(\Hb_v^u)}^2\leq C_0\epsilon^2\left(\frac{v }{(-u)^{1-\k}}\right)^{1+2\k/(1-\k)-2\delta}.
    \end{equation} 
It remains only to estimate the source terms in the preceding two displayed bounds. In the following two subarguments, every non-top-order factor is controlled by Proposition \ref{L2S_est_for_non_top_order_terms}, every metric or connection-difference factor by Proposition \ref{R_I_Estimate_metric} and Corollary \ref{Estimate_of_tilde_nabla}, and every product difference by Lemma \ref{RI_Lemma_Calculation_of_df}.

\vspace{2mm}
\noindent\textit{Proof of \eqref{Main Est of nab 3 D4phi}:} 
The terms in this expansion have been grouped by the estimate that controls them: Ricci and scalar differences use Proposition \ref{L2S_est_for_non_top_order_terms}, the curvature term uses the bootstrap bound for $K$, and the divergence-operator difference is reduced to the metric estimate by Corollary \ref{Estimate_of_tilde_nabla}.  Employing Lemma \ref{RI_Lemma_Calculation_of_df}, we obtain 
\begin{equation}\label{R_I_energy_estimate_D4phi_Nameless_eq_1}
    \begin{aligned}
        &\lnm \nabla^5F_1(e_4\phi)\rnm^2_{\LS_I^2(\R_{u,v})}\\
        \lesssim & \sum_{i\leq 5}\lnm \nabla^i\left(\df{\Omega^{-1}\tr\chi},\df{\Omega\chib}, \df\eta,\df{\Omega e_3\phi},\df{\nabla\phi},\df{\Omega\omegab}\right)\rnm^2_{\LS_I^2(\R_{u,v})}\\
        &+\sum_{i\leq 4}\lnm \nabla^i \df{K} \rnm^2_{\LS_I^2(\R_{u,v})}+\sum_{i\leq 6}\lnm \nabla^i\df{b}\rnm^2_{\LS_I^2(\R_{u,v})}+\lnm \nabla^5\left[\dv\right]_0^v\left({\nabla^2\phi} \right)_0\rnm^2_{\LS_I^2(\R_{u,v})}\\
        &+\lnm \nabla^6\left(\df{\Omega\omegab},\df{\Omega^{-1}\tr\chi},\df{\Omega\tr\chib}\right),\nabla^5\df{K} \rnm^2_{\LS_I^2(\R_{u,v})}.
    \end{aligned}
\end{equation}
By Proposition \ref{L2S_est_for_non_top_order_terms} and \ref{R_I_Estimate_metric}, for the $\psi, D\phi$ appearing in the equation we can estimate that
\begin{equation}
    \begin{aligned}
        &\sum_{i\leq 5}\lnm \nabla^i\df{\psi} ,\nabla^i\df{D\phi}\rnm^2_{\LS_I^2(S_{u,v})}+\sum_{i\leq 4}\lnm \nabla^i\df{K}\rnm^2_{\LS_I^2(S_{u,v})}+\sum_{i\leq 6}\lnm \nabla^i\df{b}\rnm^2_{\LS_I^2(S_{u,v})}\\
        &\qquad\lesssim \epsilon^2\z^{1+2\delta}.
    \end{aligned}
\end{equation}
Since $\nabla^5\left[\dv\right]_0^v (\nabla^2\phi)_0\sim \epsilon \nabla^5\left[g^{-1}\partial g\right]_0^v$, we also have  
\begin{equation}
    \lnm \nabla^5\left[\dv\right]_0^v (\nabla^2\phi)_0\rnm^2_{\LS_I^2(S_{u,v})}\lesssim \epsilon^2 \sum_{1\leq i\leq 6}\lnm \left[\partial^i g\right]_0^v\rnm^2_{\LS_I^2(S_{u,v})} \ll \epsilon^2\z^{1+2\delta}.
\end{equation}
For the top-order terms in $\nabla^5F_1(e_4\phi)$, it follows that
\begin{equation}
    \begin{aligned}
        &\lnm \nabla^6\left(\df{\Omega\omegab},\df{\Omega^{-1}\tr\chi},\df{\Omega\tr\chib}\right),\nabla^5\df{K} \rnm^2_{\LS_I^2(\R_{u,v})}\\
        \lesssim &\frac{v}{(-u)^{1-\k}}\lnm \nabla^6\left(\df{\Omega\omegab},\df{\Omega^{-1}\tr\chi},\df{\Omega\tr\chib}\right),\nabla^5\df{K} \rnm^2_{\LS_I^2(\Hb_v^u)}\\
        \lesssim &  B\epsilon^2\z^{1+2\k/(1-\k)}.
    \end{aligned}
\end{equation}
    The other terms in $F_6(D_4\phi)$, except $(\Omega\chibh+\nabla b)\nabla^6\wt{\Omega^{-1}e_4\phi}$, can be bounded by the spherical norms for non-top-order terms,
    \begin{equation}
        \begin{aligned}
            &\sum_{i\leq 4,j\leq 5}\lnm \nabla^i\left(\df{K},\df{\nabla^2\phi}\right),\nabla^j(\df{\eta},\df{\Omega\chib},\df{\Omega\omegab},\df{\Omega^{-1}D_4\phi})\rnm^2_{\LS_I^2(\R_{u,v})}\\
        &\qquad\lesssim \epsilon^2\z^{2+2\delta}.
        \end{aligned}
    \end{equation}
    Finally for $(\Omega\chibh+\nabla b)\nabla^6\wt{\Omega^{-1}e_4\phi}$, we use $\left|\Omega\chibh,\nabla b\right|\sim \epsilon$ to obtain  
    \begin{equation}
        \begin{aligned}
            \lnm(\Omega\chibh+\nabla b)\nabla^6\wt{\Omega^{-1}e_4\phi}\rnm^2_{\LS^2_I(\R_{u,v})}\lesssim \epsilon^2\lnm \nabla^6\wt{\Omega^{-1}e_4\phi}\rnm^2_{\LS^2_I(\R_{u,v})}.
        \end{aligned}
    \end{equation}
    \noindent\textit{Proof of \eqref{Main Est of nab 3 nablaphi}:}
    Using Lemma \ref{RI_Lemma_Calculation_of_df} and the bootstrap assumptions \eqref{RI_Bootstrap_Assumptions}, the following inequality holds 
    \begin{equation}
        \begin{aligned}
            \lnm G_6(\nabla\phi)\rnm^2_{\LS^2(\R_{u,v})}\lesssim & \lnm \nabla^5\ddfl{\nabla^2\phi},\nabla^6\df{\Omega^{-1}D_4\phi},\nabla^6\df{(\eta+\etab)}\rnm^2_{\LS^2(\R_{u,v})}\\
            &+\sum_{i\leq 5,j\leq 4}\lnm \nabla^j \df{K},\nabla^i\left(\df{\eta},\df{\etab},\df{\Omega^{-1}\chi},\df{D\phi}\right)\rnm^2_{\LS^2(\R_{u,v})}\\
            \lesssim & B\epsilon^2\z^{\frac{2\k}{1-\k}}.
        \end{aligned}
    \end{equation}

\end{proof}

We proceed to estimate the top-order derivatives of scalar field.
\begin{proposition}\label{R_I_Energy_Est_nabphi_D_3phi}
    We establish the following estimates for the relevant quantities:
    \begin{equation}
        \lnm\nabla^6\df{\nabla\phi}\rnm^2_{\LS_I^2(H_u^v)}+\lnm\nabla^6\ddfl{(\Omega D_3\phi)}\rnm^2_{\LS_I^2(\Hb_v^u)}\leq \epsilon^2\left(\frac{v}{(-u)^{1-\k}}\right)^{1+2\k/(1-\k)-2\delta}.
    \end{equation}
\end{proposition}
\begin{proof}
    This proposition is parallel to the previous scalar-field energy estimate, but now the differentiated pair is $(\df{\nabla\phi},\ddfl{\Omega D_3\phi})$.
    We recall the equations  
    \begin{equation}
        \begin{aligned}
            F_1(\nabla\phi):=&\Omega\nabla_3\df{\nabla^2\phi}-\nabla^2\left[\ol{\Omega D_3\phi}\right]_0^v\\
            \sim &\Lie_{\df{b}}(\nabla^2\phi)_0 -\left[\Omega\chib\right]_0^v(\nabla^2\phi)_0 -\df{\nabla\left({\Omega\chib\cdot \nabla\phi}\right)}.
        \end{aligned}
    \end{equation}
    Commutation formula \eqref{Comm_Formula_D3,nab} implies that 
    \begin{equation}
        \begin{aligned}
            F_6(\nabla\phi):=& \Omega\nabla_3\nabla^5\df{\nabla^2\phi}-\nabla^7\df{\Omega D_3\phi} \\
            =& \nabla^5F_1(\nabla\phi)+ \sum_{i_1+i_2=5}\nabla^{i_1}(\Omega\chib)\nabla^{i_2}\df{\nabla^2\phi}+\sum_{i_1+i_2+i_3=4}\nabla^{i_1}K^{i_2+1}\nabla^{i_3}\df{\Omega D_3\phi}.
        \end{aligned}
    \end{equation}
    Next, we write the relevant evolution equation for $\Omega D_3\phi$. Since $\left[\partial_v,{\partial_A}\right]=0,$ this yields 
    \begin{equation}
        \begin{aligned}
            G_1(e_3\phi):=&\Omega^{-1}\nabla_4\nabla_A\ddfl{\Omega D_3\phi}-\dv\df{\nabla_A\nabla\phi} \\
            = &-\Omega^{-1}\chi\cdot \nabla\ddfl{\Omega D_3\phi} -\frac{1}{2}\Omega^{-2}\partial_A\df{\Omega^2\Omega^{-1}\tr\chi\Omega D_3\phi}\\
            &-\frac{1}{2}\Omega^{-2}\partial_A\df{\Omega^2\Omega\tr\chib\Omega^{-1}e_4\phi}+2\Omega^{-2}\partial_A\df{\etab\cdot\nabla\phi}-\Omega^{-2}\df{\Omega^2 K\nabla_A\phi}\\
            &+\Omega^{-2}\left[\Omega^2\right]_0^v \left(\dv\nabla_A\nabla\phi\right)_0+\left[\dv\right]_0^v\left(\nabla_A\nabla\phi\right)_0 ,
        \end{aligned}
    \end{equation}
    and 
    \begin{equation}
        \begin{aligned}
            G_6(e_3\phi):=&\Omega^{-1}\nabla_4\nabla^6\ddfl{\Omega D_3\phi}-\dv\nabla^5\df{\nabla^2\phi} \\
            =&\nabla^5 G_1(\nabla\phi)+\sum_{i_1+i_2+i_3+i_4=5}\nabla^{i_1}(\eta+\etab)^{i_2}\nabla^{i_3}(\Omega^{-1}\chi)\nabla^{i_4+1}\ddfl{\Omega D_3\phi}\\
            &+\sum_{i_1+i_2+i_3=4}\nabla^{i_1}K^{i_2+1}\nabla^{i_3}\df{\nabla^2\phi}.
        \end{aligned}
    \end{equation}
    From the preceding estimates, the bound for the top-order term $\nabla^6\ddfl{\nabla\phi}$ is known:
    \begin{equation}
        \begin{aligned}
            \lnm \nabla^6\ddfl{\nabla\phi}\rnm^2_{\LS_I^2(\R_{u,v})}\lesssim & \int_0^v\frac{1}{(-u)^{1-\k}}\lnm \nabla^6\ddfl{\nabla\phi}\rnm^2_{\LS_I^2(\Hb_{v^\prime}^u)} \\
            \lesssim &   \epsilon^2\z^{2+2\k/(1-\k)-2\delta}.
        \end{aligned}
    \end{equation}
    Together with $\ddfl{\nabla\phi}-\df{\nabla\phi}=-v\cdot(\Lie_v\nabla\phi)_0$, which is of size $\epsilon\z$ in the $\LS_1^2(S_{u,v})$ norm, we thus have $$\lnm \nabla^6\df{\nabla\phi}\rnm^2_{\LS_I^2(\R_{u,v})}\lesssim \epsilon^2\z^{2+2\delta}.$$ 
    Employing the energy estimates \eqref{Energy Estimates}, we deduce that 
    \begin{equation}
        \begin{aligned}
            &\lnm \nabla^5\df{\nabla^2\phi}\rnm_{\LS_I^2(H_u^v)}^2+\lnm \nabla^6\ddfl{\Omega D_3\phi}\rnm_{\LS_I^2(\Hb_v^u)}^2\\
            \leq & \lnm \nabla^6\df{\nabla\phi} \rnm_{\LS^2_I\left(\mathcal{R}_{u,v}\right)}\lnm F_6(\nabla\phi)\rnm_{\LS^2_I\left({\mathcal{R}_{u,v}}\right)}+\lnm \nabla^6\ddfl{\Omega D_3\phi} \rnm_{\LS^2_I\left({\mathcal{R}_{u,v}}\right)}\lnm G_6(e_3\phi)\rnm_{\LS^2_I\left({\mathcal{R}_{u,v}}\right)}\\
            &+\lnm \nabla^6\df{\nabla\phi} \rnm_{\LS^2_I\left({\mathcal{R}_{u,v}}\right)}^2+\lnm \nabla^6\df{\nabla\phi} \rnm_{\LS^2_I\left(\mathcal{R}_{u,v}\right)}\lnm \nabla^7\left(\ddfl{\Omega D_3\phi}-\df{\Omega D_3\phi}\right) \rnm_{\LS^2_I\left(\mathcal{R}_{u,v}\right)}\\
            \lesssim & \epsilon\left(\frac{v}{(-u)^{1-\k}}\right)^{1+\delta}\lnm F_6(\nabla\phi)\rnm_{\LS^2_I\left(\mathcal{R}_{u,v}\right)}+C(B)\epsilon\left(\frac{v}{(-u)^{1-\k}}\right)^{1+\k/(1-\k)-\delta}\lnm G_6( D_3\phi)\rnm_{\LS^2_I\left(\mathcal{R}_{u,v}\right)}\\
            &+\epsilon^2\z^{2+2\delta}.
        \end{aligned}
    \end{equation}
    It suffices to show that 
     \begin{equation}\label{Main est F6 nabphi}
        \lnm F_6(\nabla\phi)\rnm^2_{\LS^2_I\left(\mathcal{R}_{u,v}\right)}\lesssim C(B)\epsilon^2\left(\frac{v}{(-u)^{1-\k}}\right)^{1+2\k/(1-\k)},
     \end{equation}
     \begin{equation}\label{Main est G6 D3phi}
        \lnm G_6(e_3\phi)\rnm^2_{\LS^2_I\left(\mathcal{R}_{u,v}\right)}\lesssim C(B)\epsilon^2\left(\frac{v}{(-u)^{1-\k}}\right)^{2\k/(1-\k)}.
     \end{equation}
The source estimates below use the same bookkeeping as before: products are expanded by Lemma \ref{RI_Lemma_Calculation_of_df}, non-top-order factors are controlled by Proposition \ref{L2S_est_for_non_top_order_terms}, and the remaining top-order flux terms are precisely those included in the bootstrap assumptions.
Lemma \ref{RI_Lemma_Calculation_of_df} implies:
\begin{equation}
    \begin{aligned}
        &\lnm F_6(\nabla\phi)\rnm^2_{\LS_I^2(\R_{u,v})}\\
        \lesssim & \sum_{i\leq 6}\lnm  \nabla^i\df{\Omega\chib},\nabla^i\df{\nabla\phi}\rnm^2_{\LS_I^2(\R_{u,v})}+ \sum_{i\leq 5,j\leq 4}\lnm \nabla^i\df{b}, \nabla^j\df{K},\nabla^i\df{\Omega D_3\phi}\rnm^2_{\LS_I^2(\R_{u,v})}\\
        \lesssim & \lnm \nabla^6\left(\df{\Omega\chib},\df{\nabla\phi}\right)\rnm^2_{\LS_I^2(\R_{u,v})}+\epsilon^2\z^{2+2\delta}
        \lesssim  B\epsilon^2\z^{1+\frac{2\k}{1-\k}},
    \end{aligned}
\end{equation}
and 
\begin{equation}
    \begin{aligned}
        &\lnm G_6(e_3\phi)\rnm^2_{\LS_I^2(\R_{u,v})}\\
        \lesssim &\lnm \nabla^6\left(\df{\Omega^{-1}e_4\phi},\df{\Omega e_3\phi},\df{\nabla\phi},\df{\Omega\tr\chib},\df{\Omega^{-1}\tr\chi},\df\etab\right),\nabla^5\wt{K} \rnm^2_{\LS_I^2(\R_{u,v})}\\
        &+\sum_{i\leq 5,j\leq 4}\lnm \nabla^i\df{\psi},\nabla^i\df{D\phi},\nabla^j\df{K}\rnm^2_{\LS_I^2(\R_{u,v})} \\
        &+\epsilon^2\lnm \nabla^5\left(\Omega^{-2}\left[\Omega^2\right]_0^v,\left[g^{-1}\partial g\right]_0^v\right)\rnm^2_{\LS_I^2(\R_{u,v})}\\
        \lesssim & B\epsilon^2\z^{\frac{2\k}{1-\k}}.
    \end{aligned}
\end{equation}
We thus finish the proof.
\end{proof}

\subsection{Top-order energy estimates for curvature components} 
The curvature estimates use the same energy mechanism, now applied to the Bianchi pairs. The main point is that the approximants remove the leading singular part, leaving source terms that are either lower order or already controlled by the bootstrap assumptions.

\begin{proposition}
    We establish the following estimates for the relevant quantities:
    \begin{equation}\label{estimate:wt_alpha_wt_beta}
        \lnm\nabla^5\wt{\Omega^{-2}\alpha}\rnm^2_{\LS_I^2(H_u^v)}+\lnm\nabla^5\wt{\Omega^{-1}\beta^r}\rnm^2_{\LS_I^2(\Hb_v^u)}\lesssim \epsilon^2\left(\frac{v}{(-u)^{1-\k}}\right)^{2\delta}.
    \end{equation}
\end{proposition}
\begin{proof} 
This is the first top-order curvature estimate in Region I where the explicitly constructed approximating variables are essential. Equations \eqref{Eq_for_ot_alpha} and \eqref{Eq_for_ot_beta} remove the most singular background part, so the remaining energy pair contains only terms with either a genuine difference factor or an explicit increment from $v=0$. We write the equations as:
    \begin{equation}\label{eq:R_I_nab1_alpha_beta}
        \begin{aligned}
            F_1(\alpha):=&\left(\Omega\nabla_3+\frac{1}{2}\Omega\tr\chib-8\Omega\omegab\right)\wt{\Omega^{-2}\alpha}-\nabla\hat\otimes\wt{\Omega^{-1}\beta^r}\\
            = & \nabla\hat\otimes \left(\ot{\Omega^{-1}\beta^r}\right)-\left(\Lie_{\df{b}}+2\left[\Omega\chib\right]_0^v+\frac{1}{2}\left[\Omega\tr\chib\right]_0^v-8\left[\Omega\omegab\right]_0^v\right)\ot{\Omega^{-2}\alpha}+\varphi^2\ol\varphi+\ol\varphi\Psi_2 ,\\
            G_1(\beta):=& \Omega^{-1}\nabla_4\wt{\Omega^{-1}\beta^r}-\dv\wt{\Omega^{-2}\alpha}\\
            =&\Omega^{-1}\chi\wt{\Omega^{-1}\beta^r}+\Omega^{-2}\left[\Omega^{2}\dv\right]_0^v\ot{\Omega^{-2}\alpha}+(2\eta+\etab)\wt{\Omega^{-2}\alpha}+\left[2\eta+\etab\right]_0^v \ot{\Omega^{-2}\alpha}\\
            &+\varphi\Omega^{-1}\beta^r +\varphi^2\ol\varphi-\Omega^{-1}e_4\phi \nabla\left(\Omega^{-1}e_4\phi\right).
        \end{aligned}
    \end{equation}
    We use the schematical notation $\varphi$ to represent the Ricci coefficients except $\Omega^{-1}\omega$ and derivatives of scalar field, and $\Psi_2$ represent $\Omega^{-1}\beta^r,K,\sigma^r$ in these equations. $\ol{\varphi}$ stands for $\Omega^{-1}\chih, \eta,\etab$, which itself can be seen as difference.
    Commutation formulae gives that 
    \begin{equation}\label{eq:R_I_nab6_alpha_beta}
        \begin{aligned}
            F_6(\alpha):=&\left(\Omega\nabla_3+3\Omega\tr\chib+5\Omega\chibh-8\Omega\omegab\right)\wt{\Omega^{-2}\alpha}-\nabla\hat\otimes\left(\nabla^5\wt{\Omega^{-1}\beta^r}\right)\\
            =& \nabla^{5}F_1(\alpha)+\sum_{i=0}^4\nabla^{i+1}(\Omega\chib,\Omega\omegab)\nabla^{4-i}\wt{\Omega^{-2}\alpha}+\sum_{i_1+i_2+i_3=4} \nabla^{i_1}K^{i_2+1}\nabla^{i_3}\wt{\Omega^{-1}\beta^r},\\
            G_6(\beta):=&\Omega^{-1}\nabla_4\nabla^5\wt{\Omega^{-1}\beta^r}-\dv\nabla^5\wt{\Omega^{-2}\alpha}\\
            =&\nabla^5G_1(\beta)+\sum_{i=0}^5\Omega^{-2}\nabla^{i}(\Omega^2\Omega^{-1}\chi)\nabla^{5-i}\wt{\Omega^{-1}\beta^r}+\sum_{i_1+i_2+i_3=4} \nabla^{i_1}K^{i_2+1}\nabla^{i_3}\wt{\Omega^{-2}\alpha}.
        \end{aligned}
    \end{equation}
    The crucial observation here is that the vast majority of the terms in $F_6(\alpha)$ and $G_6(\beta)$ inherently contain difference quantities, with the notable exceptions of $\nabla\ot{\Omega^{-1}\beta},\varphi^2\ol\varphi$, $\ol\varphi\Psi_2$, $\varphi\Omega^{-1}\beta^r$, and $\nabla(\Omega^{-1}e_4\phi)^2$ arising from $F_1$ and $G_1$. Because these specific terms are identically zero in any exactly spherically symmetric spacetime, they can be rigorously treated as essentially difference quantities themselves. Consequently, their $\LS^2_1(S_{u,v})$ norms are straightforwardly bounded by $O(\epsilon)$, implying that their corresponding $\LS^2_I(\R_{u,v})$ norms are comfortably bounded by $\epsilon\z^{\delta}$. Another potentially dangerous term that requires careful handling is $\left[b,\Omega\chib,\Omega\omegab\right]_0^v\ot{\Omega^{-2}\alpha}$. Drawing upon our prior estimates for $\ot{\Omega^{-1}\chih}$, we know that $\lnm\ot{\Omega^{-2}\alpha}\rnm_{\LS^2_1(S_{u,v})}\lesssim\epsilon \z^{\k/(1-\k)-\delta -1}$. The leading $\left[\cdot \right]_0^v$ difference factor provides a crucial extra power of $\z$. Therefore, this composite term exhibits an $\LS^2_1(S_{u,v})$ norm bounded by $\epsilon\z^{\k/(1-\k)-\delta}$, meaning its full spacetime $\LS^2_I(\R_{u,v})$ norm is securely bounded by $\epsilon\z^{\k/(1-\k)}$. With these structural insights established, we now proceed to the rigorous details of the estimate. The structural paragraph above explains why the sources can be treated as small after measuring them in the weighted spacetime norm. The energy inequality below uses this fact together with Proposition \ref{L2S_est_for_non_top_order_terms}, Corollary \ref{Estimate_of_tilde_nabla}, and the bootstrap bounds for the remaining curvature components. 
    By energy estimates \eqref{Energy Estimates}, we deduce that 
        \begin{equation}
            \begin{split}
                &\lnm \nabla^5\wt{\Omega^{-2}\alpha}\rnm_{\LS_I^2(H_u^v)}^2+\lnm \nabla^5\wt{\Omega^{-1}\beta^r}\rnm_{\LS_I^2(\Hb_v^u)}^2\\
                \leq & (o_{\epsilon,\epsilon_1}(1) -1-2-\k-(1-\k)(1-2\delta)+2+4\k)\lnm \nabla^5\wt{\Omega^{-2}\alpha}\rnm_{\LS_I^2(\mathcal{R}_{u,v})}^2\\
                &+\lnm \nabla^5\wt{\Omega^{-2}\alpha}\rnm_{\LS_I^2(\mathcal{R}_{u,v})}\lnm F_6(\alpha)\rnm_{\LS_I^2(\mathcal{R}_{u,v})}+\lnm \nabla^5\wt{\Omega^{-1}\beta^r}\rnm_{\LS_I^2(\mathcal{R}_{u,v})}\lnm G_6(\beta)\rnm_{\LS_I^2(\mathcal{R}_{u,v})}\\
                \leq & \lnm \nabla^5\wt{\Omega^{-2}\alpha}\rnm_{\LS_I^2(\mathcal{R}_{u,v})}\lnm F_6(\alpha)\rnm_{\LS_I^2(\mathcal{R}_{u,v})}\\
                &+\frac{v}{(-u)^{1-\k}}\lnm \nabla^5\wt{\Omega^{-1}\beta^r}\rnm_{\LS_I^2(\Hb_v^u)}\lnm G_6(\beta)\rnm_{\LS_I^2(\mathcal{R}_{u,v})}-\frac{1}{2}\lnm \nabla^5\wt{\Omega^{-2}\alpha}\rnm_{\LS_I^2(\mathcal{R}_{u,v})}^2\\
                \leq & \lnm F_6(\alpha)\rnm^2_{\LS_I^2(\R_{u,v})}+\delta \lnm \nabla^5\wt{\Omega^{-1}\beta^r}\rnm^2_{\LS_I^2(\Hb_v^u)}\\
            &+\frac{1}{\delta}\z^2\lnm G_6(\beta)\rnm^2_{\LS_I^2(\R_{u,v})}-\frac{1}{4}\lnm \nabla^5\wt{\Omega^{-2}\alpha}\rnm_{\LS_I^2(\mathcal{R}_{u,v})}^2.
            \end{split}
        \end{equation} 
The preceding energy inequality has reduced the argument to two source estimates. The estimate for $F_6(\alpha)$ uses the low-order Ricci and scalar bounds, while the estimate for $G_6(\beta)$ additionally uses the lapse increment estimate \eqref{Estimate_Omega/Omega_0} for the terms containing the explicit approximating curvature. 
It remains to prove that 
\begin{equation}
    \lnm F_6(\alpha)\rnm_{\LS_I^2(\mathcal{R}_{u,v})}^2\lesssim \epsilon^2\z^{2\delta},\ \lnm G_6(\beta)\rnm^2_{\LS_I^2(\R_{u,v})}\lesssim\epsilon^2.
\end{equation}
Since we already have fifth order $\LS^2_1(S_{u,v})$-estimate for Ricci coefficients and derivatives of scalar field from Proposition \ref{L2S_est_for_non_top_order_terms}, $$\sum_{i\leq 5}\lnm \nabla^i\left(\ol{\psi},\ol{D\phi}\right)\rnm^2_{\LS^2_1(S_{u,v})}\lesssim\epsilon^2,$$
Lemma \ref{RI_Lemma_Calculation_of_ol} implies that
\begin{equation}
    \begin{aligned}
        &\lnm G_6(\beta)\rnm^2_{\LS^2_I(R_{u,v})}\\
        \lesssim & \sum_{i\leq 5,j\leq 4}\lnm\nabla^i \left({\Omega^{-1}\beta^r},\ol{\varphi},\nabla\left(\Omega^{-1}e_4\phi\right)\right),\nabla^j\wt{\Omega^{-2}\alpha}\rnm^2_{\LS^2_I(R_{u,v})}\\
        &+\lnm\nabla^5 \left(\Omega^{-2}\left[\Omega^2\dv\right]_0^v\ot{\Omega^{-2}\alpha},\left[2\eta+\etab\right]_0^v\ot{\Omega^{-2}\alpha}\right)\rnm^2_{\LS^2_I(R_{u,v})}+\lnm(\eta,\etab)\nabla^5 \wt{\Omega^{-2}\alpha}\rnm^2_{\LS^2_I(R_{u,v})}.
    \end{aligned}
\end{equation}
Because of the $\z$-power in $\left[\cdot\right]_0^v$, the terms with $\ot{\Omega^{-2}\alpha}$ can be controlled with estimate:
\begin{equation}
\begin{aligned}
        &\lnm\nabla^5 \left(\Omega^{-2}\left[\Omega^2\dv\right]_0^v\ot{\Omega^{-2}\alpha},\left[2\eta+\etab\right]_0^v\ot{\Omega^{-2}\alpha}\right)\rnm^2_{\LS^2_I(R_{u,v})}\\
        \lesssim & \int_0^v\int_{-1}^u\frac{1}{(-u^\prime)^{2-\k}}{\left(\frac{v^\prime}{(-u^\prime)^{1-\k}}\right)}^{2+2\k/(1-\k)-2\delta-2+2\delta-1}\epsilon^2\lesssim\epsilon^2\z^{2\k/(1-\k)}.
\end{aligned}
\end{equation}
The last term is bounded by $\epsilon^2\lnm \nabla^5\wt{\Omega^{-2}\alpha}\rnm^2_{\LS^2_I(\R_{u,v})}$ for the smallness of $\eta,\etab$. For the first term on the right hand side, we note that 
\begin{equation}
    \ol{\psi}=\df{\psi}+\ol{\psi}_0=\ddfl{\psi}+\ol{\psi}_0+v\cdot\left(\Lie_v\ol{\psi}\right)_0=\wt{\psi}-\left(\ot{\psi}-\psi^c\right).
\end{equation}
From the bootstrap assumptions, and the initial value construction, we deduce that 
\begin{equation}
    \begin{aligned}
        &\sum_{i\leq 5}\lnm\nabla^i \left({\Omega^{-1}\beta^r},\ol{\varphi},\nabla\left(\Omega^{-1}e_4\phi\right)\right)\rnm^2_{\LS^2_I(R_{u,v})}\\
        \lesssim & \sum_{i\leq 5}\lnm\nabla^i \left(\wt{\Omega^{-1}\beta^r},\df{\varphi},\nabla\left(\wt{\Omega^{-1}e_4\phi}\right)\right)\rnm^2_{\LS^2_I(R_{u,v})}+\int_0^v \int_{-1}^u\frac{1}{(-u^\prime)^{2-\k}}\left(\frac{v^\prime}{(-u^\prime)^{1-\k}}\right)^{2\delta-1}\epsilon^2\\
        \lesssim &\epsilon^2\z^{2\delta}.
    \end{aligned}
\end{equation}
Therefore we have proved that $\lnm G_6(\beta)\rnm^2_{\LS^2_I(\R_{u,v})}\lesssim \epsilon^2\z^{2\delta}$.
As for $F_6(\alpha)$, the estimates read
\begin{equation}
    \begin{aligned}
        \lnm F_6(\alpha)\rnm^2_{\LS_I^2(\R_{u,v})}
        \lesssim & \sum_{i\leq 4,j\leq 5}\lnm \nabla^i\left(\wt{\Omega^{-2}\alpha},\wt{\Omega^{-1}\beta^r}\right),\nabla^j\left(\ol{\varphi},\ol{\varphi}\Psi_2\right)\rnm^2_{\LS_I^2(\R_{u,v})}\\
        &+\sum_{\psi\in\{ \Omega\chib,\Omega\omegab\}}\lnm \nabla^6\ot{\Omega^{-1}\beta^r},\nabla^5\left(\Lie_{\left[b\right]_0^v}+\left[\psi\right]_0^v\right)\ot{\Omega^{-2}\alpha}\rnm^2_{\LS^2_I(\R_{u,v})}.
    \end{aligned}
\end{equation}
The $\ot{\Omega^{-2}\alpha}$-term is bounded similarly as $\left[2\eta+\etab\right]_0^v\ot{\Omega^{-2}\alpha}$, and using $\lnm \nabla^i\ot{\Omega^{-1}\beta^r}\rnm_{\LS_1(S_{u,v})}\lesssim\epsilon$, we have its $\LS_I^2(\R_{u,v})$ norm is bounded by $\epsilon\z^\delta$. 
\begin{equation}
    \sum_{j\leq 5}\lnm \nabla^j\left(\ol{\varphi}\Psi_2\right)\rnm^2_{\LS_I^2(\R_{u,v})}\lesssim \sum_{i,j\leq 5}\sup_{u^\prime,v^\prime}\lnm \nabla^i\ol{\varphi}\rnm^2_{\LS_1^2(S_{u,v})}\lnm \nabla^j\Psi_2\rnm^2_{\LS_I^2(\R_{u,v})}\lesssim\epsilon^2\z^{2\delta}.
\end{equation}
We finally obtain that 
\begin{equation}
    \lnm\nabla^5\wt{\Omega^{-2}\alpha}\rnm^2_{\LS_I^2(H_u^v)}+\lnm\nabla^5\wt{\Omega^{-1}\beta^r}\rnm^2_{\LS_I^2(\Hb_v^u)}\lesssim \epsilon^2\left(\frac{v}{(-u)^{1-\k}}\right)^{2\delta}.
\end{equation}

\end{proof}

The next proposition closes the bootstrap assumption for other curvature components.
\begin{proposition}\label{R_I_Estimate_beta_sigma_K_betab_alphab}
    For $\Psi_1\in\{\Omega^{-1}\beta^r,\sigma^r,K,\Omega\betab^r\}$ and $\Psi_2\in\{\sigma^r,K,\Omega\betab^r,\Omega^2\alphab\}$, we have
    \begin{equation}\label{R_I_Top_order_Est_beta_sigma_K_betab_alphab}
        \begin{split}
            &\lnm\nabla^5\df{\Psi_1}\rnm^2_{\LS_I^2(H_u^v)}+\lnm\nabla^5\df{\Psi_2}\rnm^2_{\LS_I^2(\Hb_v^u)}
            \leq  \epsilon^2\left(\frac{v}{(-u)^{1-\k}}\right)^{2\k/(1-\k)}.
        \end{split}
    \end{equation}
\end{proposition}
\begin{proof}
    For pair $(\Psi_1,\Psi_2)$ taking values in $(\Omega^{-1}\beta^r,(\sigma^r,K^r))$, $((\sigma^r,K^r),\Omega\betab^r)$,\\
    and $(\Omega\betab^r,\Omega^2\alphab)$, the schematic equations are 
    \begin{equation}
        \begin{aligned}
            F_1(\Psi_1):=&\Omega\nabla_3\df{\Psi_1}-\D \df{\Psi_2}\\
            =& \df{\psi\Psi}+\df{\varphi\nabla\varphi}+\df{\varphi^3} \\
            &+\left[\D\right]_0^v \ol{\Psi_2}(0)+\ol{\Omega\chib}\left[\Psi_1^c\right]_0^v+\left(\left[\Omega\chib\right]_0^v+\Lie{\left[b\right]_0^v}\right)\ol{\Psi_1}(0)+\df{\Omega\chib}\Psi_1^c(0),\\
            G_1(\Psi_2):=&\Omega^{-1}\nabla_4\df{\Psi_2}+{}^*\D\df{\Psi_1}\\
            = &\ol{\psi\Psi}+\ol{\varphi\nabla\varphi}+\ol{\varphi^3}  -{}^*\D\ol{\Psi_1}(0) -\ol{\Omega^{-2}}\Lie_v\Psi_2^c+\ol{\Omega^{-1}\chi}\Psi_2^c,
        \end{aligned}
    \end{equation}
    where $\varphi$ can be $\Omega\chib,\Omega^{-1}\chi,\eta,\etab,\Omega\omegab$ and $\nabla\phi,\Omega\nabla_3\phi,\Omega^{-1}D_4\phi$ and $\Psi_2$ represents $\Omega^{-1}\beta^r$, $\sigma^r$, $K^r$, $\Omega\betab^r$, $\Omega^2\alphab$. The three pairs are handled together because each Bianchi system has the same energy structure after the spherical symmetry terms have been subtracted. The estimates used below are the preceding top-order control of $\wt{\Omega^{-1}\beta^r}$, the non-top-order bounds, and Lemma \ref{RI_Lemma_Calculation_of_df}. By commutation formulae, it follows that 
\begin{equation}
    \begin{aligned}
        F_6(\Psi_1):=& \Omega\nabla_3\nabla^5\df{\Psi_1}-\D \nabla^5\df{\Psi_2}\\
        =& \sum_{{}^{i_1+i_2+i_3+i_4=5}}\nabla^{i_1}\eta^{i_2}\nabla^{i_3}\Omega\chib\nabla^{i_4}\df{\Psi_1}+\sum_{{}^{i_1+i_2+i_3=4}}\nabla^{i_1}K^{i_2+1}\nabla^{i_3}\df{\Psi_2}+\nabla^5 F_1(\Psi_1),\\
        G_6(\Psi_2):=& \Omega^{-1}\nabla_4\nabla^5\df{\Psi_2}+{}^*\D\nabla^5\df{\Psi_1}\\
        =& \sum_{i_1+i_2+i_3+i_4=5}\nabla^{i_1}(\eta,\etab)^{i_2}\nabla^{i_3}\Omega^{-1}\chi\nabla^{i_4}\df{\Psi_2}+\sum_{i_1+i_2+i_3=4}\nabla^{i_1}K^{i_2+1}\nabla^{i_3}\df{\Psi_1}\\
        &+\sum_{i_1+i_2+i_3=5}\nabla^{i_1}(\eta,\etab)^{i_2}\nabla^{i_3}G_1(\Psi_2).
    \end{aligned}
\end{equation}
The energy estimate \eqref{Energy Estimates} gives that 
\begin{equation}
    \begin{aligned}
        &\lnm \nabla^5\df{\Psi_1}\rnm^2_{\LS_I^2(H_u^v)}-\lnm \nabla^5\df{\Psi_1}\rnm^2_{\LS_I^2(H_{-1}^v)}+\lnm \nabla^5\df{\Psi_2}\rnm^2_{\LS_I^2(\Hb_v^u)}-\lnm \nabla^5\df{\Psi_2}\rnm^2_{\LS_I^2(\Hb_v^u)}\\
        \lesssim & \lnm \nabla^5\df{\Psi_1}\rnm^2_{\LS_I^2(\R_{u,v})}+\lnm \nabla^5\df{\Psi_1}\rnm_{\LS_I(\R_{u,v})}\lnm F_6(\Psi_1)\rnm_{\LS_I^2(\R_{u,v})}\\
        &+\lnm \nabla^5\df{\Psi_2}\rnm_{\LS_I^2(\R_{u,v})}\lnm G_6(\Psi_2)\rnm_{\LS_I^2(\R_{u,v})}.
    \end{aligned}
\end{equation}
Because of $\df{\Omega^{-1}\beta^r}-\wt{\Omega^{-1}\beta^r}=\ot{\Omega^{-1}\beta^r}-\left(\Omega^{-1}\beta^r\right)_0\sim_{\LS^2_1(S_{u,v})} \epsilon\z^{\k/(1-\k)-\delta/2}, $ we can use \eqref{estimate:wt_alpha_wt_beta} to obtain the bulk estimate for $\df{\Omega^{-1}\beta^r}$,
\begin{equation}\label{R_I_ENERGY_EST_df_betar}
    \begin{aligned}
        &\lnm \nabla^5\df{\Omega^{-1}\beta^r}\rnm^2_{\LS_I^2(\R_{u,v})}\\
        \lesssim & \lnm \nabla^5\wt{\Omega^{-1}\beta^r}\rnm^2_{\LS_I^2(\R_{u,v})}+\int_{0}^v\int_{-1}^u(-u^\prime)^{\k-2}\left(\frac{v^\prime}{(-u^\prime)^{1-\k}}\right)^{2\delta-1}\cdot \epsilon^2 \left(\frac{v^\prime}{(-u^\prime)^{1-\k}}\right)^{\frac{2\k}{1-\k}-\delta}\\
        \lesssim &\epsilon^2\z^{1+2\delta}+ \epsilon^2\z^{\delta+\frac{2\k}{1-\k}},
    \end{aligned}
\end{equation}
and by bootstrap assumption, we can derive that
\begin{equation}
    \lnm \nabla^5\df{\Psi_2}\rnm^2_{\LS_I^2(\R_{u,v})}\lesssim B\epsilon^2\z^{1+\frac{2\k}{1-\k}}.
\end{equation}
It remains to show 
\begin{equation}
    \lnm  G_6(\Psi_2)\rnm_{\LS_I^2(\mathcal{R}_{u,v})}^2\lesssim \epsilon^2,\quad 
    \lnm F_6(\Psi_1)\rnm_{\LS_I^2(\mathcal{R}_{u,v})}^2\lesssim \epsilon^2\left(\frac{v}{(-u)^{1-\k}}\right)^{\frac{2\k}{1-\k}}.
\end{equation}
First for $G_6$, we deduce that 
\begin{equation}
    \begin{aligned}
        &\lnm G_6(\Psi_2)\rnm^2_{\LS_I^2(\R_{u,v})}\\\lesssim & \sum_{i\leq 5}\lnm \nabla^i G_1(\Psi_2)\rnm^2_{\LS_I^2(\R_{u,v})}+\sum_{i\leq 5}\lnm \nabla^i(\eta,\etab)\rnm^2_{\LS_I^2(\R_{u,v})}\\
        &+\sum_{i\leq 5}\lnm \nabla^i\left(\df{\eta},\df\etab,\df{\Omega^{-1}\chi},\df{{\Psi_2}}\right)\rnm^2_{\LS_I^2(\R_{u,v})}+\sum_{i\leq 4}\lnm \nabla^i\ol{\Psi_1}\rnm^2_{\LS_I^2(\R_{u,v})}.
    \end{aligned}
\end{equation}
The last three terms can be estimated by the bootstrap assumption so that 
\begin{equation}
    \begin{aligned}
        &\sum_{i\leq 5}\lnm \nabla^i(\eta,\etab)\rnm^2_{\LS_I^2(\R_{u,v})}+\sum_{i\leq 5}\lnm \nabla^i\left(\df{\eta},\df\etab,\df{\Omega^{-1}\chi},\df{{\Psi_2}}\right)\rnm^2_{\LS_I^2(\R_{u,v})}\\
        &\qquad+\sum_{i\leq 4}\lnm \nabla^i\ol{\Psi_1}\rnm^2_{\LS_I^2(\R_{u,v})}
        \lesssim \epsilon^2\z^{2\delta}<\epsilon^2,
    \end{aligned}
\end{equation}
and the first term satisfies
\begin{equation}
    \begin{aligned}
        \sum_{i\leq 5}\lnm \nabla^iG_1(\Psi_2)\rnm^2_{\LS_I^2(\R_{u,v})}\lesssim & \sum_{i\leq 5}\lnm \nabla^i\left(\ol{\Psi},\ol{\varphi},\ol{\nabla\varphi},O_{H^{N-3}(\epsilon)}\right)\rnm^2_{\LS_I^2(\R_{u,v})}\\
        \lesssim & \epsilon^2\z^{2\delta}+\sum_{i\leq 5}\lnm \nabla^i\left(\df{\Psi},\df{\varphi},\df{\nabla\varphi}\right)\rnm^2_{\LS_I^2(\R_{u,v})}\\
        \lesssim & \epsilon^2\z^{2\delta}<\epsilon^2.
    \end{aligned}
\end{equation} 
The estimate \eqref{R_I_ENERGY_EST_df_betar} is the only place where the approximation of $\Omega^{-1}\beta^r$ enters this combined curvature argument. It upgrades the tilde estimate from the previous proposition to the $\df{\cdot}$-estimate required in the present Bianchi pair. 
Using estimate \eqref{R_I_ENERGY_EST_df_betar}, the bound for $F_6(\Psi_1)$ is controlled by:
\begin{equation}
    \begin{aligned}
        \sum_{i_1+i_2+i_3+i_4=5}&\lnm\nabla^{i_1}\eta^{i_2}\nabla^{i_3}\Omega\chib\nabla^{i_4}\df{\Psi_1}\rnm^2_{\LS_I^2(\R_{u,v})}+\sum_{i_1+i_2+i_3=4}\lnm\nabla^{i_1}K^{i_2+1}\nabla^{i_3}\df{\Psi_2}\rnm^2_{\LS_I^2(\R_{u,v})}\\
        \lesssim & \sum_{i\leq 5}\lnm \nabla^i\left(\df{\eta},\df{\Omega\chib},\df{\Psi_2}\right)\rnm^2_{\LS_I^2(\R_{u,v})}+\sum_{i\leq 5}\lnm \nabla^i\df{\Omega^{-1}\beta^r}\rnm^2_{\LS_I^2(\R_{u,v})}\\
        \lesssim & B\epsilon^2\z^{1+\frac{2\k}{1-\k}}+ \epsilon^2\z^{\frac{2\k}{1-\k}+\delta}\ll \epsilon^2\z^{\frac{2\k}{1-\k}},
    \end{aligned}
\end{equation}
and 
\begin{equation}
    \begin{aligned}
        &\lnm \nabla^5F_1(\Psi_1) \rnm^2_{\LS_I^2(\R_{u,v})}\\
        \lesssim & \sum_{i\leq 5}\lnm \left(\df{\varphi},\df{\nabla\varphi},\df{\Psi}\right) \rnm^2_{\LS_I^2(\R_{u,v})}+\lnm \nabla^5\left(\left[\D\right]_0^v\ol{\Psi_2}(0) \right)\rnm^2_{\LS_I^2(\R_{u,v})}\\
        \lesssim & \epsilon^2\z^{\frac{2\k}{1-\k}+\delta}+\sum_{i\leq 5}\lnm \left(\wt{\varphi},\wt{\nabla\varphi},\wt{\Psi}\right) \rnm^2_{\LS_I^2(\R_{u,v})}+\epsilon^2\sum_{i\leq 5}\lnm \partial^i\left(\left[g\right]_0^v+\left[g\partial g\right]_0^v\right) \rnm^2_{\LS_I^2(\R_{u,v})}\\
        \lesssim & \epsilon^2\z^{\frac{2\k}{1-\k}+\delta}+B\epsilon^2\z^{1+2\delta}+B\epsilon^2\z^{2+2\delta}\\
        \ll & \epsilon^2\z^{\frac{2\k}{1-\k}}.
    \end{aligned}
\end{equation}
We thus finish the proof.

\end{proof}

To successfully close the energy estimates for all related Ricci coefficients, we require an enhanced result specifically for $\Omega\betab^r$. 
\begin{proposition}
    More precisely, the following top-order estimate holds:
    \begin{equation}\label{R_I_Top_order_est_for_ddfl_betab}
        \lnm \nabla^5(\df{\sigma^r},\df{K})\rnm^2_{\LS^2(H_{u}^v)}+\lnm \nabla^5\ddfl{\Omega\betab^r}\rnm^2_{\LS^2(\Hb_v^u)}\leq \epsilon^2\left(\frac{v}{(-u)^{1-\k}}\right)^{1+2\k/(1-\k)-\delta}.
    \end{equation}
\end{proposition}
\begin{proof}
    This proof separates the incoming curvature component $\Omega\betab^r$ because the first $v$-Taylor term must be removed before the energy estimate has the desired gain.
    Recall that  
    \begin{align*}
        G_1(\betab):=&\Omega^{-1}\nabla_4 \ddfl{\Omega\betab^r} -\nabla \df{K}-{}^*\nabla\df{\sigma^r}\\
    =& \Omega^{-2}\left[\Omega^2\right]_0^v\nabla\left(\ol{\Omega^{-1}\nabla_4\Omega\betab^r}(0)\right)+\varphi\ddfl{\Omega\betab^r} +\df{\varphi\Psi_2}+\df{\varphi^3}+\df{\varphi\nabla\varphi},
    \end{align*}
    where $\varphi\in\{\Omega\chib,\Omega^{-1}\chi,\etab,\nabla\phi,\Omega D_3\phi,\Omega^{-1}D_4\phi\}$.
\begin{align*}
    G_6(\betab) :=&\Omega^{-1}\nabla_4 \nabla^5\ddfl{\Omega\betab^r} -\nabla\nabla^5 \df{K}-{}^*\nabla\nabla^5\df{\sigma^r} \\
    =& \sum_{i_1+i_2+i_3=5}\nabla^{i_1}(\eta+\etab)^{i_2}\nabla^{i_3}G_1(\betab)+\sum_{i_1+i_2+i_3+i_4=5}\nabla^{i_1}(\eta,\etab)^{i_2}\nabla^{i_3}(\Omega^{-1}\chi)\nabla^{i_4}\ddfl{\Omega\betab^r}\\
    &+\sum_{i_1+i_2+i_3=4}\nabla^{i_1}K^{i_2+1}\nabla^{i_3}\left(\df{K},\df{\sigma^r}\right).
    \end{align*}
    Also, we obtain 
    \begin{equation}
        \begin{aligned}
            &F_6(\sigma,K)=: \Omega\nabla_3\nabla^5\left(\df{\sigma^r},\df{K}\right)+{}^*\left(\nabla,{}^*\nabla\right) \nabla^5\df{\Omega\betab^r}\\
        &\qquad= \sum_{i_1+i_2 =5}\nabla^{i_1}(\Omega\chib)\nabla^{i_2}\left(\df{\sigma^r},\df{K}\right)+\sum_{{i_1+i_2+i_3=4}}\nabla^{i_1}K^{i_2+1}\nabla^{i_3}\df{\Omega\betab^r}+\nabla^5 F_1(\sigma,K),
        \end{aligned}
    \end{equation}
    with 
    \begin{equation}
        F_1(\sigma,K)= \df{\varphi^3}+\df{\varphi\nabla(\eta,\Omega D_3\phi)}+\df{\psi(\sigma^r,K^r,\Omega\betab^r)}.
    \end{equation}
    By energy estimates, we have 
    \begin{align*}
        &\lnm \nabla^5\df{\sigma^r},\nabla^5\df{K}\rnm^2_{\LS_I^2(H_u^v)}+\lnm \nabla^5\ddfl{\Omega\betab^r}\rnm^2_{\LS_I^2(\Hb_v^u)}\\
        \lesssim & \lnm \nabla^5\df{\sigma^r},\nabla^5\df{K^r}\rnm_{\LS_I^2(\mathcal{R}_{u,v})}\lnm F_6(\sigma,K)\rnm_{\LS_I^2(\mathcal{R}_{u,v})}\\
        &+\lnm \nabla^5\ddfl{\Omega\betab^r}\rnm_{\LS_I^2(\mathcal{R}_{u,v})}\lnm G_6(\betab)\rnm_{\LS_I^2(\mathcal{R}_{u,v})}+\lnm \nabla^5\df{\sigma^r},\nabla^5\df{K}\rnm_{\LS_I^2(\mathcal{R}_{u,v})}^2\\
        &+\lnm \nabla^5\df{\sigma^r},\nabla^5\df{K}\rnm_{\LS_I^2(\mathcal{R}_{u,v})}\lnm \nabla^6\left(v\cdot\Lie_v\ol{\Omega\betab^r}(0)\right)\rnm_{\LS_I^2(\mathcal{R}_{u,v})}.
    \end{align*}
    The source bounds repeatedly use Proposition \ref{L2S_est_for_non_top_order_terms}, the combined curvature estimate \eqref{R_I_Top_order_Est_beta_sigma_K_betab_alphab}, and the product lemma for $\ddfl{\cdot}$. It suffices to prove that 
    \begin{equation}
        \begin{split}
            \lnm F_6(\sigma,K)\rnm_{\LS_I^2(\mathcal{R}_{u,v})}^2\leq & C(B)\epsilon^2\left(\frac{v}{(-u)^{1-\k}}\right)^{1+2\k/(1-\k)},\\
            \lnm G_6(\betab)\rnm_{\LS_I^2(\mathcal{R}_{u,v})}^2\leq & C(B)\epsilon^2\left(\frac{v}{(-u)^{1-\k}}\right)^{2\k/(1-\k)}.
        \end{split}
    \end{equation} 
The goal is now reduced to estimating the two commuted sources. In $G_6(\betab)$, terms with $\Omega\betab^r$ are controlled by the current bootstrap and then improved, while terms with $K$ and $\sigma^r$ use the combined curvature estimate just proved. 
    We first estimate $G_6(\betab)$. Using Lemma \ref{RI_Lemma_Calculation_of_df} and Proposition \ref{L2S_est_for_non_top_order_terms}, this yields 
    \begin{equation}
        \begin{aligned}
            &\sum_{i_1+i_2+i_3+i_4=5}\lnm\nabla^{i_1}(\eta,\etab)^{i_2}\nabla^{i_3}(\Omega^{-1}\chi)\nabla^{i_4}\ddfl{\Omega\betab^r}\rnm^2_{\LS_I^2(\R_{u,v})} \\
            &\qquad\qquad\qquad\qquad+\sum_{i_1+i_2+i_3=4}\lnm\nabla^{i_1}K^{i_2+1}\nabla^{i_3}\left(\df{K},\df{\sigma^r}\right)\rnm^2_{\LS_I^2(\R_{u,v})} \\
            \lesssim & \sum_{i\leq 5,j\leq 4}\lnm \nabla^i\left(\df{\eta},\df{\etab},\df{\Omega^{-1}\chi}\right),\nabla^j\left(\df{K^r},\df{\sigma^r},\ddfl{\Omega\betab^r}\right) \rnm^2_{\LS_I^2(\R_{u,v})}\\
            &+\lnm \nabla^5\ddfl{\Omega\betab^r} \rnm^2_{\LS_I^2(\R_{u,v})}\\
            \lesssim & \epsilon^2\z^{2}+\epsilon^2\z^{1+\frac{2\k}{1-\k}}\leq C_0\epsilon^2\z^{1+\frac{2\k}{1-\k}},
        \end{aligned}
    \end{equation}
    and 
    \begin{equation}
        \begin{aligned}
            &\sum_{i_1+i_2+i_3=5}\lnm\nabla^{i_1}(\eta+\etab)^{i_2}\nabla^{i_3} G_1(\betab) \rnm^2_{\LS_I^2(\R_{u,v})}\\
            \lesssim & \sum_{i\leq 5,j\leq 4}\lnm \nabla^{i}\left(\df{\varphi},\epsilon\cdot\left[\Omega^2\right]_0^v\right),\nabla^j\left(\df{\Psi_2}\right)\rnm^2_{\LS_I^2(\R_{u,v})}\\
            &\qquad\qquad\qquad+\lnm \nabla^6\df\varphi,\nabla^5\df{\Psi_2}\rnm^2_{\LS_I^2(\R_{u,v})}+\epsilon^2\z^{2+2\delta}\\
            \lesssim & \epsilon^2\z^{2+2\delta}+B\epsilon^2\z^{1+2\delta}<\epsilon^2\z^{\frac{2\k}{1-\k}}.
        \end{aligned}
    \end{equation}
    When estimating $F_6(\sigma,K)$, a key algebraic observation is that the potentially deeply singular terms ${\Omega^{-1}\beta^r},{\Omega^{-1}\chih},{\Omega^{-1}D_4\phi}$ do not appear at all in the $\Omega\nabla_3({K^r},{\sigma^r})$ evolution equations. This structural absence crucially implies that 
    \begin{equation}
        \lnm\left[\Psi^c\right]_0^v,\left[\varphi^c\right]_0^v\rnm^2_{\LS_1^2(\R_{u,v})}\lesssim \z^2.
    \end{equation}
    Using Lemma \ref{RI_Lemma_Calculation_of_df}, we obtain 
    \begin{equation}
        \begin{aligned}
            &\lnm F_6(\sigma,K) \rnm^2_{\LS_I^2(\R_{u,v})}\\
            \lesssim & \sum_{i\leq 5}\lnm\nabla^i\left(\df{\eta},\df{\Omega\chib},\df{\varphi}\right) \rnm^2_{\LS_I^2(\R_{u,v})}+\sum_{j\leq 4}\lnm\nabla^j\left(\df{\sigma^r},\df{K},\df{\Omega\betab^r}\right) \rnm^2_{\LS_I^2(\R_{u,v})}\\
            +&\lnm\nabla^6\left(\df{\eta},\df{\Omega D_3\phi}\right),\nabla^5\left(\df{\sigma^r},\df{K},\df{\Omega\betab^r}\right) \rnm^2_{\LS_I^2(\R_{u,v})}+\epsilon^2\left(\frac{v}{(-u)^{1-\k}}\right)^{2+2\delta}.
        \end{aligned}
    \end{equation}
    By Proposition \ref{L2S_est_for_non_top_order_terms}, the non top order terms can be bounded by $\epsilon^2\z^{2}$, and with the bootstrap assumption \eqref{RI_Bootstrap_Assumptions}, the top order terms are bounded by $B\epsilon^2\z^{1+\frac{2\k}{1-\k}}$. We thus finish the proof.

\end{proof}

\subsection{Top order estimates for Ricci coefficients}
After the scalar-field and curvature energies have been closed, the Ricci coefficients are recovered from their transport and elliptic equations. The estimates below complete the improvement of the Region I bootstrap assumptions.

\begin{proposition}\label{R_I_Energy_Est_omegab}
    With the previous propositions, we have estimate for $\ddfl{\Omega\omegab}$,
    \begin{equation}
        \lnm \nabla^6\ddfl{\Omega\omegab} \rnm^2_{\LS_I^2(\Hb_v^u)}\lesssim \epsilon^2\left(\frac{v}{(-u)^{1-\k}}\right)^{1+2\k/(1-\k)-\delta}.
    \end{equation}
\end{proposition}
\begin{proof}
    We introduce auxiliary function ${\Omega\omegab^\dagger}$:
        \[\Omega^{-1}\nabla_4{\Omega\omegab^\dagger}=\frac{1}{2}{\sigma^r},\quad {\Omega\omegab^\dagger}|_{\Hb_0}=0.\] 
        The auxiliary Hodge potential converts the estimate of $\Omega\omegab$ into a transport estimate for a curl-divergence pair. The estimates used in this proof are the enhanced $\Omega\betab^r$ bound \eqref{R_I_Top_order_est_for_ddfl_betab}, the non-top-order Proposition \ref{L2S_est_for_non_top_order_terms}, and the lapse estimate \eqref{Estimate_of_df_Omega}.
        It follows immediately by the transport estimate \eqref{transport_4},
        \begin{equation}\label{non_top_est_reduced_omegab}
            \begin{aligned}
                \sum_{i\leq 5}\lnm \nabla^i\ddfl{\Omega\omegab^\dagger}\rnm^2_{\LS_I^2(S_{u,v})}\lesssim & \frac{v}{(-u)^{1-\k}}\sum_{i\leq 5,j\leq 4}\lnm \nabla^i\left(\df{\sigma^r},\df{\Omega\omegab^\dagger}\right),\nabla^j\left(\df\eta,\df\etab\right)\rnm^2_{\LS_I^2(H_u^v)}\\
                \lesssim & \epsilon^2\z^{1+2\k/(1-\k)}.
            \end{aligned}
        \end{equation}
We further write 
\begin{equation}
            \begin{split}
                \wt{\omegab^r}_A=&-\nabla_A\ddfl{\Omega \omegab}+{}^*\nabla_A\ddfl{\Omega \omegab^\dagger}-\frac{1}{2}\ddf{\ol{\Omega \betab^r}_A},
            \end{split}
        \end{equation}
        and we can derive schematical equation that 
        \begin{equation}
            \begin{aligned}
                \Omega^{-1}\Lie_{e_4}\wt{\omegab^r}\sim & \Omega^{-2}\df{\Omega^2\varphi\Psi_2}+\Omega^{-2}\df{\Omega^2\varphi\nabla\varphi}+\Omega^{-2}\nabla(\Omega^2\Omega^{-1}\chi)\left(\df{\Omega\omegab},\df{\Omega\omegab^\dagger}\right).
            \end{aligned}
        \end{equation}
        where $\Psi_2\in \{\Omega^{-1}\beta^r,K,\sigma^r,\Omega\betab^r\}$, $\varphi\in\{\Omega\chib,\Omega^{-1}\chi,\eta,\etab,\Omega e_3\phi,\Omega^{-1}e_4\phi,\nabla\phi\}$. By commutation formulae, we obtain 
        \begin{equation}
            \begin{aligned}
                &\Omega^{-1}\Lie_{e_4}\nabla^5\wt{\omegab^r}
                    \\=&\sum_{i_1+i_2+i_3=5}\nabla^{i_1}(\eta+\etab)^{i_2}\nabla^{i_3}\left(\Omega^{-1}\Lie_{e_4}\wt{\omegab^r}\right)+\sum_{i_1+i_2+i_3+i_4=5}\nabla^{i_1}(\eta+\etab)^{i_2}\nabla^{i_3}(\Omega^{-1}\chi)\nabla^{i_4}\wt{\omegab^r}.
            \end{aligned}
        \end{equation} 
The right-hand side of the transport equation contains only products in which at least one factor is already a difference or a known background increment. This is why Lemma \ref{RI_Lemma_Calculation_of_df} applies directly after the commutation. Using \eqref{transport_4} and Lemma \ref{RI_Lemma_Calculation_of_df}, the following estimate holds:
    \begin{equation}
        \begin{aligned}
            &\lnm \nabla^5\wt{\omegab^r}\rnm^2_{\LS^2_I(S_{u,v})}\lesssim \z\sum_{i\leq 5}\lnm \nabla^i\left(\Lie_{\Omega^{-1}e_4}\wt{\omegab^r},\wt{\omegab^r}\right)\rnm^2_{\LS^2_I(H_{u}^v)}\\
            \lesssim & \z\sum_{i\leq 5}\lnm \nabla^i\left(\df{\varphi},\df{\nabla\varphi},\df{\Omega^2},\df{\Psi_2},\wt{\omegab^r}\right)\rnm^2_{\LS^2_I(H_u^v)}\\
            &+\epsilon^2\z^{1+2\k/(1-\k)}.
        \end{aligned}
    \end{equation}
Here $\Psi_2\in \{\Omega^{-1}\beta^r,K,\sigma^r,\Omega\betab^r\}$, $\varphi\in\{\Omega\chib,\Omega^{-1}\chi,\eta,\etab,\Omega e_3\phi,\Omega^{-1}e_4\phi,\nabla\phi\}$. Therefore it holds that  
\begin{equation}
    \lnm \left[\varphi^c\right]_0^v,\left[\Psi_2^c\right]_0^v\rnm_{\LS_1(S_{u,v})}\lesssim \z^{\k/(1-\k)},
\end{equation}
and the correpsonding $\LS^2_{I}(H_u^v)$-norms are bounded by $\z^{\k/(1-\k)+\delta}$.
To estimate $\lnm \nabla^5\wt{\omegab^r}\rnm_{\LS^2_I(\Hb_v^u)}$, it suffices to prove 
\begin{equation}
    \sum_{i\leq 5}\lnm \nabla^i\left(\df{\varphi},\df{\nabla\varphi},\df{\Omega^2},\df{\Psi_2},\wt{\omegab^r}\right)\rnm^2_{\LS^2_I(H_{u}^v)}\lesssim\epsilon^2 \z^{2\k/(1-\k)}.
\end{equation}
For $\varphi\neq \Omega^{-1}\chih$ and $\Omega^{-1}e_4\phi$, the estimate follows directly by the bootstrap assumptions for $\df{\varphi}$ or $\ddfl{\varphi}$. For $\Omega^{-1}\chih$ and $\Omega^{-1}e_4\phi$, we can use $\df{\varphi}-\wt\varphi=\left(\ot\varphi-\varphi_0\right)-\left(\varphi^c-\varphi_0\right)$ to obtain that  
\begin{equation}
    \begin{aligned}
        \lnm \nabla^6\df{\varphi}\rnm^2_{\LS_I^2(\R_{u,v})}\lesssim & \lnm \nabla^6\wt{\varphi}\rnm^2_{\LS_I^2(\R_{u,v})}+\int_0^v\int_{-1}^u\frac{1}{(-u^\prime)^{2-\k}}\left(\frac{v^\prime}{(-u^\prime)}\right)^{2\delta-1+\frac{2\k}{1-\k}-\delta}\epsilon^2\\
        \lesssim & \epsilon^2\z^{2\k/(1-\k)+\delta}.
    \end{aligned}
\end{equation}
The term $\nabla^5\wt{\omegab^r}$ appearing on the right-hand side can be easily absorbed into the left-hand side via a strictly standard bootstrap argument, while all lower-order terms are robustly controlled by the pre-established non-top-order estimates for $\Omega\omegab,\Omega\omegab^\dagger,\Omega\betab^r$ given in \ref{L2S_est_for_non_top_order_terms} and \eqref{non_top_est_reduced_omegab}. Finally, the estimate for $\df{\Omega^2}$ is completed precisely as in \eqref{Estimate_of_df_Omega}.
Therefore we conclude that 
        \begin{equation}\label{R_I_Est_omegab_r}
            \lnm \nabla^5\wt{\omegab^r}\rnm^2_{\LS^2(\Hb_v^u)}\lesssim \epsilon^2\z^{1+\frac{2\k}{1-\k}-\delta}.
        \end{equation}
        Because we have Hodge systems for $\wt{\Omega\omegab}$ and $\wt{\Omega\omegab^\dagger}$:
        $$\dv\nabla\ddfl{\Omega\omegab}=-\dv\wt{\omegab^r}-\frac{1}{2}\dv\ddfl{\Omega\betab^r},\quad \cl\nabla\ddfl{\Omega\omegab}=0,$$ $$\dv\nabla\ddfl{\Omega\omegab^\dagger}=\cl\wt{\omegab^r}+\cl\ddfl{\Omega\betab^r},\quad \cl\nabla\ddfl{\Omega\omegab^\dagger}=0,$$
        the estimate \eqref{R_I_Est_omegab_r} gives sixth order derivative estimates for $\wt{\Omega\omegab},\wt{\Omega\omegab^\dagger}$ by elliptic estimats \eqref{Ellip Est} and top order estimate \eqref{R_I_Top_order_est_for_ddfl_betab} for $\ddfl{\Omega\betab^r}$.

\end{proof}

We proceed to estimate $\Omega\chib$. Because of the constant in energy estimates \eqref{Energy Estimates} for pair $(\Omega\chib,0)$, $\df{\Omega\chib}$ is not enough to improve the bootstrap assumption. That is the main reason we need the difference $\ddfl{\cdot}$.
\begin{proposition}
    We have estimates for $\Omega\tr\chib$ and $\Omega\chibh$:
    \begin{equation}
        \lnm \nabla^6\ddfl{\Omega\tr\chib}\rnm^2_{\LS_I^2(S_{u,v})}+\lnm \nabla^6\ddfl{\Omega\chibh}\rnm^2_{\LS_I^2(\Hb_v^u)}\lesssim \epsilon^2 \left(\frac{v}{(-u)^{1-\k}}\right)^{1+2\k/(1-\k)-2\delta},
    \end{equation}
    \begin{equation}
        \lnm \nabla^6\df{\Omega\chibh}\rnm^2_{\LS_I^2(H_u^v)}\lesssim \epsilon^2 \left(\frac{v}{(-u)^{1-\k}}\right)^{2\k/(1-\k)}.
    \end{equation}
\end{proposition}
\begin{proof}
    For a tensor $t_{A_1\cdots A_r}$ on $S_{u,v}$ and $\ddfl{t}=t-t_0-v\left(\Lie_v t\right)_0$, we first deduce that 
    \begin{equation}
        \begin{aligned}
            &\ddf{\dv t}-\dv\ddf{t}\\
            =&\dv t_0-\dv_0 t_0+v\dv(\Lie_v t)_0-v(\Lie_v\dv t)_0\\
            =&\ddf{\dv t_0}+v\left(\Lie_v\left(\dv t_0\right)\right)_0+v\ddf{\dv(\Lie_v t)_0}+v^2\left(\Lie_v(\dv(\Lie_v t)_0)\right)\\
            &+v(\dv \Lie_v t)_0-v(\Lie_v\dv t)_0\\
            =&\ddf{\dv t_0}+v\ddf{\dv(\Lie_v t)_0}+v^2\left(\Lie_v(\dv(\Lie_v t)_0)\right)_0+v\left(\left[\Lie_v,\dv\right](t_0-t)\right)_0\\
            = & \ddf{\dv t_0}+v\ddf{\dv(\Lie_v t)_0}+v^2\left(\Lie_v(\dv(\Lie_v t)_0)\right)_0.
        \end{aligned}
    \end{equation}
    Here we note that $\left[\Lie_v,\dv\right] t\sim \nabla(\Omega\chi) t+\Omega\chi\nabla t$ contains no derivative with respect to $e_4$, therefore $\lim_{v\rightarrow 0}\left(\left[\Lie_v,\dv\right] (t(v)-t(0))\right)=0$.
    After making difference to Gauss-Codazzi equation \eqref{Gauss_Codazzi}, we deduce that 
    \begin{equation}
        \begin{aligned}
            \ddfl{\Omega\betab^r}=&\dv\ddfl{\Omega\chibh}-\frac{1}{2}\nabla\ddfl{\Omega\tr\chib}-\ddfl{\eta\left(\Omega\chibh-\frac{1}{2}\Omega\tr\chib\right)}\\
            &+\ddf{\dv(\Omega\chibh)_0} +v\ddf{\dv(\Lie_v(\Omega\chibh))_0}+v^2\left(\Lie_v(\dv(\Lie_v (\Omega\chibh))_0)\right)_0 .
        \end{aligned}
    \end{equation}
    In case of spherical symmetry, $\eta,\nabla\phi$ vanish, and $\Omega\tr\chib$ satisfies $$\Omega^{-1}\nabla_4(\Omega\tr\chib)+\Omega^{-1}\tr\chi\Omega\tr\chib+2K=0,$$ 
    which contains no $\Omega^{-1}e_4\phi $ and $\Omega^{-1}\omega$. Therefore we have 
    $$\lnm \ddf{(\Omega\tr\chib)^c}\rnm_{\LS_1(S_{u,v})}\lesssim\z^2.$$
    Applying elliptic estimate \eqref{Ellip Est} and using Lemma \ref{RI_Lemma_Calculation_of_ddfl}, we obtain that  
    \begin{equation}
        \begin{aligned}
            \lnm \nabla^6\ddfl{\Omega\chibh} \rnm^2_{\LS_I^2(S_{u,v})}
            \lesssim & \sum_{i\leq 5}\lnm \nabla^6\ddfl{\Omega\tr\chib},\nabla^i\ddfl{\Omega\betab^r},\nabla^i\ddfl{\eta},\nabla^i\ddfl{\Omega\chib} \rnm^2_{\LS_I^2(S_{u,v})}\\
            &+\sum_{i\leq 5}\epsilon^2\lnm \nabla^i\ddf{g^{-1}\partial g}\rnm^2_{\LS_I^2(S_{u,v})} +\epsilon^2\z^{3+2\delta} \\
            & +\z^{4}\sum_{i\leq 5}\lnm \nabla^i(\ol{\eta},\ol{\Omega\chib})\rnm^2_{\LS_I^2(S_{u,v})}.
        \end{aligned}
    \end{equation} 
All extra terms in the previous elliptic estimate have already been assigned to one of the three available controls: the non-top-order proposition, the metric estimate, or the enhanced incoming curvature estimate. Thus the only unclosed quantity is the trace term displayed on the right-hand side. Applying Proposition \ref{L2S_est_for_non_top_order_terms}, the previous estimate \ref{R_I_Estimate_metric} for $\ddfl{g}$, and the specific top-order estimate \eqref{R_I_Top_order_est_for_ddfl_betab} for $\Omega\betab^r$, we promptly deduce that 
    \begin{equation}\label{R_I_energy_est_chibh}
        \lnm \nabla^6\ddfl{\Omega\chibh}\rnm^2_{\LS_I^2(\Hb_v^u)}\lesssim \lnm \nabla^6\ddfl{\Omega\tr\chib}\rnm^2_{\LS_I^2(\Hb_v^u)}+\epsilon^2\z^{1+\frac{2\k}{1-\k}-\delta},
    \end{equation}
    and 
    \begin{equation}
        \lnm \nabla^6\ddfl{\Omega\chibh}\rnm^2_{\LS_I^2(H_u^v)}\lesssim \lnm \nabla^6\ddfl{\Omega\tr\chib}\rnm^2_{\LS_I^2(H_u^v)}+\epsilon^2\z^{\frac{2\k}{1-\k}}.
    \end{equation}
    We next estimate $\nabla^6\ddfl{\Omega\tr\chib}$. Throughout this proof, $\varphi$ only takes values in $\Omega\tr\chib,\Omega\chibh,\Omega D_3\phi$.
    From Corollary \ref{Appen nabla 3 tr chib cor}, we have 
    \begin{equation}
        \begin{aligned}
            F_0(\chib):=&\Omega\nabla_3\ddfl{\Omega\tr\chib}+(\Omega\tr\chib+4\Omega\omegab)\ddfl{\Omega\tr\chib}\\
            \sim & \Omega\chibh\cdot\ddfl{\Omega\chibh}+\Omega e_3\phi\ddfl{\Omega e_3\phi}+(\Omega\tr\chib)^c\ddfl{\Omega\omegab}+O(\epsilon)\ddfl{b}\\
            &+\sum_{\varphi\in\{\Omega\chib,\Omega\omegab,\Omega e_3\phi\}}\left((\ol{\varphi}+O(v))\ddfl{\varphi}+\ddf{\varphi^c}\ol{\varphi}\right)+O\left(\epsilon\z^2\right).
        \end{aligned}
    \end{equation}
    By Corollary \ref{increasing_rate_D4phi_c}, we find that 
    \begin{equation}
        \left|\left[(\Omega^{-1}D_4\phi)^c\right]_0^v\right| \lesssim (-u)^{-1}\z^{\frac{\k}{1-\k}},
    \end{equation}
    which implies that 
    \begin{equation}
        \partial_v\ddf{(\Omega D_3\phi)^c}=-\frac{1}{2}\left[(\Omega^2)^c\tr\chi^c (D_3\phi)^c+(\Omega^2)^c\tr\chib^c (D_4\phi)^c\right]_0^v\lesssim (-u)^{\k-2}\z^{\frac{\k}{1-\k}}.
    \end{equation}
    Therefore, the derivatives of $\ddf{\varphi^c}\ol{\varphi}$ has $\LS_I^2(S_{u,v})$ norm bounded by $\epsilon^2\z^{1+\frac{2\k}{1-\k}+2\delta}$.
    Using commutation formula, we compute $F_6(\chib)$:
    \begin{equation}
        \begin{aligned}
            F_6(\chib):=&\Omega\nabla_3\nabla^6\ddfl{\Omega\tr\chib}+\left(4\Omega\tr\chib+4\Omega\omegab\right)\nabla^3\ddfl{\Omega\tr\chib}\\
        =&\nabla^6F_0(\chib)+ \sum_{\substack{i_1+i_2+i_3+i_4=6\\1\leq i_4\leq 5}}\nabla^{i_1}\eta^{i_2}\nabla^{i_3}\Omega\chib\nabla^{i_4}\ddfl{\Omega\tr\chib}.
        \end{aligned}
    \end{equation}
    Because the constant in $\nabla_3$-transport estimate \eqref{transport_3} satisfies
    \begin{align*}
        &\left(-2+2s(\nabla^6\Omega\tr\chib)-(1-\k)(1-2\delta)-(-u)\left(8\Omega\tr\chib+8\Omega\omegab+o_{\epsilon,\epsilon_1}(1)\right)\right)\\
        &\qquad\qquad\qquad = (1-\k)(1+2\delta)+o_{\epsilon,\epsilon_1} (1),
    \end{align*}
    we have by transport estimate,
    \begin{equation}
        \begin{split}
            \lnm \nabla^6\ddfl{\Omega\tr\chib}\rnm^2_{\LS_I^2(S_{u,v})}\leq & \lnm \nabla^6\ddfl{\Omega\tr\chib}\rnm^2_{\LS_I^2(S_{-1,v})}+\lnm \nabla^6\ddfl{\Omega\tr\chib}\rnm_{\LS_I^2(\Hb_v^u)}\lnm F_6(\chib)\rnm_{\LS_I^2(\Hb_v^u)}\\
            &+\left((1-\k)(1+2\delta)+o(1)\right)\lnm \nabla^6\ddfl{\Omega\tr\chib}\rnm^2_{\LS_I^2(\Hb_v^u)}.
        \end{split}
    \end{equation}
    If we can prove that
    \begin{equation}
        \lnm F_6(\chib)\rnm_{\LS_I^2(\Hb_v^u)}^2\lesssim \epsilon^2\lnm \nabla^6\ddfl{\Omega\tr\chib}\rnm^2_{\LS_I^2(\R_{u,v})} +\epsilon^2\left(\frac{v}{(-u)^{1-\k}}\right)^{1+2\k/(1-\k)-2\delta},
    \end{equation}
    then it holds that 
    \begin{equation}
        \frac{d}{du}\left((-u)^{(1-\k)(1+2\k/(1-\k)-3\delta)}\lnm\nabla^3\ddfl{\Omega\tr\chib}\rnm^2_{\LS_I^2(\Hb_v^u)}\right)\leq C(\delta,\k) \epsilon^2 v^{1+2\k/(1-\k)-2\delta}(-u)^{-\delta-1},
    \end{equation}
    and integration yields the desired estimate:
    \[\lnm\nabla^6\ddfl{\Omega\tr\chib}\rnm^2_{\LS_I^2(\Hb_v^u)}\leq C(\delta,\k) \epsilon^2 \z^{1+2\k/(1-\k)-2\delta}.\] 
The commutator terms in $F_6(\chib)$ are lower order in the differentiated trace factor. They are controlled by \eqref{Calcu_Diff_product_ddfl} together with the non-top-order estimates, while the top differentiated source terms are deferred to the next displays. 
    With Lemma \eqref{Calcu_Diff_product_ddfl} and Proposition \ref{L2S_est_for_non_top_order_terms}, we obtain 
    \begin{equation}
        \begin{aligned}
            &\sum_{\substack{i_1+i_2+i_3+i_4=6\\1\leq i_4\leq 5}}\lnm\nabla^{i_1}\eta^{i_2}\nabla^{i_3}\Omega\chib\nabla^{i_4}\ddfl{\Omega\tr\chib}\rnm^2_{\LS_I^2(\Hb_{v}^u)}\\
            \lesssim &\sum_{i\leq 5}\lnm \nabla^i\left(\ddfl{\Omega\chib}\right)\rnm^2_{\LS_I^2(\Hb_{v}^u)} 
            \lesssim B\epsilon^2\z^{1+\frac{2\k}{1-\k}}\leq \epsilon^2\z^{1+\frac{2\k}{1-\k}-2\delta}.
        \end{aligned}
    \end{equation}
Also for $\nabla^6F_0(\chib)$, the following inequality holds 
\begin{equation}
    \begin{aligned}
        &\lnm \nabla^6 F_0(\chib)\rnm^2_{\LS_I^2(\Hb_v^u)}\\
        \leq & \epsilon^2\z^{1+\frac{2\k}{1-\k}+2\delta} + \sum_{i\leq 5}\lnm \nabla^i\left(\ddfl{\Omega\chibh},\ddfl{\Omega D_3\phi},\ddfl{\Omega\omegab}\right)\rnm^2_{\LS_I^2(\Hb_v^u)}\\
        &+\lnm \epsilon\cdot\nabla^6\ddfl{\Omega\chibh},\nabla^6\ddfl{\Omega D_3\phi},\nabla^6\ddfl{\Omega\omegab}\rnm^2_{\LS_I^2(\Hb_v^u)}\\
        &+\sum_{i\leq 6}\lnm \nabla^i\ol{\varphi}\cdot\z^{1+\frac{\k}{1-\k}},\epsilon\nabla^i\ddfl{b} \rnm^2_{\LS_I^2(\Hb_v^u)} +\sum_{i\leq 6}\lnm \nabla^i\left(\ddfl{\varphi}\ol{\varphi}\right) \rnm^2_{\LS_I^2(\Hb_v^u)}.
    \end{aligned}
\end{equation}
From estimate \ref{R_I_Estimate_metric}, we know that $\sum_{i\leq 6}\lnm \nabla^i\ddfl{b}\rnm^2_{\LS_I^2(\Hb_{v}^u)}\leq \epsilon^2\z^{1+\frac{2\k}{1-\k}-2\delta}$. 
For $\varphi=\Omega\chib,\Omega D_3\phi,\Omega\omegab$, this yields 
\begin{equation}
    \begin{aligned}
        &\sum_{i\leq 6}\lnm \nabla^i\ol{\varphi}\cdot\z^{1+\frac{\k}{1-\k}},\nabla^i\df{\varphi}\cdot\z \rnm^2_{\LS^2(\Hb_v^u)}
        \lesssim \epsilon^2\z^{1+\frac{2\k}{1-\k}+2\delta}.
    \end{aligned}
\end{equation}
The non-top-order terms are controlled by applying Proposition \ref{L2S_est_for_non_top_order_terms},
\begin{equation}
    \begin{aligned}
        \sum_{i\leq 5}\lnm \nabla^i\left(\ddfl{\Omega\chibh},\ddfl{\Omega D_3\phi},\ddfl{\Omega\omegab}\right)\rnm^2_{\LS^2(\Hb_v^u)}\lesssim & B\epsilon^2\z^{1+\frac{2\k}{1-\k}}\\
        \leq & \epsilon^2\z^{1+\frac{2\k}{1-\k}-2\delta}.
    \end{aligned}
\end{equation} 
Employing Proposition \ref{R_I_Energy_Est_nabphi_D_3phi}, Proposition \ref{R_I_Energy_Est_omegab}, and relation \eqref{R_I_energy_est_chibh}, we can estimate the top order terms.
\begin{equation}\label{temp_5-161}
    \begin{aligned}
        &\lnm \epsilon\cdot\nabla^6\ddfl{\Omega\chibh},\nabla^6\ddfl{\Omega D_3\phi},\nabla^6\ddfl{\Omega\omegab}\rnm^2_{\LS^2(\Hb_v^u)}\\
        \lesssim & \epsilon^2\lnm \nabla^6\ddfl{\Omega\tr\chib}\rnm^2_{\LS^2(\Hb_v^u)}+\epsilon^2\z^{1+\frac{2\k}{1-\k}-2\delta}.
    \end{aligned}
\end{equation}
For $\ddfl{\varphi}\ol\varphi$, we estimate:
\begin{equation}
    \begin{aligned}
        &\sum_{i\leq 6}\lnm \nabla^i\left(\ddfl{\varphi}\ol{\varphi}\right) \rnm^2_{\LS_I^2(\Hb_v^u)}\\
        \lesssim& \sum_{i\leq 6}\lnm \nabla^i\left(\ddfl{\varphi}\df{\varphi}\right)\rnm^2_{\LS_I^2(\Hb_v^u)}+\lnm \nabla^i\left(\ddfl{\varphi}\ol{\varphi_0}\right)\rnm^2_{\LS_I^2(\Hb_v^u)}\\
        \lesssim & \sum_{i\leq 4,j\leq 6}\sup_{u'<u}\lnm \nabla^i\ddfl{\varphi}\rnm^2_{\LS_1^2(S_{u',v})}\lnm \nabla^j \df{\varphi}\rnm^2_{\LS_I(\Hb_v^u)}\\
        &+ \sum_{i\leq 4,j\leq 6}\sup_{u'<u}\lnm \nabla^i\df{\varphi},\nabla^j\varphi_0\rnm^2_{\LS_1^2(S_{u',v})}\lnm \nabla^j \ddfl{\varphi}\rnm^2_{\LS_I(\Hb_v^u)}.
    \end{aligned}
\end{equation}
Using bootstrap assumptions and Proposition \ref{L2S_est_for_non_top_order_terms}, we have 
\begin{equation}
    \begin{aligned}
        & \sum_{i\leq 4,j\leq 6}\sup_{u'<u}\lnm \nabla^i\ddfl{\varphi}\rnm^2_{\LS_1^2(S_{u',v})}\lnm \nabla^j \df{\varphi}\rnm^2_{\LS_I(\Hb_v^u)}\\
        \lesssim & B\epsilon^2\z^{2+2\k/(1-\k)-2\delta} B\epsilon^2\z^{2\k/(1-\k)}\ll \epsilon^2\z^{2+2\k/(1-\k)},
    \end{aligned}
\end{equation}
and the last term is $\epsilon^2\lnm \nabla^j \ddfl{\varphi}\rnm^2_{\LS_I(\Hb_v^u)}$, which is already estimated as \eqref{temp_5-161}.
We thus conclude that 
\begin{equation}
    \lnm \nabla^6\ddfl{\Omega\tr\chib}\rnm^2_{\LS^2(S_{u,v})}\leq C_0\epsilon^2\z^{1+\frac{2\k}{1-\k}-2\delta}.
\end{equation}

\end{proof}

Next proposition establishes the top-order estimates for $\Omega^{-1}\chi$.
\begin{proposition}
    For $\Omega^{-1}\chi$, the estimates below hold,
    \begin{equation}
         \lnm \nabla^6\ddfl{\Omega^{-1}\tr\chi}\rnm^2_{\LS^2(S_{u,v})}\lesssim \epsilon^2\z^{1+2\k/(1-\k)}.
    \end{equation}
    \begin{equation}
        \lnm \nabla^6\df{\Omega^{-1}\chih} \rnm^2_{\LS^2(H_u^v)}\lesssim \epsilon^2\left(\frac{v}{(-u)^{1-\k}}\right)^{2\k/(1-\k)},\quad\lnm \nabla^6\wt{\Omega^{-1}\chih}\rnm^2_{\LS^2(\Hb_v^u)}\lesssim \epsilon^2\left(\frac{v}{(-u)^{1-\k}}\right)^{2\delta}.
    \end{equation}
\end{proposition}
\begin{proof}
Recall the Gauss-Codazzi equation $\dv(\Omega^{-1}\chih)-\frac{1}{2}\nabla(\Omega^{-1}\tr\chi)+\eta(\Omega^{-1}\chih-\frac{1}{2}\Omega^{-1}\tr\chi)=-\Omega^{-1}\beta^r$; the corresponding equation for the difference is 
\begin{equation}
    \dv\df{\Omega^{-1}\chih}=\frac{1}{2}\nabla\df{\Omega^{-1}\tr\chi}-\df{\eta\left(\Omega^{-1}\chih-\frac{1}{2}\Omega^{-1}\tr\chi\right)}-\df{\Omega^{-1}\beta^r}-\left[\dv(\Omega^{-1}\chih)_0\right]_0^v.
\end{equation}
We can use elliptic estimate \eqref{Ellip Est} to obtain that 
    \begin{equation}
        \begin{aligned}
            \lnm \nabla^6\df{\Omega^{-1}\chih}\rnm^2_{\LS_I^2(H_u^v)}
            \lesssim & \lnm \nabla^6\df{\Omega^{-1}\tr\chi},\nabla^5\df{\Omega^{-1}\beta^r}\rnm^2_{\LS_I^2(H_u^v)}\\
            &+\sum_{i\leq 5,j\leq 4}\lnm \nabla^i\left(\df{\Omega^{-1}\chi},\df{\eta},O(\epsilon)\left[{g^{-1}\partial g}\right]_0^v\right),\nabla^j\df{\Omega^{-1}\beta^r}\rnm^2_{\LS_I^2(H_u^v)}\\
            \lesssim & \lnm \nabla^6\df{\Omega^{-1}\tr\chi}\rnm^2_{\LS_I^2(H_u^v)}+\epsilon^2\z^{\frac{2\k}{1-\k}}.
        \end{aligned}
    \end{equation}
    Recall the definition \eqref{R_I_Def_ot_beta} of $\ot{\Omega^{-1}\beta^r}$, we find that 
    \begin{equation}
        \dv\wt{\Omega^{-1}\chih}=-\left[\dv\right]_0^v\ot{\Omega^{-1}\chih}+\frac{1}{2}\nabla\df{\Omega^{-1}\tr\chi}+\frac{1}{2}\df{\eta\Omega^{-1}\tr\chi}-\eta\wt{\Omega^{-1}\chih}-\df{\eta}\ot{\Omega^{-1}\chih}-\wt{\Omega^{-1}\beta^r},
    \end{equation}
    and the estimate follows similarly,
    \begin{equation}
        \begin{aligned}
            &\lnm \nabla^6\wt{\Omega^{-1}\chih}\rnm^2_{\LS_I^2(\Hb_v^u)}\\
            \lesssim & \lnm \nabla^6\df{\Omega^{-1}\tr\chi},\nabla^5\wt{\Omega^{-1}\beta^r}\rnm^2_{\LS_I^2(\Hb_v^u)}+\sum_{i\leq 5,j\leq 4}\lnm \nabla^i\left(\wt{\Omega^{-1}\chi},\df{\eta},O(\epsilon)\left[{g^{-1}\partial g}\right]_0^v\right),\nabla^j\wt{\Omega^{-1}\beta^r}\rnm^2_{\LS_I^2(\Hb_v^u)}\\
            \lesssim & \lnm \nabla^6\df{\Omega^{-1}\tr\chi}\rnm^2_{\LS_I^2(\Hb_v^u)}+\epsilon^2\z^{2\delta}.
        \end{aligned}
    \end{equation}
    Let $\varphi$ take values in $\Omega^{-1}\tr\chi,\Omega^{-1}\chih,\Omega^{-1}D_4\phi$ and $\psi\in\{\eta,\etab,\Omega^{-1}\chi\}$. We can write the equation schematically:
    \begin{equation}
        \Omega^{-1}e_4\ddfl{\Omega^{-1}\tr\chi}=\Omega^{-2}\df{\Omega^2\varphi^2},
    \end{equation}
    and using commutation formula, 
    \begin{equation}
        \Omega^{-1}\nabla_4\nabla^6\ddfl{\Omega^{-1}\tr\chi}=\sum_{\substack{i_1+i_2+i_3=6\\1\leq i_3}}\nabla^{i_1}\psi^{i_2+1}\nabla^{i_3}\ddfl{\Omega^{-1}\tr\chi}+\Omega^{-2}\nabla^6\left(\Omega^2\Omega^{-1}e_4\ddfl{\Omega^{-1}\tr\chi}\right).
    \end{equation} 
The trace equation is an outgoing transport equation, so the same transport estimate as for the metric applies. The difference product in $\Omega^{-2}\df{\Omega^2\varphi^2}$ is expanded by \eqref{Calcu_Diff_product_df}, while the trace-free and scalar-field terms are already controlled by the preceding top-order estimates. Using \eqref{transport_4} and \eqref{Calcu_Diff_product_df}, we obtain
\begin{equation}
        \begin{aligned}
            &\lnm \nabla^6\ddfl{\Omega^{-1}\tr\chi} \rnm^2_{\LS_I^2(S)}\\
            \lesssim &\frac{v}{(-u)^{1-\k}}\lnm G_6(\chi),\nabla^6\ddfl{\Omega^{-1}\tr\chi}\rnm^2_{\LS_I^2(H_u^v)}\\
            \lesssim &\frac{v}{(-u)^{1-\k}}\sum_{i\leq 5}\lnm \nabla^i\left(\df{\psi},\df{\varphi},\df{\Omega^{-2}}\right)\rnm^2_{\LS^2(H_u^v)}+\epsilon^2\z^{1+2\k/(1-\k)} \\
            &+\frac{v}{(-u)^{1-\k}}\lnm \nabla^6\left(\df{\Omega^{-1}\tr\chi},\df{\Omega^{-1}\chih},\df{\Omega^{-1}D_4\phi},\df{\Omega^{-2}}\right)\rnm^2_{\LS^2(H_u^v)}\\
            \lesssim & \z\lnm \nabla^6\df{\Omega^{-1}\tr\chi}\rnm^2_{\LS^2(H_{u,v})}+\epsilon^2\z^{1+\frac{2\k}{1-\k}}.
        \end{aligned}
    \end{equation}
    Note that $\nabla^6\left(\ddfl{\Omega^{-1}\tr\chi}-\df{\Omega^{-1}\tr\chi}\right)\sim_{\LS^2_1(S_{u,v})} \epsilon\z$, this yields 
    \begin{equation*}
        \begin{aligned}
            &\z\lnm \nabla^6\df{\Omega^{-1}\tr\chi}\rnm^2_{\LS^2(H_{u,v})}\\
            &\qquad\lesssim \z\lnm \nabla^6\ddfl{\Omega^{-1}\tr\chi}\rnm^2_{\LS^2(H_{u,v})}+\epsilon^2\z^{3+2\delta}.
        \end{aligned}
    \end{equation*}
    Therefore $\lnm \nabla^6\ddfl{\Omega^{-1}\tr\chi}\rnm^2_{\LS^2(S_{u,v})}\lesssim \epsilon^2\z^{1+2\k/(1-\k)}$.

\end{proof}

With the top-order estimates for $\chi,\chib,\omegab$, we can now bound $\eta$ and $\etab$.
\begin{proposition}
    The following estimates hold for $\eta$ and $\etab$:
    \begin{equation}
        \lnm \nabla^6\df{\eta},\nabla^6\df{\etab} \rnm^2_{\LS_I^2(H_u^v)}+\lnm \nabla^6\df{\eta},\nabla^6\df{\etab} \rnm^2_{\LS_I^2(\Hb_v^u)}\lesssim \epsilon^2\left(\frac{v}{(-u)^{1-\k}}\right)^{\frac{2\k}{1-\k}}.
    \end{equation}
\end{proposition}
\begin{proof}
    We make bootstrap assumption in the proof: 
    \begin{equation}\label{R_I_eta_etab_Bootstrap}
        \lnm \nabla^6\df{\eta},\nabla^6\df{\etab} \rnm^2_{\LS_I^2(H_u^v)}+\lnm \nabla^6\df{\eta},\nabla^6\df{\etab} \rnm^2_{\LS_I^2(\Hb_v^u)}\leq B_2\epsilon^2\left(\frac{v}{(-u)^{1-\k}}\right)^{\frac{2\k}{1-\k}}.
    \end{equation}
    We define auxiliary functions 
    \begin{equation}
        \begin{split}
            {\mub}=-\dv{\etab}+{K},\quad
            {\mu}=-\dv{\eta}+{K}.
        \end{split}
    \end{equation}
    The Hodge systems for $\eta,\etab$ are 
    \begin{equation}
        \begin{aligned}
            \dv\df\etab=-\left[\dv\etab_0\right]_0^v-\df\mub+\df{K},&\ \cl\df\etab=-\left[\cl\etab_0\right]_0^v-\df{\sigma^r},\\
            \dv\df\eta=-\left[\dv\eta_0\right]_0^v-\df\mu+\df{K},&\ \cl\df\eta=-\left[\cl\eta_0\right]_0^v+\df{\sigma^r}.
        \end{aligned}
    \end{equation}
    Combining the standard elliptic estimates \eqref{Ellip Est}, the existing estimates \ref{L2S_est_for_non_top_order_terms} for all lower-order terms, the targeted estimate \ref{R_I_Estimate_metric} specifically for $\left[g^{-1}\partial g\right]_0^v$, and the precise top-order bounds \eqref{R_I_Top_order_Est_beta_sigma_K_betab_alphab} and \eqref{R_I_Top_order_est_for_ddfl_betab} bounding $\df{K}$ and $\df{\sigma^r}$, we conclude that 
    \begin{equation}
        \begin{aligned}
            \lnm \nabla^6\df\eta\rnm^2_{\LS^2_I(H_u^v)}\lesssim & \lnm \nabla^5\df\mu\rnm^2_{\LS^2_I(H_u^v)}+\epsilon^2\z^{1+2\k/(1-\k)-\delta},\\
            \lnm \nabla^6\df\etab\rnm^2_{\LS^2_I(H_u^v)}\lesssim & \lnm \nabla^5\df\mub\rnm^2_{\LS^2_I(H_u^v)}+\epsilon^2\z^{1+2\k/(1-\k)},\\
            \lnm \nabla^6\df\eta\rnm^2_{\LS^2_I(\Hb_v^u)}\lesssim & \lnm \nabla^5\df\mu\rnm^2_{\LS^2_I(\Hb_v^u)}+\epsilon^2\z^{2\k/(1-\k)},\\
            \lnm \nabla^6\df\etab\rnm^2_{\LS^2_I(\Hb_v^u)}\lesssim & \lnm \nabla^5\df\mub\rnm^2_{\LS^2_I(\Hb_v^u)}+\epsilon^2\z^{2\k/(1-\k)}.
        \end{aligned}
    \end{equation}
    We can derive their equation by commutation formulae,
    \begin{equation}
        \begin{aligned}
            \Omega\nabla_3\mub+&\Omega\tr\chib \mub =\Omega\tr\chib \mub-\dv(\Omega\nabla_3\etab)+\Omega\nabla_3 K-\left[\Omega\nabla_3,\dv\right]\etab\\
            =&\Omega\tr\chib \mub+\dv\left(\Omega\chib(\etab-\eta)-\Omega\betab^r-\Omega e_3\phi\nabla\phi\right) +\dv\left(\Omega\betab^r-\frac{1}{2}\eta\Omega\tr\chib+\eta\Omega\chibh\right)\\
            &-\Omega\tr\chib K+\frac{1}{2}\Omega\tr\chib\dv\etab+\Omega\chibh\nabla\etab+\etab\left(\dv(\Omega\chibh)-\frac{1}{2}\nabla(\Omega\tr\chib)\right)\\
            =&2\Omega\chibh\nabla\etab+\dv(\nabla\phi\Omega e_3\phi)+2\etab\dv(\Omega\chibh)-\eta\nabla(\Omega\tr\chib).
        \end{aligned}
    \end{equation}
    Similarly, the $\Omega\nabla_4\mu$ equation is 
    \begin{equation}
        \Omega\nabla_4\mu=-\Omega\tr\chi \mu+ 2\Omega\chih\nabla\eta+\dv(\nabla\phi\Omega e_4\phi)+2\eta\dv(\Omega\chih)-\etab\nabla(\Omega\tr\chi).
    \end{equation}
    With commutation formulae, we obtain 
    \begin{equation}
        \begin{aligned}
            \Omega\nabla_3\nabla^5\df\mub+&\frac{7}{2}\Omega\tr\chib\nabla^5\df\mub\sim \sum_{i_1+i_2=4}\nabla^{i_1+1}(\Omega\chib)\nabla^{i_2}\df\mub +\Omega\chibh\df{\nabla^6\etab}+\Omega\chibh\nabla^5\df\mub\\
            &+\nabla^5\df{\etab\nabla(\Omega\chib),\eta\nabla(\Omega\chib)}+\nabla^5\df{\nabla(\nabla\phi \Omega e_3\phi)}\\
            &+\nabla^5\left((\Lie_{\left[b\right]_0^v}+\left[\Omega\chib\right]_0^v)\ol{\mub}_0+\df{\Omega\chib}\mub^c+\ol{\Omega\chib}\left[\mub^c\right]_0^v\right),\\
            \Omega^{-1}\nabla_4\nabla^5\df{\mu}\sim & \sum_{i_1+i_2=5}\Omega^{-2}\nabla^{i_1}(\Omega^2\Omega^{-1}\chi)\nabla^{i_2}\df{\mu}+\Omega^{-2}\nabla^5\left(\ol{\Omega^2\Omega^{-1}\chi\mu}+\ol{\Omega^2\varphi\nabla\varphi}\right),
        \end{aligned}
    \end{equation}
    where $\varphi\in\{\Omega^{-1}\chi,\eta,\etab,\nabla\phi,\Omega^{-1}e_4\phi\}$.
The $\mu$ equation is estimated in the outgoing direction. All products in its source contain either a previously estimated lower-order Ricci coefficient or one of the controlled top-order quantities $\Omega^{-1}\chi$, $\Omega^{-1}e_4\phi$, and $\nabla\phi$.
    Using transport estimate \eqref{transport_4}, we find that 
    \begin{equation}
        \begin{aligned}
            \lnm \nabla^5\df{\mu}\rnm^2_{\LS_I^2(S_{u,v})}\lesssim & \frac{v}{(-u)^{1-\k}}\lnm \nabla^5\df\mu,\Omega^{-1}\nabla_4\nabla^5\df{\mu}\rnm^2_{\LS_I^2(H_u^v)}\\
            \lesssim & \frac{v}{(-u)^{1-\k}}\sum_{i\leq 5,j\leq 4}\lnm \nabla^i\left(\ol{\varphi},\ol{\Omega^2}\right),\nabla^j\df{\mu}\rnm^2_{\LS_I^2(H_u^v)}\\
            &+\frac{v}{(-u)^{1-\k}}\lnm \nabla^6\left(\ol{\Omega^{-1}\chi},\ol{\Omega^{-1}e_4\phi},\ol{\nabla\phi}\right),\nabla^5\df\mu\rnm^2_{\LS_I^2(H_u^v)}\\
            \lesssim &\epsilon^2\z^{1+2\delta}\leq \epsilon^2\z^{\frac{2\k}{1-\k}}.
        \end{aligned}
    \end{equation} 
The $\mub$ equation is harder because it contains the top derivative of $\etab$. The coefficient of this term is $\Omega\chibh$, whose smallness has just been obtained from the $\Omega\tr\chib$ and $\Omega\chibh$ estimates, allowing the bootstrap constant to be improved. 
We proceed to estimate $\mub$. Since $-2+2s(\nabla^5\mub)-(1-\k)(1-2\delta)-7(-u)\Omega\tr\chib=-(1-\k)(1-2\delta)+o(1)<-\frac{1}{10}$, the transport estimate \eqref{transport_3} gives
    \begin{equation}
        \begin{aligned}
            &\lnm \nabla^{5}\df\mub\rnm^2_{\LS_I^2(S_{u,v})}\\
            \lesssim & \lnm \nabla^{5}\df\mub\rnm^2_{\LS_I^2(S_{-1,v})}+\lnm |\dv b|^{1/2}\nabla^5\df\mub,\Omega\nabla_3\nabla^5\df\mub+\frac{7}{2}\Omega\tr\chib\nabla^5\df\mub\rnm^2_{\LS_I^2(\Hb_v^u)}\\
            \lesssim & \epsilon^2\z^{\frac{2\k}{1-\k}}+\epsilon\lnm \nabla^5\df{\mub}\rnm^2_{\LS_I^2(\Hb_v^u)}+\lnm \nabla^6\left(\df{\Omega\chib},\df{\Omega D_3\phi},\df{\nabla\phi}\right)\rnm^2_{\LS_I^2(\Hb_v^u)}\\
            &+\sum_{i\leq 5}\lnm \nabla^i\left(\df{\eta},\df\etab,\df{\Omega\chib},\df{\Omega e_3\phi},\df{\nabla\phi},\df{b}\right)\rnm^2_{\LS_I^2(\Hb_v^u)}+\epsilon^2.\lnm \nabla^6\df\etab\rnm^2_{\LS_I^2(\Hb_v^u)}\\
            \lesssim & (1+\epsilon+\epsilon^2 B_2)\epsilon^2\z^{\frac{2\k}{1-\k}}\lesssim \epsilon^2\z^{\frac{2\k}{1-\k}}.
        \end{aligned}
    \end{equation}
    Here we emphasize that $\Omega\chibh\,\nabla^6\etab$ is the only top-order term of $\etab$ or $\mub$ on the right hand side of $\Omega\nabla_3\nabla^5\mub+\frac{7}{2}\Omega\tr\chib\nabla^5\mub$ equation, and the smallness of $\Omega\chibh$ improves the bootstrap constant $B_2$.
    
\end{proof}

Finally, we estimate $\nabla^6\df{\Omega^{-1}\omega}$. Although this quantity is not used in this section, it will play the role in region III. 
\begin{proposition}
    The following inequality holds:
    \begin{equation}
        \lnm \nabla^{6}\df{\Omega^{-1}\omega}\rnm^2_{\LS_I^2(H_{u}^{v})}\lesssim \epsilon^2\z^{2\k/(1-\k)}.
    \end{equation}
\end{proposition} 
\begin{proof}
    Along $u=-1$, we have $$\df{\Omega^{-1}\omega}=\ol{\Omega^{-2}}\left[(\Omega\omega)^c\right]_0^v+\df{\Omega^{-2}}(\Omega\omega)^c(0)\sim O_{H^{N-1}}\left(\epsilon v^{\k/(1-\k)}\right).$$ See \eqref{increasing_rate_omega_c} in Appendix \ref{Appendix_Chr94} for details. 
Consider auxiliary function ${\Omega^{-1}\omega^\dagger}$ defined by 
\[\Omega\nabla_3{\Omega^{-1}\omega^\dagger}=\frac{1}{2}\df{\sigma^r}+4\Omega\omegab{\Omega^{-1}\omega^\dagger},\quad {\Omega^{-1}\omega^\dagger}|_{H_{-1}}\equiv 0.\]
The $\LS_I^2(S)$ estimate follows by directly using transport estimate \eqref{transport_3} and the $\LS^2_I(\Hb_v^u)$ estimate for $\df{\sigma^r}$,
\begin{equation}
    \sum_{i\leq 5}\lnm \nabla^{i}(\Omega^{-1}\omega^\dagger)\rnm^2_{\LS_I^2(S_{u,v})}\lesssim \epsilon^2\z^{2\k/(1-\k)}.
\end{equation}
To fulfill the $\LS_I^2(H)$ estimate, we 
write $\wt{\omega^r}_A=\nabla_A\df{\Omega^{-1}\omega}+{}^*\nabla_A({\Omega^{-1}\omega^\dagger})-\frac{1}{2}\left[\ol{\Omega^{-1}\beta^r}_{A}\right]_0^v$, then two Hodge systems appear by the definition:
\begin{equation}
    \begin{split}
        \dv\nabla\df{\Omega^{-1}\omega}=\dv\wt{\omega^r}+\frac{1}{2}\dv\df{\Omega^{-1}\beta^r},&\quad \cl\nabla\df{\Omega^{-1}\omega}=0,\\
\cl{}^*\nabla(\Omega^{-1}\omega^\dagger)=\cl\wt{\omega^r}+\frac{1}{2}\cl\df{\Omega^{-1}\beta^r},&\quad \dv{}^*\nabla(\Omega^{-1}\omega^\dagger)=0.
    \end{split}
\end{equation}
It suffices to prove that 
\begin{equation}
    \lnm\nabla^5\wt{\omega^r} \rnm^2_{\LS^2(H_u^v)}\lesssim\epsilon^2\z^{2\k/(1-\k)}.
\end{equation}
We derive transport equation for $\wt{\omega^r}$,
\begin{equation}
    \begin{aligned}
        &\Omega\nabla_3\wt{\omega^r}+\left(\frac{1}{2}\Omega\tr\chib-4\Omega\omegab\right)\wt{\omega^r}\\
        \sim & \Omega\chibh\wt{\omega^r}+\nabla(\Omega\chib,\Omega\omegab)\left(\df{\Omega^{-1}\omega},\Omega^{-1}\omega^\dagger\right)+\df{\varphi^3}+\df{\varphi\nabla\varphi}+\df{\varphi\Psi_2}\\
        &+\nabla\left((\ol{\Omega\chib},\ol{\Omega\omegab})\left[(\Omega^{-1}\omega)^c\right] +\left(\left[\Omega\chib\right]_0^v+\Lie_{\left[b\right]_0^v}\right)(\Omega^{-1}\omega)_0+\df{\Omega\chib,\Omega\omegab}(\Omega^{-1}\omega)^c\right)\\
        =&:F_1(\omega),\\
        &\Omega\nabla_3\nabla^5\wt{\omega^r}+\left(3\Omega\tr\chib-4\Omega\omegab\right)\nabla^5\wt{\omega^r}\sim\nabla^5 F_1(\omega)+\Omega\chibh\nabla^5\wt{\omega^r} +\sum_{i_1+i_2=4}\nabla^{i_1+1}\varphi\nabla^{i_2}\wt{\omega^r} \\
        =&:F_6(\omega).
    \end{aligned}
\end{equation}
Here, we use $\varphi$ to collectively represent the quantities $\Omega\chib,\Omega^{-1}\chi,\eta,\etab,\Omega e_3\phi,\nabla\phi,\Omega^{-1}e_4\phi,\Omega\omegab$, while $\Psi_2$ stands for any of $\Omega^{-1}\beta^r$, $\Omega\betab^r$, $K$, or $\sigma^r$. 
Since $-2+2s(\nabla^5\omega^r)-(1-\k)(1-2\delta)-(-u)\left(6\Omega\tr\chib-8\Omega\omegab\right)=-(1-\k)(3-2\delta)<-\frac{1}{2}$,
the energy estimates yield
\begin{equation}
    \begin{aligned}
        \lnm\nabla^5\wt{\omega^r}\rnm^2_{\LS_I^2(H_u^v)}\lesssim & \epsilon^2\z^{1+2\delta}+ \lnm F_6(\omega),|\dv b|^{1/2}\nabla^5\wt{\omega^r}\rnm^2_{\LS_I^2(\mathcal{R}_{u,v})}-\frac{1}{2}\lnm \nabla^5\wt{\omega^r}\rnm^2_{\LS_I^2(\mathcal{R}_{u,v})}\\
        \lesssim & \epsilon^2\z^{1+2\delta}+\sum_{i\leq 5,j\leq 4}\lnm \nabla^i\left(\df{\varphi},\df{\Omega^{-1}\omega},{\Omega^{-1}\omega^\dagger}\right),\nabla^j\df{\Psi_2}\rnm^2_{\LS_I^2(\R_{u,v})} \\
        &+\lnm \nabla^6\left(\df{\Omega\omegab},\df{\Omega\chib},\df{\Omega^{-1}\chi},\df{\Omega D_3\phi},\df{\Omega^{-1}D_4\phi},\df{\nabla^2\phi}\right)\rnm^2_{\LS_I^2(\R_{u,v})}\\
        &+\lnm \nabla^5\df{\Psi_2}\rnm^2_{\LS_I^2(\R_{u,v})}+(\epsilon^2+\epsilon-\frac{1}{2})\lnm \nabla^5\wt{\omega^r} \rnm^2_{\LS_I^2(\R_{u,v})}\\
        \lesssim & \epsilon^2\z^{1+2\delta}.
    \end{aligned}
\end{equation}
We thus conclude $\lnm \nabla^5\wt{\omega^r} \rnm^2_{\LS_I^2(\R_{u,v})}\lesssim \epsilon^2\z^{1+2\delta}$.
\end{proof}

This concludes the proof of Proposition \ref{Main Proposition R_I}. We can furthermore study the asymptotic behavior of $\ol{(\Omega e_3)^2\phi}$. For notational simplicity, we denote ${\Omega^4(\Omega^{-1} e_3)^2\phi}$ by $\Phi$. By Proposition \ref{Appen_A_eq_D3D3phi} in Appendix A, it holds that 
\begin{equation}
    \Omega^{-1}e_4\left(\ol{\Phi}\right)+\frac{1}{2}\Omega^{-1}\tr\chi\ol{\Phi}=-\frac{1}{2}\ol{\Omega^{-1}\tr\chi}{\Phi}^c+\nabla^2(\Omega D_3\phi)+\varphi\nabla\varphi+\ol{\varphi^3}+\ol{\varphi\cdot\Psi}.
\end{equation}
We now consider the trivial weight $w\equiv 1$. The transport estimate then yields 
\begin{equation}\label{R_I_L2_est_D3D3phi}
    \begin{aligned}
        \lnm \ol{\Phi}\rnm^2_{\LS_1^2(S_{u,v})}\lesssim & \lnm \ol{\Phi}\rnm^2_{\LS_1^2(S_{u,0})}+\frac{v}{(-u)^{1-\k}}\sum_{i\leq 2}\lnm \ol{\Phi},\nabla^i\ol{\varphi},\ol{\Psi}\rnm^2_{\LS_0^2(H_u^v)}\lesssim \epsilon^2.
    \end{aligned}
\end{equation}

\subsection{Estimates along boundary of region I}
Our goal in this part is to estimate the difference quantities $\ol{\cdot}$ along the future boundary of Region I.
\begin{corollary}\label{R_I End Values}
    If $\epsilon=\epsilon(\epsilon_1)$ is sufficiently small, the following estimates hold along $z=v^{1/(1-\k)}/(-u)=\epsilon_1$:
    \begin{equation}\label{eq:R_I End Values 1}
        \begin{split}
            \sum_{i\leq 5}\int_{\SS}(-u)^{2-2s(\nabla^i\varphi)}\left|\nabla^i\ol{\varphi}\right|^2
            \lesssim \epsilon^{2-\delta/2},
        \end{split}
    \end{equation}
    and for any constant $a>0$,
    \begin{equation}\label{eq:R_I End Values 2}
        \int_{-\frac{v^{1/(1-\k)}}{z}}^u\int_{\SS} (-u^\prime)^{1+a-2s(\nabla^i\Phi)}\left|\nabla^i\ol{\Phi}\right|^2\left(u^\prime,|zu^\prime|^{1-\k}\right)\lesssim \epsilon^{2-\delta/2}v^{\frac{a}{1-\k}}.
    \end{equation}
    Here $\varphi=\Omega\chib,\Omega^{-1}\chi,\Omega\omegab,\Omega^{-1}\omega,\eta,\etab,$ $\Omega D_3\phi,\Omega^{-1} D_4\phi,\nabla\phi$, and $\Phi=\Omega^{-2}\alpha$, $\Omega^{2}\alphab$, $\Omega^{-1}\beta^r$, $\Omega\betab^r$, $\sigma^r$, $K^r,\nabla\varphi$.
\end{corollary}
\begin{proof}
    We first consider $\varphi=\Omega^{-1}\chih$ or $\Omega^{-1}e_4\phi$.
    Construction of approximating solution reveals that 
    \begin{equation}
       \begin{split}
            \lnm \ot{\varphi}-\varphi^c\rnm^2_{\LS_I^2\left(S_{u,(-\epsilon_1u)^{1-\k}}\right)}\lesssim& \int_{\SS}\left(\epsilon_1\right)^{-(1-\k)(1-2\delta)}\epsilon^2
         \leq  \epsilon^{2-\delta/2}.
       \end{split}
    \end{equation}
    Together with the estimates of $\wt{\varphi}$ and the fact $\wt{\varphi}-\ol{\varphi}=\varphi^c-\ot{\varphi}$, we deduce that 
    $$\lnm \nabla^i\ol{\varphi}\rnm^2_{\LS_I^2\left(S_{u,(-\epsilon_1u)^{1-\k}}\right)}\leq  \epsilon^{2-\delta/2}.$$
    For quantities with difference $\df{\cdot}$, we note that $\ol{\varphi}-\df{\varphi}=\ol{\varphi}_0$. The estimates follow the bound of $\df{\cdot}$ and the initial data construction. We therefore conclude \eqref{eq:R_I End Values 1}.

    To estimate $\LS_I^2(\Sigma)$ norm, we consider the Bianchi pair structure repeatedly used in the proof:
    \begin{equation}
        \begin{split}
            \Omega\nabla_3{\Phi_1}=\mathcal{D}{\Phi_2}+F,\quad \Omega^{-1}\nabla_4{\Phi_1}=-{}^*\mathcal{D}{\Phi_1}+G,
        \end{split}
    \end{equation}
    where $\mathcal{D}$ is a first order differential operator on $S_{u,v}$ and ${}^*\mathcal{D}$ is the Hodge dual operator. We use differentiation to convert the surface integration into bulk integration:
    \begin{equation}
        \begin{split}
            &(1-\k)z^{1+\k}\int_{-\frac{v^{1/(1-\k)}}{z}}^u \int_{\SS} (-u^\prime)^{1+a-2s(\Phi_1)}w\left|{\Phi_1}\right|^2(u^\prime,|zu^\prime|^{1-\k})\\
            =&\int_{|zu|^{1-\k}}^{v}\int_{\SS}\left(\frac{(v^\prime)^{1/(1-\k)}}{z}\right)^{1+a+\k-2s(\Phi_1)}w\left|{\Phi_1}\right|^2\left(-\left(\frac{(v^\prime)^{1/(1-\k)}}{z}\right),v^\prime\right)dv^\prime\\
            \sim &\int_{|zu|^{1-\k}}^{v}\int_{\SS}\left(\frac{(v^\prime)^{1/(1-\k)}}{z}\right)^{1+a-2s(\Phi_1)}\Omega^2w\left|{\Phi_1}\right|^2\left(-\left(\frac{(v^\prime)^{1/(1-\k)}}{z}\right),v^\prime\right)dv^\prime\\
            =&\int_{|zu|^{1-\k}}^{v}\int_{-\frac{v^{1/(1-\k)}}{z}}^{-\frac{(v^\prime)^{1/(1-\k)}}{z}}\int_{\SS}\partial_u\left((-u)^{1+a-2s(\Phi_1)}\Omega^2w\left|{\Phi_1}\right|^2\right)(u^\prime,v^\prime)du^\prime dv^\prime\\
            &+\int_{|zu|^{1-\k}}^{v}\int_{\SS}\left((-u)^{1+a-2s(\Phi_1)}\Omega^2w\left|{\Phi_1}\right|^2\right)\left(-\frac{v^{1/(1-\k)}}{z},v^\prime\right) dv^\prime.
        \end{split}
    \end{equation}
    For ${\Phi_2}$, we deduce that
    \begin{equation}
        \begin{split}
            & \int_{-\frac{v^{1/(1-\k)}}{z}}^u \int_{\SS} (-u^\prime)^{1+a-2s(\Phi_2)}w\left|{\Phi_2}\right|^2(u^\prime,|zu^\prime|^{1-\k})du^\prime\\
            =&\int_{-\frac{v^{1/(1-\k)}}{z}}^u\int_{|zu|^{1-\k}}^{|zu^\prime|^{1-\k}}\int_{\SS}\partial_v \left((-u)^{1+a-2s(\Phi_2)}w\left|{\Phi_2}\right|^2 \right)(u^\prime,v^\prime)dv^\prime du^\prime\\
            &+\int_{-\frac{v^{1/(1-\k)}}{z}}^u\int_{\SS}\left((-u)^{1+a-2s(\Phi_2)}w\left|{\Phi_2}\right|^2 \right)(u^\prime,|zu|^{1-\k})du^\prime\\
            =&\int_{|zu|^{1-\k}}^{v}\int_{-\frac{v^{1/(1-\k)}}{z}}^{-\frac{(v^\prime)^{1/(1-\k)}}{z}}\int_{\SS}\partial_v \left((-u)^{1+a-2s(\Phi_2)}w\left|{\Phi_2}\right|^2 \right)(u^\prime,v^\prime) du^\prime dv^\prime\\
            &+\int_{-\frac{v^{1/(1-\k)}}{z}}^u\int_{\SS}\left((-u)^{1+a-2s(\Phi_2)}w\left|{\Phi_2}\right|^2 \right)(u^\prime,|zu|^{1-\k})du^\prime.
        \end{split}
    \end{equation}
Using the above identities, the following inequalities holds 
\begin{equation}
    \begin{split}
        &\int_{|zu|^{1-\k}}^{v}\int_{-\frac{v^{1/(1-\k)}}{z}}^{-\frac{(v^\prime)^{1/(1-\k)}}{z}}\int_{\SS}\partial_u\left((-u)^{1+a-2s(\Phi_1)}\Omega^2w\left|{\Phi_1}\right|^2\right)(u^\prime,v^\prime)\\
        &\qquad\qquad\qquad +\partial_v \left((-u)^{1+a-2s(\Phi_2)}w\left|{\Phi_2}\right|^2 \right)(u^\prime,v^\prime)\\
        \lesssim & \int_{|zu|^{1-\k}}^{v}\int_{-\frac{v^{1/(1-\k)}}{z}}^{-\frac{(v^\prime)^{1/(1-\k)}}{z}}\int_{\SS} \left((-u^\prime)^{1+a-2s(\Phi_1)}\Omega^2w\left|{\Phi_1}\right|^2\right)\dv b+C (-u^\prime)^{a-2s(\Phi_1)} w\Omega^2\left|{\Phi_1}\right|^2\\
        &\quad\quad\quad\quad\quad\quad +(-u^\prime)^{a-2s(F_1)} w\Omega^2\left|{F}\right|^2+(-u^\prime)^{a-2s(F_2)} w\Omega^2\left|{G}\right|^2\quad du^\prime dv^\prime.
    \end{split}
\end{equation}
The bulk integration can be deduced from energy estimate along $H_u^v$ or $\Hb_v^u$. Specificly, if we have either 
$\lnm \Phi\rnm_{\LS_I^2(H_u^v)}\leq\epsilon$ or $\lnm \Phi\rnm_{\LS_I^2(\Hb_v^u)}\leq\epsilon$, then the bulk integration satisfies either
\begin{equation}
    \begin{split}
        &\int_{|zu|^{1-\k}}^{v}\int_{-\frac{v^{1/(1-\k)}}{z}}^{-\frac{(v^\prime)^{1/(1-\k)}}{z}}\int_{\SS}(-u^\prime)^{a-2s(\Phi)} w\Omega^2\left|{\Phi}\right|^2\\
        \leq & C(z)\int_{-\frac{v^{1/(1-\k)}}{z}}^u{(-u^\prime)^{a-1}}\lnm {\Phi}\rnm^2_{\LS_I^2\left(H_{u^\prime}^{|zu^\prime|^{1-\k}}\right)}\\
        \leq & C(z)\epsilon^2\int_{|zu|^{1-\k}}^v {(v^\prime)}^{\frac{a}{1-\k}-1}dv^\prime \leq  C(z)\epsilon^2 v^{\frac{a}{1-\k}}\ll \epsilon^{2-\delta/2}v^{\frac{a}{1-\k}},
    \end{split}
\end{equation}
or
\begin{equation}
    \begin{split}
        &\int_{|zu|^{1-\k}}^{v}\int_{-\frac{v^{1/(1-\k)}}{z}}^{-\frac{(v^\prime)^{1/(1-\k)}}{z}}\int_{\SS} (-u^\prime)^{a-2s(\Phi)} w\Omega^2\left|{\Phi}\right|^2\\
        \leq & C(z) \int_{|zu|^{1-\k}}^v{(v^\prime)^{\frac{a}{1-\k}-1}}\lnm {\Phi}\rnm^2_{\LS_I^2\left(\Hb_{v^\prime}^{-\frac{{(v^\prime)}^{1/(1-\k)}}{z}}\right)}\\
        \leq & C(z)\epsilon^2\int_{|zu|^{1-\k}}^v {(v^\prime)}^{\frac{a}{1-\k}-1}dv^\prime \leq  C(z)\epsilon^2 v^{\frac{a}{1-\k}}\ll \epsilon^{2-\delta/2}v^{\frac{a}{1-\k}}.
    \end{split}
\end{equation}
We are now ready to prove \eqref{eq:R_I End Values 2}, especially the top order estimates. 
For $(\Psi_1,\Psi_2)$ taking values in below Bianchi pairs, 
$$({\Omega^{-1}\beta^r},({\sigma^r},{K})),(({\sigma^r},{K}),{\Omega\betab^r}),({\Omega\betab^r},{\Omega^2\alphab}),(\nabla(\Omega^{-1}e_4\phi),{\nabla^2 \phi}),$$
$$(\nabla^2\phi,{\nabla(\Omega e_3\phi)}),(\nabla(\Omega\tr\chib),0),(0,{\nabla(\Omega^{-1}\tr\chi)}),({\Omega^{-1}\omega^\dagger},0), (0,{\Omega\omega^\dagger}),({\mub},0),(0,{\mu}),$$
their equations can be written as 
\begin{equation*}
    \Omega\nabla_3\Psi_1-\mathcal{D}\left[\Psi_2\right]_0^v\sim {\varphi^3} +{\varphi\Phi}+\nabla\Phi_0,\, 
\end{equation*}
and 
\begin{equation*}
    \begin{aligned}
        \Omega^{-1}\nabla_4\left[\Psi_2\right]_0^v+{}^*\mathcal{D}\ol{\Psi_1}\sim &{\varphi^3} +{\varphi\Phi}.
    \end{aligned}
\end{equation*}
Here we use $\varphi$ to represent all Ricci coefficients and first order derivatives of scalar field and $\Phi$ to represent $\nabla\varphi$ and curvature components except $\Omega^{-2}\alpha$. Besides, we recall that $\Phi_0$ means $\Phi|_{v=0}$.
We have equations for the top-order terms
\begin{equation*}
    \begin{aligned}
        &\Omega\nabla_3\nabla^5\Psi_1-\mathcal{D}\nabla^5\left[\Psi_2\right]_0^v\sim  \nabla^5(\varphi^3)+\nabla^5(\varphi\Phi)+\nabla^6\Phi_0+\sum_{i_1+i_2+i_3=4}\nabla^{i_1}K^{i_2+1}\nabla^{i_3}\df{\Psi_2},\\
        &\Omega^{-1}\nabla_4\nabla^5\df{\Psi_2}+{}^*\mathcal{D}\nabla^5\Psi_1\sim \nabla^5(\varphi^3)+\sum_{i_1+i_2+i_3=4} \nabla^{i_1}\varphi^{i_2+1}\nabla^{i_3}\Phi\\
        &\qquad\qquad\qquad\qquad\qquad\qquad\qquad+\sum_{i_1+i_2+i_3=4}\nabla^{i_1}K^{i_2+1}\nabla^{i_3}\ol{\Psi_1}.
    \end{aligned}
\end{equation*}
We prove that for every term on the right hand side satisfies that its $\LS_I^2(H_u^v)$ norm or $\LS_I^2(\Hb_v^u)$ norm is bounded by $\epsilon$. For non-top-order terms such as $\nabla^5(\varphi^3)$, initial value term $\nabla^6\Phi_0$, and $\nabla^{i_1}K^{i_2+1}\nabla^{i_3}(\ol{\Psi_1},\df{\Psi_2})$, we can use their spherical estimates to obtain that
\[\lnm l.o.t.\rnm^2_{\LS^2_1(S_{u,v})}\lesssim\epsilon^2.\]
Here we use $l.o.t.$ to represent these terms. We can thus estimate
\begin{equation}\label{eq:tmp 5-203}
    \begin{aligned}
        \lnm l.o.t.\rnm^2_{\LS^2_I(H_u^v)}\lesssim\epsilon^2\int_0^v\frac{1}{(-u)^{1-\k}}\left(\frac{v^\prime}{(-u)^{1-\k}}\right)^{2\delta-1}dv^\prime\lesssim\epsilon^2\z^{2\delta}\ll\epsilon^2.
    \end{aligned}
\end{equation}
If we decompose the top-order term $\varphi\nabla^5\Phi$ into $\varphi\nabla^5\ol\Phi_0$ and $\varphi\nabla^5\df{\Phi}$ (or $\varphi\nabla^5\ot\Phi$ and $\varphi\nabla^5\wt{\Phi}$), the first term can be dealt in the same way as \eqref{eq:tmp 5-203} and the difference term satisfies at least one of the following estimates
\[\lnm \nabla^6\df{\Phi}\rnm^2_{\LS_I^2(H_u^v)}\ll\epsilon^2,\,\lnm \nabla^6\df{\Phi}\rnm^2_{\LS_I^2(\Hb_v^u)}\ll\epsilon^2,\]
\[\lnm \nabla^6\wt{\Phi}\rnm^2_{\LS_I^2(H_u^v)}\ll\epsilon^2,\,\lnm \nabla^6\wt{\Phi}\rnm^2_{\LS_I^2(\Hb_v^u)}\ll\epsilon^2.\]
We thus arrive at \eqref{eq:R_I End Values 2} for these quantities.

The special pair is $\wt{\Omega^{-2}\alpha}$ and $\wt{\Omega^{-1}\beta}$. Recall the equations \eqref{eq:R_I_nab1_alpha_beta} and \eqref{eq:R_I_nab6_alpha_beta}. We need to deal with some extra terms:
\[\nabla^5\left[\left(\Lie_{\df{b}},\df{\Omega\chib},\df{\Omega\omegab},\df{\eta},\df{\etab},\Omega^{-2}\left[\Omega^2\dv\right]_0^v\right)\ot{\Omega^{-2}\alpha}\right],\,\nabla^6\ot{\Omega^{-1}\beta}.\]
Because $\nabla^6\ot{\Omega^{-1}\beta}$ has $ {\LS^2_1(S_{u,v})}$ norm of size $O(\epsilon)$, this term can be controlled via \eqref{eq:tmp 5-203}.
With the difference $\left[\cdot\right]_0^v$ which contributes to extra $\z$, the other term can be bounded via
\[\int_0^v\frac{1}{(-u)^{1-\k}}\left(\frac{v^\prime}{(-u)^{1-\k}}\right)^{2\delta-1}\left(\frac{v^\prime}{(-u)^{1-\k}}\right)^{2}\cdot\epsilon^2\left(\frac{v^\prime}{(-u)^{1-\k}}\right)^{2\k/(1-\k)-2\delta-2}\ll\epsilon^2.\]
The three $\z$ terms come from the weight function, $\left[\cdot\right]_0^v$ difference, and $\ot{\Omega^{-2}\alpha}$ separately. We can thus obtain \eqref{eq:R_I End Values 2} for $\Omega^{-2}\alpha$ and $\Omega^{-1}\beta$.

\end{proof}

\section{Estimates in Region II}\label{Estimates in Region II}

\subsection{From region I to II}
From now on, let $(U,V):=\left(-|u|^{1-\k},v^{\frac{1}{1-\k}}\right)$ and self-similar parameter $z:=V/(-U)^{1/(1-\k)},$
$$\mathcal{R}_{II}:=\left\{-1\leq U\leq 0,\epsilon_1\leq \frac{V}{(-U)^{1/(1-\k)}}\leq \epsilon_1^{-1}\right\}.$$ 
For a given $(U,V)$-convex domain, let $U_0(\widetilde{V})$ denote the intersection of the hypersurface $V=\widetilde{V}$ with the past boundary, and similarly let $V_0(\widetilde{U})$ denote the intersection of $U=\widetilde{U}$ with the past boundary.
We use the notation $^{(uv)}\psi, ^{(uv)}\Psi, ^{(uv)}\Upsilon$ to denote tensors evaluated in the old gauge, while $\psi, \Psi, \Upsilon$ refer to the corresponding quantities expressed in the $(U,V)$ coordinates.

We first derive some properties for Christodoulou's soluition.
\begin{lemma}
    In ${U\in(-1,0),V\in(0,\infty)}\setminus\mathcal{R}_{I}$, the following estimates hold.
    \[(\Omega^2)^{(c)}\sim V^{\k}(1+z^{-1-\k}),\]
    $$|(\Omega^{-1}\chib)^{(c)},(\Omega\chi)^{(c)},(\Omega^{-1}\omegab)^{(c)}|_g\lesssim V^{-1},\quad |(\Omega\omega)^{(c)}|_g\lesssim V^{-1},\quad |\rho^{(c)},\sigma^{(c)}|_g\lesssim V^{-2}.$$
    We write $A\lesssim_x B$ if there exists a constant $C(x)>0$ such that $A\leq C(x)B$. We say $A\sim_x B$ if and only if both $A\lesssim_x B$ and $B\lesssim_x A$ hold. The notations $A\lesssim B$ and $A\sim B$ are reserved for cases where the implicit constants are strictly universal.
\end{lemma} 
\begin{proof}
    The precise behavior of these geometric quantities in the regime $1\leq z<\infty$ is explicitly derived from Christodoulou's exact background solution, details of which will be further discussed in Appendix B. Next, we rigorously compute their limiting behavior near the specific hypersurface $V^{1-\k}/(-U)=\epsilon_1$.
    Write $(u,v)=(F(U),G(V))=(-|U|^{\frac{1}{1-\k}},V^{1-\k})$, then we find that 
\[e_3=\sqrt{\frac{F^\prime}{G^\prime}} \cdot {}^{(uv)}e_3,\quad e_4=\sqrt{\frac{G^\prime}{F^\prime}}\cdot ^{(uv)}e_4,\quad \Omega^2=F^\prime G^\prime \cdot^{(uv)}\Omega^2=\epsilon_1^{-1}\cdot ^{(uv)}\Omega^2,\]
\[\chib_{AB}=\sqrt{\frac{F^\prime}{G^\prime}}\cdot^{(uv)}\chib_{AB},\quad \chi_{AB}=\sqrt{\frac{G^\prime}{F^\prime}}\cdot^{(uv)}\chi_{AB},\]
\[\eta_A=^{(uv)}\eta_A,\quad \etab_A=^{(uv)}\etab_A,\]
\[\Omega\omega=G^\prime \cdot^{(uv)}(\Omega\omega)-\frac{G^{\prime\prime}}{4G^\prime},\quad \Omega\omegab=F^\prime \cdot^{(uv)}(\Omega\omegab)-\frac{F^{\prime\prime}}{4F^\prime},\]
\[\alpha_{AB}=\frac{G^\prime}{F^\prime}\cdot^{(uv)}\alpha_{AB},\quad \alphab_{AB}=\frac{F^\prime}{G^\prime}\cdot^{(uv)}\alphab_{AB},\]
\[\Omega\beta_A=\sqrt{\frac{G^\prime}{F^\prime}}\cdot ^{(uv)}(\Omega\beta_A),\quad \Omega\betab_A=\sqrt{\frac{F^\prime}{G^\prime}}\cdot ^{(uv)}(\Omega\betab_A),\]
\[\sigma=^{(uv)}\sigma,\quad \rho=^{(uv)}\rho.\]
Along $V/(-U)^{1/(1-\k)}=\epsilon_1$, we deduce that 
\[(\Omega^2)^{(c)}\sim |u|^\k\sim \epsilon_1^{-\k}V^{\k},\]
\[|\Omega^{-1}\chib^{(c)}|_g=\frac{1}{G^\prime}|^{(uv)}(\Omega^{-1}\chib)^{(c)}|_g\sim V^\k |u|^{-1-\k}\sim V^{-1}\epsilon_1^{1+\k},\]
\[|\Omega\chi^{(c)}|_g=G^\prime|^{(uv)}(\Omega\chi)^{(c)}|_g\sim V^{-\k}|u|^{-1+\k}\sim \epsilon_1^{1-\k} V^{-1},\]
\[|\Omega^{-1}\omegab^{(c)}|\sim \epsilon_1^{1+\k}V^{-\k}\left(-\frac{1}{1-\k}|U|^{\frac{\k}{1-\k}}|u|^{-1}+\frac{\k}{4(1-\k)}\frac{1}{-U}\right)\sim \epsilon_1^{2}V^{-1},\]
\[|(\Omega\omega)^{(c)}|\lesssim (1-\k)V^{-\k}|u|^{\k-1}+\frac{\k}{V}\lesssim V^{-1},\]
\[|\rho^{(c)},\sigma^{(c)}|\sim |u|^{-2}\sim\epsilon_1^{1+\k}V^{-2}.\]
\end{proof}

Based on the properties, we can define the $V$-signatures in region II and III.
\begin{definition}
    We define the $V$-signature $S(\psi)$ as follows:
    \[S(\Omega^2)= S(b)=1+\k,\quad S(g)=S(\phi)=1,\]
    \[S(\psi)=0,\quad \psi=\Omega^{-1}\tr\chib,\Omega^{-1}\chibh,\Omega\tr\chi,\Omega\chih,\Omega^{-1}\omegab,\Omega\omega,\eta,\etab,\]
    \[S(\Psi)=-1,\quad \Psi=\Omega^2\alpha,\Omega^{-2}\alphab,\Omega\beta^r,\Omega^{-1}\betab^r,K,\sigma^r.\]
    The signature $S$ then automatically satisfies the following simple arithmetic rules:
    \[S(\psi_1\psi_2)=S(\psi_1)+S(\psi_2)-1,\quad S(\Omega^{-1}\nabla_3\psi)=S(\Omega\nabla_4\psi)=S(\nabla\psi)=S(\psi)-1.\]
    We also define $S(V^a(-U)^b \psi)=S(\psi)+a+(1-\k)b$ and especially $S(z^a\psi)=S(\psi)$.
\end{definition}
\begin{remark}
    Note that $\Omega e_3=F^\prime {}^{(uv)}(\Omega e_3)$. We thus have $b\sim (-u)^{\k}\left({}^{(uv)}b\right)$ and $\left|b\right|\lesssim_{\epsilon_1} V^\k$ along $z=\epsilon_1$.
\end{remark}
We next introduce some minor adjustments to the norm system to properly accommodate the coordinate gauge change. Define:
\begin{equation}\label{R II norm sys}
    \begin{split}
        &\|\psi\|_{\LS_W^2(S_{U,V})}^2:=\int_{\SS}|\psi|^2V^{2-2S(\psi)}W(U,V)d\V,\\
        &\lnm \psi\rnm_{\LS_W^\infty(S_{u,v})}:=\sup_{S_{u,v}}\left|\psi\right|\cdot V^{1-S(\psi)}W(U,V)^{\frac{1}{2}},\\
        &\|\Psi\|_{\LS_W^2(H_U^V)}^2:=\int_{V_0(U)}^V\int_{\SS}|\Psi|^2(V^\prime)^{1-2S(\Psi)}W(U,V^\prime)d\V dV^\prime,\\
        &\|\Psi\|_{\LS_W^2(\Hb_V^U)}^2:=\int_{U_0(V)}^U\int_{\SS}|\Psi|^2V^{1-2S(\Psi)}W(U^\prime,V)\Omega^2 d\V dU^\prime,\\
        &\|{\Phi}\|^2_{\LS_W^2(\Sigma^z_{V})}:=\int_{0}^V\frac{1}{V^\prime}\|{\Phi}\|^2_{\LS_W^2\left(S_{-(V^\prime z^{-1})^{1-\k},V^\prime}\right)}dV^\prime,\\
        &\|\Psi\|_{\LS_W^2(\mathcal{R}_{U,V})}^2:=\int_{U_0(V)}^U\int_{V_0(U^\prime)}^V\int_{\SS}|\Psi|^2(V^\prime)^{-2S(\Psi)}\Omega^2W(U^\prime,V^\prime)d\V dV^\prime dU^\prime.
    \end{split}
\end{equation}
In particular, we define the convenient shorthand $\|\wt{\Phi}\|^2_{\LS_W^2(\Sigma^z)}=\|\wt{\Phi}\|^2_{\LS_W^2(\Sigma^z_{0,z})}.$
We use weight $W_{II}$ in region II,  
\begin{equation}
    W_{II}(U,V)=V^a\exp\left(-\frac{V^{1-\k}}{-U}D\right)=V^a W_0(U,V),
\end{equation}
where $D$ is a suitably large parameter to be determined later, and $a$ can be chosen as any fixed, strictly positive number.
We denote the norm system with weight $W_{II}$ by $\LSS$, and trivial weight $W=1$ by $\LS_1$.
Throughout this specific section, we will primarily focus on analyzing the fundamental difference quantities $\ol{\Phi}=\Phi-\Phi^c$.

The estimates in region I imply the initial data estimates in region II.
\begin{lemma}
    For the initial values of $\mathcal{R}_{II}$, the following estimates hold for $\epsilon=\epsilon(\epsilon_1)$ sufficiently small,
    \begin{equation}
        \sum_{i\leq 5}\lnm \nabla^i\ol{\psi},\nabla^i\ol{D\phi} \rnm^2_{\LSS^2\left(S_{U_0(V),V}\right)}+\lnm \nabla^i\ol{\Psi},\nabla^{i+1}\ol{D\phi},\nabla^{i+1}\ol\psi\rnm^2_{\LSS^2(\Sigma^{\epsilon_1}_{V})}\leq \epsilon^{2-\delta}V^a.
    \end{equation}
     Here 
    $\psi=\Omega^{-1}\chib,\Omega\chi,\Omega^{-1}\omegab,\Omega\omega,\eta,\etab$, $\Psi=\Omega^{2}\alpha,\Omega^{-2}\alphab,\Omega\beta^r,\Omega^{-1}\betab^r,\sigma^r,K^r$,\\ and $D\phi=\Omega^{-1}D_3\phi, \Omega D_4\phi,\nabla\phi$.
\end{lemma}
\begin{proof}
    Directly use Corollary \ref{R_I End Values} and the fact $V=\epsilon_1(-u)$. 
\end{proof}

\subsection{Bootstrap assumptions for Region II} We now state the bootstrap assumption with constant $B$ to be improved.
Assume that 
\begin{equation}\label{RII_Bootstrap_Assumptions}
    \begin{aligned}
        \sum_{i=0}^5\lnm \nabla^{i}\ol{\Phi_1}\rnm^2_{\LSS^2\left(H_U^V\right)}+\lnm \nabla^{i}\ol{\Phi_2}\rnm^2_{\LSS^2\left(\Hb_V^U\right)}\leq B\epsilon^{2-\delta}V^a \cdot W_0(U,V)^{-1}\ll  \epsilon V^a,
    \end{aligned}
\end{equation}
where $\Phi_1$ takes value in 
$$\nabla (\Omega^{-1}\chib),\nabla(\Omega\chi),\nabla(\Omega\omega),\nabla\etab,\nabla\eta, \Omega^2\alpha,\Omega\beta^r,\Omega^{-1}\betab^r,\sigma^r,K^r, \nabla^2\phi,\nabla (\Omega D_4\phi),$$
and $\Phi_2$ takes value in $$\nabla(\Omega^{-1}\chib),\nabla(\Omega\chi),\nabla(\Omega^{-1}\omegab),\nabla\etab,\nabla\eta, \Omega^{-2}\alphab,\Omega\beta^r,\Omega^{-1}\betab^r,\sigma^r,K^r, \nabla^2\phi,\nabla(\Omega^{-1} D_3\phi).$$ 
For the estimates in region II, we write conventionally that 
\begin{equation}\label{R_II_Convention}
\begin{aligned}
        &\varphi_1\in\{\Omega\chi,\Omega^{-1}\chib,\Omega\omega,\etab,\eta,\nabla\phi,\Omega D_4\phi\},\ \varphi_2\in\{\Omega\chi,\Omega^{-1}\chib,\Omega^{-1}\omegab,\etab,\eta,\nabla\phi,\Omega^{-1} D_3\phi\},\\
        &\Psi_1\in\{\Omega^2\alpha,\Omega\beta^r,\Omega^{-1}\betab^r,\sigma^r,K^r\},\Psi_2\in\{\Omega^{-2}\alphab,\Omega\beta^r,\Omega^{-1}\betab^r,\sigma^r,K^r\},
\end{aligned}
\end{equation} 
$\varphi=\varphi_1$ or $\varphi_2$, $\Psi=\Psi_1$ or $\Psi_2$
in this section to simplify the schematical equations.

\vspace{2mm}
Our main proposition for this analytical section is stated as follows:
\begin{proposition}\label{Main Proposition R_II}
    Building directly upon Proposition \ref{Main Proposition R_I}, we make the bootstrap assumption that \eqref{RII_Bootstrap_Assumptions} consistently holds throughout $\mathcal{R}_{II}$. Under this assumption, we can choose $D=D(\epsilon_1,a)$ to be sufficiently large and subsequently select $\epsilon=\epsilon(a,D,\epsilon_1)$ to be sufficiently small, such that the bootstrap constant in \eqref{RII_Bootstrap_Assumptions} can be strictly improved to $B\leq 1$.
\end{proposition}
The core estimates rely fundamentally on the following key functional lemma:
\begin{lemma}\label{R_II Key Lemma}
    Suppose that either \[\lnm \Phi\rnm^2_{\LSS^2\left(\Hb_V^U\right)}\leq C_1V^a W_0^{-1}(U,V)\ { or}\ \lnm \Phi\rnm^2_{\LSS^2\left(H_U^V\right)}\leq C_1V^a W_0^{-1}(U,V)\]
    holds.
    Then for any $\epsilon_2>0$, we can set $D=D(\epsilon_1,\epsilon_2)$ sufficiently large, so that
    \begin{equation}
        \lnm \Phi\rnm^2_{\LSS^2\left(\mathcal{R}_{U,V}\right)}\leq \epsilon_2 C_1V^a W_0^{-1}(U,V).
    \end{equation}
\end{lemma}
\begin{proof}
    We first assume $\lnm \Phi\rnm^2_{\LSS^2\left(\Hb_V^U\right)}\leq C_1V^a W_0^{-1}(U,V)$. 
    By definition, we have
    \begin{align*}
         \|\Phi\|_{\LSS^2(\mathcal{R}_{U,V})}^2W_0(U,V)
        \leq & \int_{V_0(U)}^V\frac{1}{V^\prime}\|\Phi\|^2_{\LSS^2(\Hb_{V^\prime}^U)}W_0(U,V)dV^\prime\\
        \leq & C_1\int_{V_0(U)}^V (V^\prime)^{a-1}\exp\left(D\frac{(V^\prime)^{1-\k}}{-U}-D\frac{V^{1-\k}}{-U}\right)dV^\prime\\
        =&\frac{C_1}{1-\k}\int_{\epsilon_1^{1-\k}}^{\frac{V^{1-\k}}{-U}}(-U)^{\frac{a}{1-\k}}z^{\frac{a}{1-\k}-1}\exp\left(Dz-D\frac{V^{1-\k}}{-U}\right)dz\\
        \leq &\frac{C_1}{1-\k} V^a \epsilon_1^{-a}\int_{\epsilon_1^{1-\k}}^{\epsilon_1^{\k-1}}z^{\frac{a}{1-\k}-1}\exp(Dz-D/\epsilon_1^{1-\k})dz.
    \end{align*}
    If function 
    \[F(Z)=\int_{\epsilon_1^{1-\k}}^{Z}z^{\frac{a}{1-\k}-1}\exp\left(D(z-Z)\right)dz\]
    attains its maximum at $Z_0\in (\epsilon_1^{1-\k},\epsilon_1^{\k-1})$, then we have $F^\prime(Z_0)=0$, which means 
    \[\int_{\epsilon_1^{1-\k}}^{Z}z^{\frac{a}{1-\k}-1}\exp\left(D(z-Z)\right)dz\leq \frac{Z_0^{\frac{a}{1-\k}-1}}{D}<\frac{\epsilon_1^{a-1+\k}}{D}.\]
    We can take $D=D(\epsilon_1,\epsilon_2)$ sufficiently large so that $\frac{\epsilon_1^{a-1+\k}}{D}\ll \epsilon_1^{a}\epsilon_2$. If $F(Z)$ attains its maximum on $(\epsilon_1^{\k-1},\infty)$, we can estimate $F(\epsilon_1^{\k-1})$.
    Since $\lim_{D\rightarrow\infty}\int_{\epsilon_1^{1-\k}}^{\epsilon_1^{\k-1}}z^{\frac{a}{1-\k}-1}\exp(Dz-D/\epsilon_1^{1-\k})dz=0$, we can assert that for any $D=D(a,\epsilon_1,\epsilon_2)$ chosen sufficiently large, $$ \frac{1}{1-\k}\epsilon_1^{-a}\int_{\epsilon_1^{1-\k}}^{\epsilon_1^{\k-1}}z^{-1}\exp(Dz-D/\epsilon_1)dz\leq\epsilon_2,$$
    which finishes the first part.
 
    Next, assuming instead that $\lnm \Phi\rnm^2_{\LSS^2\left(H_U^V\right)}\leq C_1V^a W_0^{-1}(U,V)$, we can similarly deduce
    \begin{equation}
        \begin{aligned}
            \lnm \Phi\rnm^2_{\LSS^2(\R_{U,V})}V^{-a}W_0(U,V)\leq &\epsilon_1^{-1+\k}\int_{U_0(V)}^U\frac{1}{-U^\prime}\exp\left(V^{1-\k}D\left(\frac{1}{-U^\prime}-\frac{1}{-U}\right)\right)dU^\prime\\
            \leq & C(\epsilon_1)\int_{-\left(\epsilon_1^{-1}V\right)^{1-\k}}^{U}\frac{1}{-U^\prime}\exp\left(D\left(\frac{V^{1-\k}}{-U^\prime}-\frac{V^{1-\k}}{-U}\right)\right)dU^\prime\\
            \leq & C(\epsilon_1)\int_{-U/V^{1-\k}}^{\epsilon_1^{\k-1}}\frac{1}{z}\exp\left(D\left(z-\frac{V^{1-\k}}{-U}\right)\right)dz.
        \end{aligned}
    \end{equation}
    Let $G(Z)=\int_{1/Z}^{\epsilon_1^{\k-1}}\frac{1}{z}\exp\left(D\left(z-Z\right)\right)dz$. $G^\prime(Z)=\frac{1}{Z^3}-D\cdot G(Z)$. If $G$ attains its maximum on $Z_0\in(\epsilon_1^{1-\k},\epsilon_1^{\k-1})$, we have 
    $G(Z_0)\leq \frac{\epsilon_1^{3\k-3}}{D}$. We can choose $D=D(\epsilon_1,\epsilon_2)$ sufficiently large so that $\frac{\epsilon_1^{3\k-3}}{D}\ll \frac{\epsilon_2}{C(\epsilon_1)}$. Since 
    $\lim_{D\rightarrow\infty}G(\epsilon_1^{\k-1})=0$, we can choose $D=D(\epsilon_1,\epsilon_2)$ such that 
    $G(\epsilon_1^{\k-1})\ll \frac{\epsilon_2}{C(\epsilon_1)}$ also holds. 
\end{proof}
Thus both null directions enjoy the same gain: a flux estimate along either family of cones produces a spacetime estimate with an arbitrarily small prefactor once $D$ is chosen sufficiently large. This small prefactor is the key tool used in the Region II bootstrap closure.
By employing identical mathematical arguments, we can rigorously prove the following powerful lemma.

\begin{lemma}\label{R_II Key Lemma Corollary}
    Suppose that $\lnm \Phi\rnm^2_{\LSS^2\left(S_{U,V}\right)}\leq C_1V^a W_0^{-1}(U,V).$
    Then for any $\epsilon_2>0$, we can set $D=D(a,\epsilon_1,\epsilon_2)$ sufficiently large, so that
    \begin{equation}
        \lnm \Phi\rnm^2_{\LSS^2\left(H_U^V\right)}+\lnm \Phi\rnm^2_{\LSS^2\left(\Hb_V^U\right)}\leq \epsilon_2 C_1V^a W_0^{-1}(U,V).
    \end{equation}

\end{lemma}
Besides these algebraic preliminaries, we must make certain necessary adjustments to the standard energy estimates, transport estimates, and trace formulae to properly align with the $\LSS$ norm systems. In a manner strictly analogous to the energy estimates derived in Region I ($\R_{I}$), we obtain the lemma below.
\begin{lemma}\label{R_II_ENERGY_ESTIMATE_LEMMA}
     Consider equations system:
    \begin{equation}
        \left\{
        \begin{aligned}
            \Omega^{-1}\nabla_3\ol{\Phi_1}+\psi_1\cdot\ol{\Phi_1}=& \mathcal{D}\ol{\Phi_2}+ F_1,\\
            \Omega \nabla_4\ol{\Phi_2}+\psi_2\cdot\ol{\Phi_2}=& -{}^*\mathcal{D}\ol{\Phi_1}+F_2.
        \end{aligned}
        \right.
    \end{equation}
    Here $\psi_1$ stands for $\Omega^{-1}\chib,\Omega^{-1}\omegab$, $\psi_2$ for $\Omega\chi,\Omega\omega$, and $\mathcal{D}$ is a differential operator with Hodge dual ${}^*\mathcal{D}$.
    Suppose that $\lnm\dv b,\psi_i\rnm_{\LS_1^\infty(S)} \leq C_0$, and 
    \begin{equation}
        \lnm \nabla\left(\frac{\sqrt{\det g(U,V)}}{V^2}\right)\rnm_{L^2(\SS)}\leq C_1.
    \end{equation}
    Then for $D=D(C_0,C_1,\epsilon_1)$ sufficiently large, the following energy estimate holds:
    \begin{equation}\label{R_II Energy Estimates}
        \begin{split}
            \lnm \ol{\Phi_1}\rnm^2_{\LSS^2\left(H_U^V\right)}+\lnm \ol{\Phi_2}\rnm^2_{\LSS^2\left(\Hb_V^U\right)}
            \leq  \lnm \ol{\Phi_1},\ol{\Phi_2}\rnm^2_{\LSS^2\left(\Sigma^{\epsilon_1}_{U,V}\right)}+\lnm F_1,F_2\rnm^2_{\LSS^2\left(\mathcal{R}_{U,V}\right)}.
        \end{split}
    \end{equation}
\end{lemma}
\begin{proof}
    Through a standard, systematic process of geometric differentiation and subsequent spacetime integration, we derive
    \begin{equation*}
        \begin{aligned}
            &\lnm \ol{\Phi_1}\rnm^2_{\LSS^2\left(H_U^V\right)}+\lnm \ol{\Phi_2}\rnm^2_{\LSS^2\left(\Hb_V^U\right)}\\ 
            \lesssim & \int_{V_0(U)}^V \int_{U_0(V^\prime)}^U\int_{\SS} 2\Omega^2 \ol{\Phi_1}\cdot F_1{V^\prime}^{-2S(\Phi_1)}W_{II}+2\ol{\Phi_2}\cdot F_2{V^\prime}^{-2S(\Phi_2)}\Omega^2\\
            &+\int_{V_0(U)}^V \int_{U_0(V^\prime)}^U\int_{\SS} \Omega^2 \left|\ol{\Phi_1}\right|^2{V^\prime}^{1-2S(\Phi_1)}W_{II}\left( \Omega^{-2}\partial_U \log\left(V^{1-2S(\Phi_1)}W_{II}\right)-2\psi_1 +\Omega^{-2}\dv b\right)\\
            &+\int_{V_0(U)}^V \int_{U_0(V^\prime)}^U\int_{\SS}\left|\ol{\Phi_2}\right|^2{V^\prime}^{1-2S(\Phi_2)}\Omega^2\left(\partial_V\log\left(V^{1-2S(\Phi_2)}\Omega^2 W_{II}\right)-2\psi_2\right)\\
            &+\lnm \ol{\Phi_1}\rnm^2_{\LSS^2\left(\Sigma^{\epsilon_1}_{U,V}\right)}+\left|\partial_VU_0\right|\cdot\lnm \ol{\Phi_2}\rnm^2_{\LSS^2\left(\Sigma^{\epsilon_1}_{U,V}\right)} + C(C_1) \lnm \ol{\Phi_1}\rnm_{\LSS^2(\R_{U,V})}\lnm \ol{\Phi_2}\rnm_{\LSS^2(\R_{U,V})} \\
            \lesssim & \int_{V_0(U)}^V \int_{U_0(V^\prime)}^U\int_{\SS} \Omega^2 \left|F_1\right|^2{V^\prime}^{1-2S(F_1)}W_{II}+\left|F_2\right|^2{V^\prime}^{1-2S(F_2)}\Omega^2W_{II}\\
            &+\int_{V_0(U)}^V \int_{U_0(V^\prime)}^U\int_{\SS} \Omega^2 \left|\ol{\Phi_1}\right|^2{V^\prime}^{1-2S(\Phi_1)}W_{II}\left( \Omega^{-2}\partial_U \log\left(V^{1-2S(\Phi_1)}W_{II}\right)-2\psi_1+C(C_1)V^{-1} \right)\\
            &+\int_{V_0(U)}^V \int_{U_0(V^\prime)}^U\int_{\SS}\left|\ol{\Phi_2}\right|^2{V^\prime}^{1-2S(\Phi_2)}\Omega^2\left(\partial_V\log\left(V^{1-2S(\Phi_2)}\Omega^2 W_{II}\right)-2\psi_2+C(C_1)V^{-1}\right)\\
            &+\lnm \ol{\Phi_1}\rnm^2_{\LSS^2\left(\Sigma^{\epsilon_1}_{U,V}\right)}+\left|\partial_VU_0\right|\cdot\lnm \ol{\Phi_2}\rnm^2_{\LSS^2\left(\Sigma^{\epsilon_1}_{U,V}\right)}.
        \end{aligned}
    \end{equation*} 
    Since $\Omega^{-2}\partial_U \log\left(V^{1-2S(\Phi_1)}W_{II}\right)\lesssim -C({\epsilon_1})D V^{-1}$, and $$\partial_V\log\left(V^{1-2S(\Phi_2)}\Omega^2 W_{II}\right)\lesssim \left(-C(\epsilon_1)D +C_0\right)V^{-1},$$ we can choose $D=D(\epsilon_1)$ sufficiently large so that 
    \[\Omega^{-2}\partial_U \log\left(V^{1-2S(\Phi_1)}W_{II}\right)\ll \psi_1-C(C_1)V^{-1},\quad \partial_V\log\left(V^{1-2S(\Phi_2)}\Omega^2 W_{II}\right)\ll \psi_2-C(C_1)V^{-1}.\]
\end{proof}
By closely following similar mathematical arguments to those utilized in Region I ($\R_I$), we can establish the necessary transport estimates. We choose to omit the detailed proof here, as the structural differences from the previous estimates \eqref{transport_3} and \eqref{transport_4} in $\R_I$ are only very slight.
\begin{lemma}
    We consider transport equations
    \begin{equation}
        \begin{aligned}
            \Omega^{-1}\nabla_3\ol{\Phi_1}+\psi_1\cdot\ol{\Phi_1}= F_1,\quad 
            \Omega \nabla_4\ol{\Phi_2}+\psi_2\cdot\ol{\Phi_2}=F_2,
        \end{aligned}
    \end{equation}
    where $\psi_1$ contains $\Omega^{-1}\chib,\Omega^{-1}\omegab$ and $\psi_2$ contains $\Omega\chi,\Omega\omega$.
    Suppose that $$\lnm\psi_i\rnm_{\LS_1^\infty(S)} \leq C_0,\ \lnm \nabla b,\Omega^2\rnm_{\LS_1^\infty(S)}\leq C_0.$$
    Then we can choose $D$ sufficiently large so that the following estimates hold:
    \begin{equation}\label{R_II transport estimate e_3}
        \lnm \ol{\Phi_1}\rnm^2_{\LSS^2\left(S_{U,V}\right)}\leq \lnm \ol{\Phi_1}\rnm^2_{\LSS^2\left(S_{U_0(V),V}\right)}+\lnm F_1\rnm^2_{\LSS^2\left(\Hb_V^U\right)},
    \end{equation}
    \begin{equation}\label{R_II transport estimate e_4}
        \lnm \ol{\Phi_2}\rnm^2_{\LSS^2\left(S_{U,V}\right)}\leq \lnm \ol{\Phi_2}\rnm^2_{\LSS^2\left(S_{U,V_0(U)}\right)}+\lnm F_2\rnm^2_{\LSS^2\left(H_U^V\right)}.
    \end{equation} 
\end{lemma}

\vspace{2mm}
Similar to estimate \eqref{Calcu_Diff_product_ol}, the following holds:
\begin{lemma}\label{R_II_Estimate_Lemma_for_difference_of_Product}
    Let $n\geq 3$. Suppose that ${\psi_i}$ satisfies $$\sum_{i\leq n-1}\lnm \ol{\psi_i}\rnm_{\LS_1^2(S_{U,V})}+\sum_{i\leq n}\lnm\nabla^i \psi_i^c\rnm^2_{\LS_1^\infty(S_{U,V})}\lesssim 1.$$ Then for $X=S_{U,V},H_U^V,\Hb_V^U,\R_{U,V}$, we have 
    \begin{equation}
        \sum_{i\leq n}\lnm \nabla^i\left(\psi_1\cdots \psi_{m-1}\ol{\psi_m}\right)\rnm^2_{\LSS^2(X)}+\sum_{i\leq n}\lnm \nabla^i\left(\ol{\psi_1\cdots \psi_m}\right)\rnm^2_{\LSS^2(X)}\lesssim \sum_{j=1}^m\sum_{i\leq n}\lnm \nabla^i\ol{\psi_j}\rnm^2_{\LSS^2(X)}.
    \end{equation}
\end{lemma}
\begin{remark}
    Since for all quantities in interest, it holds that $\nabla\psi^c=0$, we obtain have 
    \begin{equation}
        \sum_{1\leq i\leq n}\lnm \nabla^i\left({\psi_1\psi_2}\right)\rnm^2_{\LSS^2(X)}\lesssim  \sum_{j=1,2}\sum_{i\leq n}\lnm \nabla^i\ol{\psi_j}\rnm^2_{\LSS^2(X)}.
    \end{equation}
\end{remark}

\subsection{Estimates of spherical norms}

We first estimate the metric components $b,g,\Omega$.
\begin{proposition}
    The following estimate holds for the metric perturbation components:
    \begin{equation}
        \sum_{i\leq 6}\lnm \nabla^i\ol{g},\nabla^i b,\nabla^i\ol{\Omega^2}\rnm^2_{\LSS^2(S_{U,V})}\leq C(B,\epsilon_1)\epsilon^{2-\delta}V^a W_0(U,V)^{-1}.
    \end{equation}
\end{proposition}
\begin{proof}
    Recall that  
    \begin{equation}
        \begin{aligned}
            \Lie_v b=2\Omega^2(\etab-\eta),\ \Lie_v\ol{g_{AB}}=2\ol{\Omega\chi_{AB}},\ \Lie_v\ol{\Omega^2}+4\Omega\omega\ol{\Omega^2}=-4\ol{\Omega\omega}(\Omega^2)^c.
        \end{aligned}
    \end{equation}
    For $\ol{\Omega^2}$, by transport estimate \eqref{R_II transport estimate e_4},
    \begin{equation*}
        \begin{aligned}
            \lnm \ol{\Omega^2}\rnm^2_{\LSS^2(S_{U,V})}\lesssim & \epsilon^{2-\delta}V^a+\lnm \ol{\Omega\omega}\cdot(\Omega^2)^c\rnm^2_{\LSS^2(H_U^V)}\lesssim C(\epsilon_1,B)\epsilon^{2-\delta}V^a W_0(U,V)^{-1},
        \end{aligned}
    \end{equation*}
    where we use that $(\Omega^2)^c\lesssim V^\k (z^{-\k}+1)$. 
    For $\epsilon=\epsilon(\epsilon_1,B)$ sufficiently small, $\lnm \Omega^2\rnm^2_{\LSS^2(S_{U,V})}\lesssim C(\epsilon_1) W(U,V)$.
    For higher order derivatives, by commutation formula \eqref{Comm_Formula_D4,nab},
    \begin{equation}
        \Omega\nabla_4\nabla^I\Omega^2=\nabla^I(\Omega\omega\cdot\Omega^2) +\sum_{\substack{i_1+i_2=I\\i_2\geq 1}}\nabla^{i_1}(\Omega\chi)\nabla^{i_2}\Omega^2.
    \end{equation}
    Therefore 
    \begin{equation*}
        \begin{aligned}
            \sum_{1\leq i\leq 6}\lnm \nabla^i\Omega^2\rnm^2_{\LSS^2(S_{U,V})}\lesssim &\epsilon^{2-\delta}V^a+C\left(\epsilon_1\right)\sum_{i\leq 6,j\leq 5}\lnm \nabla^j(\ol{\Omega\chi}),\nabla^i\ol{\Omega\omega} \rnm^2_{\LSS^2(H_U^V)}\\
            \lesssim & C(\epsilon_1,B)\epsilon^{2-\delta}V^a W_0(U,V)^{-1}.
        \end{aligned}
    \end{equation*} 
In the source terms below, the top-order quantities $\nabla^6\varphi_2$ and $\nabla^5K$ are handled by the spacetime gain of Lemma \ref{R_II Key Lemma}, while all lower-order products are controlled by Proposition \ref{R_II_non_top_order_L2S_est}. The metric estimates from the preceding subsection justify the use of Lemma \ref{R_II_ENERGY_ESTIMATE_LEMMA}. Using commutaion formula \eqref{Comm_Formula_D4,nab} for $b$, we obtain 
    \begin{equation}
        \Omega\nabla_4 \nabla^I b\sim \nabla^I(\Omega^2(\etab-\eta))+\nabla^I(\Omega\chi b).
    \end{equation}
    We hence have transport estimate
    \begin{equation}
        \begin{aligned}
            \sum_{i\leq 6}\lnm \nabla^i b\rnm^2_{\LSS^2(S_{U,V})}\leq \epsilon^{2-\delta}V^a+\sum_{i\leq 6}\lnm \nabla^i\left(\eta,\etab,\Omega^2,\Omega\chi\right)\rnm^2_{\LSS^2(H_U^V)}\lesssim B \epsilon^{2-\delta}V^a W_0(U,V)^{-1}.
        \end{aligned}
    \end{equation}
    For $g_{AB}$, we find that 
    \begin{equation}
        \Lie_v\nabla^I\ol{g}\sim \nabla^I\left(\Omega\chi\ol{g}\right) .
    \end{equation}
    The transport estimate holds:
    \begin{equation}
        \sum_{i\leq 6}\lnm \nabla^i\ol g\rnm^2_{\LSS^2(S_{U,V})}\leq \epsilon^{2-\delta}V^a+\sum_{i\leq 6}\lnm \nabla^i\ol{\Omega\chi}\rnm^2_{\LSS^2(H_U^V)}\lesssim B\epsilon^{2-\delta}V^a W_0(U,V)^{-1}. 
    \end{equation}

\end{proof}
By Sobolev inequality \eqref{Sobolev_W22}, we can choose $\epsilon=\epsilon(\epsilon_1,B)$ sufficiently small so that 
\begin{equation}\label{R_II_metric_L4infty}
    \sum_{i\leq 4}\lnm \nabla^i\left(\ol{\Omega^2},\ol{g},{b}\right)\rnm^2_{\LS_1^\infty(S)}\leq \epsilon^{2-2\delta}.
\end{equation}
The metric condition of energy estimate \eqref{R_II Energy Estimates} are thus verified. Consequently, all later applications of the Region II energy estimate may invoke \eqref{R_II Energy Estimates} without separately reproving the bounds on $\dv b$, the lapse, or the angular volume form. 

With the estimates of metric, we can control the difference of covariant derivatives.
\begin{corollary}
    For any $(a,b)$ tensor $\psi\in H^6(\SS)$, it follows that 
    \begin{equation}
        \sum_{1\leq i\leq 6}\lnm \ol{\nabla^i}\psi\rnm^2_{L^2(\SS)}\lesssim \epsilon\lnm \psi\rnm^2_{H^5(\SS)},
    \end{equation}
    where $\ol{\nabla^i}_{A_1\cdots A_i}\psi$ is defined to be the $(a,r+b)$ tensor $(\nabla_g)^i_{A_1\cdots A_i}\psi-(\nabla_{g^c})^i_{A_1\cdots A_i}\psi$.
    
    For any function $f\in H^7(\SS)$, we also have 
    \begin{equation}
        \sum_{1\leq i\leq 7}\lnm \ol{\nabla^i}f\rnm^2_{L^2(\SS)}\lesssim \epsilon\lnm f\rnm^2_{H^6(\SS)}.
    \end{equation}
\end{corollary}
\begin{proof}
    For $i=1$, $\ol{\nabla_A}\psi=\ol{g^{-1}\partial_{\SS}g}\psi=\psi\cdot o_{L^\infty(\SS)}(\epsilon^{1-\delta}),$ the statement holds obviously,
    \begin{equation}
        \lnm \ol{\nabla}\psi\rnm^2_{L^2(\SS)}\lesssim \epsilon\lnm \psi\rnm^2_{L^2(\SS)}.
    \end{equation}
    Since we can algebraically decompose the difference as $\ol{\nabla^i}\psi=\nabla_g^{i-1}\ol{\nabla}\psi+\ol{\nabla^{i-1}}\nabla_{g^c}\psi=\nabla_g^{i-1}\left(\ol{g^{-1}\partial_{\SS}g}\psi\right)+\ol{\nabla^{i-1}}\nabla_{g^c}\psi,$ 
    we naturally proceed to estimate it inductively. For the lower-order cases $i\leq 3$, we directly apply the previously established metric estimate \eqref{R_II_metric_L4infty},
    \begin{equation}
        \begin{aligned}
            \lnm \ol{\nabla^i}\psi\rnm^2_{L^2(\SS)}\lesssim & \lnm \nabla_g^{i-1}\left(\ol{g^{-1}\partial_{\SS}g}\psi\right)\rnm^2_{L^2(\SS)}+\epsilon\lnm \nabla_{g^c}\psi\rnm^2_{H^{i-2}(\SS)}\\
            \lesssim & \sum_{j\leq i-1}\epsilon^{2-2\delta}\lnm \nabla^{j}\psi\rnm^2_{L^2(\SS)}+\epsilon\lnm \nabla_{g^c}\psi\rnm^2_{H^{i-2}(\SS)}\lesssim \epsilon\lnm \psi\rnm^2_{H^{i-1}(\SS)}.
        \end{aligned}
    \end{equation}
    When $i\geq 4$, we can use $\lnm \psi\rnm_{W^{i-3,\infty}(\SS)}\lesssim \lnm \psi\rnm_{H^{i-1}(\SS)}$ for those term cannot be estimated by \eqref{R_II_metric_L4infty} to derive that
    \begin{equation}
        \begin{aligned}
            \lnm \ol{\nabla^i}\psi\rnm^2_{L^2(\SS)}\lesssim & \lnm \nabla_g^{i-1}\left(\ol{g^{-1}\partial_{\SS}g}\psi\right)\rnm^2_{L^2(\SS)} +\epsilon\lnm \nabla_{g^c}\psi\rnm^2_{H^{i-2}(\SS)}\\
            \lesssim & \sum_{j\leq 3}\epsilon\lnm \nabla^{i-1-j}\psi\rnm^2_{L^2(\SS)}+\sum_{4\leq j\leq i}\lnm \nabla^{i-1-j}\psi\rnm^2_{L^\infty(\SS)}\epsilon +\epsilon\lnm \nabla_{g^c}\psi\rnm^2_{H^{i-2}(\SS)}\\
            \lesssim &  \epsilon\lnm \psi\rnm^2_{H^{i-1}(\SS)}.
        \end{aligned}
    \end{equation}
    To establish the second statement, we simply notice the identity $\ol{\nabla^i}f=\ol{\nabla^{i-1}}(\nabla f)$.
\end{proof}
By Lemma \ref{R_II Key Lemma Corollary}, when $D=D(\epsilon_1,a,B)$ is sufficiently large, difference of lapse, $\ol{\Omega^2}$, satisfies
\begin{equation}
    \lnm \nabla^i\ol{\Omega^2}\rnm^2_{\LSS^2(H_U^V)}\leq  \epsilon^{2-\delta} V^a\cdot W_0(U,V)^{-1},\, \text{for}\, i\leq 6.
\end{equation}

We are now formally prepared to systematically estimate the various non-top-order terms. 
\begin{proposition}\label{R_II_non_top_order_L2S_est}
    With notations \eqref{R_II_Convention}, the following relation holds 
    \begin{equation}
        \sum_{i\leq 5}\lnm\nabla^i\ol\varphi \rnm^2_{\LSS^2(S_{U,V})}+\sum_{j\leq 4}\lnm\nabla^j\ol\Psi \rnm^2_{\LSS^2(S_{U,V})}\lesssim B\epsilon^{2-\delta}W_0(U,V)^{-1}.
    \end{equation}
\end{proposition}
\begin{proof} 
    We first assume \begin{equation}
        \sum_{i\leq 5,j\leq 4}\lnm\nabla^i\ol{\varphi_2} ,\nabla^j\ol{\Psi_2}\rnm^2_{\LSS^2(S_{U,V})}\lesssim B^\prime\epsilon^{2-\delta}W_0(U,V)^{-1},
    \end{equation}
    then Lemma \ref{R_II Key Lemma Corollary} implies that 
    \begin{equation}
        \sum_{i\leq 5,j\leq 4}\lnm\nabla^i\ol{\varphi_2} ,\nabla^j\ol{\Psi_2}\rnm^2_{\LSS^2(H_U^V)}\leq \epsilon^{2-\delta}W_0(U,V)^{-1}.
    \end{equation}
    Crucially, in the full set of null structure equations, there are absolutely no isolated $\omega^2$ or $\omegab^2$ terms present. Consequently, all quadratic product terms involving the generic connection coefficients $\psi$ and scalar derivatives $D\phi$ can be cleanly rewritten in the structured form $\varphi_1\varphi_2$.
    We have schematical equations 
    \begin{equation}
        \begin{aligned}
            \Omega^{-1}\nabla_3\varphi_1+\varphi_2\varphi_1=\nabla\varphi_2+\varphi_2^2+\Psi_2,&\ \Omega\nabla_4\varphi_2+\varphi_1\varphi_2=\nabla\varphi_1+\varphi_2^2+\Psi_1,\\
            \Omega^{-1}\nabla_3\Psi_1+\varphi\Psi_1=\nabla\Psi_2+\varphi\Psi_2+\varphi\nabla\varphi_2,&\ \Omega\nabla_4\Psi_2+\varphi\Psi_2=\nabla\Psi_1+\varphi\Psi_2+\varphi\nabla\varphi_1.
        \end{aligned}
    \end{equation}
    The commutation formula are simply written as 
    \begin{equation}\label{R_II_simplified_Comm_for}
        \begin{aligned}
            \Omega^{-1}\nabla_3\nabla^i\psi=&\sum_{i_1+i_2+i_3=i}\nabla^{i_1}\varphi_2^{i_2}\nabla^{i_3}\left(\Omega^{-1}\nabla_3\psi\right)+\nabla^{i_1}\varphi_2^{i_2+1}\nabla^{i_3}\psi,\\
            \Omega\nabla_4\nabla^i\psi=&\nabla^i\left(\Omega^{-1}\nabla_3\psi\right)+\sum_{i_1+i_2+i_3=i}\nabla^{i_1}\varphi_2^{i_2+1}\nabla^{i_3}\psi.
        \end{aligned}
    \end{equation}
    With transport estimates \eqref{R_II transport estimate e_3} and \eqref{R_II transport estimate e_4}, we obtain 
    \begin{equation}
        \begin{aligned}
            \lnm \ol{\varphi_2}\rnm^2_{\LSS^2(S_{U,V})}\lesssim \epsilon^{2-\delta}V^a+\lnm \nabla\varphi_1,\ol{\varphi_2},\ol{\varphi_1},\ol{\Psi_1}\rnm^2_{\LSS^2(H_U^V)} \lesssim B\epsilon^{2-\delta}V^aW_0(U,V)^{-1},
        \end{aligned}
    \end{equation}
    and for higher order derivatives, we can use formula \eqref{R_II_simplified_Comm_for} to obtain
    \begin{equation}
        \sum_{1\leq i\leq 5}\lnm \nabla^i\varphi_2\rnm^2_{\LSS^2(S_{U,V})}\lesssim \epsilon^{2-\delta}V^a+\sum_{i\leq 5}\lnm \nabla^{i}\left(\nabla\varphi_1,\ol{\varphi_2},\ol{\varphi_1},\ol{\Psi_1}\right)\rnm^2_{\LSS^2(H_U^V)} \lesssim B\epsilon^{2-\delta}V^aW_0(U,V)^{-1}.
    \end{equation}
    Then for $\varphi_1$, the equation for the difference takes the form: $\Omega^{-1}\nabla_3\ol{\varphi_1}+\varphi_2\ol{\varphi_1}=\nabla\varphi_2+\ol{\varphi_2^2}+\ol{\Psi_2}-\ol{\Omega^{-1}\nabla_3}\varphi_1^c-\ol{\varphi_2}\varphi_1^c.$
    Because $\lnm\varphi^c,\Psi^c\rnm_{\LSS^\infty(S_{U,V})}\lesssim 1$, it holds that 
    \begin{equation}
        \lnm \ol{\varphi_1}\rnm^2_{\LSS^2(S_{U,V})}\lesssim \epsilon^{2-\delta}V^a+\lnm \nabla\varphi_2,\ol{\varphi_2},\ol{\Psi_2}\rnm^2_{\LSS^2(H_U^V)} \lesssim B\epsilon^{2-\delta}V^aW_0(U,V)^{-1}.
    \end{equation}
    With \eqref{R_II_simplified_Comm_for}, the higher order derivatives satisfy
    \begin{equation}
        \sum_{1\leq i\leq 5}\lnm \nabla^i\varphi_1\rnm^2_{\LSS^2(S_{U,V})}\lesssim \epsilon^{2-\delta}V^a+\sum_{i\leq 5}\lnm \nabla^{i}\left(\nabla\varphi_2,\ol{\varphi_2},\ol{\Psi_1}\right)\rnm^2_{\LSS^2(H_U^V)} \lesssim B\epsilon^{2-\delta}V^aW_0(U,V)^{-1}.
    \end{equation}
    For $\Psi_1,\Psi_2$, repeating the arguments of $\varphi_1$, $\varphi_2$ respectively gives the estimates.
\end{proof}

\subsection{Top order energy estimates}

We are now ready to estimate the top-order terms.
\begin{proposition}\label{R_II_top_order_est_scalar_field}
    The following inequalities hold:
    \begin{equation}
        \lnm \nabla^6(\Omega D_4\phi)\rnm^2_{\LSS^2(H_U^V)}+\lnm \nabla^7\phi\rnm^2_{\LSS^2(\Hb_V^U)}\lesssim \epsilon^{2-\delta} V^a\cdot W_0(U,V)^{-1},
    \end{equation}
    \begin{equation}
        \lnm \nabla^7\phi\rnm^2_{\LSS^2(H_U^V)}+\lnm \nabla^6(\Omega^{-1}D_3\phi)\rnm^2_{\LSS^2(\Hb_V^U)}\lesssim \epsilon^{2-\delta} V^a\cdot W_0(U,V)^{-1}+\lnm \nabla^6(\Omega\omega)\rnm^2_{\LSS^2(\mathcal{R}_{U,V})}.
    \end{equation}
\end{proposition}
\begin{proof}
    Using commutaion formula \eqref{R_II_simplified_Comm_for}, we can write schematical equations for pair $(\Omega D_4\phi,\nabla\phi)$:
    \begin{equation}
        \begin{aligned}
            &\Omega^{-1}\nabla_3\nabla^6(\Omega D_4\phi)-\dv \nabla^7\phi\\
            &\qquad\qquad=\sum_{i_1+i_2+i_3=6}\nabla^{i_1}\varphi_2^{i_2+1}\nabla^{i_3}(\Omega D_4\phi,\varphi_2)+\sum_{i_1+i_2+i_3=5}\nabla^{i_1}K^{i_2+1}\nabla^{i_3+1}\phi,\\
            &\Omega\nabla_4\nabla^7\phi-\nabla^7\Omega D_4\phi\\
            &\qquad\qquad=\sum_{i_1+i_2+i_3=6}\nabla^{i_1}\varphi_2^{i_2+1}\nabla^{i_3}\varphi_2+\sum_{i_1+i_2+i_3=5}\nabla^{i_1}K^{i_2+1}\nabla^{i_3}\Omega D_4\phi.
        \end{aligned}
    \end{equation}
    Then the energy estimates from Lemma \ref{R_II_ENERGY_ESTIMATE_LEMMA} give
    \begin{equation}
        \begin{aligned}
            &\lnm \nabla^6\Omega D_4\phi\rnm^2_{\LSS^2(H_U^V)}+\lnm \nabla^7\phi\rnm^2_{\LSS^2(\Hb_V^U)}\\
            &\qquad\qquad\lesssim  \epsilon^{2-\delta} V^a+\sum_{i\leq 6,j\leq 5}\lnm \nabla^i\ol{\varphi_2},\nabla^j\ol{\Psi_2},\nabla^j(\Omega D_4\phi)\rnm^2_{\LSS^2(\R_{U,V})}.
        \end{aligned}
    \end{equation}
    With Lemma \ref{R_II Key Lemma}, we can estimate the top order terms $\nabla^6\varphi_2$ and $\nabla^5K$. Using Corollary \ref{R_II Key Lemma Corollary} and Proposition \ref{R_II_non_top_order_L2S_est}, the non-top-order terms can be estimated. We have thus 
    \begin{equation}
        \lnm \nabla^6\Omega D_4\phi\rnm^2_{\LSS^2(H_U^V)}+\lnm \nabla^7\phi\rnm^2_{\LSS^2(\Hb_V^U)}
            \lesssim  \epsilon^{2-\delta} V^a W_0(U,V)^{-1}.
    \end{equation}
For $(\nabla\phi,\Omega^{-1} D_3\phi)$, the equations read
\begin{equation}
    \begin{aligned}
        &\Omega^{-1}\nabla_3\nabla^7\phi-\nabla^7(\Omega^{-1}D_3\phi)\\ &\qquad\qquad=\sum_{i_1+i_2=6}\nabla^{i_1}\varphi_2^{i_2+2}+\sum_{i_1+i_2+i_3=5}\nabla^{i_1}K^{i_2+1}\nabla^{i_3}(\Omega^{-1}D_3\phi),\\
        &\Omega \nabla_4\nabla^6(\Omega^{-1}D_3\phi)-\dv\nabla^7\phi\\ &\qquad\qquad=\sum_{i_1+i_2+i_3=6}\nabla^{i_1}\varphi_2^{i_2+1}\nabla^{i_3}(\Omega \omega,\varphi_2)+\sum_{i_1+i_2+i_3=5}\nabla^{i_1}K^{i_2+1}\nabla^{i_3+1}\phi.
    \end{aligned}
\end{equation}
With similar arguments, we obtain 
\begin{equation}
    \lnm \nabla^7\phi\rnm^2_{\LSS^2(H_U^V)}+\lnm \nabla^6\Omega^{-1}D_3\phi\rnm^2_{\LSS^2(\Hb_V^U)}\lesssim \epsilon^{2-\delta} V^a\cdot W_0(U,V)^{-1}+\lnm \nabla^6\Omega\omega\rnm^2_{\LSS^2(\mathcal{R}_{U,V})}.
\end{equation}
\end{proof}
\begin{proposition}
    We establish the following estimate for the relevant quantities:
    \begin{equation}
        \lnm \nabla^6(\Omega\omega)\rnm^2_{\LSS^2(H_U^V)}\lesssim \epsilon^{2-\delta} V^a\cdot W_0(U,V)^{-1}+\lnm \nabla^5(\Omega\beta^r)\rnm^2_{\LSS^2(H_U^V)}.
    \end{equation}
\end{proposition}
\begin{proof}
    Consider auxiliary function $\Omega\omega^\dagger$ defined by  
    \[\Omega^{-1}\nabla_3\Omega\omega^\dagger=\frac{1}{2}\sigma^r,\  \Omega\omega^\dagger|_{H_{-1}}\equiv 0.\]
    By transport estimate \eqref{R_II transport estimate e_3}, the following inequality holds 
    \begin{equation}
        \sum_{i\leq 5}\lnm \nabla^i(\Omega\omega^\dagger)\rnm^2_{\LSS^2(S_{U,V})}\lesssim B\epsilon^{2-\delta} V^a\cdot W_0(U,V)^{-1}.
    \end{equation}
    Let $\omega^r=\nabla\omega+{}^*\nabla\Omega\omega^\dagger-\frac{1}{2}\Omega\beta^r$. The schematic equation then takes the form
    \[\Omega^{-1}\nabla_3\omega^r+\varphi\omega^r=\varphi\Psi_2+\varphi^3+\varphi\nabla(\Omega D_4\phi)+\varphi\nabla\varphi_2,\]
    where $\varphi$ can be $\Omega\omega^\dagger$ additionally.
    Moreover, 
    \[\Omega^{-1}\nabla_3\nabla^5\omega^r+\varphi\nabla^5\omega^r=\sum_{\substack{i_1+i_2+i_3=5}}\nabla^{i_1}\varphi^{i_2}\nabla^{i_3}(\Omega^{-1}\nabla_3\omega^r)+\sum_{\substack{i_1+i_2+i_3=5\\i_3\leq 4}}\nabla^{i_1}\varphi^{i_2+1}\nabla^{i_3}\omega^r.\]
    Energy estimates gives that 
    \begin{equation}
        \begin{aligned}
            \lnm \nabla^5\omega^r\rnm^2_{\LSS^2(H_U^V)}\lesssim & \epsilon^{2-\delta}V^a+\sum_{i\leq 5,j\leq 4}\lnm \nabla^i\ol{\varphi},\nabla^j\left(\omega^r,\ol{\Psi_2}\right)\rnm^2_{\LSS^2(\R_{U,V})}\\
            &+\lnm \nabla^5\Psi_2,\nabla^6\Omega D_4\phi,\nabla^6\varphi_2\rnm^2_{\LSS^2(\R_{U,V})}.
        \end{aligned}
    \end{equation}
    $\nabla^6\Omega D_4\phi$ is bounded by Proposition \ref{R_II_top_order_est_scalar_field}, and $\nabla^5\Psi_2,\nabla^6\varphi_2$ can be estimated by Lemma \ref{R_II Key Lemma}. The non top order terms are estimated by Corollary \ref{R_II Key Lemma Corollary} and Proposition \ref{R_II_non_top_order_L2S_est}. 
    The elliptic estimate finally asserts that $\lnm \nabla^6\Omega\omega\rnm^2_{\LSS^2(H_U^V)}\lesssim \epsilon^{2-\delta} V^a\cdot W_0(U,V)^{-1}+\lnm \nabla^5(\Omega\beta^r)\rnm^2_{\LSS^2(H_U^V)}.$
\end{proof}
\begin{corollary}
    The following estimates hold for the relevant quantities:
    \begin{equation}
        \lnm \nabla^7\phi\rnm^2_{\LSS^2(H_U^V)}+\lnm \nabla^6\Omega^{-1}D_3\phi\rnm^2_{\LSS^2(\Hb_V^U)}+\lnm \nabla^6\varphi\rnm^2_{\LSS^2(\mathcal{R}_{U,V})}\lesssim \epsilon^{2-\delta} V^a\cdot W_0(U,V)^{-1}.
    \end{equation}
\end{corollary}
\begin{proof}
    It remains to estimate $\lnm \nabla^6(\Omega\omega)\rnm^2_{\LSS^2(\mathcal{R}_{U,V})}$. Using Lemma \ref{R_II Key Lemma} and the bootstrap assumption \eqref{RII_Bootstrap_Assumptions}, it follows that 
    \begin{equation}
        \lnm \nabla^6(\Omega\omega)\rnm^2_{\LSS^2(\mathcal{R}_{U,V})}\lesssim \epsilon^{2-\delta} V^a\cdot W_0(U,V)^{-1}+\lnm \nabla^5(\Omega\beta^r)\rnm^2_{\LSS^2(\R_{U,V})}\lesssim \epsilon^{2-\delta} V^a\cdot W_0(U,V)^{-1}.
    \end{equation}
\end{proof}

We proceed to estimate the curvature components.
\begin{corollary}
    We have following energy estimates:
    \begin{equation}
        \begin{split}
            &\lnm \nabla^5\left(\Omega^{2}\alpha,\Omega\beta^r,K^r,\sigma^r,\Omega^{-1}\betab^r\right)\rnm^2_{\LSS^2(H_U^V)}+\lnm \nabla^5\left(\Omega\beta^r,K^r,\sigma^r,\Omega^{-1}\betab^r,\Omega^{-2}\alphab\right)\rnm^2_{\LSS^2(\Hb_V^U)}\\
            \lesssim & \epsilon^{2-\delta} V^a\cdot W_0(U,V)^{-1}.
        \end{split}
    \end{equation}
\end{corollary}
\begin{proof}
    Note that all curvature pairs have equations in the form  
    \begin{equation}
        \left\{\begin{aligned}
            \Omega^{-1}\nabla_3\Psi_1+\varphi\cdot\Psi_1=&\mathcal{D}\Psi_2+\varphi\Psi_2+\varphi\nabla\varphi,\\
            \Omega \nabla_4\Psi_2=&-{}^*\mathcal{D}\Psi_1+\varphi\cdot\Psi_2+\varphi\nabla\varphi.
        \end{aligned}\right.
    \end{equation}
    Employing energy estimates, we deduce that
    \begin{equation}
        \begin{aligned}
            \lnm \nabla^5\Psi_1\rnm^2_{\LSS^2(H_U^V)}+\lnm \nabla^5\Psi_2\rnm^2_{\LSS^2(\Hb_V^U)}\lesssim & \epsilon^{2-\delta} V^a+\sum_{i\leq 6,j\leq 5}\lnm \nabla^i\ol{\varphi},\nabla^j\ol{\Psi_2}\rnm^2_{\LSS^2(\R_{U,V})}\lesssim \epsilon^{2-\delta} V^a.
        \end{aligned}
    \end{equation}
\end{proof}

We next estimate $\Omega^{-1}\omegab$, which is avoided in other equations in region II and III.
\begin{proposition}
    We establish the following estimate:
    \begin{equation}
        \lnm \nabla^6\left(\Omega^{-1}\omegab\right)\rnm^2_{\LSS^2(\Hb_V^U)}\lesssim \epsilon^{2-\delta} V^a\cdot W_0(U,V)^{-1}.
    \end{equation}
\end{proposition}
\begin{proof}
    We consider auxiliary function defined by  $$\Omega\nabla_4\left(\Omega^{-1}\omegab^\dagger\right)=\frac{1}{2}\sigma^r+4\Omega\omega\Omega^{-1}\omegab^\dagger,\  \Omega\omegab^\dagger|_{\Hb_0}\equiv 0.$$  
    Using transport estimate \eqref{R_II transport estimate e_4}, one can easily show that 
    \begin{equation}
        \sum_{i\leq 5}\lnm \nabla^i\left(\Omega^{-1}\omegab^\dagger\right)\rnm^2_{\LSS^2(S_{U,V})}\lesssim B\epsilon^{2-\delta} V^a\cdot W_0(U,V)^{-1}.
    \end{equation}
    Let $\omegab^r=-\nabla\left(\Omega^{-1}\omegab\right)+{}^*\nabla\left(\Omega^{-1}\omegab^\dagger\right)-\frac{1}{2}\Omega^{-1}\betab^r$. The corresponding equation is
    \begin{equation}
        \Omega\nabla_4\omegab^r=\varphi\Psi_2+\varphi\nabla\varphi+\varphi^3,
    \end{equation}
    where $\varphi$ can be $\Omega^{-1}\omega^\dagger$ additionally.
    Using energy estimate, we deduce that 
    \begin{equation}
        \begin{split}
            \lnm \nabla^5\omegab\rnm^2_{\LSS^2(\Hb_V^U)}\lesssim & \sum_{i\leq 6,j\leq 5}\lnm \nabla^i\varphi,\nabla^j\Psi_2\rnm^2_{\LSS^2(\mathcal{R}_{U,V})}+\epsilon^{2-\delta} V^a\cdot W_0(U,V)^{-1}\\
            \lesssim & \epsilon^{2-\delta} V^a\cdot W_0(U,V)^{-1}+\lnm \nabla^6\left(\Omega^{-1}\omegab,\Omega^{-1}\omegab^\dagger\right)\rnm^2_{\LSS^2(\mathcal{R}_{U,V})}.
        \end{split}
    \end{equation}
    Elliptic estimate reveals that 
    \begin{equation}
        \lnm \nabla^6\left(\Omega^{-1}\omegab,\Omega^{-1}\omegab^\dagger\right)\rnm^2_{\LSS^2(\Hb_V^U)}\lesssim \epsilon^{2-\delta} V^a\cdot W_0(U,V)^{-1}+\lnm \nabla^5\omegab^r,\nabla^5(\Omega^{-1}\betab^r)\rnm^2_{\LSS^2(\Hb_V^U)}.
    \end{equation}
    By Lemma \ref{R_II Key Lemma}, $\lnm \nabla^5\omegab^r,\nabla^5(\Omega^{-1}\betab^r)\rnm^2_{\LSS^2(\Hb_V^U)}\lesssim \epsilon^{2-\delta} V^a\cdot W_0(U,V)^{-1}$. We thus finish the proof.

\end{proof}

We proceed to estimate other Ricci coefficients.
\begin{proposition}
    We derive the following estimates for $\Omega\chi,\Omega^{-1}\chib,\eta,\etab$:
    \begin{equation}
        \lnm \nabla^6(\Omega\chi,\Omega^{-1}\chib,\eta,\etab)\rnm^2_{\LSS^2(H_U^V)}+\lnm \nabla^6(\Omega\chi,\Omega^{-1}\chib,\eta,\etab)\rnm^2_{\LSS^2(\Hb_V^U)}\lesssim \epsilon^{2-\delta} V^a\cdot W_0(U,V)^{-1}.
    \end{equation}
 \end{proposition}
 \begin{proof}
    We firstly estimate $\Omega\chi,\Omega^{-1}\chib$.  
    Assume that 
    \begin{equation}\label{R_II_Bootstrap_trchi_trchib}
        \lnm \nabla^6\Omega\tr\chi,\nabla^6\Omega^{-1}\tr\chib\rnm^2_{\LSS^2(S_{U,V})}\leq B_1\epsilon^{2-\delta} V^aW_0(U,V)^{-1}.
    \end{equation}
    By elliptic estimate \eqref{Ellip Est}, Codazzi equations give the estimates for $\Omega\chih,\Omega^{-1}\chibh$: 
    \begin{equation}
        \begin{split}
            \lnm \nabla^6\Omega\chih\rnm^2_{\LSS^2(S_{U,V})}\lesssim &\lnm \nabla^6\Omega\tr\chi,\nabla^5\Omega\beta^r,\nabla^5K,\nabla^5(\etab\cdot\Omega\chi)\rnm^2_{\LSS^2(S_{U,V})},\\
            \lnm \nabla^6\Omega^{-1}\chibh\rnm^2_{\LSS^2(S_{U,V})}\lesssim &\lnm \nabla^6\Omega^{-1}\tr\chib,\nabla^5\Omega^{-1}\betab^r,\nabla^5(\etab\cdot\Omega^{-1}\chib)\rnm^2_{\LSS^2(S_{U,V})}.
        \end{split}
    \end{equation}
    The energy estimates for $\Omega\beta^r$, $\Omega^{-1}\betab^r$, the non top order $\LSS^2(S)$ estimates Proposition \ref{R_II_non_top_order_L2S_est}, and assumption \eqref{R_II_Bootstrap_trchi_trchib} together give that 
    \begin{equation}
        \begin{split}
            \lnm \nabla^6(\Omega\chih,\Omega^{-1}\chibh)\rnm^2_{\LSS^2(H_U^V)}+
            \lnm \nabla^6(\chih,\chibh)\rnm^2_{\LSS^2(\Hb_V^U)}\lesssim & \epsilon^{2-\delta} V^a\cdot W_0(U,V)^{-1}.
        \end{split}
    \end{equation}
    From $\nabla_4\tr\chi$ and $\nabla_3\tr\chib$ equations, we obtain 
    \begin{equation*}
        \begin{aligned}
            \lnm \nabla^6\Omega\tr\chi\rnm^2_{\LSS^2(S_{U,V})}\lesssim & \lnm \nabla^6\tr\chi\rnm^2_{\LSS^2(S_{U,V_0(U)})}+\sum_{i\leq 5}\lnm \nabla^i\ol\varphi\rnm^2_{\LSS^2(H_U^V)}\\
            &+\lnm \nabla^6\left(\Omega D_4\phi,\Omega\omega,\Omega\tr\chi\right)\rnm^2_{\LSS^2(H_U^V)} +\lnm \nabla^6\left(\Omega\chih\cdot\Omega\chih\right)\rnm^2_{\LSS^2(H_U^V)},
        \end{aligned}
    \end{equation*}
    \begin{equation*}
        \begin{aligned}
            \lnm \nabla^6\Omega^{-1}\tr\chib\rnm^2_{\LSS^2(S_{U,V})}\lesssim & \lnm \nabla^6\tr\chib\rnm^2_{\LSS^2(S_{U_0(V),V})}+\sum_{i\leq 5}\lnm \nabla^i(\eta,\etab,\ol{\Omega^{-1}\chib}) \rnm^2_{\LSS^2(\Hb_V^U)}\\
            &+\lnm \nabla^6(\Omega^{-1}D_3\phi,\Omega^{-1}\tr\chib)\rnm^2_{\LSS^2(\Hb_V^U)}
            +\lnm \nabla^6(\Omega^{-1}\chibh\cdot\Omega^{-1}\chibh)\rnm^2_{\LSS^2(\Hb_V^U)}.
        \end{aligned}
    \end{equation*}
    Using Corollary \ref{R_II Key Lemma Corollary}, the result follows immediately. 
    To estimate $\eta,\etab$, we consider 
    \[\mub=-\dv\etab+K^r,\quad \mu=-\dv\eta+K^r.\]
    Then the equations read 
    \begin{equation}
        \begin{split}
            \Omega^{-1}\nabla_3\mub+\varphi\mub=&\varphi\Psi_2+\varphi\nabla(\Omega^{-1}\chib,\eta,\nabla\phi,\Omega^{-1} D_3\phi)+\varphi^3,\\
            \Omega\nabla_4\mu+\varphi\mu=&\varphi\Psi_1+\varphi\nabla(\Omega\chi,\etab,\nabla\phi,\Omega D_4\phi)+\varphi^3.
        \end{split}
    \end{equation}
    Employing energy estimate \eqref{R_II Energy Estimates}, it follows that 
    \begin{equation}\label{R_II est eta etab 1}
        \begin{split}
            \lnm \nabla^5\mub\rnm^2_{\LSS^2(S_{U,V})}\lesssim & \lnm \nabla^6(\Omega^{-1}\chib,\eta,\nabla\phi,\Omega^{-1} D_3\phi),\nabla^5\Psi_2\rnm^2_{\LSS^2(\Hb_V^U)}+\epsilon^{2-\delta} V^a\cdot W_0(U,V)^{-1}\\
            \lesssim & \lnm \nabla^6\eta\rnm^2_{\LSS^2(\Hb_V^U)}+\epsilon^{2-\delta} V^a\cdot W_0(U,V)^{-1},\\
            \lnm \nabla^5\mu\rnm^2_{\LSS^2(S_{U,V})}\lesssim & \lnm \nabla^6(\Omega\chi,\etab,\nabla\phi,\Omega D_4\phi),\nabla^5\Psi_2\rnm^2_{\LSS^2(H_U^V)}+\epsilon^{2-\delta} V^a\cdot W_0(U,V)^{-1}\\
            \lesssim & \lnm \nabla^6\etab\rnm^2_{\LSS^2(H_U^V)}+\epsilon^{2-\delta} V^a\cdot W_0(U,V)^{-1}.
        \end{split}
    \end{equation}
    Elliptic estimate \eqref{Ellip Est} asserts that 
    \begin{equation}\label{R_II est eta etab 2}
        \begin{split}
            \lnm \nabla^6\eta\rnm^2_{\LS^2(H_U^V)}\lesssim & \lnm \nabla^5\mu\rnm^2_{\LS^2(H_U^V)}+\epsilon^{2-\delta} V^a\cdot W_0(U,V)^{-1},\\
            \lnm \nabla^6\etab\rnm^2_{\LS^2(H_U^V)}\lesssim & \lnm \nabla^5\mub\rnm^2_{\LS^2(H_U^V)}+\epsilon^{2-\delta} V^a\cdot W_0(U,V)^{-1},\\
            \lnm \nabla^6\eta\rnm^2_{\LS^2(\Hb_V^U)}\lesssim & \lnm \nabla^5\mu\rnm^2_{\LS^2(\Hb_V^U)}+\epsilon^{2-\delta} V^a\cdot W_0(U,V)^{-1},\\
            \lnm \nabla^6\etab\rnm^2_{\LS^2(\Hb_V^U)}\lesssim & \lnm \nabla^5\mub\rnm^2_{\LS^2(\Hb_V^U)}+\epsilon^{2-\delta} V^a\cdot W_0(U,V)^{-1}.\\
    \end{split}
    \end{equation} 
    Combining \eqref{R_II est eta etab 1} and \eqref{R_II est eta etab 2}, we obtain 
    \begin{equation}
        \begin{split}
            \lnm \nabla^5\mub\rnm^2_{\LSS^2(S_{U,V})}\lesssim & \lnm \nabla^5\mu\rnm^2_{\LSS^2(\Hb_V^U)}+\epsilon^{2-\delta} V^a\cdot W_0(U,V)^{-1}\\
            \lesssim & \lnm \nabla^5\mub\rnm^2_{\LSS^2(\mathcal{R}_{U,V})}+\epsilon^{2-\delta} V^a\cdot W_0(U,V)^{-1},\\
            \lnm \nabla^5\mu\rnm^2_{\LSS^2(S_{U,V})}\lesssim & \lnm \nabla^5\mub\rnm^2_{\LSS^2(H_U^V)}+\epsilon^{2-\delta} V^a\cdot W_0(U,V)^{-1}.
        \end{split}
    \end{equation}
    Finally, a straightforward bootstrap argument combined with Lemma \ref{R_II Key Lemma Corollary} allows us to definitively conclude that 
    \begin{equation}
        \lnm \nabla^5\mub,\nabla^5\mu\rnm^2_{\LSS^2(S_{U,V})}\lesssim \epsilon^{2-\delta} V^a\cdot W_0(U,V)^{-1}.
    \end{equation}
    This rigorously completes the proof.
 \end{proof}
\section{Estimates in Region III}\label{Estimates in Region III}
In this section, we continue to employ the double null gauge $(U,V)$ introduced in $\mathcal{R}_{II}$, alongside the rescaled self-similar coordinate $z=V/(-U)^{1/(1-\k)}$. We direct our focus to the specific region defined by $$\mathcal{R}_{III}=\left\{\frac{V}{(-U)^{\frac{1}{1-\k}}}\geq \epsilon_1^{-1},-1<U<0\right\}.$$ 
Christodoulou's solution satisfies the following lemma.
\begin{lemma}\label{R_III_Chr_sol_lemma}
    In $\mathcal{R}_{III}$, this yields 
    \[{g^c}_{AB}\sim V^2{g^{\SS}}_{AB} ,\quad  (\Omega^{c})^2\sim V^\k,\quad (\Omega\tr\chi)^{c}\sim \frac{1}{V},\quad (\Omega^{-1}\tr\chib)^{c}\sim -\frac{1}{V},\quad {K}^c\sim \frac{1}{V^2}.\]
    Especially, we have $$\lim_{z\rightarrow\infty}V\cdot (\Omega\tr\chi)^{c}=2,\quad\lim_{z\rightarrow\infty}V\cdot (\Omega\omega)^c=-\frac{\k}{4}.$$
\end{lemma}
\subsection{From II to III}
For any point $(U,V)\in\mathcal{R}_{III}$, let $U_0(\widetilde{V})$ denote the $U$-coordinate of the intersection between the incoming null hypersurface $V=\widetilde{V}$ and the past boundary of $\mathcal{R}_{III}$. Furthermore, we let $U_{-1}(\widetilde{V})$ denote its intersection with the past boundary of $\mathcal{R}_{II}$. The quantities $V_0(U)$ and $V_{-1}(U)$ are defined entirely analogously.
\begin{lemma}
    Building upon the results of Proposition \ref{Main Proposition R_II}, we can systematically select $\epsilon=\epsilon(\epsilon_1,D)$ to be sufficiently small, such that for any $(U,V)\in\mathcal{R}_{III}$ the following initial value estimates hold:
    \begin{equation}\label{R III initial value 1}
        \sum_{i\leq 2}\int_{\SS}\left|\nabla^i\ol{\varphi}\right|^2V^{a+1-2S(\Phi)}\Omega^2(U_0(V),V)d\V\leq \epsilon^{2-3\delta/2}V^a,
    \end{equation}
    \begin{equation}\label{R III initial value 2}
        \sum_{i\leq 2}\int_{V_0(U)}^V-U_0^\prime(V^\prime)\int_{\SS}\left|\nabla^i\ol{\Phi}\right|^2(V^\prime)^{a+1-2S(\Phi)}\Omega^2d\V dV^\prime\leq \epsilon^{2-3\delta/2}V^a,
    \end{equation}
    \begin{equation}\label{R III initial value 3}
        \sum_{i\leq 2}\int_{U_0(V)}^U -V_0^\prime(U^\prime)\int_{\SS}\left|\nabla^i\ol{\Phi}\right|^2\left(V_0(U^\prime)\right)^{a+1-2S(\Phi)}d\V dU^\prime\leq \epsilon^{2-3\delta/2}V^a,
    \end{equation}
    where $\varphi=\psi,D\phi$ and $\Phi=\nabla\varphi,\Psi$.
\end{lemma}
\begin{proof}
    The bound \eqref{R III initial value 1} is an immediate and straightforward consequence of Proposition \ref{Main Proposition R_II}.  
    We restrict our detailed analysis strictly to the case where $V\leq \epsilon_1^{-1}$, given that $U_0(V)\equiv -1$ identically whenever $V>\epsilon_1^{-1}$. In this regime, the boundary curves are explicitly given by $V_0(U)=\frac{\left(-U\right)^{1/(1-\k)}}{\epsilon_1}$ and $U_0(V)=-\left(\epsilon_1 V\right)^{1-\k}$.
    Through straightforward direct computation, we obtain
\begin{equation}
    \begin{aligned}
        \int_{V_0(U)}^V-U_0^\prime &(V^\prime)\int_{\SS}|\ol{\Phi}|^2(V^\prime)^{a+1-2S(\Phi)}\Omega^2d\V dV^\prime\\
        &=(1-\k)\epsilon_1^{1-\k}\int_{V_0(U)}^V\int_{\SS}|\ol{\Phi}|^2(V^\prime)^{a+1-\k-2S(\Phi)}\Omega^2d\V dV^\prime,
    \end{aligned}
\end{equation}
and 
\begin{equation*}
    \begin{aligned}
        &\int_{U_0(V)}^U -V_0^\prime(U^\prime)\int_{\SS}|\ol{\Phi}|^2[V_0(U^\prime)]^{a+1-2S(\Phi)}d\V dU^\prime\\
        &\qquad\qquad\qquad=\frac{1}{1-\k}\int_{V_0(U)}^V\int_{\SS}|\ol{\Phi}|^2(V^\prime)^{a+1-2S(\Phi)}d\V dV^\prime.
    \end{aligned}
\end{equation*}
To establish the remaining estimates \eqref{R III initial value 2} and \eqref{R III initial value 3}, we analyze the following simplified schematic system of equations:
\begin{equation}\label{R III initial est eq sys}
    \begin{split}
        \Omega\nabla_3\ol{\Phi_1}=\mathcal{D}\ol{\Phi_2}+F_1,\quad \Omega^{-1}\nabla_4\ol{\Phi_1}=-{}^*\mathcal{D}\ol{\Phi_1}+F_2.
    \end{split}
\end{equation}
A differentiation and integration argument yields:
\begin{equation}\label{R III initial est 1.1}
    \begin{split}
        &\int_{V_0(U)}^V F(U_0(V^\prime),V^\prime)dV^\prime\\
    =&\int_{V_0(U)}^{V_{-1}\circ U_0(V)}\int_{U_{-1}(V^\prime)}^{U_0(V^\prime)}\partial_U F(U^\prime,V^\prime)dU^\prime dV^\prime+\int_{V_{-1}\circ U_0(V)}^V\int_{U_0(V)}^{U_0(V^\prime)}\partial_U F(U^\prime,V^\prime)dU^\prime dV^\prime \\
    &+\int_{V_0(U)}^{V_{-1}\circ U_0(V)}F(U_{-1}(V^\prime),V^\prime)dV^\prime+\int_{V_{-1}\circ U_0(V)}^V F(U_0(V),V^\prime)dV^\prime,
    \end{split}
\end{equation}
and
\begin{equation}\label{R III initial est 1.2}
    \begin{split}
        &\int_{V_0(U)}^V F(U_0(V^\prime),V^\prime)dV^\prime\\
    =&-\int_{U_0(V)}^U F(U^\prime,V_0(U^\prime))V_0^\prime(U^\prime)dU^\prime\\
    \sim &\epsilon_1^{\k-1} \int_{U_0(V)}^UF(U^\prime,V_0(U^\prime))\Omega^2dU^\prime\\
    =&\epsilon_1^{\k-1} \int_{U_0(V)}^{U_{-1}\circ V_0(U)} \int_{V_{-1}\circ U_0(V)}^{V_0(U^\prime)}\partial_V(\Omega^2 F)(U^\prime,V^\prime)dV^\prime dU^\prime \\
    &+ \epsilon_1^{\k-1}  \int_{U_{-1}\circ V_0(U)}^{U}\int_{V_0(U)}^{V_0(U^\prime)}\partial_V(\Omega^2 F)(U^\prime,V^\prime)dV^\prime dU^\prime\\
    &+\epsilon_1^{\k-1}  \int_{U_0(V)}^{U_{-1}\circ V_0(U)}\Omega^2 F (U_{-1}\circ V_{0}(U^\prime),U^\prime)dU^\prime +\epsilon_1^{\k-1}  \int_{U_{-1}\circ V_0(U)}^{U}\Omega^2 F(U_0(V),U^\prime)dU^\prime,
    \end{split}
\end{equation}
where we used that $-V_0^\prime(U)=\frac{1}{1-\k}\epsilon_1^{\k-1}(V_0(U))^\k\sim\frac{1}{1-\k}\epsilon_1^{\k-1}\Omega^2(U,V_0(U))$.
Assuming the bound $\lnm \ol\Phi\rnm^2_{\LSS^2\left(\Hb_V^U\right)}\lesssim C(\epsilon_1,D)\epsilon^{2-\delta} V^a$ holds, we can directly deduce that 
\begin{equation}\label{R III initial est 1.3}
    \begin{split}
        &\left(\int_{V_0(U)}^{V_{-1}\circ U_0(V)}\frac{1}{V^\prime}\int_{U_{-1}(V^\prime)}^{U_0(V^\prime)}+\int_{V_{-1}\circ U_0(V)}^V\frac{1}{V^\prime}\int_{U_0(V)}^{U_0(V^\prime)}\right) |\ol\Phi|^2(V^\prime)^{a+1-2S(\Phi)}\Omega^2W_0(U^\prime,V^\prime) dU^\prime dV^\prime\\
    \lesssim &C(\epsilon_1,D)\int_{V_0(U)}^{V}\frac{1}{V^\prime} \lnm \ol\Phi\rnm^2_{\LSS^2\left(\Hb_{V^\prime}^{U_0(V^\prime)}\right)}\\
    \lesssim & C(\epsilon_1,D)\epsilon^{2-\delta} V^a.
    \end{split}
\end{equation}
If $\lnm \ol\Phi\rnm^2_{\LSS^2\left(H_U^V\right)}\lesssim C(\epsilon_1,D)\epsilon^{2-\delta} V^a,$ we have also
\begin{equation}\label{R III initial est 1.4}
    \begin{split}
        &\left(\int_{U_0(V)}^{U_{-1}\circ V_0(U)} \int_{V_{-1}\circ U_0(V)}^{V_0(U^\prime)} +  \int_{U_{-1}\circ V_0(U)}^{U}\int_{V_0(U)}^{V_0(U^\prime)}\right)|\ol\Phi|^2(V^\prime)^{a-2S(\Phi)}\Omega^2W_0(U^\prime,V^\prime) dU^\prime dV^\prime\\
        \lesssim & \int_{U_{0}(V)}^U \frac{1}{\left(V_{-1}(U^\prime)\right)^{1-\k}} \lnm \ol\Phi\rnm^2_{\LSS^2\left(H_{U^\prime}^{V_0(U^\prime)}\right)}\\
        \lesssim & C(\epsilon_1, D)\epsilon^{2-\delta} V^a.
    \end{split}
\end{equation} 
Combining \eqref{R III initial est 1.1}, \eqref{R III initial est 1.2}, \eqref{R III initial est 1.3}, \eqref{R III initial est 1.4}, and Proposition \ref{Main Proposition R_II}, we have for equations system \eqref{R III initial est eq sys}
\begin{equation}
  \begin{split}
    &\int_{V_0(U)}^V-U_0^\prime(V^\prime)\int_{\SS}\left|\nabla^i\ol{\Phi_1}\right|^2(V^\prime)^{a+1-2S(\Phi_1)}\Omega^2d\V dV^\prime\\
    &\qquad +
      \int_{U_0(V)}^U -V_0^\prime(U^\prime)\int_{\SS}\left|\nabla^i\ol{\Phi_2}\right|^2\left(V_0(U^\prime)\right)^{a+1-2S(\Phi_2)}d\V dU^\prime\\
      \lesssim & \left(\int_{U_0(V)}^{U_{-1}\circ V_0(U)} \int_{V_{-1}\circ U_0(V)}^{V_0(U^\prime)} +  \int_{U_{-1}\circ V_0(U)}^{U}\int_{V_0(U)}^{V_0(U^\prime)}\right)|F_1|^2(V^\prime)^{a+\k-2S(F_1)}W_0(U^\prime,V^\prime) dU^\prime dV^\prime\\
      &+\left(\int_{U_0(V)}^{U_{-1}\circ V_0(U)} \int_{V_{-1}\circ U_0(V)}^{V_0(U^\prime)} +  \int_{U_{-1}\circ V_0(U)}^{U}\int_{V_0(U)}^{V_0(U^\prime)}\right)|F_2|^2(V^\prime)^{a+\k-2S(F_2)}W_0(U^\prime,V^\prime) dU^\prime dV^\prime\\
      \lesssim& C(\epsilon_1, D)\epsilon^{2-\delta} V^a\leq \epsilon^{2-3\delta/2} V^a .
  \end{split}
\end{equation}

\end{proof}

\subsection{Norms and lemmas}

We are now firmly positioned to systematically establish the core estimates within Region III. The functional definitions for the norms $\lnm \cdot\rnm_{\LS_W^2(S_{U,V})}$, $\lnm \cdot\rnm_{\LS_W^\infty(S_{U,V})}$, $\lnm \cdot\rnm_{\LS_W^2(H_U^V)}$, and $\lnm \cdot\rnm_{\LS_W^2(\Hb_V^U)}$ remain entirely identical to those presented in \eqref{R II norm sys}. However, we modify the definition of $\lnm \cdot\rnm_{\LS_W^2(\R_{U,V})}$:

\begin{equation*}
            \begin{split}
        \lnm \psi\rnm^2_{\LS_W^2(\R_{U,V})}=&\int_{U_0(V)}^U\int_{V_0(U^\prime)}^V (V^\prime)^{\k-2} \lnm \psi\rnm^2_{\LS_W^2(S_{U^\prime,V^\prime})}dU^\prime dV^\prime.
    \end{split}
\end{equation*}
Here, compared to the corresponding definition in \eqref{R II norm sys}, we have explicitly substituted the volume element factor $\Omega^2$ with $V^\k$ to better reflect the asymptotic behavior. Furthermore, we define an analogous integrated norm along outgoing cones: 
        \begin{equation*}
            \begin{split}
        \lnm \psi\rnm^2_{\LS_W^2(\Rb_{U,V})}=&\int_{U_0(V)}^U\int_{U_0(V)}^{U^\prime} V^{2\k-2}\lnm \Omega\psi\rnm^2_{\LS_W^2(S_{U^{\prime\prime},V})}dU^{\prime\prime} dU^\prime.
    \end{split}
\end{equation*}
In order to effectively estimate the combined sum $\lnm \Psi_1\rnm^2_{\LS_W^2(\R_{U,V})}+\lnm \Psi_2\rnm^2_{\LS_W^2(\Rb_{U,V})}$ specifically for a Bianchi pair $(\Psi_1,\Psi_2)$, we must define the corresponding specialized bulk energy norm:
        \begin{equation*}
            \begin{split}
        \lnm \psi\rnm^2_{\LS_W^2(\D_{U,V})}=& \int_{U_0(V)}^U\int_{V_0(U^\prime)}^V\int_{U_0(V^\prime)}^U (V^\prime)^{2\k-3} \lnm \Omega\psi\rnm^2_{\LS_W^2(S_{U^{\prime\prime},V^\prime})}dU^{\prime\prime} dV^\prime dU^\prime,
    \end{split}
\end{equation*}
and the initial energy along $z=\epsilon_1^{-1}$,
        \begin{equation*}
            \begin{split}
        \lnm \psi\rnm^2_{\LS_W^2(\Sigma_{U,V})}=&\int_{V_0(U)}^V\frac{1}{V^\prime} \lnm \psi\rnm^2_{\LS_W^2(S_{U_0(V^\prime),V^\prime})}dV^\prime,
    \end{split}
\end{equation*}
and especially $\lnm \psi\rnm^2_{\LS_W^2(\Sigma_{V})}=\lnm \psi\rnm^2_{\LS_W^2(\Sigma_{0,V})}$. Moreover, we define 
        \begin{equation*}
            \begin{split}
        \lnm \psi\rnm^2_{\LS_W^2(\Xi_{V})}=& \int_{U_0(V)}^0\int_{U_0(V)}^{U^\prime} \frac{1}{V_0(U^{\prime\prime})^{2-2\k}}\lnm \psi\rnm^2_{\LS_W^2(S_{U^{\prime\prime},V_0(U^{\prime\prime})})}dU^{\prime\prime} dU^\prime\\
        &+\int_{U_0(V)}^0\int_{V_0(U^\prime)}^V (V^\prime)^{\k-2}\lnm \psi\rnm^2_{\LS_W^2(S_{U_0(V^\prime),V^\prime})} dV^\prime dU^\prime.
    \end{split}
\end{equation*}
In this section, we use weight \begin{equation}
    W_{III}(U,V)=V^a\cdot\left(\frac{-U}{V^{1-\k}}\right)^\tau,
\end{equation} 
and we will denote the corresponding newly re-weighted norms uniformly by $\LT$. In parallel, we will continue to use $\LS_1$ to unambiguously denote those norms evaluated with the trivial constant weight $W\equiv 1$.
Furthermore, we establish the corresponding properly weighted inner products: 
\begin{equation}
    \begin{split}
        \langle \psi_1,\psi_2\rangle_{S_{U,V}}&=\int_{\SS}V^{2-S(\psi_1)-S(\psi_2)}W(U,V)\langle\psi_1,\psi_2\rangle d\V,\\
        \langle \psi_1,\psi_2\rangle_{H_U^V}&=\int_{V_0(U)}^V(V^\prime)^{-1}\langle \psi_1,\psi_2\rangle_{S_{U,V^\prime}}dV^\prime,\\
        \langle \psi_1,\psi_2\rangle_{\Hb_V^U}&=\int_{U_0(V)}^U V^{\k-1}\langle\Omega\psi_1,\Omega\psi_2\rangle_{S_{U^\prime,V}} dU^\prime,\\
        \langle \psi_1,\psi_2\rangle_{\R_{U,V}}&=\int_{V_0(U)}^V\int_{U_0(V^\prime)}^U (V^\prime)^{\k-2}\langle \psi_1,\psi_2\rangle_{S_{U^\prime,V^\prime}}dU^\prime  dV^\prime.
    \end{split}
\end{equation} 
Before embarking on the rigorous sequence of estimates, we first analyze some illustrative toy models to clearly demonstrate the core methodology continuously employed throughout this section.
\begin{lemma}
    Consider the below transport equations 
    \begin{equation}
        \Omega^{-1}\nabla_3\varphi_1=F,\quad \Omega \nabla_4\varphi_2+\frac{1}{2}\lambda \Omega\tr\chi\varphi_2+4\Lambda\Omega\omega\varphi_2=G.
    \end{equation}
    Suppose that $\lnm \ol{\Omega\tr\chi},\ol{\Omega\omega}\rnm_{\LS_1^\infty(S)}=o_{\epsilon}(1)$,
    then the following estimates hold
    \begin{equation}\label{R III transport estimate nabla 3}
        \begin{aligned}
            \lnm \varphi_1\rnm^2_{\LT^2(S_{U,V})}\leq &\lnm \varphi_1\rnm^2_{\LT^2(S_{U_0(V),V})}+ \lnm \varphi_1\rnm^2_{\LT^2(\Hb_V^U)}+ \lnm F\rnm^2_{\LT^2(\Hb_V^U)}\\
            &+ \lnm \varphi_1\lnm \Omega^{-2}\dv b\rnm^{1/2}_{\LS_1^\infty(S)}\rnm^2_{\LT^2(\Hb_V^U)},
        \end{aligned}
    \end{equation}
    and for any $\delta_1>0$, we deduce that
    \begin{equation}\label{R III transport estimate nabla 4}
        \begin{split}
            \lnm \varphi_2\rnm^2_{\LT^2(S_{U,V})} \leq & \lnm \varphi_2\rnm^2_{\LT^2(S_{U,V_0(U)})}+\frac{1}{\delta_1}\lnm G\rnm^2_{\LT^2(H_U^V)}\\
            &+\left(2-2S(\varphi_2)+a-\tau(1-\k)-2\lambda +2\Lambda \k+\delta_1+o_{\epsilon_1,\epsilon}(1)\right)\cdot \\
            &\qquad\qquad\lnm \varphi_2\rnm^2_{\LT^2(H_U^V)}.
        \end{split}
    \end{equation}
\end{lemma}
\begin{proof}
    In a manner strictly analogous to standard transport estimates, we integrate along the null cones to find 
    \begin{equation}
        \begin{split}
            &\lnm \varphi_1\rnm^2_{\LT^2(S_{U,V})}\\
            = &\lnm \varphi_1\rnm^2_{\LT^2(S_{U_0(V),V})}+\int_{U_0(V)}^U\int_{\SS}\left(\Omega^2\cdot\Omega^{-1}e_3-b^A\nabla_A\right)\left(V^{2-2S(\varphi_1)}\left|\varphi_1\right|^2W_{III}\right)\\
            \leq & \lnm \varphi_1\rnm^2_{\LT^2(S_{U_0(V),V})}+2\langle \varphi_1,F\rangle_{\Hb_V^U}+\int_{U_0(V)}^U\int_{\SS}\dv b V^{2-2S(\varphi_1)}\left|\varphi_1\right|^2W_{III}\\
            \leq & \lnm \varphi_1\rnm^2_{\LT^2(S_{U_0(V),V})}+ \lnm \varphi_1\rnm^2_{\LT^2(\Hb_V^U)}+ \lnm F\rnm^2_{\LT^2(\Hb_V^U)}+ \lnm \varphi_1\lnm \dv b\Omega^{-2}\rnm^{1/2}_{\LS_1^\infty(S)}\rnm^2_{\LT^2(\Hb_V^U)}.
        \end{split}
    \end{equation}
    For the signature calculation, we note that 
    \begin{equation}
        \begin{aligned}
            \int_{U_0(V)}^U\int_{\SS}\dv b V^{2-2S(\varphi_1)}\left|\varphi_1\right|^2W_{III}\leq &  \int_{U_0(V)}^U\sup_{S_{U^\prime,V}}\left|\Omega^{-2}\dv b \cdot V\right|\int_{\SS} V^{1-2S(\varphi_1)}\left|\Omega\varphi_1\right|^2W_{III}\\
            =& \lnm \varphi_1\lnm \dv b\Omega^{-2}\rnm^{1/2}_{\LS_1^\infty(S)}\rnm^2_{\LT^2(\Hb_V^U)}.
        \end{aligned}
    \end{equation}
    For $\varphi_2$, we use Lemma \ref{R_III_Chr_sol_lemma} and the assumption $\lnm \ol{\Omega\tr\chi},\ol{\Omega\omega}\rnm_{\LS_1^\infty(S)}=o_{\epsilon}(1)$ to obtain that 
    \begin{equation}
        \begin{split}
            &\lnm \varphi_2\rnm^2_{\LT^2(S_{U,V})} \\\leq & \lnm \varphi_2\rnm^2_{\LT^2(S_{U,V_0(U)})}+\frac{1}{\delta_1}\lnm G\rnm^2_{\LT^2(H_U^V)}\\
            &+\left(2-2S(\varphi_2)+V\partial_V\log W_{III}-\lambda V\cdot\Omega\tr\chi-8\Lambda V\cdot\Omega\omega+\delta_1\right)\lnm \varphi_2\rnm^2_{\LT^2(H_U^V)}\\
            \leq & \lnm \varphi_2\rnm^2_{\LT^2(S_{U,V_0(U)})}+\frac{1}{\delta_1}\lnm G\rnm^2_{\LT^2(H_U^V)}\\
            &+\left(2-2S(\varphi_2)+a-\tau(1-\k)-2\lambda +2\Lambda \k+\delta_1+o_{\epsilon_1,\epsilon}(1)\right)\lnm \varphi_2\rnm^2_{\LT^2(H_U^V)}.
        \end{split}
    \end{equation}
\end{proof}

This lemma will be invoked repeatedly for spherical estimates. In applications to $\nabla_4$-transport, the choice of $\tau$ and the limiting values of $V\Omega\tr\chi$ and $V\Omega\omega$ determine whether the bulk term has a favorable sign.
\begin{lemma}\label{R III energy est lemma H Hb}
    Consider below Bianchi pair equations 
    \begin{equation}
        \left\{
        \begin{aligned}
            &\Omega^{-1}\nabla_3\Phi_1+\D \Phi_2=F,\\
            &\Omega \nabla_4\Phi_2+\frac{1}{2}\lambda \Omega\tr\chi\Phi_2+4\Lambda\Omega\omega\Phi_2-{}^*\D \Phi_1=G.
        \end{aligned}
        \right.
    \end{equation}
    If $1+\k-2S(\Phi_2)+a-\tau(1-\k)-2\lambda +2\Lambda \k<-2\delta_1<0$, then the estimate holds
    \begin{equation}\label{R III energy estimate}
        \begin{aligned}
            &\lnm\Phi_1\rnm^2_{\LT^2(H_U^V)}+\lnm \Phi_2\rnm^2_{\LT^2(\Hb_V^U)} -\lnm\Phi_1\rnm^2_{\LT^2(\Sigma_{U,V})} - \left|V^\k\partial_VU_0\right|\cdot\lnm \Omega\Phi_2\rnm^2_{\LT^2(\Sigma_{U,V})}\\
            \leq & \langle \Omega^2\Phi_1,F\rangle_{\R_{U,V}}+\lnm \Phi_1 \lnm\dv b\rnm_{\LS_1^\infty}^{1/2}\rnm^2_{\LT^2(\R_{U,V})}+\frac{1}{\delta_1}\lnm \Omega G\rnm^2_{\LT^2(\R_{U,V})}\\
            &+\langle\left|\Omega\Phi_1(\eta+\etab)\right| ,\left|\Omega\Phi_2\right|\rangle_{\R_{U,V}}-\delta_1\lnm \Omega\Phi_2\rnm^2_{\LT^2(\R_{U,V})}.
        \end{aligned}
    \end{equation}
\end{lemma}
\begin{proof}
    Following the systematic approach of standard energy estimates, we deduce
    \begin{equation}
        \begin{aligned}
            &\lnm\Phi_1\rnm^2_{\LT^2(H_U^V)}\\ 
            \leq & \lnm\Phi_1\rnm^2_{\LT^2(\Sigma_{U,V})}+2\int_{V_0(U)}^V\frac{1}{V^\prime}\langle\Phi_1,F-\D \Phi_2\rangle_{\Hb_{V^\prime}^{U}}+\lnm \Phi_1\lnm\dv b\rnm_{\LS_1^\infty}^{1/2}\rnm^2_{\LT^2(\R_{U,V})}\\
            \leq & \lnm\Phi_1\rnm^2_{\LT^2(H_{U_0(V)}^V)}+\langle \Omega^2\Phi_1,F\rangle_{\R_{U,V}}+\lnm \Phi_1 \left|\dv b\right|^{1/2}\rnm^2_{\LT^2(\R_{U,V})}\\
            &-2\int_{U_0(V)}^U\int_{V_0(U^\prime)}^V V^{\k-2}\langle\Omega^2 \Phi_1, \D\Phi_2\rangle_{S_{U^\prime,V^\prime}} .
        \end{aligned}
    \end{equation}
    With coordinate change $U^\prime=U_0(V^\prime)$, this yields 
    \begin{equation}
        \begin{aligned}
            &\int_{U_0(V)}^U \frac{1}{(V_0(U^\prime))^{1-\k}}\lnm \Omega \Phi_2\rnm^2_{\LT(U^\prime,V_0(U^\prime))}d U^\prime\\
            &\qquad\qquad=(-V^\k\partial_V U_0)\int_{V_0(U)}^U\frac{1}{V^\prime}\lnm \Omega \Phi_2\rnm^2_{\LT(U_0(V^\prime),V^\prime)}d V^\prime,
        \end{aligned}
    \end{equation}
    and therefore 
    \begin{equation}
        \begin{aligned}
            &\lnm \Phi_2\rnm^2_{\LT^2(\Hb_V^U)}\\
            \leq & \left|V^\k\partial_V U_0\right|\cdot\lnm \Omega\Phi_2 \rnm^2_{\LT^2(\Sigma_{U,V})}\\
            &+\frac{1}{\delta_1}\lnm \Omega G\rnm^2_{\LT^2(\R_{U,V})} +2\int_{U_0(V)}^U\int_{V_0(U^\prime)}^V {(V^\prime)}^{\k-2}\langle \Omega^2\Phi_2,{}^*\D\Phi_1\rangle_{S_{U^\prime,V^\prime}}\\
            &+\left(1-2S(\Phi_2)+V\partial_V\log\Omega^2+V\partial_V\log W-\lambda V\Omega\tr\chi-8\Lambda V\Omega\omega+\delta_1\right)\lnm \Omega\Phi_2\rnm^2_{\LT^2(\R_{U,V})}.
        \end{aligned}
    \end{equation}
    Adding them together and applying Cauchy inequality, the following inequality holds 
    \begin{equation}
        \begin{aligned}
            &\lnm\Phi_1\rnm^2_{\LT^2(H_U^V)}+\lnm \Phi_2\rnm^2_{\LT^2(\Hb_V^U)} -\lnm\Phi_1\rnm^2_{\LT^2(\Sigma_{U,V})} - \left|V^\k\partial_VU_0\right|\cdot\lnm \Omega\Phi_2\rnm^2_{\LT^2(\Sigma_{U,V})}\\
            \leq & \langle \Omega^2\Phi_1,F\rangle_{\R_{U,V}}+\lnm \Phi_1 \lnm\dv b\rnm_{\LS_1^\infty}^{1/2}\rnm^2_{\LT^2(\R_{U,V})}+\frac{1}{\delta_1}\lnm \Omega G\rnm^2_{\LT^2(\R_{U,V})}\\
            &-2\langle\Phi_1\D\Omega ,\Omega\Phi_2\rangle_{\R_{U,V}}-\delta_1\lnm \Omega\Phi_2\rnm^2_{\LT^2(\R_{U,V})}.
        \end{aligned}
    \end{equation}
    Using $\nabla_A\Omega=\frac{1}{2}\Omega\left(\eta_A+\etab_A\right)$, we obtain the result.
\end{proof}

The following fundamental lemma provides the essential control over the initial energy as measured by the specifically tailored $\LT^2(\Xi_{V})$ norm.
\begin{lemma}
    Let $a$ be a positive number with $a<1-\k$. Suppose that 
    \begin{equation}
        \int_0^V\frac{1}{V^\prime}\int_{\SS}(V^\prime)^{2-2S(\psi)+a/2}\left|\psi\right|^2(U_0(V^\prime),V^\prime)\leq C_1 V^{a/2},
    \end{equation}
    then we deduce that 
    \begin{equation}
        \lnm \psi\rnm^2_{\LT^2(\Xi_{V})}\leq C(\epsilon_1)\cdot C_1 V^a.
    \end{equation}
\end{lemma}
\begin{proof}
    First we note that $(-V^\k\partial_V U_0)=(1-\k)\epsilon_1^{\k-1}$ and compute 
    \begin{equation}
        \begin{aligned}
            &\int_{U_0(V)}^0\frac{1}{V_0(U^\prime)^{1-\k-a/2}} dU^\prime
            = \int_0^V (V^\prime)^{a/2-1}(-(V^\prime)^\k\partial_V U_0(V^\prime)) dV^\prime \leq C(\epsilon_1,a)V^{a/2}.
        \end{aligned}
    \end{equation}
    For the second part in the definition of $\LT^2(\Xi_V)$, we relax $(V^\prime)^{a+\k-1}$ into\\ $(V^\prime)^{a/2}\cdot\frac{1}{V_0(U^\prime)^{1-\k-a/2}}$ and obtain
    \begin{equation}
        \begin{split}
            & \int_{U_0(V)}^0\int_{V_0(U^\prime)}^V (V^\prime)^{\k-2}\lnm \psi\rnm^2_{\LT^2(S_{U_0(V^\prime),V^\prime})} dV^\prime dU^\prime\\
            \leq &C(\epsilon_1)\left(\int_{U_0(V)}^0\frac{1}{V_0(U^\prime)^{1-\k-a/2}} dU^\prime\right)\cdot\left(\int_{0}^V (V^\prime)^{-1}\int_{\SS}(V^\prime)^{2-2S(\psi)+a/2}|\psi|^2 dV^\prime\right)\\
            \leq &    C(\epsilon_1,a)C_1 V^{a}.
        \end{split}
    \end{equation}
    The first part of $\lnm \psi\rnm_{\LT^2(\Sigma^V)}$ satisfies
    \begin{equation}
        \begin{split}
            &\int_{U_0(V)}^0\int_{U_0(V)}^{U^\prime} \frac{1}{V_0(U^{\prime\prime})^{2-2\k}}\lnm \psi\rnm^2_{\LT^2(S_{U^{\prime\prime},V_0(U^{\prime\prime})})}dU^{\prime\prime} dU^\prime\\
            = &  C(\epsilon_1)\int_{U_0(V)}^0\int_{V_0(U^\prime)}^{V}(V^\prime)^{\k-2}\lnm \psi\rnm^2_{\LT^2(S_{U_0(V^\prime),V^\prime})} dV^\prime dU^\prime,
        \end{split}
    \end{equation}
    which is transformed into the second part.
\end{proof}
The following crucial lemma explicitly reveals the analytical mechanism by which we can strictly improve the bootstrap constants governing the bulk energy distributions.
\begin{lemma}\label{R III Key Lemma}
    Suppose that for $\tau<1$, it holds 
    \begin{equation*}
        \lnm \Omega\psi\cdot\left(\frac{-U}{V^{1-\k}}\right)^{\tau/2}\rnm^2_{\LT^2(\R_{U,V})}\leq C_1 V^a,
    \end{equation*}
     then
    we for $\lnm \psi\rnm^2_{\LT^2(\D_{U,V})}$ have the estimate:
    \begin{equation}
        \lnm \psi\rnm^2_{\LT^2(\D_{U,V})}\leq \epsilon_1^{(1-\k)(1-\tau)}C_1V^a.
    \end{equation}
\end{lemma}
\begin{proof}
    Recall that the signature of $\psi\cdot\left(\frac{-U}{V^{1-\k}}\right)^\lambda$ for any $\lambda\in\mathbb{R}$ is equal to $S(\psi)$. 
    We compute directly and obtain
    \begin{equation*}
        \begin{aligned}
            &\lnm \psi\rnm^2_{\LT^2(\D_{U,V})}\\
            \leq & \int_{V_0(U)}^V\int_{U_0(V^\prime)}^U\frac{(-U^\prime)^{-\tau}}{(V^\prime)^{(1-\k)(1-\tau)}}dU^\prime \left(\int_{U_0(V^\prime)}^{U^\prime} (V^\prime)^{\k-2}\lnm \Omega\psi\left(\frac{-U}{V^{1-\k}}\right)^{\tau/2}\rnm^2_{\LT^2(S_{U^{\prime\prime},V^\prime})} d U^{\prime\prime}\right)dV^\prime\\
            \leq & \sup_{V^\prime}\left(\frac{-U_0(V^\prime)}{(V^\prime)^{1-\k}}\right)^{1-\tau}\lnm\Omega\psi\cdot\left(\frac{-U}{V^{1-\k}}\right)^{\tau/2}\rnm^2_{\LT^2(\R_{U,V})}\leq \epsilon_1^{(1-\k)(1-\tau)}C_1V^a.
        \end{aligned}
    \end{equation*}
\end{proof} 
The small factor in this lemma is the reason for the separate $\mathcal{D}$-norm. It converts an already known $\mathcal{R}$-type estimate into a genuinely smaller bulk estimate, which is then inserted back into the energy inequalities for top-order pairs. Besides, the bound of $\LT^2(\Hb_V^U)$ norm also yiels the result. See the corollary below.
\begin{corollary}\label{R III Key Lemma Cor}
    If $\tau<1$ and 
    \begin{equation*}
        \lnm \psi\cdot \left(\frac{-U}{V^{1-\k}}\right)^{\tau/2}\rnm^2_{\LT^2(\Hb_V^U)}\leq C_1 V^a,
    \end{equation*}
    we also have estimate for $\lnm \psi\rnm^2_{\LT^2(\D_{U,V})}$:
    \begin{equation}
        \lnm \psi\rnm^2_{\LT^2(\D_{U,V})}\leq \epsilon_1^{(1-\k)(1-\tau)}a^{-1}C_1V^a.
    \end{equation}
\end{corollary}
\begin{proof}
    Just notice that \begin{equation}
        \lnm \Omega\psi\rnm^2_{\LT^2(\R_{U,V})}=\int_{V_0(U)}^V(V^\prime)^{-1}\lnm \psi\rnm^2_{\LT^2(\Hb_{V^\prime}^U)}\leq C_1 a^{-1}V^a.
    \end{equation}
\end{proof}
The next lemma explains how the spacetime norms \(\R_{U,V}\) and \(\Rb_{U,V}\) are recovered from the boundary norm \(\Xi_V\) and from the differentiated source terms.

\begin{lemma}
    Suppose that $\lnm \Omega^{-2}\nabla b \rnm_{\LS_1^\infty(S_{U,V})}<\left(\frac{-U}{V^{1-\k}}\right)^{-\tau/2}$. Then we have estimates
    \begin{equation}\label{R III energy est R norm}
        \begin{split}
            \lnm \psi\rnm^2_{\LT^2(\R_{U,V})}\leq & C(\epsilon_1)\lnm \psi\rnm^2_{\LT^2(\Xi_V)}+ \lnm \left(\frac{V^{1-\k}}{-U}\right)^{\tau/2}\psi\rnm^2_{\LT^2(\D_{U,V})}\\
            &+\lnm \psi\rnm^2_{\LT^2(\D_{U,V})}+\lnm  \Omega^{-1} \nabla_3\psi\rnm^2_{\LT^2(\D_{U,V})},
        \end{split}
    \end{equation}
    \begin{equation}\label{R III energy est Rb norm}
        \begin{split}
            \lnm \psi\rnm^2_{\LT^2(\Rb_{U,V})}\leq & C(\epsilon_1)\lnm \psi\rnm^2_{\LT^2(\Xi_V)}+\lnm \psi\rnm^2_{\LT^2(\D_{U,V})}+\lnm \nabla_4(\Omega\psi)\rnm^2_{\LT^2(\D_{U,V})}.
        \end{split}
    \end{equation}
\end{lemma}
\begin{proof} We compute the norms from definition and obtain
    \begin{equation}
        \begin{split}
            \lnm \psi\rnm^2_{\LT^2(\R_{U,V})}=& \int_{U_0(V)}^U\int_{V_0(U^\prime)}^V\int_{U_0(V^\prime)}^{U^\prime}\int_{\SS}(V^\prime)^{\k-2}\left( \Omega e_3 -b\right)\left((V^\prime)^{2-2S(\psi)}W(V^\prime)\left|\psi\right|^2\right)\\
            &+\int_{U_0(V)}^U\int_{V_0(U^\prime)}^V\int_{\SS}(V^\prime)^{\k-2}\lnm \psi\rnm^2_{\LT^2(S_{U_0(V^\prime),V^\prime})}\\
            \leq & C(\epsilon_1)\lnm \psi\rnm^2_{\LT^2(\Xi_V)}+\lnm \psi\rnm^2_{\LT^2(\D_{U,V})}+\lnm \Omega^{-1}\nabla_3\psi\rnm^2_{\LT^2(\D_{U,V})}\\
            &+\int_{U_0(V)}^U\int_{V_0(U^\prime)}^V\int_{U_0(V^\prime)}^{U^\prime}\int_{\SS}(V^\prime)^{\k-2}\Omega^{-2}\dv b(V^\prime)^{2-2S(\psi)}W(V^\prime)\left|\Omega\psi\right|^2\\
            \leq & C(\epsilon_1)\lnm \psi\rnm^2_{\LT^2(\Xi_V)}+\lnm \left(\frac{V^{1-\k}}{-U}\right)^{\tau/2}\psi\rnm^2_{\LT^2(\D_{U,V})}+\lnm \psi\rnm^2_{\LT^2(\D_{U,V})}\\
            &+\lnm \Omega^{-1}\nabla_3\psi\rnm^2_{\LT^2(\D_{U,V})},
        \end{split}
    \end{equation}
    \begin{equation}
       \begin{split}
         \lnm \psi\rnm^2_{\LT^2(\Rb_{U,V})}=&\int_{U_0(V)}^U\int_{U_0(V)}^{U^\prime} V^{2\k-2} \lnm \Omega\psi\rnm^2_{\LT^2(S_{U^{\prime\prime},V})}dU^{\prime\prime} dU^\prime\\
         =&\int_{U_0(V)}^U\int_{U_0(V)}^{U^\prime} V_0(U^{\prime\prime})^{2\k-2} \lnm \Omega\psi\rnm^2_{\LT^2(S_{U^{\prime\prime},V_0(U^{\prime\prime})})}dU^{\prime\prime} dU^\prime\\
         &+\lnm \Omega\psi\rnm^2_{\LT^2(\D_{U,V})}+\lnm \Omega\nabla_4\Omega\psi\rnm^2_{\LT^2(\D_{U,V})}\\
         \leq & C(\epsilon_1)\lnm  \psi\rnm^2_{\LT^2(\Xi_V)}+\lnm  \psi\rnm^2_{\LT^2(\D_{U,V})}+\lnm  \nabla_4(\Omega\psi)\rnm^2_{\LT^2(\D_{U,V})}.
       \end{split}
    \end{equation}
\end{proof}

We can now rigorously synthesize the complete set of bootstrap assumptions that serve as the foundation of our local analysis:\\
Let $\psi_1=\Omega^{-1}\tr\chib,\Omega\chi,\etab,\Omega\omega$, $\psi_2=\Omega^{-1}\chib,\Omega\chi,\etab,\eta$, $\Psi_1=\Omega^{2}\alpha,\Omega\beta^r,K,\sigma^r$,\\ $\Psi_2=\Omega\beta^r,K,\sigma^r,\Omega^{-1}\betab^r$.
\begin{equation}\label{R_III_BOOTSTRAP_ASSUMPTIONS_1}
    \begin{split}
        &\sum_{i\leq 6,j\leq 5}\lnm \nabla^{i}\left(\ol\psi_1,\ol{\Omega e_4\phi},\nabla\phi\right)\rnm^2_{\LT^2(H_U^V)}+\lnm \nabla^{i}\left(\ol\psi_2,\ol{\Omega^{-1} e_3\phi},\nabla\phi\right)\rnm^2_{\LT^2(\Hb_V^U)}\\
        &+\lnm \nabla^{i}\left(\ol\psi_1,\ol{\Omega e_4\phi},\nabla\phi\right)\rnm^2_{\LT^2(\R_{U,V})} +\lnm \nabla^{j}\ol\Psi_1\rnm^2_{\LT^2(H_U^V)}+\lnm \nabla^{j}\ol\Psi_1\rnm^2_{\LT^2(\R_{U,V})}+\lnm \nabla^{j}\ol\Psi_2\rnm^2_{\LT^2(\Hb_V^U)}\\
        &\leq  B\epsilon^{2-2\delta} V^a.
    \end{split}
\end{equation}
Moreover, we assume for those quantities with $\nabla_3$ equations except $\chibh,\betab^r$, it holds that 
\begin{equation}\label{R_III_BOOTSTRAP_ASSUMPTIONS_2}
    \sum_{i\leq 4}\lnm \nabla^i(\ol{\Omega\chi},\ol{\Omega^{-1}\tr\chib},\etab,\ol{\Omega\omega},\ol{\Omega D_4\phi},\nabla\phi)\rnm^2_{\LT^2(S_{U,V})}+\sum_{j\leq 3}\lnm \nabla^j\ol{\Psi_1}\rnm^2_{\LT^2(S_{U,V})}\leq B^\prime\epsilon^{2-2\delta} W_{III}(U,V).
\end{equation}
Finally, we assume 
\begin{equation}\label{R_III_BOOTSTRAP_ASSUMPTIONS_3}
    \sum_{i\leq 3}\lnm \nabla^i(\ol{\Omega^{-1}e_3\phi},\Omega^{-1}\chibh)\rnm^2_{\LT^2(S_{U,V})} \leq B^\prime\epsilon^{2-2\delta} W_{III}(U,V).
\end{equation}
In this final critical phase, our overarching objective will be to systematically demonstrate the following crucial proposition:
\begin{proposition}
    Let $W_{III}(U,V)=V^a\cdot\left(\frac{-U}{V^{1-\k}}\right)^\tau$ with $0<a\ll \tau(1-\k)\ll 1$ fixed. Based on Proposition \ref{Main Proposition R_II} and assuming \eqref{R_III_BOOTSTRAP_ASSUMPTIONS_1}, \eqref{R_III_BOOTSTRAP_ASSUMPTIONS_2}, and \eqref{R_III_BOOTSTRAP_ASSUMPTIONS_3} hold, we can choose $\epsilon_1=\epsilon_1(a,\tau)$ and $\epsilon=\epsilon(\epsilon_1,a)$ sufficiently small so that the bootstrap constants can be improved to $B+B^\prime\leq C_0$.
\end{proposition}
\begin{remark}
    We use additional notation for those with enhanced estimates in \eqref{R_III_BOOTSTRAP_ASSUMPTIONS_2}: $$\varphi_1\in\{\Omega\chi,\Omega^{-1}\tr\chib,\etab,\Omega\omega,\Omega e_4\phi,\nabla\phi\},$$
    and for all connection components and first order derivatives of scalar field: $$\varphi\in\{\psi_1,\psi_2, \nabla\phi,\Omega^{-1} e_3\phi,\Omega e_4\phi\}.$$
    In the whole section, we will avoid $\Omega^{-1}\omegab$ and no notation refers to it. The three bootstrap assumptions are hierarchical. The first controls flux and ordinary spacetime norms, the second gives stronger spherical control for quantities with favorable $\nabla_3$ equations, and the third records the weaker but still sufficient control for the incoming components $\Omega^{-1}\chibh$ and $\Omega^{-1}e_3\phi$.
\end{remark}

\subsection{Estimates of spherical norms}
To securely utilize the previously established preparatory lemmas, we formulate the additional metric bootstrap assumptions:
\begin{equation}
    \lnm \Omega^{-2}\nabla b, \ol{\Omega^2}\rnm_{\LS_1^\infty(S_{U,V})}\leq B^{\prime\prime}\epsilon^{1-\delta}\left(\frac{V^{1-\k}}{-U}\right)^{\tau/8},
\end{equation}
and we will vigorously derive a strictly stronger optimal estimate mathematically bounding both $b$ and $\Omega^2$ later in the analysis, which will retroactively justify these initial hypotheses.
We commence the rigorous analysis with a structural product lemma deeply analogous to \eqref{Calcu_Diff_product_ol}.
\begin{lemma}\label{R_III_Lemma_Estimate_without_eta}
    With the bootstrap assumptions \eqref{R_III_BOOTSTRAP_ASSUMPTIONS_1}, \eqref{R_III_BOOTSTRAP_ASSUMPTIONS_2}, \eqref{R_III_BOOTSTRAP_ASSUMPTIONS_3} and notations above, it follows that 
    \begin{equation}
        \sum_{i\leq 5}\lnm \nabla^i\left(\ol{\varphi^{(1)}_{\neq \eta}\varphi^{(2)}_{\neq \eta}}\right)\rnm^2_{\LT^2(S_{U,V})}\lesssim \sum_{i\leq 5}\lnm \nabla^i\ol{\varphi^{(1)}_{\neq \eta}}\rnm^2_{\LT^2(S_{U,V})}+\sum_{i\leq 5}\lnm \nabla^i\ol{\varphi^{(2)}_{\neq \eta}} \rnm^2_{\LT^2(S_{U,V})}.
    \end{equation}
    Here $\varphi^{(i)}_{\neq \eta}$ represents $D\phi$ and the connection components except $\eta$. Note that all arguments in this section will not involve $\Omega^{-1}\omegab$.
\end{lemma}
\begin{proof}
    Since $\ol{\varphi^{(1)}\varphi^{(2)}}=\varphi^{(1)}\ol{\varphi^{(2)}}+\left(\varphi^{(1)}\right)^c \ol{\varphi^{(2)}}$, we have for the first term  
\begin{equation}
    \begin{aligned}
        \begin{aligned}
            \sum_{i\leq 5}\lnm \nabla^i\left(\varphi^{(1)}\ol{\varphi^{(2)}}\right)\rnm^2_{\LT^2(S_{U,V})}\lesssim & \sum_{i_1\leq 3,i_2\leq 5}\lnm \nabla^{i_1}{\varphi^{(1)}}\rnm^2_{\LS_1^2(S_{U,V})}\lnm \nabla^{i_2}{\ol{\varphi^{(2)}}}\rnm^2_{\LT^2(S_{U,V})}\\
            &+\sum_{i_1\leq 3,i_2\leq 5}\lnm \nabla^{i_1}{\varphi^{(2)}}\rnm^2_{\LS_1^2(S_{U,V})}\lnm \nabla^{i_2}{\ol{\varphi^{(1)}}}\rnm^2_{\LT^2(S_{U,V})}.
        \end{aligned}
    \end{aligned}
\end{equation}
For $\varphi^{(i)}\neq \eta$, we have by \eqref{R_III_BOOTSTRAP_ASSUMPTIONS_2} and \eqref{R_III_BOOTSTRAP_ASSUMPTIONS_3} that  
\begin{equation}
    \sum_{j\leq 3}\lnm \nabla^{j}{\varphi^{(i)}}\rnm^2_{\LS_1^2(S_{U,V})}\lesssim 1.
\end{equation}
Then we obtain the estimate for $\varphi^{(1)}\ol{\varphi^{(2)}}$, and the other term is similarly estimated.
\end{proof}

We are ready to estimate the non-top-order terms.
\begin{proposition}\label{R III L2S est non top order}
    The following estimates hold for the non-top-order Ricci coefficients and scalar field derivatives:
    \begin{equation}
        \sum_{i\leq 5}\lnm \nabla^i\ol\psi_{\neq\eta},\nabla^i\ol{D\phi}\rnm^2_{\LT^2(S_{U,V})}+\sum_{i\leq 4}\lnm \nabla^i\ol\Psi,\nabla^i\eta\rnm^2_{\LT^2(S_{U,V})}\lesssim B\epsilon^{2-2\delta} V^a,
    \end{equation}
    and
    \begin{equation}
        \lnm \nabla^5\eta\rnm^2_{\LT^2(S_{U,V})}\lesssim B\epsilon^{2-2\delta} V^a+\lnm \eta\nabla^5(\Omega\chi)\rnm^2_{\LT^2(S_{U,V})}\lesssim B\epsilon^{2-2\delta} V^a\left(1+\epsilon\left(\frac{V^{1-\k}}{-U}\right)^\tau\right) .
    \end{equation}
\end{proposition}
\begin{proof}
This proof uses the Region III transport estimates \eqref{R III transport estimate nabla 3} and \eqref{R III transport estimate nabla 4}, the initial-value estimates \eqref{R III initial value 1}--\eqref{R III initial value 3}, and the product Lemma \ref{R_III_Lemma_Estimate_without_eta}. The only component not covered directly by that product lemma is $\eta$, which is treated separately below.
    We write $\varphi_1\in\{\Omega\chi,\Omega^{-1}\tr\chib,\etab,\Omega\omega,\Omega D_4\phi,\nabla\phi\}$. For $\varphi_{2,\neq\eta}=\Omega^{-1}\chib,\Omega D_3\phi$, the schematic equation reads 
    \begin{equation}
        \Omega\nabla_4\varphi_{2,\neq\eta}\sim \varphi_1\varphi_{2,\neq\eta}+\varphi_1^2+\nabla\etab+\Psi_1.
    \end{equation}
    After making difference, we deduce that 
    \begin{equation}
        \Omega\nabla_4\ol{\varphi_{2,\neq\eta}}\sim {\varphi_1}\ol{\varphi_{2,\neq\eta}}+\ol{\varphi_1}\ol{\varphi_{2,\neq\eta}}+\ol{\varphi_1^2}+\nabla\etab+\ol{\Psi_1}.
    \end{equation}
    By bootstrap assumption \eqref{R_III_BOOTSTRAP_ASSUMPTIONS_2} and \eqref{R_III_BOOTSTRAP_ASSUMPTIONS_3}, we have $\lnm \ol{\varphi_1},\ol{\varphi_{2,\neq\eta}}\rnm^2_{\LS_1^\infty(S)}\leq B^\prime\epsilon^{2-2\delta}\ll 1$ . Transport estimate \eqref{R III transport estimate nabla 4} implies that 
     \begin{equation}
        \begin{aligned}
            \lnm \ol\varphi_2\rnm^2_{\LT^2(S_{U,V})}\lesssim & \lnm \ol\varphi_2\rnm^2_{\LT^2(S_{U,V_0(U)})}+\lnm \ol\varphi_2,\nabla\etab,\ol{\Psi}_1\rnm^2_{\LT^2(H_U^V)}\lesssim B \epsilon^{2-2\delta} V^a.
        \end{aligned}
     \end{equation}
    
    For $I\geq 1$,
    we use commutation formula to obtain that  
    \begin{equation}
        \begin{split}
            \Omega\nabla_4\nabla^I\varphi_{2,\neq\eta}\sim &\nabla^I\left(\varphi_1\varphi_{2,\neq\eta}+\nabla\etab+\Psi_1\right)+\nabla^I(\Omega\chi\varphi_{2,\neq\eta}).
        \end{split}
    \end{equation}
    By transport estimate \eqref{R III transport estimate nabla 4} and Lemma \ref{R_III_Lemma_Estimate_without_eta}, we deduce that 
    \begin{equation}
        \begin{split}
            \sum_{1\leq i\leq 5}\lnm \nabla^i\varphi_{2,\neq\eta}\rnm^2_{\LT^2(S_{U,V})}\leq & B\epsilon^{2-2\delta}V^a\\
            &+\sum_{i\leq 5}\lnm\nabla^{i}\left(\ol{\varphi_1},\ol{\varphi_{2,\neq\eta}},\ol{\Psi_1}\right)\rnm^2_{\LT^2(H_U^V)}+\lnm\nabla^{6}\etab\rnm^2_{\LT^2(H_U^V)},
        \end{split}
    \end{equation}
    which is bounded by $B\epsilon^{2-2\delta}V^a$ by assumption \eqref{R_III_BOOTSTRAP_ASSUMPTIONS_1}. 
For $\eta$, we recall equation 
\begin{equation}
    \Omega\nabla_4\eta+\Omega\chi\cdot\eta\sim \ol{\varphi_1\varphi_1}-\Omega\beta^r.
\end{equation}
Then it follows that 
\begin{equation}
    \begin{aligned}
        \partial_V \lnm \eta\rnm^2_{\LT^2(S_{U,V})}\leq& \int_{\SS} \left((2+a-(1-\k)\tau)/V-\Omega\tr\chi+2|\Omega\chih|\right)V^{2-2S(\eta)}W_{III}|\eta|^2 \\
        &\qquad\qquad+\frac{2}{V}\langle \eta,\ol{\varphi_1\varphi_1}-\Omega\beta^r\rangle_{S_{U,V}} \\
        \leq & \int_{\SS} \frac{1}{V}\langle \eta,\ol{\varphi_1\varphi_1}-\Omega\beta^r\rangle_{S_{U,V}}.
    \end{aligned}
\end{equation}
Here we use that $(2+a-(1-\k)\tau)/V-\Omega\tr\chi+2|\Omega\chih|=(a-(1-\k)\tau+o_{\epsilon,\epsilon_1}(1))/V<0$. 
Hence the following relation holds 
\begin{equation}
    \lnm \eta\rnm^2_{\LT^2(S_{U,V})}\lesssim \lnm \eta\rnm^2_{\LT^2(S_{U,V_0(U)})}+\lnm \ol{\varphi_1},\Omega\beta^r\rnm^2_{\LT^2(H_U^V)}\lesssim B\epsilon^{2-2\delta}V^a.
\end{equation}
For $1\leq i\leq 5$, the commutation formula \eqref{Comm_Formula_D4,nab} gives
\begin{equation}
    \Omega\nabla_4\nabla^i\eta+(1+i)\Omega\chi\nabla^i\eta \sim \nabla^i\left(\varphi_1\varphi_1-\Omega\beta^r\right)+\sum_{i_1+i_2=i-1}\nabla^{i_1+1}(\Omega\chi)\nabla^{i_2}\eta.
\end{equation}
When $i\leq 4$, we can similarly obtain
\begin{equation}
    \begin{aligned}
        \lnm \nabla^i\eta\rnm^2_{\LT^2(S_{U,V})}\lesssim & \epsilon^{2-2\delta}V^a+ \sum_{j\leq i}\lnm \nabla^j(\ol{\varphi_1},\Omega\beta^r)\rnm^2_{\LT^2(H_U^V)}\\
        &+\sum_{i_1\leq 4,i_2\leq i-1} \sup\left(\lnm \nabla^{i_1}\Omega\chi\rnm^2_{\LS_1^2(S)}\right) \lnm \nabla^{i_2}\eta\rnm^2_{\LT^2(H_U^V)}\\
        \lesssim & B \epsilon^{2-2\delta} V^a,
    \end{aligned}
\end{equation}
and for $i=5$, it holds
\begin{equation}
    \begin{aligned}
        \lnm \nabla^5\eta\rnm^2_{\LT^2(S_{U,V})}\lesssim & \epsilon^{2-2\delta}V^a+ \sum_{j\leq 5}\lnm \nabla^j(\ol{\varphi_1},\Omega\beta^r)\rnm^2_{\LT^2(H_U^V)}+\lnm \eta\nabla^5(\Omega\chi)\rnm^2_{\LT^2(H_U^V)} \\
        &+\sum_{i_1\leq 4,i_2\leq 4} \sup\left(\lnm \nabla^{i_1}\Omega\chi\rnm^2_{\LS_1^2(S)}\right) \lnm \nabla^{i_2}\eta\rnm^2_{\LT^2(H_U^V)}\\
        \lesssim & B \epsilon^{2-2\delta} V^a+ \lnm \eta\nabla^5(\Omega\chi)\rnm^2_{\LT^2(H_U^V)}.
    \end{aligned}
\end{equation}
Since $\lnm \eta\rnm^2_{\LS_1^\infty(S_{u,v})}\lesssim B\epsilon^{2-2\delta}\left(\frac{V^{1-\k}}{-U}\right)^\tau$, we deduce that 
\begin{equation}
    \begin{aligned}
        \lnm \nabla^5\eta\rnm^2_{\LT^2(S_{U,V})}
        \lesssim & B \epsilon^{2-2\delta} V^a\left(1+\epsilon\left(\frac{V^{1-\k}}{-U}\right)^\tau\right),
    \end{aligned}
\end{equation}
As for $\varphi_1$, the following relation holds 
\begin{equation}
    \Omega^{-1}\nabla_3\varphi_1=\varphi_1\varphi_2+\varphi_2^2+\nabla\eta+\Psi_2,
\end{equation}
and using commutation formula \eqref{Comm_Formula_D3,nab}, we obtain
\begin{equation}
    \begin{split}
        \Omega^{-1}\nabla_3\nabla^I\varphi_1=&\sum_{i_1+i_2+i_3=I}\nabla^{i_1}(\eta+\etab)^{i_2}\nabla^{i_3} \left(\varphi_1\varphi_2+\varphi_2^2+\nabla\eta+\Psi_2\right)\\
        &+\sum_{i_1+i_2+i_3+i_4=I}\nabla^{i_1}(\eta+\etab)^{i_2}\nabla^{i_3}(\Omega^{-1}\chib)\nabla^{i_4}\varphi_1 \\
        =&\sum_{i_1+i_2=I}\nabla^{i_1}\varphi^{2+i_2}+\nabla^{I+1}\eta+\nabla^I\Psi_2+\nabla^I\Upsilon_2.
    \end{split}
\end{equation} 
For the $\varphi_1$ family we now switch to the incoming transport estimate \eqref{R III transport estimate nabla 3}. The quantities $\nabla\eta$, $\Psi_2$, and $\Upsilon_2$ in the source have already been controlled in the preceding paragraphs or by the bootstrap, and the remaining integral gains the small $\epsilon_1$-factor from the geometry of Region III. 
By transport estimate \eqref{R III transport estimate nabla 3}, we have the following estimate
\begin{equation}
    \begin{split}
            \lnm \nabla^I\varphi_1\rnm^2_{\LT^2(S_{U,V})}\leq & \lnm \nabla^{I}(\nabla\eta,\Psi_2,\Upsilon_2)\rnm^2_{\LT^2(\Hb_V^U)}+\sum_{i_1+i_2=I}\lnm \nabla^{i_1}\varphi^{2+i_2}\rnm^2_{\LT^2(\Hb_V^U)}+\epsilon^{2-2\delta} V^a\\
            \lesssim & B\epsilon^{2-2\delta} V^a + \int_{U_0(V)}^U\frac{1}{V^{1-\k}} C(B)\left(\frac{V^{1-\k}}{-U^\prime}\right)^{N(I)\tau} \epsilon^{2-2\delta} V^a\\
            \lesssim & B\epsilon^{2-2\delta} V^a+B\epsilon^{2-2\delta} V^a C(B)\epsilon_1^{(1-\k)(1-N(I)\tau)}\\
            \lesssim & B\epsilon^{2-2\delta} V^a.
        \end{split}
\end{equation}
For $\Psi=\Omega^2\alpha,K$, the equations are
\begin{equation}
    \begin{split}
        \Omega^{-1}\nabla_3(\Omega^2\alpha,K)=&\nabla\Psi_2+\varphi\Psi+\varphi^3+\varphi\nabla\varphi,
    \end{split}
\end{equation}
and the estimates are similar. The estimates for $\Omega\beta^r,\sigma^r,\Omega^{-1}\betab^r$ come from equations
\begin{equation}
    \begin{aligned}
        &\Omega^{-1}\betab^r=\dv(\Omega^{-1}\chibh)+\etab\Omega^{-1}\chibh-\frac{1}{2}\nabla\left(\Omega^{-1}\tr\chi\right)-\frac{1}{2}\etab\Omega^{-1}\tr\chi,\\
        &\Omega\beta^r=-\dv(\Omega\chih)+\frac{1}{2}\nabla(\Omega\tr\chi)+\etab\Omega\chih-\frac{1}{2}\etab\Omega\tr\chi,\ \ 
        \sigma^r=-\cl\etab,
    \end{aligned}
\end{equation}
which contains no $\eta$.

\end{proof}

We preceed to derive the other desired bounds.
\begin{proposition}\label{R III L2S est of varphi_1 Phi_1}
    The following relation holds:
    \begin{equation}
        \sum_{i\leq 4}\lnm \nabla^i(\ol{\Omega\chi},\ol{\Omega^{-1}\tr\chib},\etab,\ol{\Omega\omega},\ol{\Omega D_4\phi},\nabla\phi)\rnm^2_{\LT^2(S_{U,V})}+\sum_{i\leq 3}\lnm \nabla^i\ol{\Psi_1}\rnm^2_{\LT^2(S_{U,V})}\lesssim \epsilon^{2-2\delta} W(U,V).
    \end{equation}
\end{proposition}
\begin{proof}
    The estimates used in this proposition are the non-top-order spherical bounds of Proposition \ref{R III L2S est non top order} and the integrated estimate \eqref{R III energy est R norm}. The smallness in the last line comes from choosing $\epsilon_1$ after the power $N(I)\tau$ has been fixed. Using the previous lemma, we obtain  
    \begin{equation}
        \sum_{i\leq 5}\int_{\SS}V^{2-2S(\nabla^I\varphi)}\left|\nabla^I\ol{\varphi}\right|^2\lesssim C(B,B^\prime)\epsilon^{2-2\delta} \left(\frac{V^{1-\k}}{-U}\right)^\tau,
    \end{equation}
    and 
    \begin{equation}
        \sum_{i\leq 4}\int_{\SS}V^{2-2S(\nabla^I\Psi)}\left|\nabla^I\ol{\Psi}\right|^2\lesssim C(B,B^\prime)\epsilon^{2-2\delta} \left(\frac{V^{1-\k}}{-U}\right)^\tau.
    \end{equation}
    Therefore, we have
    \begin{equation}
        \begin{split}
            &\int_{\SS}V^{a+2-2S(\nabla^I\varphi_1)}\left|\nabla^I\ol{\varphi_1}\right|^2\\ \lesssim & \epsilon^{2-2\delta} V^a+\int_{U_0(V)}^U \int_{\SS}V^{a+1-2S(\nabla^I\varphi_1)}\Omega^2\left|\nabla^I\ol{\varphi_1}\right|^2\\
            &+\int_{U_0(V)}^U \int_{\SS}V^{a+1-2S(\Omega^{-1}\nabla_3\nabla^I\varphi_1)}\Omega^2\left|\Omega^{-1}\nabla_3\nabla^I\ol{\varphi_1}\right|^2\\
            \lesssim & \epsilon^{2-2\delta} V^a+\epsilon_1^{(1-\k)(1-\tau)}\sup_{U_0(V)\leq U^\prime\leq U}\left(\int_{\SS}V^{a+2-2S(\nabla^I\varphi_1)}\left|\nabla^I\ol{\varphi_1}\right|^2(U^\prime,V)\right)\\
            &+\int_{U_0(V)}^U V^{a+\k-1} \epsilon^{2-2\delta} C(B,B^\prime)\left(\frac{V^{1-\k}}{-U^\prime}\right)^{N(I)\tau}\\
            \lesssim & \epsilon^{2-2\delta} V^a+\epsilon_1^{(1-\k)(1-\tau)}B^\prime \epsilon^{2-2\delta} V^a+V^a \epsilon^{2-2\delta} C(B,B^\prime)\epsilon_1^{(1-\k)(1-N(I)\tau)}.
        \end{split}
    \end{equation}
    Let $\epsilon_1$ be sufficiently small, we have the estimates for $\varphi_1$. $\Psi_1$ follows similarly.
\end{proof}
Using Sobolev embedding, one immediately obtains
\begin{corollary}
    For $\Omega\chi,\Omega\omega$, we establish the following bound:
    \begin{equation}
        \lnm \ol{\Omega\chi},\ol{\Omega\omega}\rnm^2_{\LS_1^\infty(S_{U,V})}\lesssim \epsilon^{2-2\delta}.
    \end{equation}
\end{corollary}

With the accurate estimate for $\varphi_1$, we then improve the estimate for $\Omega^{-1}\chibh,\ol{\Omega^{-1}D_3\phi}$.
\begin{proposition}\label{R III psi_1 W32 est}
    For $\Omega^{-1}\chibh$ and $\ol{\Omega^{-1}D_3\phi}$, we have
    \begin{equation*}
        \sum_{i\leq 3}\lnm \nabla^i\left(\Omega^{-1}\chibh,\ol{\Omega^{-1}D_3\phi}\right)\rnm^2_{\LT^2(S_{U,V})}\leq C(B)\epsilon^{2-2\delta}W(U,V).
    \end{equation*}
\end{proposition}
\begin{proof}
    We use $\Omega\nabla_4(\Omega^{-1}\chibh)$ equation to obtain that  
    \begin{equation*}
        \begin{aligned}
            &\int_{\SS}V^{a+2-2S(\Omega^{-1}\chibh)}\left|\Omega^{-1}\chibh\right|^2\\
            \leq &\epsilon^{2-2\delta}V^a+\int_{V_0(U)}^V\int_{\SS}(a+2-V^\prime\cdot\Omega\tr\chi +8V^\prime\cdot \Omega\omega)(V^\prime)^{a+1-2S(\Omega^{-1}\chibh)}\left|\Omega^{-1}\chibh\right|^2\\
            &+\int_{V_0(U)}^V\int_{\SS}V^{a+2-2S(\Omega^{-1}\chibh)}\left|\Omega^{-1}\chibh\right|\left|\nabla\varphi_1+\varphi_1\ol{\varphi_1}\right|\\
            \leq & \epsilon^{2-2\delta}V^a+(a-\k)\int_{V_0(U)}^V \frac{1}{V^\prime}\int_{\SS} (V^\prime)^{a+2-2S(\Omega^{-1}\chibh)}\left|\Omega^{-1}\chibh\right|^2 +\frac{C_0}{\k} \int_{V_0(U)}^V \epsilon^{2-2\delta} (V^\prime)^{a-1}\\
            \leq & C(B) \epsilon^{2-2\delta} V^a.
            \end{aligned}
    \end{equation*}
    Since $\Omega^{-1}D_3\phi$ satisfies equation 
    \begin{equation}
        \Omega e_4(\Omega^{-1}D_3\phi)+\frac{1}{2}\Omega\tr\chi\Omega^{-1}D_3\phi-4\Omega\omega\Omega^{-1}D_3\phi=\Delta \phi-\frac{1}{2}\Omega^{-1}\tr\chib\Omega D_4\phi+2\etab\nabla\phi=\nabla\varphi_1+\varphi_1\varphi_1,
    \end{equation}
    the signature satisfies 
    \begin{equation*}
        a+2-2S(\Omega^{-1}D_3\phi)-V\cdot\Omega\tr\chi +8V\cdot\Omega\omega<0.
    \end{equation*}
    The estimate is similar after making difference.
    For higher order derivatives, if $I\leq 3$,
    \begin{equation}
        \sum_{i_1+i_2+i_3\leq I+1}\int_{\SS} V^{2-2S}\left|\nabla^{i_1}\varphi_1^{i_2}\nabla^{i_3}\ol{\varphi_1}\right|^2\leq C(B)\epsilon^{2-2\delta}.
    \end{equation}
    Hence the lemma holds.
\end{proof}
We cannot do so for $\eta$ for the the energy of $\eta$ along $H_U^V$ cannot be canceled. Compared to Proposition \ref{R III L2S est non top order}, next proposition reveals that we can still improve the $\LT^2(S)$ norm for lower order terms.

\begin{proposition}\label{R III psi_2 Omega b W32 est}
    For any fixed $\lambda>0$, we can choose $\epsilon=\epsilon(\lambda),\epsilon_1=\epsilon_1(\lambda)$ sufficiently small, so that  
    \begin{equation}
        \sum_{i\leq 3}\lnm \nabla^ib,\nabla^i\eta,\nabla^i\ol{\Omega^2},\nabla^i\ol{\Omega^{-2}}\rnm^2_{\LT^2(S_{U,V})}\leq C(B,\lambda) \epsilon^{2-2\delta}W\cdot\left(\frac{V^{1-\k}}{-U}\right)^{\lambda}.
    \end{equation}
\end{proposition}
\begin{proof}
    We first write down the equations: 
    \begin{equation}
        \begin{split}
            \Omega\nabla_4b+\Omega\chi\cdot b=&2\Omega^2(\etab-\eta),\\
            \Omega\nabla_4\Omega^{2n}+4n\Omega\omega\Omega^{2n}=& 0, \quad \forall n\in\mathbb{Z},\\
            \Omega\nabla_4\eta+\Omega\chi\cdot\eta=&\Omega\chi\cdot\etab-\Omega\beta^r-\Omega D_4\phi \nabla\phi.
        \end{split}
    \end{equation}
    We choose $\epsilon,\epsilon_1$ sufficiently small such that in $\R_{III}$ it holds
    \begin{equation*}
        \begin{aligned}
            &\lnm \ol{\Omega\omega},\ol{\Omega\tr\chi}\rnm_{\LT^\infty(S)}+\lnm (\Omega\omega)^c(U,V)-(\Omega\omega)^c(0,V)\rnm_{\LS_1^\infty(S)}\\
            &\qquad\qquad+\lnm (\Omega\tr\chi)^c(U,V)-(\Omega\tr\chi)^c(0,V)\rnm_{\LS_1^\infty(S)}\ll \lambda.
        \end{aligned}
    \end{equation*}
    For $\Omega^2$, we find that 
    \begin{equation}
        \begin{split}
            &\int_{\SS}V^{a+2-2S(\nabla^I\Omega^2)}\left(\frac{-U}{V^{1-\k}}\right)^\lambda \left|\nabla^I{\Omega^2}\right|^2\\
            \lesssim & \epsilon^{2-2\delta}V^a+(a-\lambda(1-\k)/2)\int_{V_0(U)}^V\int_{\SS}(V^\prime)^{a+1-2S(\nabla^I \Omega^2)}\left|\nabla^I\ol{\Omega^2}\right|^2\left(\frac{-U}{{V^\prime}^{1-\k}}\right)^\lambda \\
            &+\frac{\lambda(1-\k)}{4}\int_{V_0(U)}^V\int_{\SS}(V^\prime)^{a+1-2S(\nabla^I \Omega^2)}\left|\nabla^I\ol{\Omega^2}\right|^2\left(\frac{-U}{{V^\prime}^{1-\k}}\right)^\lambda \\
            &+\frac{4}{\lambda(1-\k)}\sum_{\substack{i_1+i_2+i_3= I\\ i_3\leq I-1}}\int_{V_0(U)}^V\int_{\SS}(V^\prime)^{a+1-2S(\nabla^I \Omega^2)}\left| \nabla^{i_1}\varphi_1^{i_2+1}\nabla^{i_3}\ol{\Omega^2}\right|^2\left(\frac{-U}{{V^\prime}^{1-\k}}\right)^\lambda \\
            \leq & \epsilon^{2-2\delta}V^a+\frac{4}{\lambda(1-\k)}\sum_{\substack{i_1+i_2+i_3= I\\ i_3\leq I-1}}\int_{V_0(U)}^V\int_{\SS}(V^\prime)^{a+1-2S(\nabla^I \Omega^2)}\left| \nabla^{i_1}\varphi_1^{i_2+1}\nabla^{i_3}\ol{\Omega^2}\right|^2\left(\frac{-U}{{V^\prime}^{1-\k}}\right)^\lambda\\
            &+(a-\lambda(1-\k)/4)\int_{V_0(U)}^V\int_{\SS}(V^\prime)^{a+1-2S(\nabla^I \Omega^2)}\left|\nabla^I\ol{\Omega^2}\right|^2\left(\frac{-U}{{V^\prime}^{1-\k}}\right)^\lambda.
        \end{split}
    \end{equation} 
    Inductively, for $I=0$ the desired result is direct. Suppose we have proved
    \begin{equation}
        \sum_{i\leq I-1}\int_{\SS}V^{a+2-2S(\nabla^i\Omega^2)}\left(\frac{-U}{V^{1-\k}}\right)^\lambda \left|\nabla^i\ol{\Omega^2}\right|^2
            \leq C(B,\lambda)\epsilon^{2-2\delta}V^a,
    \end{equation}
    then for $F_I(V):=\int_{V_0(U)}^V\int_{\SS}(V^\prime)^{a+1-2S(\nabla^I \Omega^2)}\left|\nabla^I\ol{\Omega^2}\right|^2\left(\frac{-U}{{V^\prime}^{1-\k}}\right)^\lambda$ we have 
    \begin{equation}
        V\partial_V F\leq C(B,\lambda)\epsilon^{2-2\delta}V^a+(a-\lambda(1-\k)/4) F,
    \end{equation}
    or equivalently, 
    \begin{equation}
        \partial_V\left(V^{\lambda(1-\k)/4-a}F\right) \leq C(B,\lambda)\epsilon^{2-2\delta} V^{\lambda(1-\k)/4-1}.
    \end{equation}
    We obtain $$F\leq C(B,\lambda)\epsilon^{2-2\delta} V^{a}$$ since $F(V_0(U))=0$, and 
    \begin{equation}
        \int_{\SS}V^{a+2-2S(\nabla^I\Omega^2)}\left(\frac{-U}{V^{1-\k}}\right)^\lambda \left|\nabla^I{\Omega^2}\right|^2
            \leq C(B,\lambda)\epsilon^{2-2\delta}V^a.
    \end{equation}
    $\Omega^{-2}$ can be estimated in the same way as $\Omega^2$. For $\eta$, the following inequality holds 
    \begin{equation}
        \begin{split}
            &\int_{\SS}V^{a+2-2S(\nabla^I\eta)}\left(\frac{-U}{V^{1-\k}}\right)^\lambda \left|\nabla^I\eta\right|^2\\
            \lesssim & \epsilon^{2-2\delta}V^a+(a-\lambda(1-\k)/2)\int_{V_0(U)}^V\int_{\SS}(V^\prime)^{a+1-2S}\left| \nabla^I\eta\right|^2\left(\frac{-U}{(V^\prime)^{1-\k}}\right)^\lambda\\
            &+\frac{4}{\lambda(1-\k)}\sum_{\substack{i_1+i_2+i_3= I\\ i_3\leq I-1}}\int_{V_0(U)}^V\int_{\SS}(V^\prime)^{a+1-2S}\left| \nabla^{i_1}\varphi_1^{i_2+1}\nabla^{i_3}\eta\right|^2\left(\frac{-U}{(V^\prime)^{1-\k}}\right)^\lambda\\
            &+\frac{\lambda(1-\k)}{4}\int_{V_0(U)}^V\int_{\SS}(V^\prime)^{a+1-2S}\left| \nabla^I\eta\right|^2\left(\frac{-U}{(V^\prime)^{1-\k}}\right)^\lambda\\
            &+\frac{4}{\lambda(1-\k)}\sum_{i_1+i_2=I}\int_{V_0(U)}^V\int_{\SS}(V^\prime)^{a+1-2S}\left|\nabla^{i_1}\varphi_1^{i_2+2}\right|^2\left(\frac{-U}{(V^\prime)^{1-\k}}\right)^\lambda\\
            \leq & \epsilon^{2-2\delta}V^a+\frac{1}{2(\lambda(1-\k)-a)}\sum_{\substack{i_1+i_2+i_3= I\\ i_3\leq I-1}}\int_{V_0(U)}^V\int_{\SS}(V^\prime)^{a+1-2S}\left| \nabla^{i_1}\varphi_1^{i_2+1}\nabla^{i_3}\eta\right|^2\left(\frac{-U}{(V^\prime)^{1-\k}}\right)^\lambda\\
            &+(a-\lambda(1-\k)/4)\int_{V_0(U)}^V\int_{\SS}(V^\prime)^{a+1-2S}\left| \nabla^I\eta\right|^2\left(\frac{-U}{(V^\prime)^{1-\k}}\right)^\lambda.
        \end{split}
    \end{equation}
    By induction, we finish the estimate for $\eta$:
    \begin{equation}
        \int_{\SS} V^{2-2S(\nabla^I\eta)}|\nabla^I\eta|^2+V^{2-2S(\nabla^I\ol{\Omega^2})}|\nabla^I\ol{\Omega^2}|^2d\V\leq C(B)\epsilon^{2-2\delta}\left(\frac{V^{1-\k}}{-U}\right)^{\lambda/2}.
    \end{equation}
We have estimate for $b$ that 
    \begin{equation}
        \begin{split}
            \int_{\SS} V^{a+2I-2\k}|\nabla^I b| \lesssim & \epsilon^{2-2\delta}V^a + \int_{V_0(U)}^V\int_{\SS}(V^\prime)^{a+1-2S}\left|\nabla^I\left(\Omega^2(\etab-\eta)\right)\right|^2\\
            &+\sum_{\substack{i_1+i_2+i_3= I\\ i_3\leq I-1}}\int_{V_0(U)}^V\int_{\SS}(V^\prime)^{a+1-2S}\left|\nabla^{i_1}\varphi_1^{i_2+1}\nabla^{i_3}b\right|^2\\
            \lesssim & \epsilon^{2-2\delta} V^a+\sum_{\substack{i_1+i_2+i_3= I\\ i_3\leq I-1}}\int_{V_0(U)}^V\int_{\SS}(V^\prime)^{a+1-2S}\left|\nabla^{i_1}\varphi_1^{i_2+1}\nabla^{i_3}b\right|^2\\
            &+C(B)\int_{V_0(U)}^V \left(\frac{(V^\prime)^{1-\k}}{-U}\right)^{\lambda/2+\lambda/2}(V^\prime)^{a-1}\\
            \lesssim &C(B) \left(\frac{V^{1-\k}}{-U}\right)^{\lambda}V^a.
        \end{split}
    \end{equation}
\end{proof}

\subsection{Energy estimates} In this section, we establish the top-order energy estimates.
\begin{proposition}\label{R III Phi_1 R Phi_2 Rb Est}
    Let $(\Phi_1,\Phi_2)$ be $(\nabla \Omega D_4\phi,\nabla^2\phi)$, $(\nabla^2\phi,\nabla\Omega^{-1}D_3\phi)$, \\
    $(\Omega^2\alpha,\Omega\beta^r)$,
    $(\Omega\beta^r,(K,\sigma^r))$, $((K,\sigma^r),\Omega^{-1}\betab^r)$, $(\nabla(\Omega^{-1}\tr\chib),0)$, $(\nabla\Omega\omega+{}^*\nabla\Omega\omega^\dagger-\frac{1}{2}\Omega\beta^r,0)$, $(-\dv\etab+K,0)$,
    $(0,-\dv\eta+K),(0,\nabla(\Omega\tr\chi))$. Given any $0<\lambda<\frac{\tau}{2}$, for $\epsilon_1$ and $\epsilon=\epsilon(\epsilon_1)$ sufficiently small we have 
    \begin{equation}
        \sum_{I\leq 5}\lnm \nabla^I\ol{\Phi_1}\left(\frac{V^{1-\k}}{-U}\right)^\lambda\rnm^2_{\LT^2({\R_{U,V}})}+\lnm \nabla^I\ol{\Phi_2}\left(\frac{V^{1-\k}}{-U}\right)^\lambda\rnm^2_{\LT^2({\Rb_{U,V}})}\leq \epsilon_1^{1/2}\epsilon^{2-2\delta}V^a.
    \end{equation}
    Here we recall that the signature of $\psi\cdot\left(\frac{V^{1-\k}}{-U}\right)^\lambda$ is $S(\psi)$ for any $\lambda\in\mathbb{R}$.
\end{proposition}
\begin{proof}
    It suffices to estimate the top order case. 
    Firstly, notice that 
    \begin{equation*}
       \begin{split}
         \lnm \nabla^I\Phi_2\left(\frac{V^{1-\k}}{-U}\right)^\lambda\rnm^2_{\LT^2({\Rb_{U,V}})}\leq & \int_{U_0(V)}^U\frac{1}{V^{1-\k}}\left(\frac{V^{1-\k}}{-U^{\prime}}\right)^{2\lambda}\lnm \nabla^I\Phi_2\rnm^2_{\LT^2(\Hb_V^{U^\prime})}\\
        \lesssim& \epsilon_1^{(1-\k)(1-2\lambda)} B\epsilon^{2-2\delta} V^a\leq \epsilon_1^{1/2}\epsilon^{2-2\delta} V^a. 
       \end{split}
    \end{equation*}
    We consider equations 
    \begin{equation}
        \begin{split}
            \Omega^{-1}\nabla_3\Phi_1- \mathcal{D}\Phi_2=&\varphi\Phi+\varphi^3+\varphi\nabla\varphi,\\
            \Omega \nabla_4\Phi_2+{}^*\mathcal{D}\Phi_1=&\varphi\Phi+\varphi^3+\varphi\nabla\varphi.
        \end{split}
    \end{equation}
    Commutation formula gives 
    \begin{equation}
        \begin{split}
            \Omega^{-1}\nabla_3\nabla^5\Phi_1-\mathcal{D}\nabla^5\Phi_2=&\sum_{i_1+i_2+i_3=5}\nabla^{i_1}\varphi^{i_2+1}\nabla^{i_3}\Phi+\sum_{i_1+i_2=6}\nabla^{i_1}\varphi^{i_2+2},\\
            \Omega\nabla_4\nabla^5\Phi_2+{}^*\mathcal{D}\nabla^5\Phi_1=&\sum_{i_1+i_2+i_3=5}\nabla^{i_1}\varphi^{i_2+1}\nabla^{i_3}\Phi+\sum_{i_1+i_2=6}\nabla^{i_1}\varphi^{i_2+2}.
        \end{split}
    \end{equation}
The next computation is written directly at the level of the integrated $\mathcal{R}$ and $\Rb$ norms. This avoids having to reapply the energy identity separately on each cone and makes the cancellation between the Hodge-dual top derivatives visible.
    Regarding the combined energy summation:
    \begin{equation}
        \begin{aligned}
            &\lnm\nabla^5{\Phi_1}\left(\frac{V^{1-\k}}{-U}\right)^\lambda\rnm^2_{\LT^2({\R_{U,V}})}+\lnm \nabla^5{\Phi_2}\left(\frac{V^{1-\k}}{-U}\right)^\lambda\rnm^2_{\LT^2({\Rb_{U,V}})}\\
            =& \int_{U_0(V)}^U\int_{V_0(U^\prime)}^V \frac{1}{(V^\prime)^{2-\k}}(V^\prime)^{a+2-2S(\nabla^5\Phi_1)}\left(\frac{-U^\prime}{(V^\prime)^{1-\k}}\right)^{\tau-2\lambda}\left|\nabla^5\Phi_1\right|^2\\
            &+\int_{U_0(V)}^{U}\int_{U_0(V)}^{U^\prime}\frac{1}{V^{2-\k}}V^{a+2-2S(\nabla^5\Phi_2)}\Omega^2\left(\frac{-U^{\prime\prime}}{V^{1-\k}}\right)^{\tau-2\lambda}\left|\nabla^5\Phi_2\right|^2,
        \end{aligned}
    \end{equation}
    we first systematically estimate the initial integral term as follows:
    \begin{equation*}
        \begin{aligned}
            &\int_{U_0(V)}^U\int_{V_0(U^\prime)}^V \frac{1}{(V^\prime)^{2-\k}}(V^\prime)^{a+2-2S(\nabla^5\Phi_1)}\left(\frac{-U^\prime}{(V^\prime)^{1-\k}}\right)^{\tau-2\lambda}\left|\nabla^5\Phi_1\right|^2\\
            =&\int_{U_0(V)}^U\int_{V_0(U^\prime)}^V\int_{U_0(V^\prime)}^{U^\prime}\frac{1}{(V^\prime)^{2-\k}}(V^\prime)^{a+2-2S(\nabla^5\Phi_1)}(\Omega e_3-b)\left[\left(\frac{-U^{\prime\prime}}{(V^\prime)^{1-\k}}\right)^{\tau-2\lambda}\left|\nabla^5\Phi_1\right|^2\right]\\
            &+\int_{U_0(V)}^U\int_{V_0(U^\prime)}^V\frac{1}{(V^\prime)^{2-\k}}(V^\prime)^{a+2-2S(\nabla^5\Phi_1)}\left(\frac{-U_0(V^\prime)}{(V^\prime)^{1-\k}}\right)^{\tau-2\lambda}\left|\nabla^5\Phi_1\right|^2\\
            \leq & \int_{U_0(V)}^U\int_{V_0(U^\prime)}^V\int_{U_0(V^\prime)}^{U^\prime}\frac{1}{(V^\prime)^{3-2\k}}(V^\prime)^{a+2-2S(\nabla^5\Phi_1)}\lnm \dv b\rnm_{\LT^\infty(S_{U^{\prime\prime},V^\prime})}\left(\frac{-U^{\prime\prime}}{(V^\prime)^{1-\k}}\right)^{\tau-2\lambda}\left|\nabla^5\Phi_1\right|^2\\
            &+\int_{U_0(V)}^U\int_{V_0(U^\prime)}^V\int_{U_0(V^\prime)}^{U^\prime}(V^\prime)^{a+\k-2S(\nabla^5\Phi_1)}2\langle \nabla^5\Phi_1,\mathcal{D}\nabla^5\Phi_2\rangle+\lnm \nabla^5\Phi_1\rnm^2_{\LT^2(\Xi^V)}\\
            &+\lnm \sum_{i_1+i_2+i_3=5}\nabla^{i_1}\varphi^{i_2+1}\nabla^{i_3}\Phi+\sum_{i_1+i_2=6}\nabla^{i_1}\varphi^{i_2+2}\rnm^2_{\LT^2(\mathcal{\D}_{U,V})}\\
            \lesssim & \lnm \nabla^5\Phi_1 \left(\frac{V^{1-\k}}{-U}\right)^{\lambda+\tau/4}\rnm^2_{\LT^2(\D_{U,V})}+\lnm \sum_{i_1+i_2+i_3=5}\nabla^{i_1}\varphi^{i_2+1}\nabla^{i_3}\Phi+\sum_{i_1+i_2=6}\nabla^{i_1}\varphi^{i_2+2}\rnm^2_{\LT^2(\mathcal{\D}_{U,V})}\\
            &+\epsilon^{2-2\delta}V^a +2\int_{U_0(V)}^U\int_{V_0(U^\prime)}^V\int_{U_0(V^\prime)}^{U^\prime}(V^\prime)^{a+\k-2S(\nabla^5\Phi_1)}\langle \nabla^5\Phi_1,\mathcal{D}\nabla^5\Phi_2\rangle.
        \end{aligned}
    \end{equation*}
    Following a similar line of reasoning, for $\lnm \nabla^5\Phi_2\rnm^2_{\LT^2(\Rb_{U,V})}$ we establish that
    \begin{equation}
        \begin{aligned}
            &\int_{U_0(V)}^{U}\int_{U_0(V)}^{U^\prime}\frac{1}{V^{2-\k}}V^{a+2-2S(\nabla^5\Phi_2)}\Omega^2\left(\frac{-U^{\prime\prime}}{V^{1-\k}}\right)^{\tau-2\lambda}\left|\nabla^5\Phi_2\right|^2\\
            =& \int_{U_0(V)}^{U}\int_{U_0(V)}^{U^\prime}\int_{V_0(U^{\prime\prime})}^V (\Omega e_4)\left((V^\prime)^{a+\k-2S(\nabla^5\Phi_2)}\Omega^2\left(\frac{-U^{\prime\prime}}{(V^\prime)^{1-\k}}\right)^{\tau-2\lambda}\left|\nabla^5\Phi_2\right|^2\right)\\
            &+\int_{U_0(V)}^{U}\int_{U_0(V)}^{U^\prime}V_0(U^{\prime\prime})^{a+\k-2S(\nabla^5\Phi_2)}\Omega^2\left(\frac{-U^{\prime\prime}}{V_0(U^{\prime\prime})^{1-\k}}\right)^{\tau-2\lambda}\left|\nabla^5\Phi_2\right|^2\\
            \lesssim & \lnm \nabla^5\Phi_2\rnm^2_{\LT^2(\Xi^V)}+\lnm \nabla^5\Phi_2\rnm^2_{\LT^2(\D_{U,V})}\\
            &-2\int_{U_0(V)}^U\int_{V_0(U^\prime)}^V\int_{U_0(V^\prime)}^{U^\prime}(V^\prime)^{a+\k-2S(\nabla^5\Phi_1)}\langle {}^*\mathcal{D}\nabla^5\Phi_1,\nabla^5\Phi_2\rangle\\
            &+\lnm \sum_{i_1+i_2+i_3=5}\nabla^{i_1}\varphi^{i_2+1}\nabla^{i_3}\Phi+\sum_{i_1+i_2=6}\nabla^{i_1}\varphi^{i_2+2}\rnm^2_{\LT^2(\mathcal{\D}_{U,V})}.
        \end{aligned}
    \end{equation}
    By summing these two resulting inequalities, the critical top-order ($6$th-order) derivative terms effectively cancel out.  
    Furthermore, the auxiliary $\LT^2(\D)$ norm can be straightforwardly estimated utilizing Lemma \ref{R III Key Lemma} and Corollary \ref{R III Key Lemma Cor}, given that every involved geometric quantity already admits either a robust $\LT^2(\R)$ or $\LT^2(\Hb)$ estimate. 

\end{proof}

We have the following corollaries from the estimates above.
\begin{corollary}
    Let $0<\lambda< \frac{\tau}{2}$. We find that 
    \begin{equation}\label{R_III_Enhanced_Estimate_e_3_eq_Connections}
        \begin{aligned}
            \sum_{i\leq 6}\lnm \nabla^i\left(\ol{\Omega^{-1}\tr\chib},\ol{\Omega\omega},\etab,\nabla\phi,\ol{\Omega e_4\phi}\right)\left(\frac{V^{1-\k}}{-U}\right)^\lambda\rnm^2_{\LT^2({\R_{U,V}})}\leq \epsilon^{2-2\delta}V^a.
        \end{aligned}
    \end{equation}
\end{corollary}
\begin{corollary}\label{R III Phi_1 R Est Cor}
    By Propositions \ref{R III psi_1 W32 est} and \ref{R III psi_2 Omega b W32 est}, for any fixed integer $N\geq 0$, we can choose $\epsilon_1,\epsilon$ sufficiently small so that
    \begin{equation}
        \begin{split}
            &\sum_{i,j,k\leq 1}\sum_{l,m,n\leq N}\lnm \nabla^5\left(\Phi_1,\nabla(\Omega^{-1}\tr\chib),\nabla\Omega\omega,\nabla\etab\right)\cdot{\left|\nabla^i\varphi\right|^m} |\nabla^jb|^n\left|\nabla^k\Omega^2\right|^l\rnm^2_{\LT^2(\R_{U,V})}\\
            &\qquad\leq \epsilon_1^{1/2}\epsilon^{2-2\delta}V^a.
        \end{split}
    \end{equation}
\end{corollary}

We now proceed to establish the estimates for the $\LT^2(H)$ and $\LT^2(\Hb)$ norms corresponding specifically to the top-order terms.

\begin{proposition}
    We establish the following estimates for the relevant quantities:
    \begin{equation}
        \lnm \nabla^6(\Omega e_4\phi)\rnm^2_{\LT^2(H_U^V)}+\lnm \nabla^7\phi\rnm^2_{\LT^2(\Hb_V^U)}\lesssim \epsilon^{2-2\delta}V^a.
    \end{equation}
\end{proposition}
\begin{proof}
    This proof now uses Proposition \ref{R III Phi_1 R Phi_2 Rb Est} for the enhanced $\mathcal{R}$-control of the pair\\ $(\nabla\Omega D_4\phi,\nabla^2\phi)$. The remaining lower-order products are controlled by Proposition \ref{R III L2S est non top order} and the spherical improvements from the preceding subsection. We first write the equations:
    \begin{equation}
        \begin{aligned}
         &   \Omega^{-1}\nabla_3\nabla^6(\Omega e_4\phi)-\dv \nabla^7\phi=\sum_{i_1+i_2=6}\nabla^{i_1}\varphi^{i_2+2}+\sum_{\substack{i_1+\cdots +i_5=6\\ i_2+i_4\geq 1}} \nabla^{i_1}(\eta+\etab)^{i_2}\nabla^{i_3}K^{i_4}\nabla^{i_5+1}\phi ,\\
            &\Omega\nabla_4\nabla^7\phi-\nabla(\nabla^6(\Omega e_4\phi))= \sum_{i_1+i_2=6}\nabla^{i_1}\varphi^{i_2+2}+\sum_{\substack{i_1+i_2+i_3=4}} \nabla^{i_1}K^{i_2+1}\nabla^{i_3+1}(\Omega e_4\phi).
        \end{aligned}
    \end{equation}
    Since $\nabla^6(\Omega e_4\phi),\nabla^7\phi$ both admits enhanced $\LT(\R)$ estimates, the following inequalities holds 
    \begin{equation}
        \begin{aligned}
            &\lnm \nabla^6(\Omega e_4\phi)\rnm^2_{\LT^2(H_U^V)}+\lnm \nabla^7\phi\rnm^2_{\LT^2(\Hb_V^U)}\\
            \lesssim & \epsilon^{2-2\delta}V^a+\lnm \left(\frac{V^{1-\k}}{-U}\right)^{\tau/3}\Omega\nabla^6\left(\Omega e_4\phi,\nabla\phi\right)\rnm^2_{\LT^2(\R_{U,V})}\\
            &+\sum_{i\leq 6,j\leq 4}\lnm \left(\frac{-U}{V^{1-\k}}\right)^{\tau/3}\Omega\left(\nabla^i\varphi,\nabla^j K\right)\rnm^2_{\LT^2(\R_{U,V})}\\
            \lesssim & \epsilon^{2-2\delta} V^a.
        \end{aligned}
    \end{equation}
\end{proof}
\begin{proposition}
    For $\Omega\tr\chi$, we have the estimate:
    \begin{equation}\label{R_III_Estimate_Omegatrchi_S}
        \lnm \nabla^6(\Omega \tr\chi)\rnm^2_{\LT^2(S_{U,V})}\lesssim B \epsilon^{2-2\delta}V^a\left(\frac{V^{1-\k}}{-U}\right)^{C\tau},
    \end{equation}
    where the implicit constant $C$ is strictly bounded by a universal geometric constant.
\end{proposition}
\begin{proof}
    We commence by explicitly writing down the corresponding top-order evolution equation: 
    \begin{equation}
        \begin{aligned}
            \Omega\nabla_4\nabla^6(\Omega\tr\chi)&+(4\Omega\tr\chi+4\Omega\omega+C\Omega\chih)\nabla^6(\Omega\tr\chi)\sim\Omega\chih\nabla^6(\Omega\chih)\\
            &+\Omega e_4\phi\nabla^6(\Omega e_4\phi)+\Omega\tr\chi\nabla^6(\Omega\omega)+\text{non top order terms}.
        \end{aligned}
    \end{equation}
    The estimates invoked in the source terms are the scalar-field estimate just proved, the enhanced $\Omega\omega$ and $\etab$ estimates in \eqref{R_III_Enhanced_Estimate_e_3_eq_Connections}, and the curvature estimates from Proposition \ref{R III Phi_1 R Phi_2 Rb Est}. The trace-free term $\Omega\chih$ is treated through the Codazzi equation and the enhanced $\Omega\chi$ estimate is obtained immediately afterward.
    Crucially, since the leading algebraic coefficient of the borderline top-order term is strictly negative, i.e.,
    \[\frac{2+\k-2S(\nabla^6\Omega\tr\chi)+a+(2\lambda-\tau)(1-\k)}{V}-8\Omega\tr\chi-8\Omega\omega-C|\Omega\chih|<-1/V,\]
    we hence have the estimate 
    \begin{equation}
        \begin{aligned}
            &\lnm \nabla^6(\Omega \tr\chi) \rnm^2_{\LT^2(S_{U,V})}\\
            \lesssim & \epsilon^{2-2\delta} V^a-\frac{1}{2}\lnm \Omega\nabla^6(\Omega\tr\chi) \rnm^2_{\LT^2(H_U^V)}+o_{\epsilon}(1)\lnm \Omega\nabla^6(\Omega\chih)\rnm^2_{\LT^2(H_U^V)}\\
            &+\lnm \Omega\nabla^6(\Omega e_4\phi,\Omega\omega)\rnm^2_{\LT^2(H_U^V)}+\int_{V_0(U^\prime)}^V \frac{1}{V^\prime}\left(\frac{(V^\prime)^{1-\k}}{-U^\prime}\right)^{C\tau}\epsilon^{2-2\delta}(V^\prime)^a.
        \end{aligned}
    \end{equation}
\end{proof}
Leveraging the result of the preceding proposition, we can now rigorously establish the enhanced global $\LT^2(\R)$ energy estimate specifically for $\Omega\chi$.
\begin{corollary}
    Let $0<\lambda<\tau/2$. For $\Omega\chi$, the below inequality holds
    \begin{equation}\label{R_III_Enhanced_Estimate_OMegachi}
        \lnm \nabla^6(\Omega \chi)\left(\frac{V^{1-\k}}{-U}\right)^{\lambda}\rnm^2_{\LT^2(\R_{U,V})}\lesssim \epsilon^{2-2\delta}V^a.
    \end{equation}
\end{corollary}
\begin{proof}  
The first displayed estimate converts the spherical trace estimate \eqref{R_III_Estimate_Omegatrchi_S} into an $\mathcal{R}$-norm. The elliptic estimate then adds the trace-free part using the curvature estimate for $\Omega\beta^r$ from Proposition \ref{R III Phi_1 R Phi_2 Rb Est}. Thus, we derive
    \begin{equation}
        \begin{aligned}
            & \lnm \nabla^6(\Omega \tr\chi)\left(\frac{V^{1-\k}}{-U}\right)^{\lambda}\rnm^2_{\LT^2(\R_{U,V})}
            \\
            \lesssim& \int_{V_0(U)}^V\int_{U_0(V^\prime)}^U\frac{1}{\left(V^\prime\right)^{2-\k}}B\epsilon^{2-2\delta} (V^\prime)^a\left(\frac{(V^\prime)^{1-\k}}{-U^\prime}\right)^{C\tau}\lesssim \epsilon^{2-2\delta} V^a.
        \end{aligned}
    \end{equation}
    Using elliptic estimate, it follows that
    \begin{equation}
        \begin{aligned}
            & \lnm \nabla^6(\Omega \chih)\left(\frac{V^{1-\k}}{-U}\right)^{\lambda}\rnm^2_{\LT^2(\R_{U,V})}\\
            \lesssim &   \lnm \left(\nabla^6(\Omega \tr\chi),\nabla^5(\Omega\beta^r),l.o.t.\right)\left(\frac{V^{1-\k}}{-U}\right)^{\lambda}\rnm^2_{\LT^2(\R_{U,V})}\lesssim \epsilon^{2-2\delta }V^a.
        \end{aligned}
    \end{equation}
\end{proof}

Leveraging the estimates for $\LT^2(\R_{U,V})$ norms, we can obtain the energy estimates for $\LT^2(H_U^V)$, $\LT^2(\Hb_V^U)$.
\begin{proposition}
    For $(\Phi_1,\Phi_2)\in\{(\nabla^2\phi,\nabla(\Omega^{-1}e_3\phi))$, $(\Omega^2\alpha,\Omega\beta^r)$, $(\Omega\beta^r,(K,\sigma^r)) $,\\ $((K,\sigma^r),\Omega^{-1}\betab^r),$ $(\nabla(\Omega^{-1}\tr\chib),0)$,
     $(-\dv\etab+ K,0),(\nabla(\Omega\omega)+{}^*\nabla(\Omega\omega^\dagger)-\frac{1}{2}\Omega\beta^r,0)\}$, it follows that 
    \begin{equation}
        \lnm \nabla^5\Phi_1\rnm^2_{\LT^2(H_U^V)}+\lnm \nabla^5\Phi_2\rnm^2_{\LT^2(\Hb_V^U)}\lesssim \epsilon_1^{1/2}\epsilon^{2-2\delta}V^a.
    \end{equation}
\end{proposition}

\begin{proof}
    This proposition upgrades the enhanced $\mathcal{R}$-estimates to flux estimates. It uses \eqref{R_III_Enhanced_Estimate_e_3_eq_Connections}, \eqref{R_III_Enhanced_Estimate_OMegachi}, the $\Omega^{-1}\chibh$ estimate obtained below, and the non-top-order product bounds from Proposition \ref{R III L2S est non top order}.
    We begin our analysis with the fundamental schematic structure equations: 
    \begin{equation}
        \begin{split}
            \Omega\nabla_4\Phi_2+&\frac{1}{2}\lambda \Omega\tr\chi \Phi_2+4\Lambda \Omega\omega \Phi_2+C_0\Omega\chih\cdot\Phi_2-\D \Phi_1\\
            &=\varphi\Phi_1+\varphi\nabla\varphi_1+\varphi^3+\Omega\chih\nabla(\Omega^{-1}\chibh),\\ 
            \Omega^{-1}\nabla_3\Phi_1+&{}^*\D \Phi_2=\varphi\Phi+\sum_{i_1+i_2=1}\nabla^{i_1}\varphi^{i_2+2} ,
        \end{split}
    \end{equation}
    where $\varphi_1$ stands for $\Omega\chi,\etab,\nabla\phi,\Omega e_4\phi,\Omega\omega$ with enhanced $\LT(\R)$ estimates \eqref{R_III_Enhanced_Estimate_e_3_eq_Connections} and \eqref{R_III_Enhanced_Estimate_OMegachi}. Thus every occurrence of $\varphi_1$ in the top-order source is attached to an already enhanced $\mathcal{R}$-estimate. The only remaining delicate term is the product involving $\Omega^{-1}\chibh$, because that component is controlled through an elliptic estimate rather than directly by the $\nabla_3$-enhanced family. Applying standard algebraic commutation formulae subsequently yields 
    \begin{equation}
        \begin{aligned}
            &\Omega \nabla_4\nabla^5\Phi_2+\frac{1}{2}(\lambda+5)\Omega\tr\chi\nabla^5\Phi_2+4\Lambda\Omega\omega\nabla^5\Phi_2+C_5\Omega\chih\cdot\nabla^5\Phi_2-\mathcal{D}\nabla^5\Phi_1\\
            =& \underline{\varphi\nabla^{5}\Phi_1+\varphi\nabla^6\varphi_1+\Omega\chih\nabla^6(\Omega^{-1}\chibh)}_{G_1}\\
            &\underline{+\sum_{\substack{i_1+i_2+i_3=5\\ i_3\leq 4}}\nabla^{i_1}\varphi^{i_2+1}\nabla^{i_3}\Phi +\sum_{\substack{i_1+i_2+i_3=6\\ i_1,i_3\leq 5}}\nabla^{i_1}\varphi_1^{i_2+1}\nabla^{i_3}\varphi +\sum_{i_1+i_2+i_3=4}\nabla^{i_1}K^{i_2+1}\nabla^{i_3}\Phi_1}_{G_2},\\
            &\Omega^{-1}\nabla_3\nabla^5\Phi_1+{}^*\D\nabla^5\Phi_2\\
            =& \underline{\varphi\nabla^5\Phi+\varphi\nabla^6\varphi}_{F_1}\\ &+\underline{\sum_{\substack{i_1+i_2+i_3=5\\ i_3\leq 4}}\nabla^{i_1}\varphi^{i_2+1}\nabla^{i_3}\Phi+\sum_{\substack{i_1+i_2=6\\i_1,i_2\leq 5}}\nabla^{i_1}\varphi\nabla^{i_2}\varphi^{i_3+1} +\sum_{\substack{i_1+\cdots +i_5=5\\ i_2+i_4\geq 1}} \nabla^{i_1}(\eta+\etab)^{i_2}\nabla^{i_3}K^{i_4}\nabla^{i_5}\Phi_2}_{F_2}.
        \end{aligned}
    \end{equation}
    We carefully note that the highest top-order terms arising in the $\Omega\nabla_4$ evolution equations are precisely given by $\varphi\nabla^5\Phi_1,\varphi\nabla^6\varphi_1$, and the coupled differences $\nabla^6(\Omega\chih\Omega^{-1}\chibh)$ alongside $\Omega^{-1}\tr\chib\nabla^6(\Omega\tr\chi)$.
    The constants $\lambda,\Lambda$ take values:
    \begin{equation}
        \begin{array}{c|c|c|c|c|c|c}
            \Phi_2 & \Omega^{-1} D_3\nabla\phi  & \Omega\beta^r & K &\sigma^r & \Omega^{-1}\betab^r    \\
            \hline
            \lambda & 2  & 4 & 2 & 2 & 2 \\
            \Lambda & -1  & 1 & 0 & 0 & -1 \\
        \end{array}
    \end{equation}
    Consequently, substituting these precise tabular values into the signature formula, we deduce
    \begin{equation*}
        1-2S(\nabla^I\Phi_2)+a-(1-\k)\tau-V\cdot(\lambda+I)\Omega\tr\chi-8V\cdot \Lambda \Omega\omega = 3+a-(1-\k)\tau-2\lambda +2\k\Lambda +O(\epsilon)<-1/4.
    \end{equation*}
    This strictly negative upper bound mathematically justifies the direct application of the bulk energy estimate formulated in \ref{R III energy est lemma H Hb}:

\begin{equation}
        \begin{aligned}
            &\lnm \nabla^5\Phi_1\rnm^2_{\LT^2(H_U^V)}+\lnm \nabla^5\Phi_2\rnm^2_{\LT^2(\Hb_V^U)}\\
            \lesssim & \epsilon^{2-2\delta} V^a-\frac{1}{4}\lnm \Omega\nabla^5\Phi_2\rnm^2_{\LT^2(\R_{U,V})}+\lnm \lnm\Omega^2\nabla b\rnm_{\LS_1^\infty(S)}^{1/2}\nabla^5\Phi_1\rnm^2_{\LT^2(\R_{U,V})}\\
            &+\langle \Omega\nabla^5\Phi_1,\Omega F_1\rangle_{\R_{U,V}}+\langle \Omega\nabla^5\Phi_1,\Omega F_2\rangle_{\R_{U,V}}+\langle \Omega\nabla^5\Phi_2,\Omega G_1\rangle_{\R_{U,V}}+\langle \Omega\nabla^5\Phi_2,\Omega G_2\rangle_{\R_{U,V}}\\
            \lesssim & \epsilon^{2-2\delta} V^a +\lnm \left(\frac{V^{1-\k}}{-U}\right)^{\tau/3}\nabla^5\Phi_1\rnm^2_{\LT^2(\R_{U,V})}+\lnm \left(\frac{-U}{V^{1-\k}}\right)^{\tau/3}\Omega F_1\rnm^2_{\LT^2(\R_{U,V})}\\
            &+\lnm \Omega F_2,\Omega G_2\rnm^2_{\LT^2(\R_{U,V})} +\lnm \Omega G_1\rnm^2_{\LT^2(\R_{U,V})}-\frac{1}{8}\lnm \Omega\nabla^5\Phi_2\rnm^2_{\LT^2(\R_{U,V})}.
        \end{aligned}
    \end{equation} 
    Drawing upon our carefully established previous estimates, the integral term $\lnm \left(\frac{V^{1-\k}}{-U}\right)^{\tau/3}\nabla^5\Phi_1\rnm^2_{\LT^2(\R_{U,V})}$ is effectively bounded by $o_{\epsilon_1}(1)\epsilon^{2-2\delta}V^a$; moreover, for the third term, we explicitly exploit the limiting geometric property that $\frac{-U}{V^{1-\k}}=o_{\epsilon_1}(1)$, which formally gives
    \begin{equation}
        \begin{aligned}
            &\lnm \left(\frac{-U}{V^{1-\k}}\right)^{\tau/3}\Omega F_1\rnm^2_{\LT^2(\R_{U,V})}
            \lesssim  o_{\epsilon_1}(1)\lnm \Omega\nabla^5\Phi_2,\Omega\nabla^6\varphi_2\rnm^2_{\LT^2(\R_{U,V})}+o_{\epsilon_1}(1)\epsilon^{2-2\delta}V^a.
        \end{aligned}
    \end{equation}
    To effectively bound the extensive collection of non-top-order terms collectively denoted by $F_2$ and $G_2$, we seamlessly invoke standard Sobolev embedding inequalities in conjunction with our prior pointwise estimates, and we get
    \begin{equation}
        \begin{aligned}
            &\lnm \Omega F_2,\Omega G_2\rnm^2_{\LT^2(\R_{U,V})}\\
            \lesssim & \int_{U_0(V)}^U \int_{V_0(U^\prime)}^V \frac{1}{(V^\prime)^{2-\k}}\left(\frac{(V^\prime)^{1-\k}}{-U^\prime}\right)^{N(5)\tau}\epsilon^{2-2\delta} (V^\prime)^a\\
            \lesssim & \epsilon^{2-2\delta}\int_{V_0(U)}^V \left(\frac{-U_0(V^\prime)}{(V^\prime)^{1-\k}}\right)^{1-N(5)\tau}(V^\prime)^{a-1}\lesssim o_{\epsilon_1}(1)\epsilon^{2-2\delta} V^a.
        \end{aligned}
    \end{equation}
    Finally, we rigorously estimate the delicate composite term $G_1$. For the specific components $\varphi\nabla^5\Phi_1$ and $\varphi\nabla^6\varphi_1$, their required bounds follow effortlessly from their vastly improved $\LT(\R)$ norm estimates. Subsequently, we can estimate the remaining problematic term $\Omega\chih\nabla^6(\Omega^{-1}\chibh)$ and obtain 
    \begin{equation}
        \begin{aligned}
            & \lnm \Omega \cdot \Omega\chih\nabla^6(\Omega^{-1}\chibh)\rnm^2_{\LT^2(\R_{U,V})}\\
            \lesssim & o_\epsilon(1)\lnm \Omega \nabla^6(\Omega^{-1}\chibh)\rnm^2_{\LT^2(\R_{U,V})}\\
            \lesssim & o_\epsilon(1)\lnm \Omega \nabla^6(\Omega^{-1}\tr\chib)\rnm^2_{\LT^2(\R_{U,V})}+o_\epsilon(1)\lnm \Omega \nabla^5(\Omega^{-1}\betab^r)\rnm^2_{\LT^2(\R_{U,V})}+o_{\epsilon_1}(1)\epsilon^{2-2\delta} V^a\\
            \lesssim & o_\epsilon(1)\lnm \Omega \nabla^5(\Omega^{-1}\betab^r)\rnm^2_{\LT^2(\R_{U,V})}+\epsilon^{2-2\delta} V^a.
        \end{aligned}
    \end{equation}
For $o_{\epsilon}(1)$ small, the first term will be entirely absorbed into the remaining beneficial negative bulk term $-\frac{1}{8}\lnm \Omega\nabla^5\Phi_2\rnm^2_{\LT^2(\R_{U,V})}$, thus definitively concluding the proof.

\end{proof}

Using Codazzi equation for $\chib$, we can obtain the estimates for $\Omega^{-1}\chibh$.
\begin{corollary}
    For $\Omega^{-1}\chibh$, the below estimate holds:
    \begin{equation}
        \lnm \nabla^6(\Omega^{-1}\chibh)\rnm^2_{\LT^2(H_U^V)}\lesssim \epsilon^{2-2\delta} V^a.
    \end{equation}
\end{corollary}
\begin{proof}
    Standard application of elliptic estimates readily confirms that  
    \begin{equation}
        \lnm \nabla^6(\Omega^{-1}\chibh)\rnm^2_{\LT^2(H_U^V)}\lesssim \lnm \nabla^6(\Omega^{-1}\tr\chib),\nabla^5(\Omega^{-1}\betab^r)\rnm^2_{\LT^2(H_U^V)}+\sum_{i\leq 5}\lnm \nabla^i(\Omega^{-1}\chib,\etab)\rnm^2_{\LT^2(H_U^V)}.
    \end{equation}
\end{proof}
As a final, capstone step in this analytical sequence, we rigorously improve the outstanding critical bootstrap assumptions specifically constraining $\eta$ and $\Omega\chi$.
\begin{proposition}
    The following estimates hold for the relevant quantities:
    \begin{equation}
        \lnm \nabla^6(\eta,\Omega\chi)\rnm^2_{\LT^2(\Hb_V^U)}\lesssim \epsilon^{2-2\delta} V^a.
    \end{equation}
\end{proposition}
\begin{proof}
    Analytically, it is entirely sufficient to just estimate $\nabla^6(\Omega\tr\chi)$ along with $\nabla^5\mu$ (where $\mu=-\dv\eta+K^r$), as the full proposition then drops out immediately from standard elliptic estimates. In what follows, we will rigorously prove the bound below for any chosen $\lambda\in(0,\tau/2)$: 
    \begin{equation}
        \lnm \nabla^6(\eta,\Omega\chi)\rnm^2_{\LT^2(\Hb_V^U)}\lesssim \lnm \Omega^{-1}\nabla^6(\eta,\Omega\chi)\left(\frac{V^{1-\k}}{-U}\right)^\lambda\rnm^2_{\LT^2(\Hb_V^U)}\lesssim \epsilon^{2-2\delta} V^a.
    \end{equation}
    We initiate this specific analysis by carefully considering $\Omega\tr\chi$, which strictly obeys the following nonlinear equation:
    \begin{equation}
        \Omega\nabla_4\nabla^6(\Omega\tr\chi)\sim \sum_{i=0}^6\nabla^{i}\varphi_1\nabla^{6-i}\varphi_1,
    \end{equation}
    with $\varphi_1\in\{\Omega\chi,\Omega\omega,\Omega e_4\phi\}$.
    Therefore this yields 
    \begin{equation}
        \begin{aligned}
            &\lnm \Omega^{-1}\nabla^6(\Omega\tr\chi)\left(\frac{V^{1-\k}}{-U}\right)^\lambda\rnm^2_{\LT^2(\Hb_V^U)}\\
            \lesssim & \epsilon^{2-2\delta}V^a+\lnm \nabla^6\varphi_1\left(\frac{V^{1-\k}}{-U}\right)^\lambda\rnm^2_{\LT^2(\R_{U,V})}+\int_{U_0(V)}^U\int_{V_0(U^\prime)}^V\frac{1}{(V^\prime)^{2-\k}}\epsilon^{2-2\delta}(V^\prime)^a\left(\frac{(V^\prime)^{1-\k}}{-U^\prime}\right)^{C\tau}\\
            \lesssim & \epsilon^{2-2\delta} V^a.
        \end{aligned}
    \end{equation}

    Turning our analytical attention to $\mu$, we confront the evolution equation
\begin{equation}
    \Omega\nabla_4\mu+\Omega\tr\chi\mu+\Omega\chih\cdot\nabla\eta = \varphi\Phi_1+\varphi^3+\varphi \nabla(\Omega\tr\chi,\etab),
\end{equation}
and applying the exact commutation formula subsequently reveals
\begin{equation}
    \begin{aligned}
        \Omega\nabla_4\nabla^5\mu+&\frac{7}{2}\Omega\tr\chi\mu+\Omega\chih\cdot\nabla^6\eta\\
        =&\nabla^{5}(\varphi\Phi_1)+\varphi\nabla^6(\Omega\tr\chi,\etab)+\sum_{\substack{i_1+i_2+i_3=6\\i_1,i_3\leq 5}}\nabla^{i_1}\varphi^{i_2+1}\nabla^{i_3}\varphi.
    \end{aligned}
\end{equation}
Since 
\begin{equation*}
    1+a+(2\lambda-\tau)(1-\k)-2S(\nabla^5\mu)-7 V\cdot\Omega\tr\chi \leq -1/2,
\end{equation*}
we can safely employ a specialized energy estimate adapted for the degenerate pair $(0,\nabla^5\mu)$:
\begin{equation*}
    \begin{aligned}
        &\lnm \Omega^{-1}\nabla^5\mu\left(\frac{V^{1-\k}}{-U}\right)^\lambda\rnm^2_{\LT^2(\Hb_V^U)}\\
        \lesssim & \epsilon^{2-2\delta}V^a-\frac{1}{4} \lnm \nabla^5\mu\left(\frac{V^{1-\k}}{-U}\right)^\lambda\rnm^2_{\LT^2(\R_{U,V})}+C(B)\epsilon^{2-2\delta}\lnm \nabla^6\eta\left(\frac{V^{1-\k}}{-U}\right)^\lambda\rnm^2_{\LT^2(\R_{U,V})}\\
        &+\lnm \Omega\varphi \nabla^5\Phi_1,\Omega\varphi\nabla^6\etab\rnm^2_{\LT^2(\R_{U,V})}+C(B)\int_{U_0(V)}^U\int_{V_0(U^\prime)}^V (V^\prime)^{\k-2} \epsilon^{2-2\delta} {(V^\prime)}^a\left(\frac{{(V^\prime)}^{1-\k}}{-U^\prime}\right)^{C\tau}dU^\prime dV^\prime\\
        \lesssim & \epsilon^{2-2\delta}V^a-\frac{1}{4} \lnm \nabla^5\mu\left(\frac{V^{1-\k}}{-U}\right)^\lambda\rnm^2_{\LT^2(\R_{U,V})}+C(B)\epsilon^{2-2\delta}\lnm \nabla^6\eta\left(\frac{V^{1-\k}}{-U}\right)^\lambda\rnm^2_{\LT^2(\R_{U,V})}.
    \end{aligned}
\end{equation*} 
From the fundamental properties of elliptic estimates, we similarly deduce
\begin{equation*}
    \lnm  \nabla^6\eta\left(\frac{V^{1-\k}}{-U}\right)^\lambda\rnm^2_{\LT^2(\R_{U,V})}\lesssim \lnm  \nabla^5(\sigma^r, K,\mu)\left(\frac{V^{1-\k}}{-U}\right)^\lambda\rnm^2_{\LT^2(\R_{U,V})},
\end{equation*}
so by explicitly selecting an $\epsilon$ that inherently satisfies $C(B)\epsilon^{2-2\delta}<1/4$, we guarantee the validity of the final requested estimate:
\begin{equation}
    \lnm \Omega^{-1}\nabla^5\mu\left(\frac{V^{1-\k}}{-U}\right)^\lambda\rnm^2_{\LT^2(\Hb_V^U)}\lesssim \epsilon^{2-2\delta}V^a.
\end{equation}
\end{proof}
\begin{remark}
    Because 
    $\lnm \Omega^{-1}\nabla^6\psi\rnm^2_{\LT^2(\Hb_V^U)}=\int_{U_0(V)}^U \frac{1}{V^{1-\k}}\lnm \nabla^6\psi\rnm^2_{\LT^2(S_{U^\prime,V})}$,
    we use $\Omega^{-1}\nabla^6\psi$ to cancel the $\Omega^2$ in the definition of $\LT^2(\Hb)$ and we can thus apply elliptic estimate without $\nabla\Omega$ appearing.
\end{remark}

\section{Asymptotic Analysis near Inner Cauchy Horizon}\label{Further Result near Cauchy Inner Horizon}
This final section extracts the geometric consequences of the estimates proved in the three regions.  The first point is that the constructed spacetime reaches the limiting inner boundary in finite affine time, so the singularity is visible from the regular region.  The second point is more precise: the scalar field develops a transverse blow-up at the inner Cauchy horizon strong enough to obstruct the expected higher regularity extensions.

Combining the local existence results and the a priori estimates, we establish the following proposition.
\begin{proposition}
    With the initial data prescribed as \eqref{334466}, the spacetime $(\mathcal{M},g,\phi)$ can be solved globally and contains a naked singularity.
\end{proposition}
\begin{proof}
    Based on the preceding estimates, we choose $\{U=-1\}$ as the asymptotically flat hypersurface. Since $D_{\Omega^{-1}e_4}(\Omega^{-1}e_4)=0$, the companion null vector field along $H_{-1}$ is given by $\Omega e_3$, and the null geodesic parameterized by the affine parameter $s$ is expressed as 
    \begin{equation}
        \gamma_{(-1,V,\theta)}(s)=\exp_{(-1,V,\theta)}\left(s\cdot (\Omega e_3)_{(-1,V,\theta)}\right).
    \end{equation}
    Given that $D_{\Omega^{-1}e_3}(\Omega^{-1}e_3)=0$, we deduce that $\partial_s=(\Omega^2)|_{U=-1}\Omega^{-1}e_3$, which implies $\partial_u s=\Omega^2(\Omega^{-2})|_{U=-1}$. For $V\geq \epsilon_1^{-1}$, the affine length satisfies the relation
    \begin{equation}
        s(0,V,\theta)-s(-1,V,\theta)= \int_{-1}^0 \Omega^2(U^\prime,V,\theta)\Omega^{-2}(-1,V,\theta).
    \end{equation}
    
    \begin{equation}
        \left|\ol{\Omega\omega}(U,V)-\ol{\Omega\omega}(U_0(V),V)\right|\lesssim \epsilon^{1/3}\int_{-1}^U\frac{1}{V^{2-\k}}\left(\frac{V^{1-\k}}{-U^\prime}\right)^{\lambda}\lesssim \epsilon^{1/3}\frac{(-U_0(V))^{1-\lambda}}{V^{1+(1-\k)(1-\lambda)}}.
    \end{equation}
Applying the chain rule yields
\begin{equation}
    \begin{aligned}
        &\partial_V\ol{\log\left(\Omega^2(U,V)\Omega^{-2}(U_0(V),V)\right)}\\
        =&-4\left(\ol{\Omega\omega}(U,V)-\ol{\Omega\omega}(U_0(V),V)\right)+\left(4\ol{\Omega\omegab}(U_0(V),V)-\ol{b\nabla\log\Omega^{-2}}(U_0(V),V)\right)\partial_V U_0.
    \end{aligned}
\end{equation}
    In the regime where $V<\epsilon_1^{-1}$, the following estimate holds: 
    \begin{equation}
        \begin{aligned}
            \left|\ol{\log\left(\Omega^2(U,V)\Omega^{-2}(U_0(V),V)\right)}\right|=&\left|\int_{V_0(U)}^V \partial_V\ol{\log\left(\Omega^2(U,V^\prime)\Omega^{-2}(U_0(V^\prime),V^\prime)\right)}dV^\prime\right|\\
            \leq & \epsilon^{1/4}\int_{V_0(U)}^V \frac{1}{V^\prime}dV^\prime\leq \epsilon^{1/4}\log\left(\frac{\epsilon_1^{-1}}{V_0(U)}\right).
        \end{aligned}
    \end{equation}
    Conversely, for $V\geq \epsilon_1^{-1}$, we deduce that 
    \begin{equation}
        \begin{aligned}
            \left|\ol{\log\left(\Omega^2(U,V)\Omega^{-2}(U_0(V),V)\right)}\right|
            \leq & \epsilon^{1/4}\int_{V_0(U)}^{\epsilon_1^{-1}} \frac{1}{V^\prime}dV^\prime+\epsilon^{1/4}\int_{\epsilon_1^{-1}}^V \frac{1}{(V^\prime)^{1+(1-\k)(1-\lambda)}}.
        \end{aligned}
    \end{equation}
    Consequently, for any $V\in (0,\infty)$, we arrive at the estimate 
    \begin{equation}
        \left|\ol{\log\left(\Omega^2(U,V)\Omega^{-2}(U_0(V),V)\right)}\right|\leq \log\left((-U)^{-\frac{\epsilon^{1/4}}{1-\k}}\exp({\epsilon^{1/4}})\right).
    \end{equation}
    Since $(\Omega^{2})^c(U,V)(\Omega^{-2})^c(U_0(V),V)=1+o(1)$, we thereby obtain the following estimate, which is uniform in $V$:
    \begin{equation}
        \frac{1}{2}(-U)^{-\frac{\epsilon^{1/4}}{1-\k}}\leq \Omega^{2}(U,V)\Omega^{-2}(U_0(V),V) \leq 2(-U)^{-\frac{\epsilon^{1/4}}{1-\k}}.
    \end{equation}
    Finally, calculating the affine length yields $s(0,V,\theta)-s(-1,V,\theta)=1+o(1)$, demonstrating that the future null infinity is incomplete.
\end{proof}

The subsequent proposition demonstrates that the $C^{1,\frac{\k+\delta}{1-\k}}$ continuity cannot be further improved near the Cauchy horizon.
\begin{proposition}
    For any fixed $\delta>0$, provided that $\epsilon_1=\epsilon_1(\delta)$ and $\epsilon=\epsilon(\epsilon_1,\delta)$ are chosen sufficiently small, we obtain the following upper bound for $(\Omega^{-1}e_3)^2\phi$:
    \begin{equation}
        V^{2-2S(\Omega^{-1}\nabla_3\Omega^{-1}\nabla_3\phi)}\left|(\Omega^{-1}e_3)^2\phi\right|^2(U,V)\lesssim \epsilon^{2-4\delta} \left(\frac{V}{(-U)^{\frac{1}{1-\k}}}\right)^{2-4\k+2\delta},
    \end{equation}
    and its blow-up rate as $U\rightarrow 0$,
    \begin{equation}
        \lim_{U\rightarrow 0}(-U)^{\frac{2-4\k-2\delta}{1-\k}} V^{2-2S(\Omega^{-1}\nabla_3\Omega^{-1}\nabla_3\phi)}\left|(\Omega^{-1}e_3)^2\phi\right|^2(U,V)=\infty.
    \end{equation}
\end{proposition}
\begin{proof}
    According to \eqref{R_I_L2_est_D3D3phi}, upon performing a gauge transformation from $\R_I$ to $\R_{II}$, the difference in $(\Omega e_3)^2\phi$ along the hypersurface $z=\epsilon_1$ remains small, satisfying:
    $$\lnm\ol{(\Omega e_3)^2\phi}\rnm^2_{\LS_1^\infty\left(S_{-(V\epsilon_1^{-1})^{1-\k},V}\right)}\lesssim \epsilon^{2-\delta}.$$ 
    By Proposition \ref{Appen_A_eq_D3D3phi} in Appendix A, the term $(\Omega^{-1}e_3)^2\phi$ satisfies the evolution equation
    \begin{equation}
        \begin{aligned}
            \Omega e_4(\Omega^{-1}e_3)^2\phi+\frac{1}{2}\Omega\tr\chi (\Omega^{-1}e_3)^2\phi-8\Omega\omega(\Omega^{-1}e_3)^2\phi=\Delta\Omega^{-1}e_3\phi+\varphi\nabla\varphi+\varphi^3+\varphi\cdot\Psi:=F.
        \end{aligned}
    \end{equation}
    Here $\varphi\in\{\etab,\Omega\chi,\Omega^{-1}\chib,\nabla\phi,\Omega^{-1}D_3\phi,\Omega D_4\phi\}$ and $\Psi\in\{\Omega^{-1}\betab^r,K^r\}$. 
    The transport estimate \eqref{R III transport estimate nabla 4} in $\R_{II}$ readily implies that on the future boundary of $\R_{II}$, the following inequality holds: 
    \begin{equation}
        \lnm \ol{(\Omega e_3)^2\phi}\rnm^2_{\LS_1^\infty\left(S_{-(\epsilon_1 V)^{1-\k},V}\right)}\lesssim \epsilon^{2-2\delta}.
    \end{equation}
    Since the background solution satisfies
    \begin{equation}      \left|\left((\Omega^{-1}e_3)^2\phi\right)^c\right|\sim \frac{1}{V^{2}}\left(\frac{V}{(-U)^{\frac{1}{1-\k}}}\right)^{1-2\k},
    \end{equation}
    we thus deduce 
    \begin{equation}
        \partial_V \left|\ol{(\Omega^{-1}e_3)^2\phi}\right|+\frac{1+2\k-\delta}{V}\left|\ol{(\Omega^{-1}e_3)^2\phi}\right|\lesssim \epsilon^{1-\delta/2}\frac{1}{V^3}\left(\frac{V}{(-U)^{1/(1-\k)}}\right)^{1-2\k},
    \end{equation}
    and upon integration, this yields 
    \begin{equation}
        \begin{aligned}
            \left|\ol{(\Omega^{-1}e_3)^2\phi}\right|(U,V)\lesssim & \epsilon^{1-\delta}V^{-1-2\k+\delta} V_0(U)^{-1+2\k-\delta}+\epsilon^{1-\delta/2}V^{-1-2\k+\delta}V_0(U)^{-\delta}(-U)^{-\frac{1-2\k}{1-\k}}\\
            \leq & \epsilon^{1-2\delta} V^{-2}\left(\frac{V^{1-\k}}{-U}\right)^{\frac{1-2\k+\delta}{1-\k}}.
        \end{aligned}
    \end{equation}
    Next, we establish the lower bound for $(\Omega^{-1}e_3)^2\phi$:
    \begin{equation*}
        \begin{aligned}
            &V^{2-2S((\Omega e_3)^2\phi)}\left|(\Omega^{-1}e_3)^2\phi\right|^2\cdot V^{2-4\k+\delta}\left(\frac{-U}{V^{1-\k}}\right)^{\frac{2-4\k}{1-\k}}\\
            \geq & \left(V^{\delta+2-2S((\Omega e_3)^2\phi)}\left|(\Omega^{-1}e_3)^2\phi\right|^2\right)\left(U,V_0(U)\right)\cdot \left(-U\right)^{\frac{2-4\k}{1-\k}}\\
            &+\frac{\delta}{2} \int_{V_0(U)}^V  {V^\prime}^{2-2S((\Omega e_3)^2\phi)}\left|(\Omega^{-1}e_3)^2\phi\right|^2\cdot {V^\prime}^{\delta-1}\left(-U\right)^{\frac{2-4\k}{1-\k}}-\frac{4}{\delta}\int_{V_0(U)}^V  {V^\prime}^{2-2S}\left|F\right|^2\cdot {V^\prime}^{\delta-1}\left(-U \right)^{\frac{2-4\k}{1-\k}}\\
            \geq & \frac{1}{C_0}V_0(U)^\delta+\frac{\delta}{2} \int_{V_0(U)}^V  {V^\prime}^{2-2S((\Omega e_3)^2\phi)}\left|(\Omega^{-1}e_3)^2\phi\right|^2\cdot {V^\prime}^{\delta-1}\left(-U\right)^{\frac{2-4\k}{1-\k}}-C(\delta) V^{\delta}\left(-U\right)^{\frac{2-4\k}{1-\k}}.
        \end{aligned}
    \end{equation*}
    Fixing $V_1$ and taking $-U$ sufficiently small such that $(-U)^{\frac{\delta+4\k-2}{1-\k}}\gg V_1^{\delta}$, we find that for $V\in[V_0(U),V_1)$, the following relation holds: 
    \begin{equation*}
         V^{2-2S((\Omega e_3)^2\phi)}\left|(\Omega^{-1}e_3)^2\phi\right|^2\cdot V^{2-4\k+\delta}\left(\frac{-U}{V^{1-\k}}\right)^{\frac{2-4\k}{1-\k}}\geq \frac{1}{C_0}\frac{V_0(U)^\delta}{V^\delta}(-U)^{-\frac{2-4\k}{1-\k}}\geq \frac{1}{C(\epsilon_1)}V^{-\delta}(-U)^{-\frac{2-4\k-\delta}{1-\k}}.
    \end{equation*}
    Taking the limit as $-U \to 0$ for a fixed $V$ yields the desired result.

\end{proof}

\begin{remark}[Exact scale-invariant normalization]
Since $S(\phi)=1$ and $S(\Omega^{-1}\nabla_3\psi)=S(\psi)-1$, we have
$S(\Omega^{-1}\nabla_3\Omega^{-1}\nabla_3\phi)=-1$.  Thus the prefactor in (8.9)--(8.10) is $V^{2-2S}=V^4$, and for fixed $V>0$ this normalization does not affect the $U\to0$ divergence.
\end{remark}

\begin{remark}[$C^{1,\alpha}$-inextendibility]
Taking square roots in (8.10) gives
\[
    (-U)^{\frac{1-2\k-\delta}{1-\k}}
    \left|(\Omega^{-1}e_3)^2\phi\right|(U,V)
    \longrightarrow\infty
\]
for each fixed $V>0$.  A $C^{1,\alpha}$ extension across $U=0$ would force
$(-U)^{1-\alpha}|(\Omega^{-1}e_3)^2\phi|$ to remain bounded along the generator.  Comparing the exponents gives
\[
    1-\alpha=\frac{1-2\k-\delta}{1-\k},\qquad
    \alpha=\frac{\k+\delta}{1-\k}.
\]
Hence the Cauchy horizon is not $C^{1,\frac{\k+\delta}{1-\k}}$ extendible, and, after taking $\delta>0$ arbitrarily small, this is the stated $C^{1,\frac{\k}{1-\k}+}$ obstruction.
\end{remark}

\newpage

\appendix
\section{Derivation of equations}\label{appen_derive_eq}

This appendix records the structural identities used earlier: the self-similar system on the initial hypersurface, the difference equation for $\Omega\tr\chib$, and the transport identities for the scalar and torsion quantities.

{We begin with the self-similar identities that organize the later appendix computations.}

\begin{proposition}\label{Self similar equation near v equal 0}
    Along the initial hypersurface $u=-1$ for $v\geq 0$, the assumed self-similar spacetime geometry inherently satisfies the following coupled system of equations:
\begin{equation}\label{Equation_set}
    \left\{
    \begin{aligned}
        v\partial_v(\Omega^{-1}\chih_{AB})=&-\Lie_b(\Omega^{-1}\chih_{AB})+\l\frac{1}{2}\dv b+4v\Omega\omega-2\Lie_b\log\Omega\r\Omega^{-1}\chih_{AB}\\
        &+\frac{1}{2}(\nabla\hat\otimes b)^C_{\  A}(\Omega^{-1}\chih)_{BC}+ \frac{1}{2}(\nabla\hat\otimes b)^C_{\  B}(\Omega^{-1}\chih)_{AC}  +(\nabla\hat\otimes\eta)_{AB}+(\eta\hat\otimes\eta)_{AB}\\
        &-\frac{1}{2}(\nabla\hat\otimes b)_{AB}(\Omega^{-1}\tr\chi)+D_A\phi D_B\phi-\frac{1}{2}g_{AB}D_C\phi D^C\phi\\
        &+v\Omega^{2}\ \Omega^{-1}\tr\chi\ \Omega^{-1}\chih_{AB}+v\Omega^2\Omega^{-1}\chih_{AC}\Omega^{-1}\chih^{C}_{\ B}+v\Omega^2\Omega^{-1}\chih_{BC}\Omega^{-1}\chih^{C}_{\ A},\\
        v\partial_v(\Omega^{-1}\tr\chi)=&-\Lie_b(\Omega^{-1}\tr\chi)+\Omega^{-1}\tr\chi(1+\k-\dv b-2\Lie_b\log\Omega)-2K+2\dv\eta+2|\eta|^2\\
        &+D_A\phi D^A\phi-v\Omega^2\ (\Omega^{-1}\tr\chi)^2,\\
        v\partial_v(v^\k\Omega^2)=&v^\k\Omega^2(\k-4v\Omega\omega),\\
        v\partial_v b^A=&2v\Omega^2(\etab^A-\eta^A),\\
    v\partial_vg_{AB}=&2v\Omega^2\ \Omega^{-1}\chi_{AB},\\
        v\partial_v(v\Omega\omega)=&-b^A\nabla_A(v\Omega\omega)+v\Omega^2\l 2\nabla\log\Omega\cdot(\nabla\log\Omega-\eta)-\eta\cdot(2\nabla\log\Omega-\eta)\\
        &+\frac{1}{2}|\eta|^2+\frac{1}{2}\rho+\frac{1}{6}\Omega D_3\phi \ \Omega^{-1}D_4\phi+\frac{1}{12}D_A\phi D^A\phi\r,\\
    v\Lie_v\zeta_A=&-\Lie_b\zeta_A-\left(-\frac{2}{-u}+\dv b+\frac{v}{-u}\Omega\tr\chi\right)\zeta_A\\
            &-2(\nabla\hat\otimes b +\frac{v}{-u}\Omega\chih)_A^C\nabla_C\log\Omega-\frac{\k}{-2u}-\Lie_b\log\Omega-2\left(\frac{v}{-u}\Omega\omega-\frac{\k}{-4u}\right)\\
            &-\frac{1}{2}\dv\left(\nabla\hat\otimes b\right)-\dv\left(-\frac{v}{-u}\Omega\chih\right)+\frac{1}{2}\nabla\dv b+\frac{1}{2}\nabla\left(\frac{v}{-u}\tr\chi\right),\\
        v\partial_v(\Omega^{-1}D_4\phi)=&\l 4v\Omega\omega-\frac{1}{2}v\Omega^2\ \Omega^{-1}\tr\chi+\frac{1}{2}\dv b-b\cdot\nabla\log\Omega\r\cdot\Omega^{-1}D_4\phi-b\cdot\nabla(\Omega^{-1}D_4\phi)\\
        &+\Delta\phi+2(4\nabla\log\Omega-3\eta)\cdot\nabla\phi-\frac{1}{2}(\Omega D_3\phi+v\Omega D_4\phi)\Omega^{-1}\tr\chi,\\
    v\partial_v (\Omega D_3\phi)=&v\Omega^2\left(-\frac{1}{2}\Omega^{-1}D_4\phi\left(v\Omega^2\ \Omega^{-1}\tr\chi-2+\dv b\right)+\Delta \phi-\Omega^{-1}\tr\chi \Omega D_3\phi+2\etab\cdot\nabla\phi\right),\\v\partial_v\nabla_A\phi=&2v\Omega^2(\nabla_A\log\Omega\cdot \Omega^{-1}D_4\phi+\nabla_A(\Omega^{-1}D_4\phi)).
\end{aligned}
\right.
\end{equation}
\end{proposition}
\begin{proof}
    To derive the evolution equation for $\zeta$, we start directly from the fundamental $\nabla_3\etab$ equation:
    \begin{equation}
        \begin{aligned}
            \Omega\nabla_3\zeta+&\frac{3}{2}\Omega\tr\chib\zeta+\Omega\chibh\zeta=\Omega\nabla_3\nabla\log\Omega-\Omega\dv\chibh+\frac{1}{2}\Omega\nabla\tr\chib\\
            =&-\Omega\chib\nabla\log\Omega-2\nabla(\Omega\omegab)-\dv\left(\Omega\chibh\right)+\frac{1}{2}\nabla\left(\Omega\tr\chib\right)-\nabla\log\Omega\cdot\chibh+\frac{1}{2}\nabla\log\Omega\cdot\tr\chib\\
            =&-2\Omega\chibh\nabla\log\Omega-2(\Omega\omegab)-\dv\left(\Omega\chibh\right)+\frac{1}{2}\nabla\left(\Omega\tr\chib\right).
        \end{aligned}
    \end{equation}
    Exploiting the essential scaling symmetry $\zeta_A(u,v)=\zeta_A(-1,v/(-u))$, we deduce
    \begin{equation}
        \begin{aligned}
            \left(v\Lie_v+\Lie_b\right)\zeta_A+&\left(-\frac{2}{-u}+\dv b+\frac{v}{-u}\Omega\tr\chi\right)\zeta_A\\
            =&-2(\nabla\hat\otimes b +\frac{v}{-u}\Omega\chih)_A^C\nabla_C\log\Omega-\frac{\k}{-2u}-\Lie_b\log\Omega-2\left(\frac{v}{-u}\Omega\omega-\frac{\k}{-4u}\right)\\
            &-\frac{1}{2}\dv\left(\nabla\hat\otimes b\right)-\dv\left(-\frac{v}{-u}\Omega\chih\right)+\frac{1}{2}\nabla\dv b+\frac{1}{2}\nabla\left(\frac{v}{-u}\tr\chi\right).
        \end{aligned}
    \end{equation}
    \providecommand{\dv}{\operatorname{div}}
\providecommand{\Lie}{\mathcal L}

We proceed to work on other equations.
On $u=-1$, the algebraic self-similar relations are
\[
\Omega\tr\chib
 =
 -2+\dv b+v\Omega^2\Omega^{-1}\tr\chi,
\]
\[
\Omega\chibh_{AB}
 =
 \frac12(\nabla\widehat\otimes b)_{AB}
 +v\Omega^2(\Omega^{-1}\hat\chi)_{AB},
\]
\[
\Omega\underline{\omega}
 =
 v\Omega\omega-\frac12\Lie_b\log\Omega,
 \qquad
\underline{\eta}=2\nabla\log\Omega-\eta.
\]

First, we derive the equation for $\Omega^{-1}\hat\chi_{AB}$.  From the
$\nabla_3\hat\chi$ equation in \((2.5)\),
\[
\nabla_3\hat\chi_{AB}
+\frac12\tr\chib\,\hat\chi_{AB}
 =
(\nabla\widehat\otimes\eta)_{AB}
+2\underline{\omega}\hat\chi_{AB}
-\frac12\tr\chi\,\chibh_{AB}
+(\eta\widehat\otimes\eta)_{AB}
+D_A\phi D_B\phi-\frac12g_{AB}D_C\phi D^C\phi.
\]
Equivalently,
\[
\begin{aligned}
\Omega\nabla_3(\Omega^{-1}\hat\chi_{AB})
={}&
(\nabla\widehat\otimes\eta)_{AB}
+(\eta\widehat\otimes\eta)_{AB}
+D_A\phi D_B\phi-\frac12g_{AB}D_C\phi D^C\phi        \\
&+\left(4\Omega\underline{\omega}
       -\frac12\Omega\tr\chib\right)
       \Omega^{-1}\hat\chi_{AB}
-\frac12(\Omega^{-1}\tr\chi)\Omega\chibh_{AB}.
\end{aligned}
\]
Using
\[
\begin{aligned}
\Omega\nabla_3(\Omega^{-1}\hat\chi_{AB})
={}&
-\Omega^{-1}\hat\chi_{AB}
+v\partial_v(\Omega^{-1}\hat\chi_{AB})
+\Lie_b(\Omega^{-1}\hat\chi_{AB})       \\
&-\Omega\chib^{C}{}_{A}(\Omega^{-1}\hat\chi)_{CB}
 -\Omega\chib^{C}{}_{B}(\Omega^{-1}\hat\chi)_{AC},
\end{aligned}
\]
and substituting the identities above, we obtain
\[
\begin{aligned}
v\partial_v(\Omega^{-1}\hat\chi_{AB})
={}&
-\Lie_b(\Omega^{-1}\hat\chi_{AB})
+\left(
\frac12\dv b+4v\Omega\omega-2\Lie_b\log\Omega
\right)\Omega^{-1}\hat\chi_{AB}        \\
&+\frac12(\nabla\widehat\otimes b)^C{}_{A}
   (\Omega^{-1}\hat\chi)_{BC}
 +\frac12(\nabla\widehat\otimes b)^C{}_{B}
   (\Omega^{-1}\hat\chi)_{AC}          \\
&+(\nabla\widehat\otimes\eta)_{AB}
 +(\eta\widehat\otimes\eta)_{AB}
 -\frac12(\nabla\widehat\otimes b)_{AB}
   (\Omega^{-1}\tr\chi)                \\
&+D_A\phi D_B\phi
 -\frac12g_{AB}D_C\phi D^C\phi          \\
&+v\Omega^2(\Omega^{-1}\tr\chi)(\Omega^{-1}\hat\chi_{AB})
 +v\Omega^2(\Omega^{-1}\hat\chi_{AC})
           (\Omega^{-1}\hat\chi)^C{}_{B}       \\
&+v\Omega^2(\Omega^{-1}\hat\chi_{BC})
           (\Omega^{-1}\hat\chi)^C{}_{A}.
\end{aligned}
\]

Next, we derive the equation for $\Omega^{-1}\tr\chi$.  From the
$\nabla_3\tr\chi$ equation in \((2.5)\), 
we write
\[
\begin{aligned}
\Omega\nabla_3(\Omega^{-1}\tr\chi)
={}&
\left(
4\Omega\underline{\omega}
-\Omega\tr\chib
\right)\Omega^{-1}\tr\chi       \\
&-2K+2\dv\eta+2|\eta|^2+D_A\phi D^A\phi.
\end{aligned}
\]
Using self-similarity for $\Omega^{-1}\tr\chi$ and the identity
\[
v\partial_v(v^\kappa\Omega^2)
 =
v^\kappa\Omega^2(\kappa-4v\Omega\omega),
\]
the preceding equation becomes
\[
\begin{aligned}
v\partial_v(\Omega^{-1}\tr\chi)
={}&
-\Lie_b(\Omega^{-1}\tr\chi)
+\left(
1+\kappa-\dv b-2\Lie_b\log\Omega
\right)\Omega^{-1}\tr\chi        \\
&-2K+2\dv\eta+2|\eta|^2+D_A\phi D^A\phi
-v\Omega^2(\Omega^{-1}\tr\chi)^2 .
\end{aligned}
\]

We now derive the equation for $v\Omega\omega$.  From the
$\nabla_3\omega$ equation in \((2.6)\),
\[
\nabla_3\omega
 =
2\omega\underline{\omega}
-\eta\cdot\underline{\eta}
+\frac12|\underline{\eta}|^2
-\frac12K-\frac18\tr\chi\tr\chib
+\frac14\hat\chi\cdot\chibh
+\frac14D_3\phi D_4\phi
+\frac14D_A\phi D^A\phi,
\]
we derive
\[
\begin{aligned}
\Omega\nabla_3(\Omega\omega)
 =
\Omega^2\Big(
&-\eta\cdot\underline{\eta}
+\frac12|\underline{\eta}|^2
+\frac12\rho
+\frac16D_3\phi D_4\phi
+\frac1{12}D_A\phi D^A\phi
\Big).
\end{aligned}
\]
Since $v\Omega\omega$ is self-similar of degree zero, we have
\[
\Omega\nabla_3(v\Omega\omega)
 =
v\partial_v(v\Omega\omega)
+b^A\nabla_A(v\Omega\omega),
\]
and hence we deduce the equation
\[
\begin{aligned}
v\partial_v(v\Omega\omega)
={}&
-b^A\nabla_A(v\Omega\omega)        \\
&+v\Omega^2\Big(
2\nabla\log\Omega\cdot(\nabla\log\Omega-\eta)
-\eta\cdot(2\nabla\log\Omega-\eta)      \\
&\qquad\qquad
+\frac12|\eta|^2
+\frac12\rho
+\frac16\Omega D_3\phi\,\Omega^{-1}D_4\phi
+\frac1{12}D_A\phi D^A\phi
\Big).
\end{aligned}
\]

Finally, we derive the equation for $\Omega^{-1}D_4\phi$.  The wave
equation in \((2.8)\), written as a $\nabla_3D_4\phi$ equation, gives 
\[
\begin{aligned}
\Omega\nabla_3(\Omega^{-1}D_4\phi)
={}&
\left(
4\Omega\underline{\omega}
-\frac12\Omega\tr\chib
\right)\Omega^{-1}D_4\phi
+\Delta\phi                                      \\
&-\frac12(\Omega^{-1}\tr\chi)\Omega D_3\phi
+2\underline{\eta}\cdot\nabla\phi .
\end{aligned}
\]
Using self-similarity for $\Omega^{-1}D_4\phi$, together with
$\underline{\eta}=2\nabla\log\Omega-\eta$ and the preceding algebraic
relations, this becomes
\[
\begin{aligned}
v\partial_v(\Omega^{-1}D_4\phi)
={}&
\left(
4v\Omega\omega
-\frac12v\Omega^2\Omega^{-1}\tr\chi
+\frac12\dv b
-b\cdot\nabla\log\Omega
\right)\Omega^{-1}D_4\phi        \\
&-b\cdot\nabla(\Omega^{-1}D_4\phi)
+\Delta\phi
+2(4\nabla\log\Omega-3\eta)\cdot\nabla\phi       \\
&-\frac12\left(\Omega D_3\phi+v\Omega D_4\phi\right)
       \Omega^{-1}\tr\chi .
\end{aligned}
\]
\end{proof}

{We next isolate the second-order subtraction formula for the incoming expansion.}

\begin{proposition}\label{Appen nabla 3 tr chib}
    The following relation holds:
    \begin{equation}
        \begin{split}
            \Omega\nabla_3\ddfl{\Omega\tr\chib}=& -\ddfl{b}\cdot\nabla\left(\ol{\Omega\tr\chib}(0)+v\Lie_v\ol{\Omega\tr\chib}(0)\right)-v\Lie_v b(0)\cdot\nabla\left(v\Lie_v\ol{\Omega\tr\chib}(0)\right)\\
            &-4\ddfl{\Omega\omegab\Omega\tr\chib}-\ddfl{\left|\Omega\chib\right|^2}-\ddfl{(\Omega D_3\phi)^2}.
        \end{split}
    \end{equation}
\end{proposition}
\begin{proof}
We compute directly:
    \begin{equation}
        \begin{aligned}
            &\Omega\nabla_3\ddfl{\Omega\tr\chib}\\
            =&\ol{\Omega\nabla_3(\Omega\tr\chib)}-(\ol{\Omega\nabla_3(\Omega\tr\chib)})_0-\df{b}\cdot\left(\nabla\ol{\Omega\tr\chib}\right)_0\\
            &-v\left(\Omega e_3\partial_v\ol{\Omega\tr\chib}\right)_0-v\df{b}\cdot\nabla(\Lie_v\ol{\Omega\tr\chib})_0\\
           =&\ol{\Omega\nabla_3(\Omega\tr\chib)}-(\ol{\Omega\nabla_3(\Omega\tr\chib)})_0-\ddfl{b}\cdot\left(\nabla\ol{\Omega\tr\chib}\right)_0-v(\Lie_v b)_0 \cdot\left(\nabla\ol{\Omega\tr\chib}\right)_0\\
            &-v\left(\Lie_v\ol{\Omega\nabla_3(\Omega\tr\chib)}\right)_0+v\left(\Lie_v b\cdot \nabla\ol{\Omega\tr\chib}\right)_0\\
            &-v\ddfl{b}\cdot\nabla\left(\Lie_v(\ol{\Omega\tr\chib})\right)_0-v^2\left(\Lie_v b\cdot\nabla\partial_v\ol{\Omega\tr\chib}\right)_0\\
            =& \ddfl{\Omega\nabla_3(\Omega\tr\chib)}-\ddfl{b}\cdot\nabla\left((\ol{\Omega\tr\chib})_0+v(\Lie_v\ol{\Omega\tr\chib})_0\right)-v^2(\Lie_v b)_0\nabla\left(\Lie_v\ol{\Omega\tr\chib}\right)_0.
        \end{aligned}
    \end{equation}
    In the computation, we have used the fact that $$\Omega e_3(\Omega\tr\chib)^c=(\Omega e_3)_0(\Omega\tr\chib)^c=(\Omega e_3(\Omega\tr\chib))^c.$$
    
\end{proof}
To use this equation, we make it more specific:
\begin{corollary}\label{Appen nabla 3 tr chib cor}
    If we replace the known quantities, which are $(\cdot)_0,\ol\cdot,(\cdot)^c, \df{(\cdot)^c}$ by suitable $\epsilon, \z$ terms, and ignore the $u$-signature, that readers can check that the consistency of the signature in the equations, for convenience, we have schematic equatioin
    \begin{equation}
        \begin{aligned}
            &\Omega\nabla_3\ddfl{\Omega\tr\chib}+(\Omega\tr\chib+4\Omega\omegab)\ddfl{\Omega\tr\chib}\\
            \sim & \Omega\chibh\cdot\ddfl{\Omega\chibh}+\Omega e_3\phi\ddfl{\Omega e_3\phi}+(\Omega\tr\chib)^c\ddfl{\Omega\omegab}+O(\epsilon)\ddfl{b}\\
            &+\sum_{\varphi\in\{\Omega\chib,\Omega\omegab,\Omega e_3\phi\}}\left((\ol{\varphi}+O(v))\ddfl{\varphi}+\ddf{\varphi^c}\ol{\varphi}\right)+O\left(\epsilon\z^2\right).
        \end{aligned}
    \end{equation}
\end{corollary}
\begin{proof}
    From the proposition above, we have
    \begin{equation}
        \begin{aligned}
            \Omega\nabla_3\ddfl{\Omega\tr\chib}\sim & O(\epsilon)\cdot\ddfl{b}+O\left(\epsilon^2\z^2\right)\\
            &-4\ddfl{\Omega\omegab\Omega\tr\chib}-\ddfl{\left|\Omega\chibh\right|^2}-\frac{1}{2}\ddfl{(\Omega\tr\chib)^2}-\ddfl{(\Omega D_3\phi)^2}.
        \end{aligned}
    \end{equation}
    For $\psi_i\in\{\Omega\omegab,\Omega\chib,\Omega D_3\phi\}$, the $\LS_1(S)$ norm of $\left[\psi_i^c\right]_0^v$ is controlled by $O\z$. Using Lemma \ref{Product_rule_for_difference}, we have 
    \begin{equation}
        \begin{aligned}
            \ddfl{\psi_1\psi_2}=&\psi_1\ddfl{\psi_2}-\ol{\psi_1}\ddfl{\psi_2}+\ddfl{\psi_1}\ol{\psi_2}+\ddf{\psi_1^c}\ol{\psi_2}+(\psi_2^c)_0\ddfl{\psi_1}+\ol{\psi_1}\ddf{\psi_2^c}\\
            &+O(v)\ddfl{\psi_1}+O(v)\ddfl{\psi_2}+O\left(\epsilon\z^2\right).
        \end{aligned}
    \end{equation}
    Slightly changing the expression yields:
    \begin{equation}
        \begin{aligned}
            &\ddfl{\psi_1\psi_2}-\psi_1\ddfl{\psi_2}-\ddfl{\psi_1}\psi_2\\
            \sim & (\ol{\psi_2}+O(v))\ddfl{\psi_1}+(\ol{\psi_1}+O(v))\ddfl{\psi_2}+\ddf{\psi_2^c}\ol{\psi_1}+\ddf{\psi_1^c}\ol{\psi_2}+O\left(\epsilon\z^2\right).
        \end{aligned}
    \end{equation}
    The proof is finished after plugging this into Proposition \ref{Appen nabla 3 tr chib}.
\end{proof}

{We now record the commuted scalar-field equation needed for the Cauchy-horizon analysis.}

\begin{proposition}\label{Appen_A_eq_D3D3phi}
    We explicitly derive the following higher-order evolution equation for the repeated derivative $(\Omega^{-1}e_3)^2\phi$:
    \begin{equation}
        \begin{aligned}
            \Omega e_4(\Omega^{-1}e_3)^2\phi+\frac{1}{2}\Omega\tr\chi (\Omega^{-1}e_3)^2\phi-8\Omega\omega(\Omega^{-1}e_3)^2\phi=\Delta\Omega^{-1}e_3\phi+\varphi\nabla\varphi+\varphi^3+\varphi\cdot\Psi.
        \end{aligned}
    \end{equation} 
    Here, the placeholder $\varphi$ systematically denotes any geometric factor taking values in the prescribed set $\{\etab,\Omega\chi,\Omega^{-1}\chib,\nabla\phi,\Omega^{-1}D_3\phi,\Omega D_4\phi\}$, and analogously $\Psi\in\{\Omega^{-1}\betab^r,K^r\}$.
\end{proposition}
\begin{proof}
    By actively commuting the differential operator $\Omega^{-1}e_3$ through the established $\Omega e_4(\Omega^{-1}e_3)\phi$ evolution equation, we explicitly obtain
    \begin{equation}
        \begin{aligned}
            \Omega e_4(\Omega^{-1}e_3)^2\phi=& 4\Omega\omega\cdot(\Omega^{-1}e_3)^2\phi+2(\etab-\eta)\cdot\nabla(\Omega^{-1}e_3\phi)+\Omega^{-1}e_3(\Delta\phi)-\frac{1}{2}\Omega\tr\chi\cdot (\Omega^{-1}e_3)^2\phi\\
            &-\frac{1}{2}\Omega^{-1}e_3\phi\cdot\Omega^{-1}\nabla_3(\Omega\tr\chi)-\frac{1}{2}\Omega^{-1}\nabla_3(\Omega^{-1}\tr\chib)\cdot\Omega e_4\phi-\frac{1}{2}\Omega^{-1}\tr\chib\cdot(\Omega^{-1}e_3(\Omega e_4)\phi)\\
            &+2\Omega^{-1}\nabla_3\etab\cdot\nabla\phi+2\etab\cdot\Omega^{-1}\nabla_3\nabla\phi+4\Omega^{-1}e_3(\Omega\omega)\Omega^{-1}e_3\phi+4\Omega\omega(\Omega^{-1}e_3)^2\phi.
        \end{aligned}
    \end{equation}
    Our primary analytical objective here is distinctly tracking the problematic linear $\eta$ terms; therefore, we can safely and directly ignore all lower-order nonlinear combinations of the schematic forms $\varphi\nabla\varphi$, $\varphi^3$, and $\varphi\cdot\Psi$, precisely as aggregated in the formal proposition statement. Using the rigorous commutation identity \eqref{Exact_Comm_Formula}, this yields
    \begin{equation}
        \begin{aligned}
            \Omega^{-1}\nabla_3\Delta\phi=&\Omega^{-1}[\nabla_3,\dv]\nabla\phi+\Omega^{-1}\dv\nabla_3\nabla\phi\\
            =&-\frac{1}{2}\Omega^{-1}\tr\chib\Delta\phi-\Omega^{-1}\chibh\cdot\nabla^2\phi+\Omega^{-1}\betab^r\cdot\nabla\phi+(\eta+\etab)\cdot\Omega^{-1}\nabla_3\nabla\phi-\eta\cdot\Omega\chibh\cdot\nabla\phi\\
            &+\frac{1}{2}\Omega^{-1}\tr\chib\eta\cdot\nabla\phi+\dv(\Omega^{-1}\nabla_3\nabla\phi)\\
            =& -2\eta\cdot\Omega^{-1}\chibh\cdot\nabla\phi+(|\eta|^2+2\eta\cdot\etab)\Omega^{-1}D_3\phi+\dv\eta\cdot\Omega^{-1}D_3\phi+2\eta\cdot\nabla\Omega^{-1}D_3\phi\\
            &+\varphi\nabla\varphi+\varphi^3+\varphi\cdot\Psi.
        \end{aligned}
    \end{equation}
    By algebraically expanding the terms situated on the second line, the expression formally reduces to
    \begin{equation}
        \begin{aligned}
            &-\frac{1}{2}\Omega^{-1}e_3\phi\cdot\Omega^{-1}\nabla_3(\Omega\tr\chi)-\frac{1}{2}\Omega^{-1}\nabla_3(\Omega^{-1}\tr\chib)\cdot\Omega e_4\phi-\frac{1}{2}\Omega^{-1}\tr\chib\cdot(\Omega^{-1}e_3(\Omega e_4)\phi)\\
            =&-\Omega^{-1}D_3\phi(\dv\eta+|\eta|^2)-\eta\cdot\nabla\phi\Omega^{-1}\tr\chib+\varphi\nabla\varphi+\varphi^3+\varphi\cdot\Psi.
        \end{aligned}
    \end{equation}
    Systematically substituting the full set of coupled null structure equations into the right-hand side, the remaining algebraic components cleanly yield
    \begin{equation}
      \begin{aligned}
          &-2\eta\cdot\nabla(\Omega^{-1}e_3\phi)+2\Omega^{-1}\nabla_3\etab\cdot\nabla\phi+2\etab\cdot\Omega^{-1}\nabla_3\nabla\phi+4\Omega^{-1}e_3(\Omega\omega)\Omega^{-1}e_3\phi\\
          =&\Omega^{-1}\tr\chib\eta\cdot\nabla\phi+2\eta\cdot\Omega^{-1}\chibh\cdot\nabla\phi+2\eta\cdot\etab\Omega^{-1}D_3\phi-2\eta\cdot\nabla(\Omega^{-1}D_3\phi)-4\eta\cdot\etab\Omega^{-1}D_3\phi\\
          &+\varphi\nabla\varphi+\varphi^3+\varphi\cdot\Psi.
      \end{aligned}
    \end{equation}
Upon carefully combining these disparate calculated segments, one remarkably finds that every single isolated $\eta$-dependent term algebraically cancels out in exact totality.

\end{proof}

{We also spell out the curl transport identities for the torsion forms.}

\begin{proposition}
    We rigorously establish the following pair of coupled transport equations:
    \begin{equation}
        \begin{aligned}
            &\Omega\nabla_3\cl\etab+\Omega\tr\chib\cl\etab=\cl\left(\Omega\betab^r\right)+\cl\left(\Omega\chib\cdot\eta\right)-\cl(\Omega\Rc)_{3\cdot},\\
            &\Omega\nabla_4\cl\eta+\Omega\tr\chi\cl\eta=-\cl\left(\Omega\beta^r\right)+\cl\left(\Omega\chi\cdot\etab\right)-\cl(\Omega\Rc)_{4\cdot}.
        \end{aligned}
    \end{equation}
\end{proposition}
\begin{proof}
    We provide the detailed proof for the first equation; the rigorous derivation for the second one proceeds in a strictly similar symmetric fashion. 
    Starting fundamentally from the structural equation
    \[\Omega\nabla_3\etab_A=-\Omega\chib_A^B(\etab-\eta)_B+\Omega\betab^r_A-\Omega\Rc_{3A},\]
    we systematically apply the specific commutation formula \eqref{Comm_Formula_D3,nab} and subsequently obtain
    \begin{equation*}
        \begin{aligned}
            \Omega\nabla_3\cl\etab=&\cl\left(\Omega\nabla_3\etab\right)+\seps^{AB}\left(-\nabla_B(\Omega\chib)_A^C+\nabla^C(\Omega\chib)_{AB}\right)\etab_C-\seps^{AB}\Omega\chib_A^C\nabla_C\etab_B\\
            =&\cl(\Omega\betab^r)-\cl(\Omega\Rc)_{3\cdot}+\cl(\Omega\chib\cdot\eta)-\cl(\Omega\chib\cdot\etab)+\cl(\Omega\chib)\cdot\etab-\seps^{AB} \Omega\chib_A^C\nabla_C\etab_B\\
            =&\cl(\Omega\betab^r)-\cl(\Omega\Rc)_{3\cdot}+\cl(\Omega\chib\cdot\eta)-\Omega\tr\chib\cl\etab-\seps^{AB}\left(\nabla_A\etab_C\Omega\chibh_B^C+\nabla_C\etab_B\Omega\chibh_A^C\right).
        \end{aligned}
    \end{equation*}
    Through straightforward tensorial algebra, one can easily verify that the explicitly anti-symmetrized contraction $\seps^{AB}\left(\nabla_A\etab_C\Omega\chibh_B^C+\nabla_C\etab_B\Omega\chibh_A^C\right)$ vanishes identically to zero.
\end{proof}

{Finally, we derive the mass-aspect transport equations used in the renormalized Bianchi estimates.}

\begin{proposition}
     For the specific combined geometric quantities $\dv\etab-K$ and $\dv\eta-K$, we rigorously derive the governing evolution equations:
    \begin{equation}
        \begin{aligned}
            &\Omega\nabla_3\left(\dv \etab-K\right)+\Omega\tr\chib\left(\dv \etab-K\right)=-2\dv(\Omega\chibh\cdot\etab)-\dv(\Omega\Rc_3)+2\dv(\eta\cdot\Omega\chibh),\\
            &\Omega\nabla_4\left(\dv \eta-K\right)+\Omega\tr\chi\left(\dv \eta-K\right)=-2\dv(\Omega\chi\cdot\eta)-\dv(\Omega\Rc_4)+2\dv(\etab\cdot\Omega\chih).
        \end{aligned}
    \end{equation}

\end{proposition}
\begin{proof}
    We explicitly compute the differential term $\nabla_3\dv\etab$; the symmetric calculation for $\nabla_3\dv\eta$ follows an identical algebraic pattern. Submitting these intermediate results into the coupled structure equations \eqref{Equations_K} then immediately generates the desired final result.
    \begin{equation}
        \begin{aligned}
            \Omega\nabla_3\dv\etab=&g^{AB}\nabla_A\left(-\Omega\chib_B^C\etab_C+\Omega\chib_B^C\eta_C+\Omega\chib^r_B-\Omega\Rc_{3B}\right)\\
            &+g^{AB}\left(-\nabla_B(\Omega\chib)_A^C+\nabla^C(\Omega\chib)_{AB}\right)\etab_C-g^{AB}\Omega\chib_A^C\nabla_C\etab_B\\
            =&-\dv(\Omega\chibh\cdot\etab)-\frac{1}{2}\nabla(\Omega\tr\chib)\cdot\etab-\frac{1}{2}\Omega\tr\chib\dv\etab+\dv(\Omega\chib\cdot\eta)+\dv(\Omega\betab^r)-\dv(\Omega\Rc_3)\\
            &+\nabla(\Omega\tr\chib)\cdot\etab-\dv(\Omega\chibh)\cdot\etab-\frac{1}{2}\nabla(\Omega\tr\chib)\cdot\etab-\frac{1}{2}\Omega\tr\chib\cdot\dv\etab-\Omega\chibh\cdot\nabla\etab\\
            =&-2\dv(\Omega\chibh\cdot\etab)+\dv(\Omega\chib\cdot\eta)+\dv(\Omega\betab^r)-\dv(\Omega\Rc_3).
        \end{aligned}
    \end{equation}
\end{proof}
\section{Christodoulou's naked singularity}\label{Appendix_Chr94}
In this appendix, we systematically study the geometric properties of Christodoulou's fundamental naked singularity solution presented in \cite{Christodoulou1994}. 
\subsection{Near incoming null cone} Exploiting the inherent $\k$-self-similarity, the spacetime metric can be explicitly written as
\begin{equation}
    g(u,v,\theta)=-\left(\frac{-u}{v}\right)^\k\check{\Omega}^2\left(\frac{v}{-u}\right)\left(du\otimes dv+dv\otimes du\right)+(-u)^2\check{g}_{AB}\left(\frac{v}{-u}\right)d\theta^A\otimes d\theta^B,
\end{equation}
with scalar field function 
\begin{equation}
    \phi(u,v)=\sqrt{2\k}\log(-u)+\check{\phi}\left(\frac{v}{-u}\right).
\end{equation}
where the boundary conditions $\check{\Omega}^2(0)=1$ and $\check{g}_{AB}(0)={g^{\SS}}_{AB}$ necessarily hold according to the detailed analysis in \cite{Christodoulou1994}.
Invoking Proposition \ref{Self similar equation near v equal 0}, we rigorously derive the following reduced system of equations:
\begin{lemma}
    Along $u=-1$, the following algebraic relations hold:
    \begin{equation}
        \begin{aligned}
            \Omega\tr\chib=-2+v\Omega^2\cdot\Omega^{-1}\tr\chi,\quad \Omega\omegab=v\Omega\omega,\quad \Omega D_3\phi=-\sqrt{2\k}+v\Omega^2\cdot\Omega^{-1}D_4\phi,
        \end{aligned}
    \end{equation}
    and equations below are satisfied
    \begin{equation}
        \left\{
        \begin{aligned}
            v\partial_v(\Omega^{-1}\tr\chi)=&(1+\k)\Omega^{-1}\tr\chi-2K-v\Omega^2\left(\Omega^{-1}\tr\chi\right)^2,\\
            v\partial_v(v\Omega\omega)=& v\Omega^2\left(-\frac{1}{2}K-\frac{1}{8}\tr\chi\tr\chib+\frac{1}{4}D_3\phi D_4\phi\right),\\
            v\partial_v(v^\k\Omega^2)=&v^\k\Omega^2\left(\k-4v\Omega\omega\right),\\
            v\partial_v(\Omega^{-1}D_4\phi)=& (4v\Omega\omega-v\Omega^2\cdot \Omega^{-1}\tr\chi)\Omega^{-1}D_4\phi+\frac{\sqrt{2\k}}{2}\Omega^{-1}\tr\chi,\\
            v\partial_v K=&-v\Omega^2\cdot\Omega^{-1}\tr\chi K,
        \end{aligned}
        \right.
    \end{equation}
    with initial values 
    \begin{equation}
        K(0)=1,\  v\Omega\omega(0)=\k/4,\ v^\k\Omega^2(0)=1,\ \Omega D_3\phi(0)=-\sqrt{2\k} .
    \end{equation}
\end{lemma}
Since a well-defined $\nabla_4\tr\chi$ evolution equation exists for the generic ESE system, we automatically know that the derivative $\Omega^{-1}e_4(\Omega^{-1}\tr\chi)$ remains strictly bounded, which in turn implies the limit $(v\partial_v)\Omega^{-1}\tr\chi(0)=0$. Therefore, evaluating at the axis, we deduce that
\begin{equation*}
    \Omega^{-1}\tr\chi(0)=\frac{2K(0)}{1+\k}=\frac{2}{1+\k}.
\end{equation*}
Because the governing structure equation dictates that $\Omega^{-2}\partial_v\Omega^{-1}\tr\chi=-\frac{1}{2}(\Omega^{-1}\tr\chi)^2-(\Omega^{-1}D_4\phi)^2$, we can asymptotically expand $\Omega^{-1}\tr\chi=\frac{2}{1+\k}+v^{1-\k}O(1)$. By systematically plugging this rigorous expansion into the $v\partial_v(\Omega^{-1}D_4\phi)$ equation, the following relation holds
\begin{equation*}
    v\partial_v(\Omega^{-1}D_4\phi)=\k \Omega^{-1}D_4\phi +\frac{\sqrt{2\k}}{1+\k}+v^{1-\k}O(1).
\end{equation*}
Thus evaluating this relation at the origin conclusively yields $\Omega^{-1}D_4\phi(0)=-\frac{1}{1+\k}\sqrt{\frac{2}{\k}}$.

{The next identity links the scalar derivative to the normalized lapse derivative near the incoming cone.}

\begin{proposition}
    Near $v=0$, the following relation holds
    \begin{equation}
        \Omega^{-1}D_4\phi=-\frac{4}{\sqrt{2\k}}\frac{v\Omega\omega-\k/4}{v\Omega^2}.
    \end{equation}
\end{proposition}
\begin{proof}
    Starting from the fundamental geometric equation for $\nabla_3\tr\chib$, we restrict our analysis along the initial hypersurface $u=-1$ to obtain
    \begin{equation*}
        \begin{split}
            &\frac{-2}{(-u)^2}+\frac{v}{(-u)^2}\Omega\tr\chi+\frac{v}{-u}\partial_v\left(\frac{v}{-u}\Omega\tr\chi\right)+\frac{2}{(-u)^2}-2\frac{v}{(-u)^2}\Omega\tr\chi+\frac{1}{2}\left(\frac{v}{-u}\Omega\tr\chi\right)^2\\
            =&-4\frac{v}{-u}\Omega\omega(\frac{-2}{-u}+\frac{v}{-u}\Omega\tr\chi)-\frac{2\k}{(-u)^2} +2\frac{\sqrt{2\k}v}{(-u)^2} \Omega D_4\phi-\left(\frac{v}{-u}\Omega D_4\phi\right)^2.
        \end{split}
    \end{equation*}
After carefully simplifying this expression by substituting the corresponding exact relation from the $\nabla_4\tr\chi$ equation, we decisively have
    \begin{equation*}
        \begin{split}
            \Omega^{-1}D_4\phi = -\frac{1}{\sqrt{2\k}}\frac{4v\Omega\omega-\k}{v\Omega^2},
        \end{split}
    \end{equation*}
and therefore,
\begin{equation*}
    (\Omega^{-1}D_4)^2\phi=\frac{1}{\sqrt{2\k}}\frac{4v\Omega\omega-\k}{(v\Omega^2)^2}\cdot\frac{1-4v\Omega\omega}{-u}-\frac{1}{\sqrt{2\k}}\frac{4\Omega^{-1}\nabla_4\Omega\omegab}{\frac{v}{-u}\Omega^2}=-(-u)^{-2}\left(\frac{-u}{v}\right)^{1-\k}\left(\frac{2(1-\k)}{(1+\k)\sqrt{2\k}}+o(1)\right).
\end{equation*}
\end{proof}

If we formally consider the rescaled coordinate change defined by $(u,\hat{v})=(u,\frac{1}{1-\k}v^{1-\k})$, we readily compute the transformed operators:
\begin{equation*}
    \partial_v=v^{-\k}\partial_{\hat{v}},\  v\partial_v=(1-\k)\hat{v}\partial_{\hat{v}},\ \hat{\Omega}^2= v^{\k}\Omega^2,\ \Omega^{-1}e_4=\hat{\Omega}^{-1}\hat{e}_4,
\end{equation*}
\begin{equation*}
    \hat{\Omega}^{-1}\hat{\omega}=-\frac{1}{4}\hat{\Omega}^{-1}\hat{e}_4\log\hat{\Omega}^2=-\frac{1}{4}\Omega^{-2}\partial_v\log(v^\k\Omega^2)=\Omega^{-1}\omega-\frac{\k}{4v\Omega^2}=\frac{v\Omega\omega-\k/4}{v\Omega^2}.
\end{equation*}
Therefore we find that 
\begin{equation*}
    \hat\Omega^{-1}\hat{e}_4\phi = -\frac{4}{\sqrt{2\k}}\hat{\Omega}^{-1}\hat\omega,
\end{equation*}
\begin{equation*}
    \hat{\Omega}\hat{e}_3=\Omega e_3,\ \hat{\Omega}\hat\omegab=\Omega\omegab.
\end{equation*}

From the proposition, it follows that 
\begin{equation}
    \begin{aligned}
        \Omega^{-1}e_4(\Omega e_3)\phi =& -\frac{4}{\sqrt{2\k}}\Omega^{-1}\nabla_4\Omega\omegab.
    \end{aligned}
\end{equation}
On the other hand, we also have 
\begin{equation}
    \begin{aligned}
        &-\frac{1}{2}\left(\Omega^{-1}\tr\chi\Omega D_3\phi+\Omega\tr\chib\Omega^{-1}D_4\phi\right)\\
        =& -\frac{1}{2}\left(\Omega^{-1}\tr\chi\left(-\sqrt{2\k}-\frac{4}{\sqrt{2\k}}v\Omega^2 \hat\Omega^{-1}\hat\omega\right)-(-2+v\Omega\tr\chi)\frac{4}{\sqrt{2\k}}\hat\Omega^{-1}\hat\omega\right)\\
        =& \frac{\sqrt{2\k}}{2}\Omega^{-1}\tr\chi +\frac{4}{\sqrt{2\k}}\Omega^{-1}\tr\chi \left(\Omega\omegab-\k/4\right)-\frac{4}{\sqrt{2\k}}\hat\Omega^{-1}\hat\omega\\
        =&-\frac{4}{\sqrt{2\k}}\hat\Omega^{-1}\hat\omega+\frac{4}{\sqrt{2\k}}\Omega^{-1}\tr\chi \Omega\omegab.
    \end{aligned}
\end{equation}
\begin{equation}
    \begin{aligned}
        \Omega^{-1}\nabla_4\Omega\omegab+\Omega^{-1}\tr\chi\Omega\omegab=\hat\Omega^{-1}\hat\omega.
    \end{aligned}
\end{equation}
Since $\Omega^{-1}e_4\Omega\omegab=\hat\Omega^{-1}\hat{e}_4\hat\Omega\hat\omegab=\partial_u\left(\hat\Omega^{-1}\hat\omega\right)-4\Omega\omegab\hat\Omega^{-1}\hat\omega=\hat\Omega^{-1}\hat\omega+v\partial_v\hat\Omega^{-1}\hat\omega-4\Omega\omegab\hat\Omega^{-1}\hat\omega$, it follows that 
\begin{equation}\label{Sec App B v partial v Omega -1 omega}
    \begin{split}
        v\partial_v\hat\Omega^{-1}\hat\omega=&(4\hat\Omega^{-1}\hat\omega-\Omega^{-1}\tr\chi)\left(\frac{\k}{4}+v\Omega^2\cdot\hat\Omega^{-1}\hat\omega\right)\\
        =&\k\hat\Omega^{-1}\hat\omega-\Omega^{-1}\tr\chi\left(\frac{\k}{4}+v\Omega^2\cdot\hat\Omega^{-1}\hat\omega\right)+4v\Omega^2(\hat\Omega^{-1}\hat\omega)^2\\
        =& \k \hat\Omega^{-1}\hat\omega+F,
    \end{split}
\end{equation}
where 
$F=4v\Omega^2(\hat\Omega^{-1}\hat\omega)^2-\left(\frac{\k}{4}+v\Omega^2\cdot\hat\Omega^{-1}\hat\omega\right)\Omega^{-1}\tr\chi$.
We formally gather the complete algebraic equation system governing the restricted $\k$-self-similar and strictly spherically symmetric ESE model:
\begin{equation}\label{equations system for the self-similar and spherical symmetric ESE}
    \left\{
    \begin{aligned}
        & v\partial_v(\hat\Omega^{-1}\hat\omega)=(4\hat\Omega^{-1}\hat\omega-\Omega^{-1}\tr\chi)\left(\frac{\k}{4}+v\Omega^2\cdot\hat\Omega^{-1}\hat\omega\right)\\
        &\partial_v\Omega^{-1}\tr\chi=\Omega^2\left(-\frac{1}{2}(\Omega^{-1}\tr\chi)^2-\frac{8}{\k}(\hat\Omega^{-1}\hat\omega)^2\right),\\
        &\partial_v\log v^\k\Omega^2=-4\Omega^2\cdot\hat\Omega^{-1}\hat\omega^{-1}.
    \end{aligned}
    \right.
\end{equation}
We can algebraically transform the differential equation \eqref{Sec App B v partial v Omega -1 omega} into the exact conservation law format $\partial_v\left(v^{-\k}\left[\hat\Omega^{-1}\hat\omega\right]_0^v+\int_0^v v^{-1-\k}(F+\k\hat\Omega^{-1}\hat\omega(0))\right)=0$, which immediately implies the uniform bound
\begin{equation}\label{increasing_rate_omega_c}
    \left[\hat\Omega^{-1}\hat\omega\right]_0^v\lesssim v^{\k}.
\end{equation}
Let us purposefully define the auxiliary variables $A(v)=(\Omega^{-1}\tr\chi)^{-1}(-1,v)$ and $B(v)=\frac{\hat\Omega^{-1}\hat\omega}{\Omega^{-1}\tr\chi}$. We then deduce directly from the coupled system \eqref{equations system for the self-similar and spherical symmetric ESE} that
\begin{equation}
    \left\{
    \begin{aligned}
        &\partial_v A=\Omega^2\left(\frac{1}{2}+\frac{8}{\k}B^2\right),\\
        &\partial_v \left(v^{-\k}(B-B(0))\right)=v^{-1-\k}\cdot v\Omega^2\left(\frac{8}{\k}B^2+4B-\frac{1}{2}\right)\frac{B}{A}.
    \end{aligned}
    \right.
\end{equation}
Since $\lim_{v\rightarrow\infty}v^{-\k}\Omega^{2}(-1,v)=1$ and $\lim_{v\rightarrow\infty}v\Omega\omega(-1,v)=-\k/4$, $\lim_{v\rightarrow\infty}v\Omega\tr\chi=2$, the following relation holds 
\begin{equation*}
    B(0)=\frac{1}{4},\ A(0)=\frac{1+\k}{2},\ B(\infty)=-\frac{\k}{4},\ \lim_{v\rightarrow\infty}v^{-1-\k}A(v)=\frac{1}{2}.
\end{equation*}
Analyzing this dynamical system, we rigorously observe that the variable $A$ is strictly increasing and bounded positively, while the algebraic quadratic expression $\frac{8}{\k}B^2+4B-\frac{1}{2}$ near the boundary value $B=\frac{1}{4}$ is also strictly positive. Consequently, the fraction $\frac{B-B(0)}{v^\k}$ must be strictly increasing in an immediate neighborhood of $v=0$. There necessarily exists a critical point $v_0>0$, identified as the first positive value for which $B(v_0)=B(0)$, so that
\begin{equation*}
    \lim_{v\rightarrow 0}\frac{B-B(0)}{v^\k}=-\int_{0}^{v_0}v^{-1-\k}\cdot v\Omega^2\left(\frac{8}{\k}B^2+4B-\frac{1}{2}\right)\frac{B}{A}dv<0.
\end{equation*}
Therefore, evaluated along the specific null hypersurface $H_{-1}$, we robustly possess the asymptotic relationship
$\left[\hat\Omega^{-1}\hat\omega\right]_0^v\sim v^\k$.

{This asymptotic immediately gives the singular rate of the iterated outgoing null derivative.}

\begin{corollary}\label{increasing_rate_D4phi_c}
    For any $\delta>0$, the following asymptotic estimate holds near $\frac{v}{-u}=0$:
    \begin{equation}
        \left|(\Omega^{-1}D_4)^2\phi\right|\sim \left(\frac{-u}{v}\right)^{1-2\k}(-u)^{-2}.
    \end{equation}
\end{corollary}
\begin{proof}
    Directly recalling the established equivalence $\Omega^{-1}D_4\phi\sim \hat\Omega^{-1}\hat\omega$, we subsequently compute
    \begin{equation}
        \left|(\Omega^{-1}D_4)^2\phi\right|\sim \left(\frac{(-u)^{-\k}}{v^{1-\k}}\right)\left|\left[\hat\Omega^{-1}\hat\omega\right]_0^v\right|\sim \left(\frac{-u}{v}\right)^{1-2\k}(-u)^{-2}.
    \end{equation}
\end{proof}
\subsection{Near outgoing null cone}
We systematically proceed to analyze the asymptotic spacetime geometry in the complementary neighborhood near $\frac{-u}{v}=0$. Exploiting the fundamental $\k$-similarity and spherical symmetry, we readily deduce that
\begin{equation}
    g(u,v,\theta)=-\left(\frac{v}{-u}\right)^\k\check{\Omega}^2(-u/v)\left(du\otimes dv+dv\otimes du\right)+v^2\check{g}_{AB}(-u/v)d\theta^A\otimes d\theta^B,
\end{equation}
with scalar function 
\begin{equation}
    \phi(u,v)=\sqrt{2\k}\log v+\check{\phi}(-u/v),
\end{equation}
where $\check{\Omega}^2(0)=1$ and $\check{g}_{AB}(0)={g^{\SS}}_{AB}$ according to the analysis in \cite{Christodoulou1994}. 
From the elementary differential identity $\partial_v\left(v^kf(-u/v)\right)=\frac{k}{v}\cdot v^kf(-u/v)+\frac{-u}{v}\partial_u\left(v^kf(-u/v)\right)$, we can explicitly calculate
\begin{equation}
    \Omega\tr\chi=\frac{2}{v}+\frac{-u}{v}\Omega^2\cdot\Omega^{-1}\tr\chib,\ \Omega\omega=\frac{(-u)\Omega\omegab}{v},\ \Omega D_4\phi = \sqrt{2\k}\frac{1}{v}+\frac{-u}{v}\Omega^2\cdot \Omega^{-1}D_3\phi,
\end{equation}
and along $u=0$, $\Omega\tr\chi=2v^{-1}$, $\Omega\omega=-\k v^{-1}/4$, $\Omega D_4\phi=\sqrt{2\k}v^{-1}$.

{We record the corresponding reduced ODE system on the outgoing side.}

\begin{lemma}
    The following ODE system governs the self-similar geometry along $v=1$:
    \begin{equation}
        \left\{
        \begin{aligned}
            \Omega^{-1}e_3(\Omega^{-1}\tr\chib)=&-\frac{1}{2}(\Omega^{-1}\tr\chib)^2-(\Omega^{-1}D_3\phi)^2,\\
            (-u)\partial_u\Omega^{-1}\tr\chib=&(4\Omega\omega-1-(-u)\Omega^2\Omega^{-1}\tr\chib)\Omega^{-1}\tr\chib-2K,\\
            (-u)\partial_u\Omega^{-1}D_3\phi=&\left(4\Omega\omega-(-u)\Omega^2\cdot \Omega^{-1}\tr\chib\right)\Omega^{-1}D_3\phi-\frac{\sqrt{2\k}}{2}\Omega^{-1}\tr\chib,\\
            \partial_u K=&-\Omega^2\Omega^{-1}\tr\chib K,\\
            (-u)\partial_u \log\Omega^2=&-4\Omega\omega.
        \end{aligned}
        \right.
    \end{equation}
    Moreover, along $u=0$, we deduce that 
    \begin{equation}
        K(0,v)=\frac{1}{v^2} ,\ \Omega^{-1}\tr\chib(0,v)=-\frac{2}{1+\k},\ \Omega^{-1}D_3\phi(0,v)=\sqrt{\frac{2}{\k}}\frac{1}{1+\k}.
    \end{equation}
\end{lemma}
By systematically plugging the derived algebraic formulae directly into the structural equation
\begin{equation*}
    \Omega \nabla_4\Omega\tr\chi+\frac{1}{2}(\Omega\tr\chi)^2+4\Omega\omega\Omega\tr\chi+(\Omega D_4\phi)^2=0,
\end{equation*}
we consequently obtain, through an application of the corresponding $\Omega^{-1}\nabla_3\Omega^{-1}\tr\chib$ equation, the exact relation
\begin{equation*}
    \Omega^{-1}D_3\phi=-\frac{4}{\sqrt{2\k}}\frac{(-u)\Omega\omegab+v\k/4}{(-u)\Omega^2}.
\end{equation*}
Consider the analogous rescaled coordinate transformation $(\hat{u},v)=\left(-\frac{1}{1-\k}(-u)^{1-\k},v\right)$; it then identically holds that
\begin{equation*}
    \partial_{u}=(-u)^{-\k}\partial_{\hat{u}}, \ u\partial_u=(1-\k)\hat{u}\partial_{\hat{u}},\ \Omega^2=(-u)^{-\k}\hat{\Omega}^2,\ \Omega^{-1}e_3=\hat\Omega^{-1}\hat{e}_3,\ \Omega e_4=\hat\Omega \hat{e}_4.
\end{equation*}
\begin{equation*}
    \hat\Omega^{-1}\hat{\omegab}=-\frac{1}{4}\hat\Omega^{-1}\hat{e}_3\left(\log \hat\Omega^2\right)=-\frac{1}{4}\Omega^{-1}e_3\left(\log\Omega^2+\k\log(-u)\right)=\frac{(-u)\Omega\omegab+\k/4}{(-u)\Omega^2}.
\end{equation*}
Therefore we can write 
\begin{equation*}
    \Omega^{-1}D_3\phi=-\frac{4}{\sqrt{2\k}}\hat\Omega^{-1}\hat\omegab,
\end{equation*}
\begin{equation*}
    \Omega D_4\phi=\frac{\sqrt{2\k}}{v}-\frac{4}{\sqrt{2\k}}\left(\Omega\omega+\frac{\k}{4v}\right),
\end{equation*}
Along $v=1$, we obtain  
\begin{equation*}
    \begin{aligned}
        \Omega^{-1}e_3(\Omega e_4)\phi=-\frac{4}{\sqrt{2\k}}\Omega^{-1}\nabla_3\Omega\omega,
    \end{aligned}
\end{equation*}
and on the other hand 
\begin{equation*}
    \begin{aligned}
        \Omega^{-1}e_3(\Omega e_4)\phi=&-\frac{1}{2}\Omega^{-1}\tr\chib\Omega D_4\phi-\frac{1}{2}\Omega\tr\chi\Omega^{-1}D_3\phi\\
        =&-\frac{1}{2}\Omega^{-1}\tr\chib\left(\sqrt{2\k}-\frac{4}{\sqrt{2\k}}(-u)\Omega^2\hat\Omega^{-1}\hat\omegab\right) -\frac{1}{2}\left(2+(-u)\Omega^2\cdot \Omega^{-1}\tr\chib\right)\left(-\frac{4}{\sqrt{2\k}}\hat\Omega^{-1}\hat\omegab\right)\\
        =&-\frac{\sqrt{2\k}}{2} \Omega^{-1}\tr\chib+\frac{4}{\sqrt{2\k}}(-u)\Omega^2\Omega^{-1}\tr\chib\hat\Omega^{-1}\hat\omegab+\frac{4}{\sqrt{2\k}}\hat\Omega^{-1}\hat\omegab,
    \end{aligned}
\end{equation*}
so by algebraically comparing these expressions, we finally obtain
\begin{equation*}
    \Omega^{-1}\nabla_3\Omega\omega+\Omega^{-1}\tr\chib\Omega\omega+\hat\Omega^{-1}\hat\omegab=0.
\end{equation*}
\begin{equation*}
    \begin{aligned}
        (-u)\partial_u\hat\Omega^{-1}\hat\omegab=& \partial_v\hat\Omega^{-1}\hat\omegab+\hat\Omega^{-1}\hat\omegab=4\Omega\omega \hat\Omega^{-1}\hat\omegab+\Omega^{-1}\nabla_3\Omega\omega+\hat\Omega^{-1}\hat\omegab\\
        =&(4\hat\Omega^{-1}\hat\omegab-\Omega^{-1}\tr\chib)\Omega\omega\\
        =&-\k\hat\Omega^{-1}\hat\omegab+(\Omega^{-1}\tr\chib)\left(\frac{\k}{4}-(-u)\Omega^2\hat\Omega^{-1}\hat\omegab\right)-4(-u)\Omega^2\left(\hat\Omega^{-1}\hat\omega\right)^2.
    \end{aligned}
\end{equation*}
The overarching dynamical equations for the $\k$-self-similar and spherically symmetric ESE system securely reduce to the following compact form:
\begin{equation}
    \left\{\begin{aligned}
        (-u)\partial_u\hat\Omega^{-1}\hat\omegab=&-\k\hat\Omega^{-1}\hat\omegab+(\Omega^{-1}\tr\chib)\left(\frac{\k}{4}-(-u)\Omega^2\hat\Omega^{-1}\hat\omegab\right)-4(-u)\Omega^2\left(\hat\Omega^{-1}\hat\omega\right)^2,\\
        \partial_u \Omega^{-1}\tr\chib=& -\Omega^2\left(\frac{1}{2}(\Omega^{-1}\tr\chib)^2+\frac{8}{\k}(\hat\Omega^{-1}\hat\omegab)^2\right),\\
        \partial_u\log\left((-u)^\k\Omega^2\right)=&-4\Omega^2\cdot \hat\Omega^{-1}\hat\omegab.
    \end{aligned}\right.
\end{equation}
Employing an analytical argument rigorously parallel to that utilized near $\frac{v}{-u}=0$, we formally define
$A(u)=\frac{1}{\Omega^{-1}\tr\chib}(u,1)$ and $B(u)=\frac{\hat\Omega^{-1}\hat\omegab(u,1)}{\Omega^{-1}\tr\chib(u,1)}$, the resulting equations are
\begin{equation}
    \begin{aligned}
        \partial_uA =&\Omega^2\left(\frac{1}{2}+\frac{8}{\k}B^2\right),\\
        (-u)\partial_u B=&-\k (B-B(0))-(-u)\Omega^2\frac{B}{A}\left(\frac{8}{\k}B^2+4B-\frac{1}{2}\right).
    \end{aligned}
\end{equation}
The second differential equation can be identically reformulated into the exact conservation-like format
\begin{equation*}
    \partial_u\left((-u)^{-\k}(B-B(0))\right)=-(-u)^{-\k-1}(-u)\Omega^2\frac{B}{A}\left(\frac{8}{\k}B^2+4B-\frac{1}{2}\right).
\end{equation*}
Since the prescribed boundary values dictate $A(0)=-\frac{1+\k}{2},B(0)=\frac{1}{4},$ and $B(-\infty)=-\frac{\k}{4}$, in a localized neighborhood of $u=0$ we firmly establish that the scaled quantity $(-u)^{-\k}(B-B(0))$ is strictly decreasing with respect to $u$. There necessarily exists a largest negative value $u_0<0$ such that $B(u_0)=B(0)$, and furthermore the quantities $-A$ and $\frac{8}{\k}B^2+4B-\frac{1}{2}$ remain strictly positive throughout the interval $(u_0,0)$, so
\begin{equation*}
    \lim_{u\rightarrow 0}(-u)^{-\k}(B-B(0))=\int_{u_0}^0(-u)^{-\k-1}(-u)\Omega^2\frac{B}{-A}\left(\frac{8}{\k}B^2+4B-\frac{1}{2}\right)>0,
\end{equation*}
which fundamentally implies the desired boundary behavior
\begin{equation}
    \left|\hat\Omega^{-1}\hat\omega(u,1)-\hat\Omega^{-1}\hat\omega(0,1)\right|\sim (-u)^\k.
\end{equation}

{The outgoing-side analogue yields the corresponding blow-up rate for the incoming iterated derivative.}

\begin{corollary}
    Near $\frac{-u}{v}=0$, the following asymptotic estimate holds for the iterated null derivative:
    \begin{equation}
        \left|(\Omega^{-1}D_3)^2\phi\right|\sim (-u)^{2\k-1}v^{-2}.
    \end{equation}
\end{corollary}

\bibliographystyle{plain} 
\bibliography{ref}

\end{document}